\documentclass{article}
\usepackage{amsmath,amsfonts,amssymb,amsthm}
\usepackage{mathtools}
\usepackage{hyperref}
\usepackage{subcaption}
\usepackage{float}
\usepackage{multicol}
\usepackage{textcomp}
\usepackage[x11names]{xcolor}
\usepackage{balance}
\usepackage{xspace}
\RequirePackage{xparse}
\usepackage[numbers]{natbib}
\usepackage{pifont}
\usepackage{makecell}

\usepackage{thm-restate}

\usepackage{paper}
\usepackage{cleveref}
\usepackage[inline]{enumitem}
\usepackage{sbnf}
\renewcommand{\bnfcmt}[1]{\emph{#1}}
\renewcommand{\bnfmid}{\mathrel{\mbox{\Large{$\mid$}}}}
\usepackage[margin=1.2in]{geometry}
\usepackage{tikz}
\usepackage{authblk}

\usetikzlibrary{decorations.pathmorphing}
\usetikzlibrary{positioning}
\usetikzlibrary{calc}
\usetikzlibrary{fit}

\newcommand{\cmark}{\ding{51}}%
\newcommand{\xmark}{\ding{55}}%
\newcommand{\diff}[1]{#1}
\renewcommand{\lbl}{\ell}

\newcommand{\arxiv}{}

\newcommand{\DD}[1]{\ifdefined\commentaire{{{\color{orange!60!black!55} \texttt{[#1]}}}}\fi}
\newcommand{\TR}[1]{\ifdefined\commentaire{{{\color{blue!60!black!55} \texttt{[#1]}}}}\fi}

\author{Davide Davoli}
\author{Martin Avanzini}
\author{Tamara Rezk}
\affil{Inria, Université Côte d'Azur}

\newcommandx{\wfp}{\mathsf{wf}}
\newcommandx{\wf}[3][1=\system,2=\lay]{#1\vdash \wfp_{#2}(#3)}
\newcommandx{\fencerel}{\precsim}

\newcommand{\HERE}{{{~}\par\color{red}{HERE\hrule}}}
\title{Comprehensive Kernel Safety in the Spectre Era: Mitigations and Performance Evaluation (Extended Version)}

\begin{document}

\maketitle
\begin{abstract}
  The efficacy of address space layout randomization has been formally demonstrated in a shared-memory model by Abadi et al., contingent on specific assumptions about victim programs. However, modern operating systems, implementing layout randomization in the kernel, diverge from these assumptions and operate on a separate memory model with communication through system calls. 
In this work, we relax Abadi et al.'s language assumptions while demonstrating that layout randomization offers a comparable safety guarantee in a system with memory separation. 
However, in practice, speculative execution and side-channels are recognized threats to layout randomization. 
We show that kernel safety cannot be restored for attackers capable of using side-channels and speculative execution, and introduce enforcement mechanisms that can guarantee speculative kernel safety for safe system calls in the Spectre era. We show that kernel safety cannot be restored for attackers capable of using side-channels and speculative execution, and introduce enforcement mechanisms that can guarantee speculative kernel safety for safe system calls in the Spectre era. We implement three suitable mechanisms and we evaluate their performance overhead on the Linux kernel.
\end{abstract}

\section{Introduction}
\label{sec:introduction}
Memory safety violations on kernel memory can result in serious ramifications for security, such as e.g. arbitrary code execution, privilege escalation, or information leakage. 
In order to mitigate safety violations, operating systems such as Linux employ 
address space layout randomization~\cite{LWNKASLR,MacOSXASLR,AndroidASLR,iOSASLR,OpenBSD6.3,FreeBSDASLR}. This protection measure can prevent attacks that depend on knowledge of specific data or procedure location, as it introduces randomization of  these addresses.

On the one hand, the efficacy of layout randomization has been formally demonstrated in Abadi et al.'s line of work~\cite{Abadi, Abadi2, Abadi3}, as a protective measure within a {\it shared-memory model} between the attacker and the victim. These results, however, are contingent on specific assumptions regarding victim programs, notably the absence of pointer arithmetic,  introspection, or indirect jumps. These precise constraints shaped a controlled environment where memory safety could be enforced effectively via layout randomization.
However, operating systems employing layout randomization on kernel (a.k.a. KASLR in Linux e.g.~\cite{LWNKASLR}) diverge from these assumptions. Notably, they operate on a  separate memory model, wherein, kernel code---acting as the victim---resides on
kernel memory, while user code---acting as the potential attacker---resides in user space. The interaction between the two occurs through a limited set of procedures provided via system calls~\cite{Tanenbaum}.
In the operating system's realm, system calls may be written in C and assembly code, further deviating from 
  the restricted conditions outlined by Abadi et al.  This introduces a distinction not only in the expressiveness of victim code considered but also in the underlying memory model.
  
Hence, our first research question  emerges: can we relax the language assumptions proposed by Abadi et al.~\cite{Abadi,Abadi2,Abadi3} while concurrently demonstrating that layout randomization offers a comparable safety guarantee in a system  with memory 
separation? We affirmatively respond to this question by showcasing that layout randomization probabilistically ensures kernel safety within a classic attacker model, where users of an operating system execute without privileges and victims can feature pointer arithmetic, introspection, and indirect jumps.

On the other hand, in the current state-of-the-art of security, often referred to as the Spectre era, speculative execution and side-channels are well known to be effective vectors for compromising layout randomization~\cite{BlindSide,Meltdown,cacheKASLR,EntryBleed,TagBleed,EchoLoad}.\TR{add example Blindside here?? } 
Indeed, our first result neglects the impact of speculative execution and side-channels. Recognizing this limitation, our second research question arises:  can we restore a similar safety result in the Spectre era?

In this regard, we formally acknowledge that by relying solely on layout randomization it is not possible to restore kernel safety. We then introduce a new condition, called \emph{side-channel layout non-interference} akin to speculative constant-time~\cite{CTFundations}, which intuitively asserts that victims should not unintentionally leak information on the kernel's layout through side-channels. Our research formally demonstrates that under this assumption, the system is safe, and perhaps surprisingly, without the necessity of layout randomization.
Later, we show that side-channel layout non-interference is not a necessary requirement, and this motivates us to study how safety can be enforced without requiring that property.

\diff{
Our third contribution is to show that kernels can be protected even without requiring {\it side-channel layout non-interference}. We do so by relating safety in the classic execution model to the speculative one.
%
%
Specifically, we show that a kernel that is safe against
classic attackers can be protected against speculative attackers
by applying specific program transformations. However, the initial
safety requirement that we impose on the kernel cannot be granted solely
by layout randomization. Finally, we show the soundness of three
such transformations, and we implement them in order to evaluate their
overhead in terms of performance.  
}

This marks the first formal step toward strengthening kernel safety in the presence of speculative and side-channel vulnerabilities,
and the surpassing of layout randomization as a system level protection mechanism.
In summary, our contributions are: 
\begin{itemize}
\item  We formally demonstrate the effectiveness of layout randomization to provide kernel safety for a classic operating system scenario, with system calls offered as 
interfaces to attackers and different privilege execution modes, as well as kernel and user memory separation. 
\item We empower the attackers of our first scenario to execute side-channel attacks and control speculative execution. We demonstrate that kernel safety is not maintained under this more potent attacker model, and we subsequently present a sufficient condition to ensure kernel safety. 
\item We show that it is possible to enforce safety against speculative attackers on a system that enjoys weaker security guarantees by the application of specific program transformations.
  \diff{
  \item We implement these transformations, and we measure their overhead on the Linux kernel. The experimental evaluation shows that they impose low performance overheads on computationally heavy user-space tasks.
  }\end{itemize}

\diff{This paper is an extended version of the ACM CCS conference paper \cite{ConferencePaper}
  on kernels' safety in presence of speculative execution. The additional material in this
  work, compared to the conference version, is the following:
  \begin{itemize}
    \item We model indirect branch prediction, related to Spectre v2, significantly strengthening the attack model compared to the conference paper~(\Cref{sec:safety2}). Incorporating indirect branch prediction required extending the attacker's language with controls over branch targets, and adapting our speculative semantics accordingly. To contrast unwanted indirect branch predictions, we modeled a class of jump instructions whose targets cannot be speculative, and that that are crucial for the effectiveness of the mitigations in \Cref{sec:sksenforcement}. This addition does not affect sections where speculative execution is not considered, notably \Cref{sec:language,sec:safety1}. 
  \item We review state-of-art mitigations against speculative execution attacks, showing that none of them can easily be adopted to enforce kernel's safety in presence of speculative execution~(\Cref{sec:discussion}).
  \item We extend the transformation given in the conference version of this paper~\cite[Section 7]{ConferencePaper} to protect against indirect branch speculation.
  \item We define two additional transformations. The first one is an optimized version of the transformation given in the conference version of this paper~\cite[Section 7]{ConferencePaper}, and the other one blocks all forms of \emph{kernel-space} speculation taken into account in this work~(\Cref{sec:sksenforcement}). 
  \item We implement the transformations as LLVM passes, and we offer them as open source~\cite{repo}.
  \item We evaluate the performance overhead of these transformations on computationally heavy \emph{user-space} and \emph{kernel-space} tasks, and on I/O-bound \emph{user-space} workloads. 
  \end{itemize}
}

The paper is structured as follows: in \Cref{sec:motivations} we give
an overview of the contributions of this paper, motivated by a
concrete example.  In \Cref{sec:language}, we introduce our execution
model by giving its language and semantics; in \Cref{sec:threatmodel},
we establish threat models.  \Cref{sec:safety1} is devoted to showing
that layout randomization protects against attacks that do not rely on
speculative execution and side-channel observations.  In
\Cref{sec:safety2} we extend the model of \Cref{sec:language} with
side-channel info leaks and speculative execution, and we show that
layout randomization is not a viable protection mechanism in this
scenario.
\diff{ State-of-art mitigations against speculative attacks are
  reviewed in~\Cref{sec:discussion}.  In \Cref{sec:sksenforcement} we
  show that any system that is safe against classic attackers can be
  transformed into an equivalent system that is safe against speculative
  attackers, and we propose three suitable program transformations for
  this task.  In \Cref{sec:experiment}, we estimate the overhead of
  this transformation on real hardware.}  Finally, we discuss related
work in \Cref{sec:relatedwork}, and we conclude in
\Cref{sec:conclusion}.
\ifdefined\conference{Omitted proofs are available
in the extended version of
this paper~\cite{ExtendedJournal}.}\fi
\ifdefined\arxiv{Omitted proofs are in
\Cref{sec:appendix}.}\fi


\section{Overview}
\label{sec:motivations}
\begin{figure}[t]
\begin{code}[emph={buf,recv,send,valid},morekeywords={socket,size_t}]
int buf[K+1][H];

int recv(socket* s, size_t idx) {
  if (valid(s, idx)) return buf[*s][idx];
  return 0;
}

void send(socket* s, size_t idx, int msg) {
  if (valid(s, idx)) {
    buf[*s][idx] = msg;
    if (buf[K][0] != NULL)
      (*buf[K][0])(s, idx);
  }
}
\end{code}
  \vspace{-3mm}
  \caption{System Calls vulnerable to memory corruption}
  \label{fig:msgpassing}
\end{figure}

Each year, dozens of vulnerabilities are found in commodity operating systems' kernels, and the majority of them are memory corruption vulnerabilities~\cite{ThreatOverview}. A kernel suffers a memory corruption vulnerability when an unprivileged user-space program, acting as the attacker, can trigger it to read or write its memory in an \emph{unexpected} way, usually, by issuing a sequence of system calls with maliciously crafted arguments. In \Cref{fig:msgpassing}, we show a pair of system calls of a hypothetical operating system that are subject to this vulnerability.
\def\ex{\cc[emph={buf, recv, send, valid, buf, native_write_cr4, sc_leak, foo}]}
The \ex{recv} and \ex{send} system calls are meant to implement a simple message passing protocol.
The implementation supports up to \ex{K} sockets, each socket can buffer up to \ex{H} messages.
A user can send messages by invoking the system call \ex{send}, and read them with the system call \ex{recv}.
These system calls employ a shared buffer \ex!buf! that stores messages, together with a hook for a customizable callback pointed by \ex{buf[K][0]}.
If specified, this callback is executed after a message is sent. Such a callback may, for instance, inform the sender on whether the message was sent correctly.

These system calls are meant to interact only with the memory
containing the buffer, the code of the called procedures and with the
resources that these procedures access.  In the following, we will
refer to the set of memory resources that a system call may access
rightfully as the \emph{\tcaps} of that system call.  Depending on the
implementation of the procedure \ex{valid}, these system calls can
suffer from memory corruption vulnerabilities.  For instance, if the
\ex{valid} function does not perform any bound checks on the value of
\ex{idx}, these two system calls can be used to perform arbitrary read
and write operations by a malicious user-space program, acting as an
attacker against the victim kernel.  In particular, if the attacker
supplies an out-of-bounds value for \ex{idx} to the \ex{recv} system
call, the system call can be used to perform an unrestricted memory
read.  Similarly, the \ex{send} system call can be used to overwrite
any value of kernel memory and, in particular, to overwrite the
pointer to the callback that is stored within the buffer.  This means
that a malicious user-space program, acting as an attacker, can turn
this memory-vulnerability into a \emph{control flow integrity}
(CFI) violation~\cite{CFI}, as it can
deviate the control flow from its intended paths.

However, if the system that implements these system calls is
protected with layout randomization---like many commodity operating
systems
do~\cite{LWNKASLR,MacOSXASLR,AndroidASLR,iOSASLR,OpenBSD6.3,FreeBSDASLR}---the
exploitation of these vulnerabilities is not a straightforward
operation. In Linux, for instance, one of the viable ways to mount a
privilege-escalation attack is to disable the \emph{Supervisor Mode
  Execution Prevention} (SMEP) by running the \ex{native_write_cr4}
procedure. When this protection is disabled, the kernel is allowed to
run any \emph{payload} stored in user-space. The attacker can trigger
the system to execute the \emph{payload} by exploiting the
vulnerability of the \ex{send} system call twice: the first time to
run the \ex{native_write_cr4} procedure, and the second time to run
the \emph{payload}.  However, in order to do so, the attacker has to
infer the address of \ex{native_write_cr4}.  In the absence of
info-leaks, an attacker can only guess such address and, due to layout
randomization, the probability of success is low.
\Cref{sec:safety1} is devoted to showing that without
side-channel leaks (and speculative execution), if a system is
protected with layout randomization, the probability that an
unprivileged attacker leads the system to perform an unsafe memory
access is very low, provided the address space is sufficiently
large. Of course, the precise probability depends on the concrete
randomization scheme.  This result is compatible
with the large number of kernel attacks that break Linux's kernel
layout randomization, e.g.\ by means of heap
overflows~\cite{HeapFengShui2022}, as the distribution of Linux's heap
addresses lacks entropy~\cite{FreelistRandomization}. In consequence,
the probability of mounting a successful attack are relatively high.

Although it was already well known that layout randomization can
provide security guarantees~\cite{Abadi, Abadi2, Abadi3}, one of
the main novelties of \Cref{sec:safety1} is showing that these
guarantees are valid even if victims can perform pointer arithmetic
and indirect jumps.
Despite this positive result, the threat model considered in \Cref{sec:safety1} is unrealistic nowadays. In particular, it does not take into account the ability of the attackers to access side-channel info-leaks and to steer speculative execution. There is evidence that, by leveraging similar features, the attackers can leak information on the kernel's layout~\cite{cacheKASLR,TLBKASLR,TagBleed, EntryBleed, EchoLoad} and compromise the security guarantees offered by layout randomization~\cite{BlindSide,Meltdown}.

In particular, if the victim system suffers from side-channel
info-leaks that involve the layout, an attacker may break the
protection offered by randomization.  For instance,
suppose the system contains the following system call:

{
\begin{code}[emph={sc_leak,native_write_cr4}]
void sc_leak(x){
  if ((void*) x  == (void*) native_write_cr4)
    for(int i = 0; i < K; i++);
}
\end{code}
}

By measuring the execution time of the system call, an attacker may deduce information on the location of \ex{native_write_cr4}.
If, for some address \ex{a}, the invocation call \ex{sc_leak(a)} takes sensibly longer to execute, the attacker can deduce that \ex{a} corresponds to the address of \ex{native_write_cr4}.
Then, the attacker can disable the SMEP protection via the vulnerable system call \ex{send}.

Side-channel leaks can also be triggered with the help of speculative execution:
in our example from \Cref{fig:msgpassing}, an attacker can use the \ex{recv} primitive to probe for readable data without crashing the system. This can be done by supplying to the system call arguments \ex{s} and \ex{idx} such that \ex{valid(s, idx)} returns false---ideally, causing an out of bound access when the return value is fetched from memory. If the attacker manages in mis-training the branch predictor, the access to \ex{buf[*s][idx]} is performed after a wrong prediction, i.e., in \emph{transient execution}. Depending on the allocation state of the address referenced by \ex{buf[*s][idx]}, two cases arise. If that address does not store any readable data, the memory violation is not raised to the architectural state, because it occurred during transient execution. However, if that address stores writable data, the load operation inserts a new line in the system's cache and, when the system detects the mis-prediction, the execution backtracks to the latest architecturally valid state. After the backtrack, the insertion of a new line in the cache can be detected from user-space. Thus, the attacker can infer that the address referenced by \ex{buf[*s][idx]} contains readable data, and it can exploit the vulnerabilities of the \ex{send} and the \ex{recv} system calls to read or write the content of that memory address. 
This form of \emph{speculative probing} is very similar to what happens, for instance, in the BlindSide attack~\cite{BlindSide} that effectively defeats Linux's KASLR.

The reader may observe that such attacks rely on the attacker's ability to reconstruct the kernel's memory layout by collecting side-channel info-leaks. For this reason, a natural question is whether these attacks can be prevented by imposing that no information of the layouts leaks to the architectural and the micro-architectural state during the execution of system calls. It turns out that this is the case, as we show in~\Cref{sec:safety2}. In practice, this mitigation is of little help though, as it would effectively rule out all system calls that access memory at runtime.

However, we show that any operating system can be pragmatically turned into another system that is functionally equivalent to the original one for user-space programs, but that is not subject to vulnerabilities due to transient execution. \diff{This can be achieved by program transformation.} With such approach, showing that a kernel is safe in the speculative execution model boils down to showing that the kernel under consideration is safe in the classic execution model.  Concretely, a suitable program transformation could block the speculative attack to the \ex{recv} system call we described above by disallowing the transient execution of the unsafe load operation with a speculation barrier: an instruction that stops speculative execution.  The efficacy of this technique is formally shown in~\Cref{sec:sksenforcement}.

\diff{Although partially blocking kernel-space speculative execution can have deleterious ramifications for performance, computer systems' execution takes place mostly in user space, amortizing the overhead that the system encounters in kernel-space. An experimental evaluation given in~\Cref{sec:experiment} confirms this claim.}


\section{Language}
\label{sec:language}
In this section, we introduce the language that we employ throughout the following to study the effectiveness of kernel address space layout randomization. We are considering a simple imperative \texttt{while} language. The address space is explicit, and segregated into user and kernel space. We start by describing the syntax.
\subsection{Syntax}

\begin{figure}
  \centering
  \begin{align*}
    \Expr \ni \expr ,\exprtwo &\bnfdef
    \val
    \bnfmid  \vx 
    \bnfmid  \ar 
    \bnfmid  \fn 
    \bnfmid  \op (\expr_1,\dots,\expr_n)\\
    \Instr \ni \stat,\stattwo &\bnfdef
    \cskip 
    \bnfmid
    \vx \ass \expr 
    \bnfmid \cmemread \vx \expr 
    \bnfmid \cmemass \expr \exprtwo 
    \bnfmid \ccall \exprtwo{\expr_1,\dots,\expr_n} 
    \bnfmid \csyscall \syscall {\expr_1,\dots,\expr_n} \\ 
    &  \bnfmid \cif{\expr}{\cmd}{\cmdtwo} 
    \bnfmid \cwhile{\expr}{\cmd} \\ 
    \Cmd \ni \cmd, \cmdtwo &\bnfdef \cnil \bnfmid \stat\sep\cmd
  \end{align*}
  \vspace{-7mm}
  \caption{Syntax of the language. Here, $\val$ is a value, $\vx$ a register, $\ar$ an array identifier, $\fn$ a procedure identifier, and $\op$ is an operator.}
  \label{fig:syntax}
\end{figure}


The set $\Cmd$ of \emph{commands} is given in \Cref{fig:syntax}.
Memories may store \emph{procedures} and \emph{arrays}, i.e.,
sequences of \emph{values} $\val \in \Val$ organized as contiguous
regions.  The set of values is left abstract, but we assume that it
encompasses at least \emph{Boolean values}
$\cBool \defsym \{\ctrue,\cfalse\}$, \emph{(memory) addresses} $\Add$,
modeled as non-negative integers, and an \emph{undefined value}
$\cnull$.  Within expressions, $\vx \in \Reg$ ranges over
\emph{registers}, $\ar \in \Ar$ and $\fn \in \Fn$ over \emph{array}
and \emph{procedure identifiers}, and $\op \in \Op$ over
\emph{operators}. \emph{Identifiers} $\Id \defsym \Ar \uplus \Fn$ are
mapped to addresses at runtime, as governed by a layout randomization
scheme.  The \emph{size} (length) of an array $\ar$ is denoted by
$\size{\ar}$ and is fixed for simplicity, i.e., we do not model
dynamic allocation and deallocation.

A command $\cmd \in \Cmd$ is a sequence of instructions, evaluated in-order.
The instruction $\vx \ass \expr$ stores the result of evaluating $\expr$ within register $\vx \in \Reg$.
To keep the semantics brief, expressions neither read nor write to
memory. Specifically, addresses are dereferenced explicitly.  To this
end, the instruction $\cmemread \vx \expr$ performs a memory read from
the address given by $\expr$, and stores the corresponding value in
register $\vx$. Dually, the instruction $\cmemass \expr \exprtwo$
stores the value of $\exprtwo$ at the address given by $\expr$.  The
instruction $\ccall \exprtwo{\expr_1,\dots,\expr_n}$ invokes the
procedure residing at address $\exprtwo$ in memory, supplying
arguments $\expr_1,\dots,\expr_n$.  Likewise,
$\csyscall \syscall {\expr_1,\dots,\expr_n}$ invokes a system call
$\syscall$ with arguments $\expr_1,\dots,\expr_n$ from a finite set of
system calls $\Sys$.  The execution of a system call engages the
\emph{privileged execution mode} and thereby the accessible address
space changes.  \diff{ To this end, the address space $\Add$ is
  partitioned into $\kappa_\um$ \emph{user-space} addresses
  $\Addu = \{0, \ldots, \kappa_\um-1\}$, visible in unprivileged mode,
  and $\kappa_{\km}$ \emph{kernel-space} addresses
  $\Addk = \{ \kappa_{\um}, \ldots, \kappa_{\um}+\kappa_{\km}-1\}$,
  visible in unprivileged mode.}  The remaining constructs are
standard.

\paragraph*{Stores} 
Let $\Arr$ denote the set of arrays, i.e., finite sequences of values
$\vec{\val}$ of fixed length $|\vec{\val}|$.  A \emph{store} is a
(well-sorted) mapping $\rfs : \Id \to \Arr \cup \Cmd$, mapping array
identifiers $\ar$ to arrays $\rfs(\ar) \in\Arr$ of length $\size{\ar}$
and procedure identifiers $\fn$ to their implementation
$\rfs(\fn) \in \Cmd$.  Let $\Idu \uplus \Idk = \Id$ be a partitioning
of identifiers int \emph{user-space} and \emph{kernel-space}
identifiers, respectively. This distinction will signify the intended
location of the corresponding objects within the memory address space.
We write $\Fnk \subseteq \Idk$ and $\Fnu \subseteq \Idu$ for the
\emph{kernel-space} and \emph{user-space procedure idenitifiers};
similar for array identifiers we use $\Ark \subseteq \Idk$ and
$\Aru \subseteq \Idu$ to denote \emph{kernel-space} and
\emph{user-space array idenitifiers} respectively.  Given a store
$\rfs$, we always assume that the address space is sufficiently large
to hold $\rfs$; that is, for every $b \in \{\um,\km\}$,
$\kappa_{b} \geq \sum_{\id \in \Id[b]} \size{\id}$.  Here, by
convention, $\size{\fn} \defsym 1$. In the following, we often work
with pairs of stores that associate a set of identifiers to the same
values---e.g., containing identical procedures. We write
$\rfs \eqon{Id} \rfs'$ if $\rfs$ and $\rfs'$ coincide on
$Id \subseteq \Id$.

\paragraph*{\Tcaps} To model safety, each system call $\syscall$
is associated with a fixed set of identifiers defining the memory
regions it can access. We refer to this set as the \emph{\tcaps} of
$\syscall$.

\paragraph*{Systems} 
In our model, kernels are modeled as triples defining system calls,
the content of kernel space memory, and the capabilities associated to every system call.  
Let $\Sys$ denote a (finite) set of \emph{system call identifiers}.
A \emph{system} for $\Sys$ is a tuple $\system = (\rfs,\syss,\caps)$, consisting of:
\begin{itemize}
\item an \emph{initial store} $\rfs : \Id \to \Arr \cup \Cmd$, relating identifiers to their initial
  value;
  \item a \emph{system call map} $\syss : \Sys \to \Cmd$ associating system calls to
  their implementation; and
\item a \emph{capability map} $\caps : \Sys \to \parts{\Idk}$ associating system calls with their capabilities.
\end{itemize}
We require that the code $\rfs(\fn)$ associated to user space
identifiers $\fn \in \Fnu$ is \emph{unprivileged}, i.e.
$\ids(\rfs(\fn)) \subseteq \Idu$, where $\ids(\cmd) \subseteq \Id$ is
the set of identifiers literally occurring in $\cmd$.  As our notion
of safety will be defined in terms of system calls' capabilities, we
furthermore require that the set of capabilities $\caps(\syscall)$ of
a system call $\syscall$ includes all identifiers that $\syscall$
refers to, directly in its body, or indirectly through procedure or
system calls.
To define this latter requirement formally, let
$\refs_\system(\cmd) \subseteq \Id \cup \Sys$ denote the set of
identifiers and system calls that the command $\cmd$ refers to,
directly or indirectly.  Specifically, we call $\mathsf{syscalls}(\cmd)$ is
the set of system calls in $\cmd$, and we define $\refs_\system(\cmd)$
as the smallest set such that:
(a)~$\ids(\cmd) \cup \mathsf{syscalls}(\cmd) \subseteq
\refs_\system(\cmd)$, where (b)~if $\fn \in \refs_\system(\cmd)$ then
$\ids(\rfs(\fn)) \subseteq \refs_\system(\cmd)$, and likewise, (c)~if
$\syscall \in \refs_\system(\cmd)$ then
$\ids(\syss(\syscall)) \subseteq \refs_\system(\cmd)$.  The
requirement can now be stated as
$\refs_\system(\syss(\syscall)) \setminus \Sys\subseteq
\caps(\syscall)$.
%

\subsection{Semantics}

We now endow our language with an operational semantics. To define our
notion of safety, we directly define an \emph{instrumented semantics},
which signals capability violations.  Semantics cannot be defined
directly on stores, which directly relate identifiers to their
values. Instead, to model address-based memory accesses, we introduce
\emph{memories}.

\paragraph*{Memories}
A \emph{memory} is a function
$\mem: \Add \to \Val \cup \Cmd \cup \{\none\}$ associating addresses
with their content, or to the special symbol $\none \not\in \Val$ if
the location is not occupied.  We denote by $\Mem$ the set of all
memories.  \diff{Arrays will be represented by sequences of values in
  contiguous memory locations.} We use $\update \mem \add \val$ to
denotes the memory that is pointwise identical to $\mem \in \Mem$,
except for the address $\add\in \Add$ that is mapped to
$\val\in \Val$.  Note that updates are restricted values, in
particular updating a memory location with a command is forbidden.
Thereby we model a W\^{}X memory protection policy, separating
writable from executable memory regions.

\paragraph{Layouts} A \emph{(memory) layout} is a function $\lay : \Id \to \Add$ that
describes where objects are placed in memory.
As we mentioned, an array $\ar$ is stored as continuous block at addresses
$\underline \lay(\ar) \defsym \{ \lay(\ar), \ldots, \lay(\ar)+\size\ar - 1 \}$
within memory. For procedure identifiers $\fn$, we set 
$\underline \lay(\fn) \defsym \{ \lay(\fn)\}$.
We overload this notation to sets of identifiers in the obvious way.
In particular, $\underline \lay(\Ar)$ and $\underline \lay(\Fn)$ refer
to the address-spaces of arrays and procedures, under the given layout.
We regard only layouts that associate identifiers with non-overlapping blocks
($\underline \lay(\id_1)\cap \underline \lay(\id_2) = \emptyset$ for all
$\id_1 \neq \id_2$) and that respect address space separation ($\underline \lay(\id) \subseteq \Add[b]$
for $\id \in \Id[b]$, $b \in \{\um,\km\}$).
The set of all such layouts is denoted by $\Lay$.
Note that, by the assumptions on the sizes
$\kappa_\um$ and $\kappa_{\km}$ of address spaces, layouts always exist.


A layout $\lay : \Id \to \Add$ now defines how a store
$\rfs : \Id \to \Arr \cup \Cmd$ is placed in memory. This memory,
denoted by $\lay \lcomp \rfs$, is defined as follows:
\[
  (\lay \lcomp \rfs)(\add) \defsym
  \begin{cases}
    \rfs(\fn) & \text{if $\add = \lay(\fn)$ for some $\fn \in\Fn$,} \\
    \vec \val[k] & \parbox{35em}{if $\add = \lay(\ar)+ k$, for some $\ar \in \Ar$ and $0 \leq k < \size{\ar}$ s.t. $\rfs(\ar) = \vec \val$,} \\
    \none & \text{otherwise,}
  \end{cases}
\]
where $\vec \val[k]$ denotes the $k$-th element of the tuple $\vec \val$, indexed
starting from $0$.

\paragraph{Randomization scheme}
Abstracting from details, we model an address space randomization
scheme through a probability distribution over layouts.  A specific
layout $\lay$ is selected at random prior to system execution.  For a
given system $\system = (\rfs,\syss,\caps)$, this choice then dictates
the initial memory configuration $\lay \lcomp \rfs$.  Although the
semantics is itself deterministic, computation can be viewed as a
probabilistic process.  In our language, instructions are
layout-sensitive: for example, the outcome of a memory load operation
at a specific address depends on whether $\lay$ places an object at
that address. Therefore, kernel's safety should be construed as a
property that holds in a probabilistic sense.


\paragraph*{Register maps} Besides memory locations, our program can
also manipulate the content of registers. In our semantics, we model
registers through functions $\regmap: \Reg \to \Val$, associating register
identifiers to their value. As for memories,
the notation $\update \regmap \vx \val$ denotes the register map that is pointwise
identical to $\regmap$, except for the register $\vx$, which is mapped to the value $\val$.
Again, this operation is only defined when $\val$
is a value, i.e., registers cannot store arrays or procedures.

\paragraph*{Semantics of expressions} To define the semantics of
expressions, we assume for each $n$-ary operator $\op$ an
interpretation $\widehat \op:\Val^n \to \Val$. The semantics of
an expression depends, besides registers, on a layout, in order to resolve identifiers.
The semantics is now defined by:
\begin{align*}
\sem\val_{\regmap, \lay} &\defsym \val & \sem\vx_{\regmap, \lay} &\defsym \regmap(\vx) & \sem \id_{\regmap, \lay} &\defsym \lay(\id) &  \sem{\op (\expr_1,\dots,\expr_n)}_{\regmap, \lay} \defsym \widehat \op (\sem{\expr_1}_{\regmap, \lay}, \ldots, \sem{\expr_n}_{\regmap, \lay} ).
\end{align*}
Let $(\cdot)^{\Add{}} : \Val \to \Add$ and
$(\cdot)^{\cBool{}} : \Val \to \cBool$ be functions that cast any
value to an address, or a Boolean, respectively. In particular,
$\sem {\expr}^{\Add}_{\regmap, \lay}$ and
$\sem {\expr}^{\cBool}_{\regmap, \lay}$ evaluate expression to
addresses and Boolean.

\paragraph*{Configurations}
Due to the presence of (possible recursive) procedures, configurations make use
of a stack of frames. 
Each such frame records the command under evaluation, the register contents
and the execution mode.
Formally, configurations are drawn from the following BNF:

{\hfill
\begin{minipage}[l]{0.25\linewidth}
  \begin{cbnf}
    \opt & \um \bnfmid \km[\syscall]
  \end{cbnf}
\end{minipage}
\begin{minipage}[l]{0.25\linewidth}
  \begin{cbnf}
    \st & \varepsilon \bnfmid \frame{\cmd}{\regmap}{\opt} : \st 
  \end{cbnf}
\end{minipage}
\begin{minipage}[l]{0.25\linewidth}
\begin{cbnf}
    \confone, \conftwo & \conf{\st,\mem} \bnfmid \err \bnfmid \unsafe. 
  \end{cbnf}
\end{minipage}
\hfill}

A \emph{configuration} of the form $\conf{\st,\mem}$, with
top-frame $\frame{\cmd}{\regmap}{\opt}$, indicates that $\cmd$ is
executed with allocated registers $\regmap$ in \emph{execution mode}
$\opt$ on memory $\mem$.  In particular, $\opt = \km[\syscall]$ indicates that
execution proceeds in privileged kernel-mode, triggered by system call
$\syscall$. The annotation of the \emph{kernel-mode} flag by a system
call name facilitates the instrumentation of the semantics.  Indeed,
every time an access to the memory is made, the semantics enforces
that address is in the capabilities of the system call that is running
(if any).  If the address can be rightfully accessed, the execution
Finally, $\err$ signals abnormal termination (for instance, when
dereferencing a pointer to kernel memory from user-mode or vice
versa).


\ifdefined\conference{\input{rules1}}\fi
\ifdefined\arxiv{\begin{figure*}[t]
  \small
  \centering
  \columnwidth=\linewidth
  \[
    \Infer[WL][Pop]{
      \step
      {\ntc{\cnil}{\regmap}{\opt} {\frame {\cmd} {\regmap'} {\opt'}\cons \st}{\mem}}
      {\ntc{\cmd}{\update{\regmap'}{\ret} {\regmap(\ret)}}{\opt'}{\st}{\mem}}
    }
    {}
  \]
  \\[-3mm]
  \[
    \Infer[WL][Skip]
    { \step
      {\ntc{\cskip\sep\cmd}{\regmap}{\opt}{\st}{\mem}}
      {\ntc{\cmd}{\regmap}{\opt}{\st}{\mem}}
    }
    {
    }
  \]
  \\[-3mm]
  \[
    \Infer[WL][Op]
    { \step
      {\ntc{\vx \ass \expr\sep\cmd}{\regmap}{\opt}{\st}{\mem}}
      {\ntc{\cmd}{\update{\regmap}{\vx}{\sem \expr_{\regmap, \lay}}}{\opt}{\st}{\mem}}
    }
    {
    }
  \]
  \\[-3mm]
  \[
    \Infer[WL][If]
    { \step
      {\ntc{\cif \expr {\cmd_\ctrue} {\cmd_\cfalse}\sep\cmdtwo}{\regmap}{\opt}{\st}{\mem}}
      {\ntc{\cmd_{\toBool{\sem \expr_{\regmap, \lay}}}\sep \cmdtwo}{\regmap}{\opt}{\st}{\mem}}
    }
    {
    }
  \]  
  \\[-3mm]
  \[
    \Infer[WL][While]
    { \step
      {\ntc{\cwhile {\expr} {\cmd}\sep\cmdtwo}{\regmap}{\opt}{\st}{\mem}}
      {\confone_{\toBool{\sem \expr_{\regmap, \lay}}}}
    }
    {
      \confone_{\ctrue} = {\ntc{\cmd\sep\cwhile {\expr} {\cmd}\sep \cmdtwo}{\regmap}{\opt}{\st}{\mem}} &
      \confone_{\cfalse} = {\ntc{\cmdtwo}{\regmap}{\opt}{\st}{\mem}}     }
  \]
  \caption{Semantics w.r.t. system $\system=(\rfs,\syss,\caps)$, first part.}
  \label{fig:stepexcerpt1}
\end{figure*}

\begin{figure*}[t]
  \small
  \centering
  \columnwidth=\linewidth
  \[
    \Infer[WL][Load]
    { \step
      {\ntc{\cmemread \vx \expr\sep\cmd}{\regmap}{\opt}{\st}{\mem}}
      {\ntc{\cmd}{\update{\regmap}{x}{\mem(\add)}}{\opt}{\st}{\mem}}
    }
    {
      \toAdd{\sem{\expr}_{\regmap, \lay}} = \add &
      \add \in \underline \lay(\Ar[\opt]) &
      \fbox{$\opt = \km[\syscall] \Rightarrow \add \in \underline \lay(\caps(\syscall))$}
    }
  \]
  \\[-3mm]
  \[
    \Infer[WL][Load-Error]
    {\step
      {\ntc{\cmemread \vx \expr\sep\cmd}{\regmap}{\opt}{\st}{\mem}}
      {\err }
    }
    {
      \toAdd{\sem{\expr}_{\regmap, \lay}} = \add &
      \add \not\in \underline \lay(\Ar[\opt])
    }
    \quad
    \Infer[WL][Load-Unsafe]
    {\step
      {\ntc{\cmemread \vx \expr\sep\cmd}{\regmap}{\km[\syscall]}{\st}{\mem}}
      {\unsafe}
    }
    {
      \toAdd{\sem{\expr}_{\regmap, \lay}} = \add &
      \add \in \underline \lay(\Ar[\km]) &
      \fbox{$\add \not\in \underline \lay(\caps(\syscall))$}
    }
  \]
  \\[-3mm]
  \[
    \Infer[WL][Store]
    { \step
      {\ntc{\cmemass \expr \exprtwo\sep\cmd}{\regmap}{\opt}{\st}{\mem}}
      {\ntc{\cmd}{\regmap}{\opt}{\st}{\update{\mem}{\add}{\sem \exprtwo_{\regmap, \lay}}}}
    }
    {
      \toAdd{\sem{\expr}_{\regmap, \lay}} = \add &
      \add \in \underline \lay(\Ar[\opt]) &
      \fbox{$\opt = \km[\syscall] \Rightarrow \add \in \underline \lay(\caps(\syscall))$}
    }
  \]
  \\[-3mm]
  \[
    \Infer[WL][Store-Error]
    {\step
      {\ntc{\cmemass \expr \exprtwo\sep\cmd}{\regmap}{\opt}{\st}{\mem}}
      {\err}
    }
    {\toAdd{\sem\expr_{\regmap, \lay}} = \add &
      \add \notin \underline \lay(\Ar[\opt])
    }
    \quad
    \Infer[WL][Store-Unsafe]
    {\step
      {\ntc{\cmemass \expr \exprtwo\sep\cmd}{\regmap}{\km[\syscall]}{\st}{\mem}}
      {\unsafe}
    }
    {
      \toAdd{\sem{\expr}_{\regmap, \lay}} = \add &
      \add \in \underline \lay(\Ar[\km]) &
      \fbox{$\add \not\in \underline \lay(\caps(\syscall))$}
    }
  \]
  \\[-3mm]
  \[
    \Infer[WL][Call]{
      \step
      {\ntc{\ccall{\expr}{\vec \exprtwo}\sep\cmd}{\regmap}{\opt}{\st}{\mem}}
      {
        \ntc
        {\mem(\add)}
        {\regmap_0[\vec \vx \upd \sem{\vec \exprtwo}_{\regmap,\lay}]}
        {\opt}
        {\frame{\cmd}{\regmap}{\opt} : \st}
        {\mem}
      }
    }
    {\toAdd{\sem{\expr}_{\regmap, \lay}}=\add &
      \add \in \underline \lay(\Fn[\opt]) &
      \fbox{$\opt = \km[\syscall] \Rightarrow \add \in \underline \lay(\caps(\syscall))$}
    }
  \]
  \\[-3mm]
  \[
    \Infer[WL][Call-Error]{
      \step
      {\ntc{\ccall{\expr}{\vec\exprtwo}\sep\cmd}{\regmap}{\opt}{\st}{\mem}}
      {\err}
    }
    {
      \toAdd{\sem{\expr}_{\regmap, \lay}} = \add &
      \add \not\in \underline \lay(\Fn[\opt])
    }
    \quad
    \Infer[WL][Call-Unsafe]{
      \step
      {\ntc{\ccall{\expr}{\vec \exprtwo}\sep\cmd}{\regmap}{\km[\syscall]}{\st}{\mem}}
      {\unsafe}
    }
    {
      \toAdd{\sem{\expr}_{\regmap, \lay}} = \add &
      \add \in \underline \lay(\Fn[\km]) &
      \fbox{$\add \not\in \underline \lay(\caps(\syscall))$}
    }
  \]
  \\[-3mm]
  \[
    \Infer[WL][SystemCall][\textsc{SC}]{
      \step
      {\ntc{\csyscall{\syscall}{\vec \exprtwo}\sep\cmd}{\regmap}{\opt}{\st}{\mem}}
      {
        \ntc
        {\syss(\syscall)}
        {\update {\regmap_0}{\vec \vx}{\sem{\vec\exprtwo}_{\regmap,\lay}}}
        {b'}
        {\frame{\cmd}{\regmap}{\opt} : \st}
        {\mem}
      }
    }
    {
      b = \um \Rightarrow b' = \km[\syscall] &
      b = \km[\syscalltwo] \Rightarrow b' = \km[\syscalltwo] 
    }
  \]
    \caption{Semantics w.r.t. system $\system=(\rfs,\syss,\caps)$, second part.}
  \label{fig:stepexcerpt2}
\end{figure*}}\fi

\paragraph{Small step operational semantics} Transitions in our semantics take the form
\[
  \step{\confone}{\conftwo},
\]
indicating that, w.r.t. system $\system$, configuration $\confone$
reduces to $\conftwo$ in one step, under layout $\lay$.
\ifdefined\conference{The most important reduction rules are defined
    in \Cref{fig:stepexcerpt}}\fi\ifdefined\arxiv{The reduction rules
    are defined in \Cref{fig:stepexcerpt1,fig:stepexcerpt2}}\fi.
%
Rule~\ref{WL:Load} implements a successful memory load $\cmemread \vx \expr$.
Expression $\expr$ is evaluated to an address $\add = \toAdd{\sem{\expr}_{\regmap, \lay}}$,
and the content of the register $\vx$ is updated with the value $\mem(\add)$. 
The side-condition $\add \in \underline \lay(\Ar[\opt])$ enforces
that $\add$ refers to a value accessible in the current execution mode $\opt$ (by slight abuse of notation, we disregard
the system call label in kernel-mode),
otherwise the instruction leads to $\err$ (see Rule~\ref{WL:Load-Error}).
As such, we are modeling unprivileged execution and SMAP protection, preventing
the access of kernel-space addresses when in user-mode, and vice versa.
The final, boxed, side-condition refers to the safety instrumentation.
In kernel-mode, triggered by system call $\syscall$ ($\opt =\km[\syscall]$), the rule
ensures that $\add$ refers to an object within the capabilities of $\syscall$ ($\add \in \underline w(\caps(\syscall))$).
When this condition is violated, unsafe execution is signaled (see Rule~\ref{WL:Load-Unsafe}).
In a similar fashion, the rules for memory writes and procedure calls
are defined.


Rule~\ref{WL:Call} deals with procedure calls.  It opens a new frame
and, places the $n$ evaluated arguments
$\exprtwo_1, \ldots, \exprtwo_n=\vec\exprtwo$ at registers
$\vx_1,\dots,\vx_n$ in an initial register environment $\regmap_0$,
summarized by the notation
$\regmap_0[\vec \vx \upd \sem{\vec \exprtwo}_{\regmap,\lay}]$. Notice
that our choice does not exclude stack-based inter-procedural
communication from our model: procedures can use a dedicated array as
stack, and pass the stack and return pointers as arguments.  System
calls, modeled by Rule~\ref{WL:SystemCall}, follow the same calling
convention.  Note that, in the newly created frame, the execution flag
is set to kernel-mode.  Once a procedure or system call finished
evaluation, Rule~\ref{WL:Pop} removes the introduced frame from the
stack. Observe how the rule permits return values through a designated
register $\ret$.  The remaining rules are standard.

Let us denote by $\lay \red \confone \to^* \conftwo$ that configuration $\confone$ reduces in zero or more steps to configuration $\conftwo$,
and by $\diverge \confone$ that $\confone$ \emph{diverges}.
In our semantics, under layout $\lay$, any non-diverging computation either
halts in a terminal configuration of the form $\conf{\frame{\cnil}{\regmap}{\opt}, \lay \lcomp \store'}$,
or abnormally terminates through an error $\err$, or safety violation $\unsafe$.
This motivates the following definition of an evaluation function:

\[
  \Eval {\cmd,\regmap,\opt,\store} \defsym
  \begin{cases}
    (\val,\store')
    & \!\!\!\text{if } \lay \calign{\red \conf{\frame{\cmd}{\regmap}{\opt},\lay \lcomp \store} \to^* \conf{\frame{\cnil}{\regmap'[\ret \mapsto \val]}{\opt}, \lay \lcomp \store'},} \\
    \err
    & \!\!\!\text{if } \lay \red \conf{\frame{\cmd}{\regmap}{\opt},\lay \lcomp \store} \to^* \err, \\
    \unsafe
    & \!\!\!\text{if } \lay \red \conf{\frame{\cmd}{\regmap}{\opt},\lay \lcomp \store} \to^* \unsafe\\
      \Omega
    & \!\!\!\text{if $\diverge{\conf{\frame{\cmd}{\regmap}{\opt},\lay \lcomp \store}}$}.
  \end{cases}
\]

Note how, in the case of normal termination, a computation produces a pair of a return value and a store.


\section{Threat Model}
\label{sec:threatmodel}
In our threat model, attackers are unprivileged user-space
programs that execute on a machine supporting two
privilege rings: user-mode and kernel-mode.
The victim is the host operating system which
runs in kernel mode and has exclusive access to its private memory.
In particular, the operating system exposes a set of procedures,
the system calls, that can be invoked by the attacker and
that have access to kernel's memory. The attacker's goal is
to trigger a system call to perform an unsafe memory
access or control flow transfer.

In \Cref{sec:safety1}, attackers are ordinary programs that do not
control speculative execution and do not have access to side-channel
info-leaks. However, the target machine has standard kernel protection
measures enabled.  In particular, these measures include \emph{Data
  Execution Protection} mechanisms (DEP), \emph{Supervisor Mode Access
  Prevention} (SMAP)~\cite{LWNSMAP} preventing kernel-mode access to
user-space data, and \emph{Supervisor Mode Execution Prevention}
(SMEP)~\cite{SMEP} preventing the execution of user-space functions
when running in kernel-mode.  More precisely, the above-mentioned
protection mechanisms are modeled in our semantics by the
preconditions of Rules \ref{WL:Call}, \ref{WL:Load},
\ref{WL:Store} that prevent the system from: (i) loading and
overwriting procedures, (ii) execute values, (iii) accessing
user-space data and procedures when the system is in kernel-mode.
Although we assume that the victim's code remains immutable, our
results naturally extend to scenarios where a victim may load new
code. Indeed, for any victim that introduces new vulnerable code at
runtime, we can consider the victim that has already loaded that code
and is at least as vulnerable as the first one. 
\diff{We also assume that the victim hardware supports
  Intel${}^{\text{\textregistered}}$\ \emph{Indirect Branch Tracking}
  (IBT)~\cite{IntelManual}, to restrict control flow transfer only to
  specific program points, typically the beginning of a procedure.
  This is modeled by restricting the target of indirect call
  instructions to the beginning of procedures only.  } Finally, the
system adopts kernel address space layout randomization, that is
modeled by executing programs with a randomly chosen memory layout.

\diff{In \Cref{sec:safety2}, we consider a stronger threat model where
  attackers have access to side-channel observations and control
  \emph{Pattern History Table} (PHT), \emph{Branch Target Buffer}
  (BTB) predictions and \emph{Store To Load} (STL) forwarding, related
  to Spectre v1, v2 and v4 vulnerabilities
  respectively~\cite{Spectre}}.  In addition, we assume that the
system supports \emph{Page Table Isolation} (PTI)~\cite{PTI} to
prevent speculative accesses to kernel-space memory by user-space
program; this is modeled by using the same preconditions of Rules
\ref{WL:Call}, \ref{WL:Load}, \ref{WL:Store} for their speculative
counterparts in \Cref{sec:internalsemantics}. We also assume that
the victim's machine supports a specific class of indirect branch
instructions where the user cannot influence the jump target
prediction---this result can be achieved by using retpoline when it is
fully effective~\cite{IntelBHIDISS} or by using
Intel${}^{\text{\textregistered}}$\ enhanced Indirect Branch
Restricted Speculation (eIBRS)~\cite{eIBRS} and
\texttt{BHI\_DIS\_S}~\cite{IntelBHIDISS}.
This is modeled by adding an instruction for speculatively safe jumps
to the victim's language in \Cref{sec:sksenforcement}.

Finally, we assume that the victim machine does not support return
address speculation (exploited by attacks such as
Retbleed~\cite{Retbleed}), and we model it by not allowing any form of
return target speculation.  Phantom
speculation~\cite{Phantom,Inception} is out of scope for this work and
the correspondent speculation mechanism is not modeled.



\section{Kernel Safety in the Classic Threat Model}
\label{sec:safety1}

In this section we show how the result of Abadi et. al. \cite{Abadi,Abadi2,Abadi3} scales
to the model introduced in Section~\ref{sec:language}.
The safety property that we aim at is defined in terms of our instrumented semantics, as follows:

\begin{definition}[Kernel safety]\label{def:ks}
  We say that a system $\system = (\rfs,\syss,\caps)$ is \emph{kernel safe},
  if for every layout $\lay$, 
  \emph{unprivileged} attacker $\cmd \in \Cmd$, 
  and registers $\regmap$, we have: 
  \[
    \lnot \left(\lay \red \conf{\frame{\cmd}{\regmap}{\um}, \lay \lcomp \rfs} \to^* \unsafe \right).
  \]
\end{definition}
\def\ex{\cc[v,s1,s2,f,s]}
Thus, safety is broken if an attacker $\cmd$, executing in unprivileged user mode, is able to trigger a system call
in such a way that it accesses, or invokes, a kernel-space object outside its capabilities.
The source of such a safety violation can be twofold:
\begin{varenumerate}
  \item\label{issue1} \textbf{Scope extrusion.}
  An obvious reason why kernel-safety may fail is due to apparent
  communication channels, specifically through the memory and procedure returns.
  As an example, consider a system $\system =(\rfs, \syss, \caps)$,
  where:
  \begin{align*}
    \syss(\syscall_1) &\defsym \cmemass \ar \fn & \syss(\syscall_2) &\defsym \cmemread \vx \ar ;\ccall \vx {}   & \caps (\syscall_1)=\caps(\syscall_2)=\{\ar\}
  \end{align*}
  A malicious program can use $\syscall_1$ to store the address of
$\fn$ at $\ar[0]$, which is a shared capability. A consecutive call
to $\syscall_2$ then breaks safety if $\fn$ is not within the capabilities of $\syscall_2$.
\item\label{issue2} \textbf{Probing.} Another counterexample is given
  by a system call accessing memory based on its input, such as the
  system call that only contains the instruction $\ccall {\vx_1}{}$,
  which directly invokes the procedure stored at the kernel-address
  corresponding to the value of its first argument $\vx_1$.  This
  system call can potentially be used as a gadget to invoke an
  arbitrary kernel-space procedure from user-space. Since an attacker
  lacks knowledge of the kernel-space layout, such an invocation needs
  to happen effectively through probing.  As any probe of an unused
  memory address leads to an unrecoverable error,\footnote{%
    This is not always the case for \emph{user-space} software
    protected with layout randomization, as some programs (e.g. web
    servers) may automatically restart after a crash to ensure
    availability.  This behavior can be exploited by attackers to
    probe the entire memory space of the victim program, thus
    compromising the protection offered by layout
    randomization~\cite{ApacheAttack}.  } the likelihood of an unsafe
  memory access is, albeit not zero, diminishingly small when the
  address-space is reasonably large.
\end{varenumerate}
To overcome Issue~\ref{issue1}, we impose a form of (layout) \emph{non-interference}
on system calls.
\begin{definition}[Layout non-interference]
  \label{def:lni}
  Given $\system=(\rfs,\syss,\caps)$,
  a system call $\syscall$ is \emph{layout non-interfering}, if,
  \[
    \Eval[\system][\lay_1]{\syss(\syscall), \regmap, \km[\syscall], \store'}
    \evaleq
    \Eval[\system][\lay_2]{\syss(\syscall), \regmap, \km[\syscall], \store'}
  \]
  for all layouts $\lay_{1},\lay_{2}$, registers $\regmap$ and stores $\store' \eqon{\Fn} \rfs$.
  Here, the equivalence $\evaleq$ extends equality with $\err\evaleq \unsafe$ and $\unsafe\evaleq \err$.
  The system $\system$ is non-interfering if all its system calls are.
\end{definition}
In effect, layout non-interfering systems do not expose layout information,
neither through the memory nor through return values.
In particular, observe how non-interference rules out Issue~\ref{issue1},
as witnessed by two layouts placing $\fn$ at different addresses in kernel-memory.

Concerning Issue~\ref{issue2}, it is well known that layout
randomization provides in general safety not in an absolute sense, but
\emph{probabilistically}~\cite{Abadi,DieHard,PaXASLR}.  Indeed, the
chance for a probe to be successful depends on the randomization
scheme.  Following \citet{Abadi}, let $\mu$ be a \emph{probability
  distribution} of layouts, i.e., a function $\mu : \Lay \to [0,1]$
assigning to each layout $\lay \in \Lay$ a probability $\mu(\lay)$
(where $\sum_{\lay \in \Lay} \mu(\lay) = 1$).  Without loss of
generality, we assume that the layout of public, i.e. user-space,
addresses is fixed.  That is, we require for each $\lay_{1},\lay_{2}$
with non-zero probability in $\mu$, that $\lay_1(\id)=\lay_2(\id)$ for
all $\id \in \Idu$.  For a distribution of layouts $\mu$ and a system
$\system = (\rfs, \syss, \caps)$, the value of $\delta_{\mu, \system}$
quantifies the smallest probability that a probe for an address
$\add \in \Addk$ fails.  To formally define $\delta_{\mu, \system}$, given a
system call $\syscall$ we denote with
$\id_1^\syscall, \dots,\id_k^\syscall$ the enumeration of its
references $\refs_\system(\syss(\syscall))\setminus \Sys$.  The value $\delta_{\mu, \system}$
can then be defined as follows:
\begin{multline*}
  \delta_{\mu, \system} \defsym
  \min \bigl \{
  \displaystyle{\Pr_{\lay\leftarrow \mu}}
  [ \add \notin \underline\lay(\Idk) \mid \lay(\id_i^\syscall)=p_i,
  \text{ for } 1\le i \le h]
  \mid \syscall \in \Sys, p,p_1,\dots,p_h \in \Addk \land {}\\
  \phantom{{} \mid} p \notin \{p_i, \ldots, p_i +\size{\id_i^\syscall}-1\}, \text{ for } 1\le i \le h
  \bigr \}.
\end{multline*}
In practice, $\delta_{\mu, \system}$ bounds the probability that,
during the execution of a system call $\syscall$, a fixed kernel
address $\add$ is not allocated, given that it does not store any
object that is in the references of that system call. Notably, if an
attacker controls the value of a kernel address $\add$,
$\delta_{\mu, \system}$ is a lower bound to the probability that its
probe for $\add$ does not hit any memory content.
This property
is reflected by the cases \ref{case:lemma12} and \ref{case:lemma2b} of
\Cref{lemma:onsyscall,lemma:onsyscallterm} below.  More precisely,
\Cref{lemma:onsyscall} proves this property for fixed-length reductions,
and \Cref{lemma:onsyscallterm} lifts it to full evaluations.
By considering the
complementary event, $\delta_{\mu, \system}$ gives an upper bound to
the probability of performing a safety violation in the presence of
layout randomization, as expressed by \Cref{thm:scenario1} below.

We now give the formal statements of
\Cref{lemma:onsyscall,lemma:onsyscallterm} together with their proofs.
To this aim, we extend $\refs_\sigma$ to frame stacks as follows:
\begin{align*}
  \refs_\sigma(\nil) &\defsym \emptyset&
  \refs_\sigma(f\cons \st) &\defsym \refs_\sigma(f) \cup \refs_\sigma(\st) &
  \refs_\sigma(\frame \cmd \regmap \opt) &\defsym \refs_\sigma(\cmd).
\end{align*}

\begin{restatable}{lemma}{onsyscall}
  \label{lemma:onsyscall}
  Let $\syscall$ be a system call of a \emph{layout non-interfering} system $\system=(\rfs, \syss, \caps)$, and let $\refs_\system(\syss(\syscall))\setminus \Sys=\{\id_1, \dots, \id_h\}$ be identifiers within the references of $\syscall$. Given a sequence of addresses $\add_1, \dots, \add_h$, an initial frame $\frame{\syss(\syscall)}{\regmap}{\km[\syscall]}$, and a store $\rfs'$ such that $\rfs'\eqon{\Fn}\rfs$,
  for every reduction length $\nat\in \Nat$ and distribution of layouts $\mu$, one of the following statements holds:
  
  \begin{enumerate}[label={(\arabic*)}]
    \item For every layout $\lay$ such that $\forall 1\le i\le h$, $\lay(\id_i)=\add_i$, we have \(
      \nstep{\nat}
      {\conf{\frame{\syss(\syscall)}{\regmap}{\km[\syscall]}, \lay \lcomp \rfs'}}
      {\conf{\overline \st, \lay \lcomp \overline \rfs}}
      \)
      for some non-empty stack $\overline \st$
      such that
      $\refs_\system(\overline\st)\subseteq \refs_\system(\syss(\syscall))$, and a store
      $\overline \rfs\eqon{\Fn}\rfs'$.
    \item\label{case:lemma12}
    \(
      \Pr_{\lay \leftarrow \mu}\Big[ \exists \nat' \le \nat. \nstep {\nat'} {\conf{\frame{\syss(\syscall)}{\regmap}{\km[\syscall]}, \lay\lcomp \rfs'}}\err \,\,\Big|\,\,
      \forall 1\le i\le h.\lay(\id_i)=\add_i\Big]\ge \delta_{\mu, \system}.
    \)

  \item For every layout $\lay$ such that $\forall 1\le i\le h$, $\lay(\id_i)=\add_i$, we have   \(
      \nstep{\nat'}
      {\conf{\frame{\syss(\syscall)}{\regmap}{\km[\syscall]}, \lay \lcomp \rfs'}}
      {\conf{\frame{\cnil}{\overline\regmap}{\km[\syscall]}, \lay \lcomp \overline \rfs}}
      \)
      for some $\nat' \le\nat$,  $\overline \regmap$ and
    store $\overline \rfs\eqon{\Fn}\rfs'$.
  \end{enumerate}
\end{restatable}

\begin{proof}[Proof sketch of \Cref{lemma:onsyscall}]
  The proof proceeds by induction on $\nat$ and case analysis on the transition rules. In the base case, it is trivial to establish claim (A). In the inductive case, by applying the IH to the initial configuration. Three cases arise:
  \begin{proofcases}
    \proofcase{A} Let $\system$ be a system, $\syscall$ a system call
    and consider the initial configuration
    $ {\conf{\frame{\syss(\syscall)}{\regmap}{\km[\syscall]}, \lay
        \lcomp \rfs'}}$.  Suppose
    $\refs_\system(\syss(\syscall))\setminus \Sys= \{\id_1, \dots,
    \id_h\}$ and fix addresses $\add_1, \dots, \add_h$.  By the IH,
    there exists a stack $\st$ and a
    store $\rfs'' \eqon{\Fn} \rfs$ satisfying the following property:
    \[
      \forall \lay. \left( \forall 1 \leq i \leq h, \lay(\id_i) = \add_i\right)  \Rightarrow
      \nstep{\nat}
      {\conf{\frame{\syss(\syscall)}{\regmap}{\km[\syscall]}, \lay \lcomp \rfs'}}
      {\conf{\st, \lay \lcomp \rfs''}}.
    \]
    Since $\st$ is non-empty, we can assume that
    $\st = \frame \cmd {\regmap'} {\km[\syscall]}:\st'$.  Furthermore,
    from the IH, we know that
    $\refs_\system(\st)\subseteq \refs_\system(\syss(\syscall))$
    (H). The proof proceeds with a case analysis on $\cmd$. We only
    focus on the representative case of procedure invocation.
    
    \begin{proofcases}
      \proofcase{$\ccall \expr {\exprtwo_1, \dots, \exprtwo_k}\sep
        \cmdtwo$} We start by observing that by (H) there exists a
      unique address $\add$ such that for every $\lay$ that satisfies
      the precondition (that stores the references of $\syscall$ at
      addresses $p_1, \ldots, p_h$), we have
      $\toAdd{\sem \expr_{\regmap,\lay}}=\add$. Similarly, we
      introduce the values $\val_1, \dots,\val_k$, which correspond to
      the semantics of $\exprtwo_1,\dots,\exprtwo_k$ evaluated under
      $\regmap$ and every layout that satisfies the precondition.
      Again for the same reason, we observe that there exists a set
      $P$ such that for each layout under consideration, it holds that
      $\underline{\lay} (\refs_\system(\syss(\syscall))\setminus
      \Sys)=P$. The proof proceeds by cases on whether $\add \in P$.
      \begin{proofcases}
        \proofcase{$\add\in P$} In this case, there is a unique identifier $\id_j$ such that for every layout $\lay$ that satisfies the precondition, we have $\add \in \underline{\lay}(\id_j)$. We analyze two cases based on whether $\id_j$ is a function identifier $ \fn$.
        \begin{proofcases}
          \proofcase{$\id_j= \fn$} In this case, from the definition
          of $\lcomp$, we deduce that for each of these layouts we
          have $\lay\lcomp \rfs''(\add) = \rfs''(\fn) = \rfs(\fn)$,
          where the last step follows from the assumption
          $\rfs''\eqon{\Fn}\rfs$.  Since $\add \in P$, we conclude
          that, independently of the specific layout, if the
          preconditions hold, then:
          \begin{equation*}
            \step
            {\conf{\frame{\ccall  \expr {\exprtwo_1, \dots, \exprtwo_k}\sep \cmdtwo}{\regmap}{\km[\syscall]}:\st', \lay\lcomp \rfs''}}{}\\
            {\conf{\frame{\rfs(\fn)}{\regmap_0'}{\km[\syscall]}:\frame{\cmdtwo}{\regmap}{\km[\syscall]}:\st', \lay\lcomp \rfs''}},
          \end{equation*}
          where $\regmap_0'= \regmap_0[\vx_1,\dots,\vx_k\upd \val_1,\dots,\val_k]$.
          Finally, we observe that, by the definition of
          $\refs_\system$, $\refs_\system(\syss(\syscall))$
          contains all the identifiers within
          $\rfs(\fn) = \rfs(\id_j)$ because $\id_j \in \refs_\system(\syss(\syscall))$ and
          $\refs$ is closed under procedure calls.
          This shows that (A) holds.
          \proofcase{$\id_j =\ar$} In this case, since the set of
          array identifiers and that of functions are disjoint, we
          conclude that for every layout $\lay$ that satisfies the
          preconditions, we have that
          $\add \notin \underline \lay(\Fn[{\km}])$.  This means that
          for each of these layouts, we can show:
          \begin{equation*}
            \step
            {\conf{\frame{\ccall  \expr {\exprtwo_1, \dots, \exprtwo_k}\sep \cmdtwo}{\regmap}{\km[\syscall]}:\st', \lay\lcomp \rfs''}}
            {\err},
          \end{equation*}
          and this means that (B) holds with probability 1.
        \end{proofcases}
        \proofcase{$\add\notin P$}
        %
        Observe that, for every layout $\lay$ such that
        $\add \notin \underline{\lay}(\Idk)$,
        only rule \ref{WL:Call-Error} applies, which shows the following transition:
        \begin{equation*}
          \step
          {\conf{\frame{\ccall  \expr {\exprtwo_1, \dots, \exprtwo_k}\sep \cmdtwo}{\regmap}{\km[\syscall]}:\st', \lay\lcomp \rfs''}}
          {\err},
        \end{equation*}
        Thus, we observe that:
        \begin{equation*}
          \Pr_{\lay \leftarrow\mu}\big[
          \step
          {\conf{\frame{\ccall  \expr {\exprtwo_1, \dots, \exprtwo_k}\sep \cmdtwo}{\regmap}{\km[\syscall]}:\st', \lay\lcomp \rfs''}}
          \err\,\, \big|\\
          \forall 1\le i\le h.\lay(\id_i)=\add_i\big]
        \end{equation*}
        is greater than
        \[
          \Pr_{\lay \leftarrow\mu}\big[ \add \notin \underline\lay(\Idk) \mid
          \forall 1\le i\le h.\lay(\id_i)=\add_i\big]
        \]
        which, by definition, is greater than $\delta_{\mu}$.  
        This shows that (B) holds.
      \end{proofcases}
    \end{proofcases}
  \end{proofcases}
\end{proof}

Observe that, in the proof of \Cref{lemma:onsyscall}, the case where
$\add \notin P$ also covers situations where memory is accessed via a
\emph{raw reference} during a system call, such as when dereferencing
a constant pointer. In this scenario, \Cref{lemma:onsyscall}
establishes claim (B), emphasizing that such practices should be
avoided in kernel code, as they are highly likely to result in memory
violations.

\begin{restatable}{lemma}{onsyscallterm}
  \label{lemma:onsyscallterm}
    Let $\syscall$ be a system call of a \emph{layout non-interfering} system $\system=(\rfs, \syss, \caps)$. Given an initial frame $\frame{\syss(\syscall)}{\regmap}{\km[\syscall]}$, and a store $\rfs'$ such that $\rfs'\eqon{\Fn}\rfs$,
  for every distribution of layouts $\mu$, one of the following statements holds:

  
  \begin{enumerate}[label={(\Alph*)}]
    \item For every layout $\lay$, we have
    \(
      \Eval{\syss(\syscall), \regmap, \rfs', \km[\syscall]} = (\overline \rfs, \overline v),
    \)
    for some $\overline \val$ and $\overline\rfs\eqon{\Fn}\rfs$.
    \item\label{case:lemma2b}
    \(
      \Pr_{\lay\leftarrow \mu}\left[\Eval{\syss(\syscall), \regmap, \km[\syscall], \rfs'} = \err\right] \ge \delta_{\mu, \system}.
    \)

    \item For every layout $\lay$, we have
    \(
      \Eval{\syss(\syscall), \regmap, \km[\syscall], \rfs'} = \Omega.
    \)
  \end{enumerate}
\end{restatable}

\begin{proof}[Proof Sketch of \Cref{lemma:onsyscallterm}]
  The proof is by case analysis: if for some layout $\lay$, $\Eval{\syss(\syscall), \regmap, \km[\syscall], \rfs'} \notin \{\err, \unsafe\}$, claims (A) or (C) must hold because of layout non-interference. On the other hand, if all layouts lead to $\err$ or to $\unsafe$, due to the finiteness of $\Lay$, there exists an upper bound $t$ on the number of steps needed to reach a terminal configuration. 
  We call $\refs(\syss(\syscall))\setminus \Sys=\{\id_1, \dots, \id_h\}$ and apply \Cref{lemma:onsyscall} with $\nat=t$. Observe that for every choice of $p_1, \ldots, p_h \in \Addk$, we can refuse conclusions (1) and (3) of that lemma because they are contradictory with $\forall \lay. \Eval{\syss(\syscall), \regmap, \km[\syscall], \rfs'} \in \{\err, \unsafe\}$, so (2) must hold.
  Observe that:
  \[
    \Pr_{\lay\leftarrow \mu}\left[\Eval{\syss(\syscall), \regmap, \km[\syscall], \rfs'} = \err\right] = 
    \Pr_{\lay \leftarrow \mu}\Big[ \exists \nat' \le \nat. \nstep {\nat'} {\conf{\frame{\syss(\syscall)}{\regmap}{\km[\syscall]}, \lay\lcomp \rfs'}}\err].
  \]
  The right-hand-side, in turn, is equal to:
    \begin{multline*}
    \sum_{\add_1, \ldots, \add_k}\Pr_{\lay \leftarrow \mu}\Big[ \exists \nat' \le \nat. \nstep {\nat'} {\conf{\frame{\syss(\syscall)}{\regmap}{\km[\syscall]}, \lay\lcomp \rfs'}}\err \,\,\Big|\,\,
    \forall 1\le i\le h.\lay(\id_i)=\add_i\Big]\cdot\\
    \Pr_{\lay \leftarrow \mu}\Big[\forall 1\le i\le h.\lay(\id_i)=\add_i\Big]. 
  \end{multline*}
  Since all the factors on the left
  are bounded by $\delta_{\mu, \system}$ by conclusion (2), their
  convex combination is also bounded, proving (B).
\end{proof}

We arrive now at the main result of this section, where we show the effectiveness of
layout randomization:

\begin{theorem}
  \label{thm:scenario1}
  Let $\system=(\rfs, \syss, \caps)$ be \emph{layout non-interfering}.
  Then, for any \emph{unprivileged} attacker $\cmd \in \Cmd$ and register map $\regmap$,
  \(
    \Prob{\lay \leftarrow \mu}\left[\Eval{\cmd, \regmap, \um, \rfs} =\unsafe\right] \leq 1 - \delta_{\mu, \system}
  \).
\end{theorem}

\begin{proof}
  The $\unsafe$ state can only be reached during the execution of a
  system call. \Cref{lemma:onsyscallterm} provides a lower bound on
  the probability of reaching $\err$, and thereby an upper bound to
  the probability of reaching $\unsafe$.
\end{proof}

\Cref{thm:scenario1}
extends the results of \cite{Abadi,Abadi2,Abadi3} by showing
that layout randomization guarantees \emph{kernel safety} probabilistically to
operating systems; in contrast with \cite{Abadi,Abadi2,Abadi3}, this holds even
when victim's code contains unsafe programming constructs,
such as arbitrary pointer arithmetic and indirect jumps. This is achieved by
replacing \Citet{Abadi}'s restrictions
on the syntax of victims with a \emph{weaker} dynamic property: \emph{layout
  non-interference}. 
Notice that the strength of the security guarantee provided by
\Cref{thm:scenario1} depends on the distribution of the layouts $\mu$.
Therefore, in practice, it is important to determine a randomization
scheme that provides a good bound.  This can be done quite easily: for
instance, if we assume that
(i)~$\kappa_{\km} > \sum_{\id \in \Id[\km]} \size \id$ and that
(ii)~$\theta \defsym \max_{\id \in \Id[\km]} (\size \id)$ divides
$\kappa_{\km}$, we can think of the kernel space address range as
divided in $\frac {\kappa_{\km}} \theta$ slots, each one large enough
to store any procedure or array.  In this setting, we can define the
distribution $\nu$ as the uniform distribution of all the layouts that
store each \emph{memory object} within a \emph{slot} starting from the
beginning of that slot.  For this simple scheme, we can approximate
the bound $\delta_\nu$ as the ratio between unallocated slots and all
the slots that do not store any object that is in the reference of a
system call:
\begin{equation*}
  \label{eq:deltanu}
  \delta_\nu \ge 
  \min_{\syscall \in \Sys}\frac{\kappa_{\km}/\theta-|{\Idk}|} {\kappa_{\km}/\theta-|{\refs_\system(\syss(\syscall))\setminus \Sys}|}.
\end{equation*}
In particular, the fraction in the right-hand side is the probability that
by choosing a slot that is not storing any object referenced by $\syscall$,
we end up with a fully unoccupied slot.
Observe that this lower-bound approaches
$1$ when $\kappa_{\km}$ goes to infinity.

\paragraph{Relation of kernel safety with security} Kernel safety
encompasses some form of \emph{spatial memory safety} and of
\emph{control flow integrity}, which are among the most critical
security properties of operating systems' kernels. This importance is
reflected in the numerous measures developed to enforce such
properties~\cite{kCFI,RustInLinux,towardkernelsafety,HyperSafe,FineGrainedkCFI,CompProtKer}.
Often, definitions of \emph{spatial memory safety} associate a
software component (a program, an instruction, or even a variable)
with a fixed memory area, that this component can access
rightfully~\cite{HighAssurance,SoftBound,MSWasm,PierceMS}. In this
realm, any load or store operation that does not fall within this area
is considered a violation of \emph{spatial memory safety}.
Our notion of \emph{kernel safety} encompasses a form of \emph{spatial
  memory safety}: if a system enjoys \emph{kernel safety}, then no
system call can access a memory region that does not appear within its
\tcaps.
In addition, \emph{kernel safety} also implies a form of \emph{control
  flow integrity}: a property which requires that the control transfer
operations performed by a program can reach only specific statically
determined targets~\cite{CFI}.  More precisely, in our semantics it is
unsafe to execute a procedure if its address does not belong to the
set of \tcaps of the current system call. This means that, if a system
$(\rfs, \syss, \caps)$ is \emph{kernel safe}, when executing a system
call $\syscall$, the control flow will remain across the procedures in
$\caps(\syscall)$.

\paragraph{Final remarks}
Kernel-space layout randomization provides \emph{kernel safety} in a
probabilistic sense. In particular, the probability of violating
safety depends on the randomization scheme, and for certain
randomization schemes, it approaches zero when the size of the address
space goes to infinity.  Finally, \emph{kernel safety} provides
essential security guarantees by enforcing both \emph{spatial memory
  safety}---restricting memory accesses to authorized regions---and
\emph{control flow integrity}---ensuring that control flow transfers occur
only to permitted procedures.


\section{Kernel Safety in the Speculative Threat Model}
\label{sec:safety2}
In this section we study \emph{kernel safety} in the presence of speculative attackers.
To this aim, in \Cref{subsec:specmodel}, we extend the model of
\Cref{sec:language} for this new scenario. More precisely,
we endow the semantics of \Cref{sec:language}
with speculative execution and side-channel observations
that reveal the accessed addresses and the value of conditional branches~\cite{HighAssurance,CTFundations,Spectector,Contracts}.
In \Cref{subsec:specsafety},
we refine the notion of \emph{kernel safety} for this model,
by defining \emph{speculative kernel safety}.

\subsection{The Speculative Execution Model}
\label{subsec:specmodel}

A popular way to model speculative attacks is by annotating
transitions with \emph{directives} that describe the choices made by
microarchitectural prediction units~\cite{HighAssurance,CTFundations,
  Spectector,Contracts}.  For instance, PHT speculation can be modeled
by the directives $\dbranch \ctrue$ and $\dbranch \cfalse$ that, in
the presence of a branching command, instruct the processor to
speculatively execute the true and the false branch respectively. In
this realm, reductions are driven by sequences of directives governing
each transition. Therefore, stating the absence of a
speculative attack boils down to stating the absence of a sequence of
directives that steer the execution to a bad event. This also
means that attacks are never explicitly represented as computational
objects, but only as sequences of possibly unrelated directives.

Also in our model, we use directives to guide predictors. In addition,
our approach allows attackers to be explicitly represented as
fully-fledged program that has primitives for steering speculative
execution and perform side-channel observations.
This permits us to naturally extend the notion of \emph{kernel safety}
to the new scenario. Besides, we believe that modeling an attack explicitly can be
interesting on its own.  The feasibility of an attack is witnessed explicitly
through a program. In this setting, for instance, assumptions on
the attacker's computational capabilities can be imposed seamlessly. 



\DD{Relate with crypto's CPA property??}

We extend the victim's language and semantics by introducing
directives and observations in \Cref{sec:internalsemantics} below, and
then we define the attacker's language in
\Cref{sec:adversarialsemantics}.

\subsubsection{Victim Language and Semantics}
\label{sec:internalsemantics}

The victims' language remains identical to the classic model except that,
we assume that load, branch and call instructions are tagged by unique labels $\lbl \in \Lbl$
in order to model the ability of attackers to influence the speculative
execution of specific instructions.
For example, a tagged load operation looks like $\cmemread[\lbl] \vx \expr$, and can be targeted only with directives with label $\lbl$.
We show labels alongside with commands only when they are relevant---e.g. the command is being executed speculatively---and they cannot be inferred from the context.

The speculative semantics is instrumented through \emph{directives}, modeling the choice made by
prediction units of the processor. Directives take the form:
\begin{cbnf}
  \dir \ni \Dir
  & \dbranch[\lbl] \bool
  \bnfmid \djump[\lbl] \add
  \bnfmid \dload[\lbl]{i}
  \bnfmid \dbt
  \bnfmid \dstep,
\end{cbnf}
where $i \in \Nat$, $\add\in \Add$,  and $\bool \in \cBool$.
The $\dbranch[\lbl] \bool$ directive causes a branch
instruction to be evaluated as if the guard resolved to
$b$, modeling PHT speculation. \diff{The $\djump[\lbl] \add$ directive causes a call
  instruction to execute the procedure stored at address $\add$ (if any),
  modeling BTB speculation.}
The directive $\dload[\lbl]{i}$ causes the load instruction to load the
$i$-th most recent value that is associated to the load address from the
(buffered) memory, modeling STL speculation. The $\dbt$ directive
directs speculations, either backtracking the most recent
mis-speculation or committing the microarchitectural state.
Finally, the $\dstep$ directive evaluates instructions
without engaging into speculation, in correspondence to the semantics
given in \Cref{sec:language}.

Transitions are also labeled with observations to model timing side-channel leakage. Observations are drawn from the following grammar:
\begin{cbnf}
  \obs, \obstwo \ni \Obs
  & \onone
  \bnfmid \obranch b
  \bnfmid \omem \add
  \bnfmid \ojump \add
  \bnfmid \obt \bool,
\end{cbnf}
where $b\in \cBool$ and $\add \in \Add$.
We use $\onone$ to label transitions that do not leak specific data.
The $\obranch \bool$ observation is caused
by branching instructions, with $\bool$ reflecting the taken branch.
The $\omem \add$ observation is caused by memory access, through loads or stores,
and contains the address of the accessed location, thus modeling data-cache leaks.
Likewise, the $\ojump \add$ observation is caused by calls to
procedures residing at address $\add$ in memory, modeling
instruction-cache leaks.
Finally, the $\obt \bool$ observation signals a
backtracking step during speculative execution.
Notice that we leak full addresses on memory accesses, and
the value of the branching instructions, i.e. we adopt
the \emph{baseline leakage model} that is
widely employed in the literature to model side-channel
info-leaks~\cite{CTPolicies, HighAssurance,
  CTFundations, SLNonInt}.

\ifdefined\conference{\input{rules2}}\fi
\ifdefined\arxiv{
\begin{figure*}[t]
  \small
  \centering
    \vspace{0.5ex}
    
  \resizebox{\textwidth}{!}{\(
    \Infer[SIE][SLoad][\textsc{SLoad-Load}]
    {\sstep
      {\confone \cons\cfstack}
      {\sframe{\frame{\cmd}{\update \regmap x \val}{\opt}\cons\st}{\bmvar}{\boolms\lor f}\cons
        \confone \cons \cfstack}  {{{\dload[\lbl]{i}}}} {\omem \add}}
    {\confone = \sframe{\frame{\cmemread[\lbl] \vx  \expr\sep\cmd}{\regmap}{\opt}\cons \st} {\bmvar}{\boolms} &
      \toAdd{\sem\expr_{\regmap, \lay}} = \add &
      \bufread {\bmvar} \add i =(\val, f) &
      \add \in \underline \lay(\Ar[\opt]) &
      \fbox{$\opt = \km[\syscall] \Rightarrow \add \in \underline \lay(\caps(\syscall))$}
    }
    \)}

  \[
    \Infer[SIE][SLoad-Error][\textsc{SLoad-Error}]
    {
      \lay \red
      \sframe{\frame{\cmemread[\lbl] \vx \expr\sep\cmd}{\regmap}{\opt}\cons \st}{\bmvar}{\boolms}\cons\cfstack
      \sto{\dir}{\onone}
      \sconf{\err, \boolms}\cons\cfstack
    }{
      \toAdd{\sem\expr_{\regmap, \lay}} = \add &
      \add \notin \underline \lay(\Ar[\opt]) &
      \dir = \dstep \lor \dir = \dload[\lbl]{i} &
    }
  \]

  \[
    \Infer[SIE][SLoad-Unsafe][\textsc{SLoad-Unsafe}]
    {\sstep
      {\sframe{\frame{\cmemread[\lbl] \vx \expr\sep\cmd}{\regmap}{\km[\syscall]}\cons \st}{\bmvar}{\boolms}\cons\cfstack}
      {{\unsafe}}   {\dir} {\omem \add}}
    {\toAdd{\sem\expr_{\regmap, \lay}} = \add &
      \add \in \underline \lay(\Ar[\km]) &
      \dir = \dstep \lor \dir = \dload[\lbl]{i} &
      \fbox{$\add \notin \underline \lay(\caps(\syscall))$} 
    }
  \]

  \[
    \Infer[SIE][Load-Step][\textsc{SLoad-Step}]
    {
      \lay \red \sframe{\frame{\cmemread \vx \expr\sep\cmd}{\regmap}{\opt}\cons\st}{\bmvar}{\boolms}\cons\cfstack
      \sto{\dstep}{\omem \add} \sframe{\frame{\cmd}{\update \regmap x \val}{\opt} \cons \st}{\bmvar}{\boolms}\cons\cfstack
    }
    {\toAdd{\sem\expr_{\regmap, \lay}} = \add &
      \bufread {\bmvar} \add 0 =\val, \bot &
      \add \in \underline \lay(\Ar[\opt]) &
      \fbox{$\opt = \km[\syscall] \Rightarrow \add \in \underline \lay(\caps(\syscall))$}
    }
  \]    

  \[
    \Infer[SIE][If-Branch][\textsc{SIf-Branch}]{
      \lay \red
      \specconfone\cons\cfstack
      \sto{\dbranch[\lbl]{d}}{\obranch{d}}
      \sframe{\frame{\cmdtwo_{d}\sep\cmdtwo}{\regmap}{\opt}\cons\st}{\bmvar}{\boolms\lor d\neq\toBool{\sem \expr_{\regmap, \lay}}}
      \cons \specconfone \cons \cfstack
    }
    {\specconfone=\sframe {\frame{\cif[\lbl] \expr {\cmdtwo_\ctrue} {\cmdtwo_\cfalse}\sep\cmdtwo}{\regmap}{\opt}\cons\st}{\bmvar}{\boolms}}
  \]

  \[
    \Infer[SIE][Backtrack-Top][\textsc{Bt}_{\top}]{
      \lay \red \specconfone \cons\cfstack \sto{\dbt}{\obt \top} \cfstack
    }{
      \specconfone = \sframe{\st}{\bmvar}{\top} \lor \specconfone = \sconf{\err,\top}
    }
    \quad
    \Infer[SIE][Backtrack-Bot][\textsc{Bt}_{\bot}]{
      \lay \red \specconfone \cons\cfstack \sto{\dbt}{\obt \bot} \specconfone \cons \nil
    }{
      \specconfone = \sframe{\st}{\bmvar}{\bot} \lor \specconfone = \sconf{\err,\bot}
      & \cfstack \neq \nil
    }
  \]
  \caption{Speculative semantics, excerpt.}
  \label{fig:scen2specsemiexcerpt}
\end{figure*}

\begin{figure*}[t]
  \small
  \centering
  \[
      \Infer[SIE][CallJmp][\textsc{SCall}]{
        \lay \red \confone\cons\cfstack
        \sto{\djump[\lbl] \add}{\ojump \add} \sframe{\frame{\mem(\add)} {\regmap_0[\vec \vx\upd\sem{\vec \exprtwo}_{\regmap,\lay}]}{\opt}\cons\frame{\cmd}{\regmap}{\opt} \cons \st}{\bm\buf \mem}{\boolms \lor (\add \neq \toAdd{\sem\expr_{\regmap, \lay}})}\cons \confone\cons \cfstack
      }
      {
        \confone = \sframe{\frame{\ccall[\lbl] \expr {\vec \exprtwo}\sep\cmd}{\regmap}{\opt}\cons\st}{\bm\buf \mem}{\boolms} &
        \add \in \underline \lay(\Fn[\opt]) &
        \fbox{$\opt = \km[\syscall] \Rightarrow \add \in \underline \lay(\caps(\syscall))$}
      }
    \]
    \\[-3mm]
    \[
      \Infer[SIE][CallJmp-Unsafe][\textsc{SCall-Unsafe}]{\sstep
        {\sframe{\frame{\ccall[\lbl] \expr{\vec \exprtwo}\sep\cmd}{\regmap}{\km[\syscall]}\cons\st}{\bmvar}{\boolms}\cons\cfstack}
        {\unsafe}
        {\djump[\lbl] \add}
        {\ojump \add}}
      { \add \in \underline \lay(\Fn[\km]) &
        \fbox{$\add \notin \underline \lay(\caps(\syscall))$}
      }
    \]
    \\[-3mm]
    \[
      \Infer[SIE][CallJmp-Error][\textsc{SCall-Error}]{
        \lay \red \confone\cons\cfstack
        \sto{\djump[\lbl] \add}{\onone} \sconf{\err,\boolms \lor (\add \neq \toAdd{\sem\expr_{\regmap, \lay}})}\cons \confone \cons \cfstack
      }
      {
        \confone = \sframe{\frame{\ccall[\lbl] \expr {\vec \exprtwo}\sep\cmd}{\regmap}{\opt}\cons\st}{\bmem}{\boolms} &
        \add \notin \underline \lay(\Fn[\opt])
      }
    \]
  \caption{Speculative semantics, speculative call instructions.}
  \label{fig:scen2specsemiexcerptscall}
\end{figure*}}\fi

A reduction step now takes the form
\[
  \lay \red \cfstack \sto{\dir}{\obs} \cfstack',
\]
indicating that for a given system $\system$, under layout $\lay \in \Lay$,
the system evolves from state $\cfstack$ with directive $\dir \in \Dir$ to $\cfstack'$ in one step,
producing the side-channel observation $\obs$. States  are described by the following BNF:

\begin{minipage}{0.2\linewidth}
  \begin{cbnf}
    \cfstack & \nil \bnfmid \specconfone \cons \cfstack
  \end{cbnf}
\end{minipage}
\begin{minipage}{0.4\linewidth}
  \begin{cbnf}
    \specconfone, \specconftwo & \sframe{\st}{\bmvar}{\boolms} \bnfmid \sconf
    {\err, \boolms} \bnfmid \unsafe
  \end{cbnf}
\end{minipage}
\begin{minipage}{0.15\linewidth}
  \begin{cbnf}
    \bmvar & \bm \buf \mem
  \end{cbnf}
\end{minipage}
\begin{minipage}{0.2\linewidth}
  \begin{cbnf}
    \buf & \nil \bnfmid \bitem \add \val \cons \buf
  \end{cbnf}
\end{minipage}

Here, $\cfstack$ is a stack of backtrackable configurations. During
the execution, the state of the system is described by the
configuration on top of the stack. To support backtracking, every time
a speculation occurs, the system pushes a new configuration to the stack
and keeps track of the current state in the configuration just below.

In a configuration $\sframe{\st}{\bmvar}{\boolms}$, $\st$ is a
call-stack as in \Cref{sec:language}, $\bmvar$ is a memory
$\mem$ equipped with a \emph{write buffer} $\buf$, and the Boolean
vale $\boolms$ is the \emph{mis-speculation} flag.
The mis-speculation flag governs backtracking of speculative
execution: when its value is $\top$, some previous speculation may be
incorrect, so a backtracking step will discard the configuration from
the stack.  
As errors are recoverable under mis-speculation,
error configurations carry a mis-speculation flag too.
As in \Cref{sec:language}, $\unsafe$ indicates
a safety violation.

\diff{Buffered
  memories are used to model STL speculation, and are based on those
  from~\cite{HighAssurance}.  } Writing a value $\val$ at address
$\add$ results in a delayed write $[\add \mapsto \val]$ that is
appended to the buffer $\buf$, in notation:
$\bm{[\add \mapsto \val]:\buf}{\mem}$.  Reading the $n$-th most recent
entry associated to $\add$ in a buffered memory is written
$\bufread {\bm{\buf}{\mem}} \add n$ and yields a value $\val$ together
with a boolean flag $f$ that is $\bot$ only if $\val$ is the
most recent value associated to address $\add$, as described in
\Cref{fig:lookup}.

\begin{figure}[t]
  \small
\begin{align*}
  \bufread {\bm{\nil}\mem} \add n &\defsym \mem(\add), \bot &  
  \bufread {\bm{\bitem {\add'} \val \cons \buf}\mem} \add n &\defsym
                                                              \begin{cases}
                                                                \bufread {\bm\buf\mem} \add n &\text{if }\add\neq\add'\\
                                                                \val', \top &\text{if }\add=\add', n>0\text{, and }\bufread {\bm\buf\mem} \add {n-1} = \val', f\\
                                                                \val, \bot &\text{if } \add=\add'\text{, and }n=0.
                                                           \end{cases}
\end{align*}
\normalsize
\vspace{-4mm}
%
\caption{Buffered memory lookup function.}
\label{fig:lookup}
\end{figure}

\ifdefined\conference{Some illustrative rules of the semantics are
    given in \Cref{fig:scen2specsemiexcerpt}, the complete set of
    rules is relegated to the extended version of this
    work~\cite{ExtendedJournal}.}\fi
\ifdefined\arxiv{Some illustrative rules of the semantics are
    given in \Cref{fig:scen2specsemiexcerpt}, the complete set of
    rules is relegated to~\Cref{sec:appsafety2}.}\fi\xspace
The rules for load instructions are very similar to the ones we give in
\Cref{sec:language}, but attackers can take advantage of
the store-to-load dependency speculation by issuing a $\dload[\lbl] i$
directive, as in Rule \ref{SIE:SLoad}. When this happens, the $i$-th most recent value
associated to the address $\add = \toAdd{\sem \expr_{\regmap, \lay}}$
is retrieved from the buffered memory.
When such value may not correspond to that of
the most recent store to the address $\add$, the flag $f$ is
$\top$. Depending on the value of $f$, Rule \ref{SIE:SLoad} may be
engaging mis-speculation and, for this reason, it keeps track of the
starting configuration in the stack and updates the mis-speculation
flag with $f$.  A successful load produces the observation
$\omem \add$ that leaks the address to the attacker.
The rules for erroneous and unsafe loads are
\ref{SIE:SLoad-Error} and \ref{SIE:SLoad-Unsafe}.
These rules are analogous to their non-speculative counterparts.
In our semantics, every command also supports the $\dstep$ directive,
which does not cause speculation. For instance,
Rule \ref{SIE:Load-Step} evaluates the $\cmemread \vx \expr$ command
by fetching the most recent value from the write buffer, instead of an
arbitrary one.


Branch instructions can be executed speculatively by issuing the
directive $\dbranch d$ by means of Rule \ref{SIE:If-Branch}.  This
directive
causes the evaluation to continue as if the guard resolved to
$d$. This operation leaks which branch is being executed by means of
the observation $\obranch d$. 

\diff{\ifdefined\conference{Rule \ref{SIE:CallJmp} models indirect
      jump speculations by accepting the directive $\djump \add$. When
      the call instruction is evaluated with such directive, the
      system invokes the procedure that is stored at address $\add$
      (if any).}\fi\ifdefined\arxiv{When a call instruction is
      encountered, the attacker can steer the control flow to any
      target function by supplying the directive $\djump \add$, as
      described by the rules in \Cref{fig:scen2specsemiexcerptscall}.
      Speculative call instructions are evaluated with Rule
      \ref{SIE:CallJmp}, which invokes the procedure that is stored at
      address $\add$ (if any).}\fi ~This operation leaks the
  observation $\ojump \add$, which reveals that some procedure is
  stored at address $\add$. To keep track of wrong speculation,
  the system checks whether the address supplied by the attacker corresponds to the intended jump location and updates the mis-speculation flag accordingly.
  The rules for erroneous and unsafe speculative call instructions
  \ref{SIE:CallJmp-Error} and \ref{SIE:CallJmp-Unsafe}
   are similar to their non-speculative counterparts,
  with the exception that the address supplied by the attacker is used to determine
  whether the access is erroneous or unsafe.
%
}

If the topmost configuration of a stack
carries the mis-speculation flag $\top$, some prior speculation may
have been incorrect, so
the configuration can be discarded with
Rule \ref{SIE:Backtrack-Top}.
If the mis-peculation flag is $\bot$, the current
state is not mis-speculating, so the whole stack
of book-kept configurations can be discarded with
Rule  \ref{SIE:Backtrack-Bot}.

We write $\sto{}{}^{*}$ for the multi-step reduction relation induced by $\sto{}{}$, i.e,
$\cfstack \sto{\nil}{\nil} \cfstack$ and $\cfstack \sto{\dir : \Ds}{\obs : \Os}^{*} \cfstack'$ if
$\cfstack \sto{\dir}{\obs} \circ \sto{\Ds}{\Os}^{*} \cfstack'$.

\begin{figure*}
  \small
  \vspace{-3mm}
  \centering
  \[
    \Infer[ALE][Poison]{
      \lay \red \aconf{\frame{\cpoison \dir\sep\speccmd}{\regmap}{\opt}\cons\st }{\mem}{\Ds}{\Os}
      \ato \aconf{\frame{\speccmd}{\regmap}{\opt}\cons\st}{\mem}{\dir:\Ds}{\Os}
    }{\strut}
  \]
  \\[-3mm]
  \[
    \Infer[ALE][Observe]{
      \lay \red \aconf{\frame{\vx \ass \eobs\sep\speccmd}{\regmap}{\opt}\cons\st }{\mem}{\Ds}{\obs:\Os}
      \ato \aconf{\frame{\speccmd}{\update \regmap\vx \obs}{\opt}\cons\st}{\mem}{\Ds}{\Os}
    }{}
  \]
  \\[-3mm]
  \[
    \Infer[ALE][Spec-Init]{
      \lay \red \aconf{\frame{\cspec\cmd\sep\speccmd}{\regmap}{\opt}\cons\st }{\mem}{\Ds}{\Os}
      \ato \hconf{\sframe{\frame{\cmd}{\regmap}{\opt}}{\bm{\nil} \mem}{\bot}}{\frame{\speccmd}{\regmap}{\opt}\cons\st}{\Ds}{\Os}
    }{}
  \]
  \\[-3mm]
  \[
    \Infer[ALE][Spec-D]{
      \lay \red \hconf{\cfstack}{\st}{\dir{\cons}\Ds}{\Os} \ato \hconf{\cfstack'}{\st}{\Ds}{\obs{\cons}\Os}
    }{
      \lay \red \cfstack \sto{\dir}{\obs} \cfstack'
    }
  \]
  \\[-3mm]
  \[
    \Infer[ALE][Spec-S]{
      \lay \red \hconf{\cfstack}{\st}{\Ds}{\Os} \ato \hconf{\cfstack'}{\st}{\Ds}{\obs{\cons}\Os}
    }{
      \nf \cfstack \Ds &
      \lay \red \cfstack \sto{\dstep}{\obs} \cfstack'
    }
    \quad
    \Infer[ALE][Spec-BT]{
      \lay \red \hconf{\cfstack}{\st}{\Ds}{\Os} \ato \hconf{\cfstack'}{\st}{\Ds}{\obs{\cons}\Os}
    }{
      \nf \cfstack \Ds &
      \nf \cfstack \dstep &
      \lay \red \cfstack \sto{\dbt}{\obs} \cfstack'
    }
  \]
  \\[-3mm]
  \[
    \Infer[ALE][Spec-Term]{
      \lay \red
      \hconf{\sframe{\frame{\cnil}{\regmap}{\opt}}{\bmvar}{\bot}}
      {\frame{\speccmd}{\regmap'}{\opt'}\cons\st}{\Ds}{\Os}
      \ato \aconf{\frame{\speccmd}{\regmap}{\opt}\cons\st}{\overline{\bmvar}}{\Ds}{\Os}
    }{}
  \]
  \\[-3mm]
  \[
    \Infer[ALE][Spec-Error]{
      \lay \red \hconf{\sconf{\err,\bot}} {\st}{\Ds}{\Os} \ato \err
    }{}
    \quad
    \Infer[ALE][Spec-Unsafe]{
      \lay \red \hconf{\unsafe} {\st}{\Ds}{\Os} \ato \unsafe
    }{}
  \]
  \vspace{-5mm}
  \caption{Semantics for speculative attackers, excerpt.}%
  \label{fig:scen2semiexcerpt}
\end{figure*}


\subsubsection{Attacker's Language and Semantics}
\label{sec:adversarialsemantics}

To give a definition of kernel safety w.r.t.\ speculative semantics,
we endow attackers with the ability to engage speculative execution,
to issue directives, and to read side-channel leaks.
To this end, we extend the instructions from \Cref{sec:language}
as follows:

\begin{minipage}[l]{0.6\linewidth}
\begin{bnf}
  \SpInstr \ni \specstat & \stat 
  \bnfmid \cspec \cmd 
  \bnfmid \cpoison \dir 
  \bnfmid \vx \ass \eobs &\\ 
\end{bnf}
\end{minipage}
\begin{minipage}[l]{0.3\linewidth}
\begin{bnf}
    \SpCmd \ni \speccmd & \cnil \bnfmid \specstat\sep\speccmd
\end{bnf}
\end{minipage}

The instruction $\cspec \cmd$ is used to execute a command $\cmd$ with
the speculative semantics defined just above.  The instruction
$\cpoison \dir$ models the attacker's ability to control
microarchitectural predictors and to control speculative execution by
issuing the directive $\dir$, which is used to control the evaluation
of instructions executed under speculative semantics.  Dual, the
instruction $\vx \ass \eobs$ is used to extract side-channel
info-leaks, collected during speculative execution of commands. To
model this operation, in the following, we assume
$\Obs \subseteq\Val$.  As an example, the snippet
\begin{equation}
  \label{spec-prob}
  \begin{aligned}
  \cblock{
    \cpoison{\dbranch[\lbl]{\top}}\sep
    \cinline{\cspec{\cwhen[\lbl]{\expr}{\csyscall{\syscall}{\add}}}} \sep
    \vx \ass \eobs
  }
  \end{aligned}
  \tag{\dag}
\end{equation}
forces the mis-speculative execution of $\csyscall{\syscall}{\add}$, independently of
the value of $\expr$. The register $\vx$ will hold the final observation leaked
through executing the system call.

The attacker's semantics is defined in terms of a relation
\[
  \lay \red \specconfone \ato \specconfone' .
\]
Configurations for the attacker's semantics are drawn from the following BNF:

\begin{minipage}[l]{0.5\linewidth}
  \begin{cbnf}
    \specconfone, \specconftwo
    & \aconf{\st}{\mem}{\Ds}{\Os}
    \bnfmid \hconf{\cfstack}{\st}{\Ds}{\Os}
    \bnfmid \err
    \bnfmid \unsafe.
  \end{cbnf}
\end{minipage}
\begin{minipage}[l]{0.25\linewidth}
  \begin{cbnf}
    \Os
    & \nil
    \bnfmid \obs\cons\Os
  \end{cbnf}
\end{minipage}
\begin{minipage}[l]{0.25\linewidth}
  \begin{cbnf}
    \Ds
    & \nil
    \bnfmid \dir\cons\Ds.
  \end{cbnf}
\end{minipage}

Within attacker's configurations $\st$ is the evaluation stack (as in
\Cref{sec:language}), and $\mem$ is the memory. The stacks $\Ds$ and
$\Os$ collect directives and observations respectively.  
We use hybrid configurations $\hconf{\cfstack}{\st}{\Ds}{\Os}$ to
model the execution of the command $\cspec \cmd$. In hybrid
configurations, the state is given by the stack of speculative
configurations $\cfstack$, and the frame stack $\st$ keeps tracks of
the procedure's activation records before the invocation of
$\cspec{\cmd}$. The stack $\Ds$ contains the directives used by the speculative
semantics of \Cref{sec:internalsemantics} to evaluate $\cfstack$. The stack $\Os$ collects the side-channel observations produced during speculative execution of $\cfstack$.

\Cref{fig:scen2semiexcerpt} shows the evaluation rules for the new
constructs.  Rules~\ref{ALE:Poison} and~\ref{ALE:Observe} define the
semantics for poisoning and side-channel observations, modeled by
respectively pushing and popping elements of the corresponding stacks.
Rule~\ref{ALE:Spec-Init} deals with the initialization $\cspec \cmd$
of speculative execution, starting from the corresponding initial
configuration of the victim $\cmd$ in an empty speculation context.  A
frame for the continuation of the attacker $\speccmd$ is left on the
call stack $F$.  This frame is used to resume execution of the
attacker, once the victim has been fully evaluated.  The victim itself
is evaluated via the speculative semantics through
Rules~\ref{ALE:Spec-D}, \ref{ALE:Spec-S} and \ref{ALE:Spec-BT}.  Note
how execution of the victim is directed through the directive stack
$\Ds$ (Rule~\ref{ALE:Spec-D}).  Should the current directive be
inapplicable, a non-speculative rewrite step (Rule~\ref{ALE:Spec-S})
or backtracking (Rule~\ref{ALE:Spec-BT}) is performed. Here, the
premise $\nf \cfstack \dir$ signifies that $\cfstack$ is irreducible
w.r.t.\ the directive $\dir$. Likewise, $\nf \cfstack \Ds$ means that
$\cfstack$ is irreducible w.r.t.\ the topmost directive of $\Ds$, or
that $\Ds$ is empty.
Also notice how side-channel leakage, modeled through observations, is
collected in the configuration via these rules.
Upon normal termination, resuming of evaluation of the attacker is
governed by Rule~\ref{ALE:Spec-Term} in the case of normal termination.
As a side-effect, this rule commits all buffered writes to memory,
as described by the following function:
\begin{align*}
  \overline {\bm{\nil}\mem} &\defsym \mem &  \overline {\bm{\bitem \add \val : \buf}\mem} &\defsym \update{\overline {\bm{\buf}\mem}}\add \val.
\end{align*}
Finally, Rules~\ref{ALE:Spec-Error} and~\ref{ALE:Spec-Unsafe} deal with abnormal termination.

Apart for the new constructs, the attacker executes under a semantics
analogous to the one given in \Cref{sec:language}.

\subsection{Speculative Kernel Safety}
\label{subsec:specsafety}
We are now ready to extend the definition of kernel safety (\Cref{def:ks})
to the speculative semantics.
\begin{definition}[Speculative kernel safety]\label{def:sks}
  We say that a system $\system = (\rfs,\syss,\caps)$ is \emph{speculative kernel safe}
  if for every unprivileged attacker $\speccmd \in \SpCmd$,
  every layout $\lay$,
  and register map $\regmap$, we have: 
  \[
    \lnot \left(\lay \red \aconf{\frame{\speccmd}{\regmap}{\um}}{\lay \lcomp \rfs}{\nil}{\nil} \ato^* \unsafe \right).
  \]
\end{definition}
It is important to note that this safety notion captures violations that
occur during transient execution. This is in line with what happens, for instance,
for Spectre and Meltdown~\cite{Spectre,Meltdown}, both exploiting unsafe
memory access under transient execution in order to
reveal confidential information.

\subsection{The Demise of Layout Randomization in the Spectre Era}
\label{sec:breakingaslr}

A direct consequence of \Cref{def:sks} is that every system that is
\emph{speculative kernel safe} is also \emph{kernel safe}. 
The inverse, of course, does not hold in general. Specifically,
the probabilistic form of safety provided by layout randomization
in \Cref{sec:safety1} does not scale to this extended threat model.
This happens for the following reasons:

\begin{varenumerate}
\item {\bf Side-channel leaks.} 
  %
  \emph{Layout non-interference} (\Cref{def:lni}) does not account for
  side-channel leakage. In practice, kernel-space gadgets that leak
  information on the kernel layout can be exploited by attacker to
  compromise layout randomization.  For instance, with the following
  gadget
  \[
    \label{eq:addleak}
    \cif{\fn = \add}{\cmd}{\cmdtwo},
    \tag{\ddag}
  \]
  an attacker who controls the value of $\add$ may gather information
  about the address of the kernel procedure $\fn$ by measuring the
  variations in the execution time of this gadget for different values
  of $\add$, assuming that the execution times of $\cmd$ and $\cmdtwo$
  are indeed different.  In our model, this form of leak is captured
  by leaking the guard of the conditional when
  executing \eqref{eq:addleak}.
    %
    %

\item {\bf Speculative execution.} With speculative execution,
  unsuccessful memory probes within transient executions do not lead
  to abnormal termination.  This fact undermines another fundamental
  assumption of \Cref{thm:scenario1} and the majority of studies
  demonstrating the efficacy of layout randomization
  (e.g.,~\cite{Abadi,Abadi2,Abadi3}), where memory access violations
  are not fully recoverable.
  As an example, consider the snippet \eqref{spec-prob} at page
  \pageref{spec-prob}. This program will never reach an unrecoverable
  error state even when the system call $\syscall$ loads an arbitrary
  address $\add$ form the memory.  If $\add$ is not allocated, the
  system performs a memory violation under transient execution that
  does not terminate the execution. If $\add$ is allocated, its
  content may be loaded into the cache (producing the observation
  $\omem \add$) before the execution backtracks.  By reading
  side-channel observations, an attacker can thus distinguish
  allocated kernel-addresses from those that are not allocated.  This
  last example, in particular, is not at all fictitious: the BlindSide
  attack~\cite{BlindSide} uses the same idea to break Linux's KASLR
  and locate the position of kernel's executable code and data. In our
  model, we account for a similar attacks by associating
  different observations to successful and unsuccessful memory
  operations (e.g., see Rules \ref{SIE:CallJmp} and
  \ref{SIE:CallJmp-Error}) and by allowing the backtrack of error states
  reached during transient execution.
\end{varenumerate}

\subsection{Side-channel Layout Non-Interference}
\label{sec:slni}

As this revised model significantly enhances the attackers' strength,
we need to implement more stringent countermeasures in order to
restore kernel safety.  To counter side-channel info-leaks, we can
impose a form of side-channel non-interference that is in line with
the notion of \emph{speculative constant-time}
from~\cite{CTFundations}, and that is meant to prevent side-channel
info-leaks from leaking information on the kernel's memory layout.

\begin{definition}[Side-channel layout non-interference]\label{def:slni}
  A system call $\syscall$ of a system $\system=(\rfs,\syss,\caps)$ is \emph{side-channel layout non-interferent} if for every $\rfs'\eqon\Fn \rfs$, directives $\Ds$, observations $\Os$, register map $\regmap$ and buffer $\mu$, we have:
  \[
    \exists \cfstack_1.\lay_1 \red \sframe{\frame{\syss(\syscall)}{\regmap}{\km[\syscall]}}{\bm{\buf}{(\lay_1 \lcomp \store')}}{\boolms} \sto{\Ds}{\Os}^* \cfstack_1
  \]
  implies
  \[
    \exists \cfstack_2.\lay_2 \red \sframe{\frame{\syss(\syscall)}{\regmap}{\km[\syscall]}}{\bm{\buf}{(\lay_2 \lcomp \store')}}{\boolms} \sto{\Ds}{\Os}^* \cfstack_2,
  \]
  for all layouts $\lay_{1},\lay_{2}$.
\end{definition}

 \emph{Side-channel layout non-interference} ensures the non-leakage of layout information throughout the side-channels by requiring the identity of the sequences of observations produced by the two reductions. However, it indirectly implies severe restrictions on memory interactions---effectively prohibiting non-static memory accesses, that are the main ingredient of layout randomization! Unsurprisingly, this form of non-interference directly establishes kernel safety of system calls:

\begin{lemma}
  \label{lemma:cttosafe}
  Suppose $\kappa_{\km}\ge \sum_{\id \in \Idk} \size\id + 2\cdot\max_{\id \in \Idk} \size\id$.
  Given $\system=(\rfs,\syss,\caps)$,
  if $\syscall$ is \emph{side-channel layout non-interfering} then
  \[
    \lnot \left(\lay \red \sframe{\frame{\syss(\syscall)}{\regmap}{\km[\syscall]}}{\bm{\buf}{(\lay \lcomp \store')}}{\boolms} \sto{\Ds}{\Os}^* \unsafe\right)
  \]
  for all layouts $\lay$ and initial configurations over stores $\store'$ coinciding with $\rfs$ on $\Fn$.
\end{lemma}

\begin{proof}[Proof Sketch]
  The condition on the size of the memory ensures that
  it is large enough to allow any array or procedure to be
  moved to a different location.  With this precondition, we just need
  to observe that if an invocation of a system call $\syscall$
  performs an unsafe memory access when executing under a layout
  $\lay$, the address $\add$ of the accessed resource is leaked; but
  the same address cannot leak if the resource is moved to another
  location, leaving $\add$ free. Therefore, if a system call is not speculatively safe, it
  cannot be side-channel layout non-interfering because different memory
  layouts produce different observations.
\end{proof}

Observe that our leakage model could be relaxed without invalidating
\Cref{lemma:cttosafe}, which is valid as long as the attacker can
distinguish a successful load, store or call operation form an
unsuccessful one.

The following result, stating that \emph{side-channel layout non-interference}
entails \emph{speculative kernel safety},
is a direct consequence of \Cref{lemma:cttosafe}.

\begin{theorem}
  \label{thm:scenario2}
  Under the assumption $\kappa_{\km}\ge \sum_{\id \in \Idk} \size\id + 2\cdot\max_{\id \in \Idk} \size\id$, if a system $\system$ is side-channel layout non-interfering, then it is \emph{speculative kernel safe}.
\end{theorem}
\begin{proof}[Proof Sketch]
  As for \Cref{thm:scenario1}, observe that the only way to reach the $\unsafe$ state is during the execution of a system call. \Cref{lemma:cttosafe} proves no system call can reach that the $\unsafe$ state under this theorem's assumptions. 
\end{proof}

\paragraph{Final remarks} In contrast with what happens in \Cref{thm:scenario1} of \Cref{sec:safety1}, the safety guarantee
provided by \Cref{thm:scenario2} is not probabilistic and does not rely on \emph{layout randomization}.
%
%
Therefore, although layout randomization is unlikely to be restored at
the software level without imposing \emph{side-channel layout
  non-interference}, in the presence of this assumption, layout
randomization is a redundant protection measure.

In general, non-interference properties do not necessarily entail
memory safety. 
In our case, \emph{side-channel layout non-interference} entails
\emph{speculative kernel safety} because the layouts are not only used
as the inputs for a computation, but they also determine \emph{where}
objects are placed in memory.

Notice that \emph{side-channel layout non-interference} is not a
necessary condition for \emph{speculative kernel safety}.  As an
example, take the system that only defines the array $\ar$ with size 1
and the system call $\syscall$ with body $\cmemass \ar \val$, and
capabilities $\{\ar\}$. This system is not \emph{side-channel layout
  non-interfering} because by executing $\syscall$ the address of
$\ar$ leaks, and this address depends on the layout.  However, this
system is speculatively safe because the only operation it can do is
writing to $\ar$ through $\syscall$, and $\ar$ is in the capabilities
of $\syscall$.

\section{Assessing the Effectiveness of Existing Mitigations for Speculative Kernel Safety}
\label{sec:discussion}
\renewcommand{\speccmdtwo}{\command{B}}

\begin{figure}[t]\centering
\[
  \begin{aligned}
    \speccmd\defsym
    \cblock{\vy \ass \kappa_{\um} ; \\
    \cdowhile {
    \cpoison {{\dbranch[\lbl] \ctrue}};\\
    \cspec{\csyscall {\syscall} {\cfalse, \vy}};\\
    \cdowhile {\vx \ass \eobs }{x \notin \{\ojump \_, \cnull\}} ;\\
    \vy \ass \vy + 1
    }{x \neq \ojump \_}
    }
  \end{aligned}
  \ 
  \begin{aligned}
    \speccmdtwo\defsym
    \cblock{\vy \ass \kappa_{\um} ; \\
    \cdowhile {
    \cpoison {\dload[\lbl] 1};\\
    \cspec{\csyscall {\syscalltwo} {\vy}}\\
    \cdowhile {\vx \ass \eobs}{x \notin \{\ojump \_, \cnull\}} ;\\
    \vy\ass \vy + 1
    }{x \neq \ojump \_}
    
    }
  \end{aligned}
  \ 
  \begin{aligned}
    \command{C} \defsym
    \cblock{\vy \ass \kappa_{\um} ; \\
    \cdowhile {
    \cpoison {\djump[\lbl] y};\\
    \cspec{\csyscall {\sym u}  {0}}\\
    \cdowhile {\vx \ass \eobs}{x \notin \{\ojump \_, \cnull\}} ;\\
    \vy\ass \vy + 1
    }{x \neq \ojump \_}
    
    }
  \end{aligned}
\]
\vspace{-4mm}
\caption{Attacks witnessing that $\system$ does not enjoy \emph{speculative kernel safety}.}
\label{fig:sksattacks}
\end{figure}

In this section we review state-of-art countermeasures
against speculative attacks to
assess to what extent they can
help to enforce \emph{speculative kernel safety}.


To this aim, we introduce a system $\system = (\rfs, \syss, \caps)$ that is
intentionally vulnerable to \emph{speculative kernel safety}
violations. The system includes a kernel function $\fn$ that simply
returns the constant $0$, and which is not included in the
capabilities of any system call. Additionally, the system defines the
system calls $\syscall$, $\syscalltwo$, and $\sym u$, that are
designed to trigger transient execution via PHT, STL, and BTB
speculation, respectively. These system calls can be used to mount
transient execution attacks that probe kernel memory during transient
execution and execute $\fn$. These system calls are
defined as follows:
\small
\begin{align*}
   \syss(\syscall) &= \cwhen[\lbl] {\vx_1} {\ccall {\vx_2} {}}
  &
    \syss(\syscalltwo) &= \cmemread \vz (\ar +\vx_1);  \cmemass \ar 1;\cmemread[\lbl] \vy \ar; \ccall {\vx_1} {}
  &
  \syss(\sym u) &= \ccall[\lbl] {\vx_1} {},
\end{align*}
\normalsize
 with  
$\caps(\syscall) =\caps(\sym u) = \emptyset$ and
$\caps(\syscalltwo) = \{\ar\}$.\footnote{In the system calls
  $\syscall$, $\syscalltwo$, and $\sym u$, the attacker has
  \emph{direct} control over the jump target, violating the design
  principle of not accepting kernel addresses as input for system
  calls~\cite[p. 414]{UnderstandingLinux}.  Although our system calls
  violate such principle, we use them to model the scenario where an
  attacker controls a jump target addresses, which is a common kernel
  vulnerability (see, e.g., CVE-2017-1000112, CVE-2017-7294, and
  CVE-2018-5332).  }

We describe in \Cref{fig:sksattacks} the attacks
$\speccmd, \speccmdtwo$ and $\command{C}$ that take advantage of
$\syscall$, $\syscalltwo$ and $\sym u$ respectively to break
\emph{speculative kernel safety}.  The attack $\speccmd$ is
reminiscent of BlindSide~\cite{BlindSide}. In the main loop, the
variable $y$ ranges over all the kernel addresses. At each iteration,
the attacks triggers the speculative execution of the conditional
instruction within the victim system call to probe the address $y$,
and it stops when it reads a side-channel leak $\ojump \add$,
revealing the address of $\fn$. When the attack reads $\ojump p$
within the inner loop, it exits both the loops and the address of
$\fn$ will be equal to $y-1$. In the inner loop, reading the value
$\cnull$ (which indicates an empty observation stack) implies that $y$
does not match the address of $\fn$, so the attack continues with the
next iteration.
The attack $\speccmdtwo$ is very similar to $\speccmd$, but it relies
on STL forwarding instead of branch direction speculation to engage
speculative execution and transiently execute $\ccall {\vx_1} {}$.
Finally, the attack $\command{C}$ takes
control of branch target speculation to probe the kernel's memory
address space for
$\fn$. 
%

\begin{table}[t]
  \centering
  \begin{tabular}{lccc}
    \toprule
    {\bf Protection mechanism} & {\bf  Attack $\speccmd$}& {\bf     Attack $\speccmdtwo$}  & {\bf Attack $\command{C}$} \\
    \midrule
    SLH and Eclipse   & \cmark & \xmark & \xmark   \\
    SCT   & \xmark    & \xmark  & \xmark  \\ 
    STT   & \xmark    & \xmark  & \xmark \\ 
    SPT   & \cmark     & \xmark & \cmark  \\ 
    NDA   & \xmark    & \xmark   & \xmark \\ 
    \textsc{ProSpeCT}   & \xmark    & \xmark  & \xmark  \\ 
    \textsc{SESES}   & \xmark    & \cmark  & \xmark \\ 
    SESES + \texttt{lvi-cfi}   & \cmark    & \cmark  & \xmark \\ 
    SESES + \texttt{lvi-cfi} + retpoline   & \cmark    & \cmark  & \cmark \\ 
    Our transformations (\Cref{sec:sksenforcement})   & \cmark  & \cmark  & \cmark  \\ 
    \bottomrule
  \end{tabular}
  \vspace{0.4em}
  \caption{Effectiveness of the protection mechanisms examined in \cref{sec:discussion,sec:sksenforcement} against the attacks in \Cref{fig:sksattacks}. }
  \label{tab:mitigations}
  \vspace{-2em}
\end{table}
\paragraph{SLH~\cite{SLH} and Eclipse~\cite{Eclipse}}
\emph{Speculative Load Hardening} (SLH) is a software-level
countermeasure that prevents speculative leaks of sensitive
information. It keeps track of a predicate indicating whether the
execution follows a wrong speculation or not. The value of this
predicate is used to instrument leaking operations---typically
loads---to obfuscate leaked values during transient execution.  In
addition, SLH can be used to obfuscate also loaded addresses and the
targets of indirect calls or jumps.  Later works like \emph{Strong
  SLH}~\cite{Exorcising} and \emph{UltimateSLH}~\cite{UltimateSLH} can
obfuscate even more operands, including store addresses and inputs of
time-variable arithmetical operations,
respectively. Eclipse~\cite{Eclipse} is a software-level protection
measure adopting a similar approach to obfuscate the operands of
instructions used by attacks like BlindSide~\cite{BlindSide} or
PACMAN~\cite{PACMAN} to perform speculative probes. Eclipse achieves
lower overhead compared to SLH by protecting only a specific set of
instructions---either indirect jumps or instructions for pointer
validation, depending on the type of speculative probing the user aims
to prevent.
However, Eclipse and SLH are specific to PHT speculation---related to
Spectre v1---and they fail to detect BTB and STL speculation---related
to Spectre v2 and v4, respectively.  Therefore, although they would be effective against attacks relying on PHT speculation
only, such as $\speccmd$, they would not stop $\speccmdtwo$ and
$\command{C}$ because they would not detect STL-dependency speculation
or branch target speculation.

\paragraph{SCT} The \emph{Constant-Time} (CT) discipline imposes that
program do not leak secret information via side-channels.
\emph{Speculative Constant Time} (SCT) extends (CT) by also taking into
account speculative execution~\cite{HighAssurance,CTFundations}.
Therefore, SCT programs are not subject to speculative attacks like
Spectre~\cite{Spectre}. However, SCT is not helpful to enforce
\emph{speculative kernel safety}, as leaking secrets speculatively
does not necessarily involve breaking \emph{speculative kernel
  safety}.  So, by imposing SCT, \emph{speculative kernel safety} is
not necessarily achieved. For instance, notice that none of the
attacks in~\Cref{fig:sksattacks} leaks secrets \emph{stored in
  memory}, although they break
\emph{speculatively kernel safety}.

In \Cref{sec:slni} we proposed a policy
based on SCT, to prevent (secret) information on the layout
to leak during speculative execution:
\emph{side-channel layout non-interference}.
It was also possible to show that
\emph{speculative kernel safety} can be obtained by imposing
\emph{side-channel layout non-interference}.
However, such property would be too restrictive
in practice, as any system call
interacting with the memory would violate it.

\paragraph{STT~\cite{STT}} \emph{Speculative Taint Tracking} (STT) is
a hardware protection mechanism to protect from speculative leakage by
preventing the propagation of transiently loaded data. However, STT
works by preventing the propagation only of data that is loaded during speculative
execution to subsequent instructions.
Therefore, it would neither block $\speccmd, \speccmdtwo$ or
$\command {C}$ from speculatively probing the victim's
memory. Specifically, STT would not prevent the
execution of the vulnerable call instruction, because the jump target
is never loaded during transient execution.

\paragraph{SPT~\cite{SPT}} \emph{Speculative Privacy Tracking} (SPT)
is a hardware protection mechanism based on STT. Instead of preventing
the propagation of \emph{speculatively loaded data}, SPT prevents the
propagation of \emph{speculative secrets}, i.e., of data that without
this measure would leak during speculative execution, but not during
normal execution.  To identify \emph{speculative secrets}, the
hardware keeps track of the data that is leaked in normal execution.
When the hardware detects that some data leaks in normal execution,
that data is also allowed to leak in speculative execution.

SPT would be effective against the attacks $\speccmd$ and $\command{C}$
because the addresses that are probed by the syscalls $\syscall$ and
$\syscalltwo$ do not leak in normal execution. However, it would not
be effective against attacks the attack $\speccmdtwo$ because the
probed address is leaked in normal execution inside $\syscalltwo$ by
the first load instruction before engaging speculative execution.

\paragraph{NDA~\cite{NDA}} \emph{Non-speculative Data Access} (NDA) is
a hardware protection measure similar to STT that restricts the
propagation of data during speculative execution. NDA supports different
policies to determine whether an instruction can broadcast its output,
however NDA is not effective against attacks leaking addresses via a
single transient operation~\cite{NDA}.
%
%
Therefore, it cannot prevent the speculative
leaks caused by a single instruction that tries to access a memory
address, such as the call instructions in the attacks we are considering.

\paragraph{\textsc{ProSpeCT}~\cite{ProSpeCT}} \textsc{ProSpeCT} is a
RISC-V processor prototype that adopts taint-tracking to prevent
speculative leaks of secret data stored in a specific address range. In
particular, an instruction cannot interact with the memory (performing
side channel leaks) if the register holding the target address depends
on secret data. The authors of \textsc{ProSpeCT} show that any CT
program running on \textsc{ProSpeCT} is SCT. Therefore
\textsc{ProSpeCT} does not enforce \emph{speculative kernel safety}
because, as we already mentioned, SCT programs are not necessarily
\emph{speculative kernel safe}.


\paragraph{SESES~\cite{SESES}} LLVM's Speculative Execution Side
Effect Suppression (SESES) aims at preventing the speculative leakage
of confidential data by placing \texttt{lfence} instructions before
load or store instructions and \emph{before} branch instructions.
This approach mitigates \emph{speculative} memory safety violations as
well as \emph{speculative }violations of constant-time execution.
However, the SESES mitigation itself does not prevent violations of
\emph{speculative kernel safety} by speculative probing with indirect
branch instructions, as it happens in our example and the
BHI~\cite{BHInjection} and BlindSide~\cite{BlindSide} attacks. When
used in combination with LLVM's \texttt{lvi-cfi} pass, indirect jumps
are replaced by a call to a thunk which places an additional
\texttt{lfence} before the jump. This means that SESES would be
effective in blocking attacks $\speccmd$ and $\speccmdtwo$. However,
even in combination with LLVM's \texttt{lvi-cfi} enforcement pass,
SESES would not prevent branch target speculation on the branch
instruction itself, as in $\command{C}$, where the attacker may still
be able to influence the jump target speculation in the jump inside
the thunk. Therefore, to achieve full protection, SESES needs to be
combined with \texttt{lvi-cfi}, and with some mechanism blocking BTB
speculation, such as Intel eIBRS or retpoline.

\paragraph{Final Remarks}
The countermeasures we just discussed are designed to prevent the
side-channel leakage of sensitive data during transient execution
at the hardware level. However, a violation of
\emph{speculative kernel safety} does not necessarily leak sensitive
data to the attacker. For instance, our attacks break
\emph{speculative kernel safety} but do not leak sensitive data, as
$\fn$ returns the constant 0, which is not a secret.

We believe that \emph{speculative kernel safety} could be enforced at
hardware level by delaying all the operations that interact with
kernel's memory until all the ongoing speculations are resolved as
correct. Notice that this is not what taint-tracking mechanisms such
as STT, SPT, NDA, and \textsc{ProSpeCT} do. For instance, by tainting
the content of kernel memory,
\textsc{ProSpeCT} would only block accesses to locations whose address
depends on kernel data, rather than blocking all speculative accesses
to kernel's memory.
%
This consideration motivates us to develop an alternative
software-based technique specifically aimed at enforcing
\emph{speculative kernel safety}.





\section{Enforcement of Speculative Kernel Safety}
\label{sec:sksenforcement}
Our last research goal is to define software-level protection
mechanism by which we can enforce \emph{speculative kernel safety} on
a kernel. 
Instead of defining a mechanism that enforces
\emph{speculative kernel safety} on a system from scratch, our goal is
rather to nullify the gap between \emph{kernel safety} and
\emph{speculative kernel} safety, by making the latter property a
consequence of the former.
This approach is similar to the one adopted by the
countermeasures of \Cref{sec:discussion}, that are
ultimately aimed to prevent undesirable events to take place under
transient execution. With this guarantee, any software that is secure
in the classic execution model becomes safe in the speculative one as
well.



%

In terms of our model, our goal boils down to finding a transformation
$\zeta$ that turns any \emph{kernel safe} system $\system$ into
another system $\zeta(\system)$ which is secure during speculative
execution but that cannot be distinguished from the first by users
that do not control speculative execution and side channel leaks.
The latter requirement is expressed by the notion of \emph{semantic
  equivalence} in \Cref{def:sempres} below.

\begin{definition}
  \label{def:sempres}
  Two systems $\system=(\rfs, \syss, \caps)$ and
  $\system'=(\rfs', \syss', \caps')$ are \emph{semantically
    equivalent} if
  \[
    \Eval[\system][\lay]{\cmd, \regmap, \um, \rfs'} \simeq
    \Eval[\system'][\lay]{\cmd, \regmap, \um, \rfs}
  \]
  for every
  unprivileged $\cmd\in \Cmd$.  Here, the equivalence is given by
  $(v,\rfs[1]) \simeq (v,\rfs[2])$ if
  $\rfs[1] \eqon{\Idu\cup \Ark} \rfs[2]$, and coincides with equality
  otherwise.
\end{definition}

In the previous definition we require
$\rfs \eqon {\Idu\cup\Ark} \rfs'$ instead of the identity between the
stores in order to allow reasoning about systems with different
kernel-space procedures.  In turn, to guarantee \emph{semantic
  equivalence}, it is sufficient to impose that the transformation
does not change the behavior of the single system calls, as expressed
by the notion of \emph{system call semantics preservation} below.

\begin{definition}
  \label{def:sempres}
  Let $\system=(\rfs, \syss, \caps)$ and
  $\systrans(\system) = (\rfs', \syss', \caps')$ be two systems.
  The transformation $\systrans$ is \emph{system call
    semantics preserving} if we have $\rfs \eqon {\Idu\cup \Ark} \rfs'$ and
  for every layout $\lay$, registers
  $\regmap$,
  $\syscall$, and stores
  $\rfstwo, \rfstwo'$:
  \[
  \label{eq:sempres}
  \rfstwo \eqon \Fnk\rfs \land
  \rfstwo' \eqon \Fnk\rfs' \land
  \rfstwo \eqon \Ark \rfstwo' \quad\text{ implies }\quad
    \Eval[\system][\lay]{\syss(\syscall), \regmap, \km[\syscall], \rfstwo}
    \simeq
    \Eval[\systrans(\system)][\lay]{\syss'(\syscall), \regmap, \km[\syscall], \rfstwo'}.
  \]
\end{definition}

Reducing \emph{speculative kernel safety}
to \emph{kernel safety} means reducing safety violations in speculative
execution to safety violations in normal execution.
Therefore, the second requirement on $\zeta$
can be expressed by asking that the system $\zeta(\system)$
can violate \emph{speculative kernel safety} only if it violates \emph{kernel safety}, as
captured by \Cref{def:cbu} below.

\begin{definition}
  \label{def:cbu}
  We say that $\zeta$  \emph{imposes speculative kernel safety}
  if, for every system $\system$ such that $\systrans(\system)=(\rfs, \syss, \caps)$,
  every buffer $\buf$ with $\dom(\buf)\subseteq\underline \lay(\Ar)$
  and store $\rfs' \eqon{\Fun} \rfs$ we have:
  \ifdefined\conference{
      \[
        \left(
          \lay \red[\systrans(\system)] \conf{\conf{\syss(\syscall), \regmap, \km[\syscall]}, \bm\buf{(\lay \lcomp \rfs')}, \boolms} \sto{\Ds}{\Os}^* \unsafe
        \right)
        \ \text{implies}\  
        \left(
          \lay \red[\systrans(\system)] \conf{\conf{\syss(\syscall), \regmap, \km[\syscall]}, \overline{\bm\buf{(\lay \lcomp \rfs')}}}
          \to^* \unsafe
        \right).
      \]
    }\fi
  \ifdefined\arxiv{
      \begin{multline*}
        \left(
          \lay \red[\systrans(\system)] \conf{\conf{\syss(\syscall), \regmap, \km[\syscall]}, \bm\buf{(\lay \lcomp \rfs')}, \boolms} \sto{\Ds}{\Os}^* \unsafe
        \right)
        \Rightarrow\\
        \left(
          \lay \red[\systrans(\system)] \conf{\conf{\syss(\syscall), \regmap, \km[\syscall]}, \overline{\bm\buf{(\lay \lcomp \rfs')}}}
          \to^* \unsafe
        \right).
  \end{multline*}
}\fi
\end{definition}


\DD{add a reference to table 1 and an informal explanation of the theorem}

We now state the main result of this section, that establishes the effectiveness of our mitigation strategy.

\begin{theorem}
  \label{prop:mitigation}
  If a system $\system$ is \emph{kernel-safe},
  and the transformation $\zeta$  (i) \emph{imposes speculative kernel
    safety} and (ii) is \emph{system call semantics preserving}, then
  (a) $\zeta(\system)$ is \emph{speculative kernel safe},
  and (b) $\zeta(\system)$ is semantically equivalent to $\system$.
\end{theorem}

\begin{proof}[Proof Sketch]
  Claim (a) follows by contraposition: assume that the system $\zeta(\system)$ is unsafe, so there are a system call, and a sequence of transition that lead to $\unsafe$ with the specualtive semantics. By combining \Cref{def:cbu} and (ii), we deduce that the system $\system$ is unsafe under the non-speculative semantics, contradicting $\system$'s safety. Claim (b) follows by induction on the length of the reduction, and by applying (ii) when a system call is encountered.
\end{proof}

Notice that \emph{kernel safety} cannot be
provided solely by the adoption of layout randomization: by
\Cref{thm:scenario1}, we know that layout randomization only provides
\emph{kernel safety} modulo a small probability of failure.

\diff{
  With \Cref{prop:mitigation} in mind, our next goal is to show
  some candidate transformations that fulfill its requirements
  and which enforce
  speculative kernel safety on classically safe kernels.
  To model these transformations, we extend our language with program instructions acting
  as speculation barriers---modeling the \texttt{lfence} instruction found in
  modern CPUs~\cite{IntelManual}---and with call instructions that do not support
  BTB speculation---similar to retpoline thunks~\cite{Retpoline}.

  The resulting language looks as follows:
  \begin{cbnf}
    \Instr \ni \stat & \ldots \bnfmid \cfence \bnfmid \cscall \expr {\expr_1, \ldots, \expr_n}.
  \end{cbnf}

}

The semantics of these additional instructions is in
\Cref{fig:scen2specsemifence}.
Architecturally, the $\cfence$ instruction is a no-op,
on the micro-architecture level
it commits all buffered writes to memory.
In particular, for consistency, a potentially mis-speculative state must
be resolved. This is why this rule requires the mis-speculation flag to
be $\bot$. This means that, if this configuration is reached
when the flag is $\top$, the semantics must backtrack
with Rule \ref{SIE:Backtrack-Top}.

\diff{ 
  The semantics of the $\cscall \expr {\vec\exprtwo}$ instruction is
  similar to the for ordinary call instructions, with the difference
  that $\djump[\lbl] \add$ directives are not accepted. Instead, the
  address of the target procedure is always obtained by evaluating
  $\expr$. In normal execution, the \emph{non-speculative} call
  instruction is interpreted in the same way as the ordinary call
  instruction.
}

To facilitate proving the \emph{system call semantics preservation}
condition, we make the following remark:

\begin{remark}
  \label{rem:seq}
  If a transformation $\zeta$ only inserts $\cfence$ instructions (at
  any positions) in the kernel code and replaces $\kwd{call}$
  instruction with safe calls, then $\zeta$ is
  \emph{system call semantic preserving}.
\end{remark}

\begin{figure*}
  \small
  \[
    \Infer[SIE][SCall][\textsc{SSCall}]{
      \lay \red \confone\cons\cfstack
      \sto{\dstep}{\ojump \add} \sframe{\frame{\mem(\add)} {\regmap_0[\vec \vx\upd\sem{\vec \exprtwo}_{\regmap,\lay}]}{\opt}\cons\frame{\cmd}{\regmap}{\opt} \cons \st}{\bmem}{\boolms }\cons \cfstack
    }
    {
      \confone = \sframe{\frame{\cscall \expr {\vec \exprtwo}\sep\cmd}{\regmap}{\opt}\cons\st}{\bmem}{\boolms} &
      \toAdd{\sem\expr_{\regmap, \lay}} = \add &
      \add \in \underline \lay(\Fn[\opt]) &
      \fbox{$\opt = \km[\syscall] \Rightarrow \add \in \underline \lay(\caps(\syscall))$}
    }
  \]
  \\[-3mm]
  \[
    \Infer[SIE][SCall-Unsafe][\textsc{SSCall-Unsafe}]{\sstep
      {\sframe{\frame{\cscall \expr{\vec \exprtwo}\sep\cmd}{\regmap}{\km[\syscall]}\cons\st}{\bmvar}{\boolms}\cons\cfstack}
      {\unsafe}
      {\dstep}
      {\ojump \add}}
    {
      \toAdd{\sem\expr_{\regmap, \lay}} = \add &
      \add \in \underline \lay(\Fn[\km]) &
      \fbox{$\add \notin \underline \lay(\caps(\syscall))$}
    }
  \]
  \\[-3mm]
  \[
    \Infer[SIE][SCall-Err][\textsc{SSCall-Error}]{
      \lay \red \confone\cons\cfstack
      \sto{\dstep}{\onone} \sconf{\err,\boolms}\cons \cfstack
    }
    {
      \confone = \sframe{\frame{\cscall \expr {\vec \exprtwo}\sep\cmd}{\regmap}{\opt}\cons\st}{\bmvar}{\boolms} &
      \toAdd{\sem\expr_{\regmap, \lay}}\notin \underline \lay(\Fn[\opt])
    }
  \]
  \\[-3mm]
  \[
    \Infer[SIE][Fence]{
      \lay \red
      \sframe{\frame{\cfence\sep\cmd}{\regmap}{\opt} \cons\st}{\bmvar}{\bot} \cons \cfstack
      \sto{\dstep}{\onone} \sframe{\frame{\cmd}{\regmap}{\opt}\cons\st}{\overline {\bmvar}}{\bot} \cons \cfstack
    }{}
  \]
  \vspace{-5mm}
  \caption{Speculative semantics, additional constructs.}
  \label{fig:scen2specsemifence}
\end{figure*}

    


\begin{figure}[t]
  \centering
  \small
  \begin{align*}
    \fencetrans( \vx \ass \expr)&\defsym \vx \ass \expr & 
    \fencetrans( \cskip)&\defsym \cskip\\
    \fencetrans( \cmemass \expr\exprtwo)&\defsym \cfence; \cmemass \expr\exprtwo & 
    \fencetrans( \cmemread \vx\expr)&\defsym\cfence\sep\cmemread \vx\expr\\ 
    \fencetrans( \ccall \expr {\vec \exprtwo})&\defsym  \cfence\sep \cscall \expr {\vec \exprtwo} & 
    \fencetrans( \cwhile{\expr}{\cmd})& \defsym \cwhile{\expr}{\fencetrans( \cmd)} \\
    \fencetrans( \csyscall \syscall {\vec \expr})&\defsym  \csyscall \syscall {\vec \expr} & 
    \fencetrans( \cif{\expr}{\cmd}{\cmdtwo})& \defsym\cif{\expr}{\fencetrans( \cmd)}{\fencetrans( \cmdtwo)}  \\
    \fencetrans( \stat\sep \cmd)&\defsym\fencetrans(\stat) ;\fencetrans(\cmd) &
    \fencetrans (\cnil) &\defsym \cnil.
  \end{align*}
  \vspace{-1em}
  \begin{align*}
    \fencetrans (\syss) &\defsym \syscall \mapsto\fencetrans (\syss(\syscall)) & 
    \fencetrans (\rfs) &\defsym  \id \mapsto
    \begin{cases}
      \rfs(\id) & \text{if } \id \in \Arr \cup \Funu\\
      \fencetrans (\rfs(\id)) & \text{otherwise }
    \end{cases} & 
    \fencetrans((\rfs, \syss, \caps)) &\defsym (\fencetrans (\rfs), \fencetrans(\syss), \caps).
  \end{align*}
  \vspace{-4mm}
  \caption{Simple fencing transformation}
  \label{fig:fencetrans}
\end{figure}

\paragraph{Simple Fencing Transformation}
\label{sec:fencetrans}
As a warm-up example we now consider the simple transformation
$\fencetrans$ shown in \Cref{fig:fencetrans}, and prove that it
satisfies the requirements of~\Cref{prop:mitigation}.  Similarly to
SESES~\cite{SESES}, this transformation places a $\cfence$ instruction
before \emph{all} the potentially unsafe operations, including call
instructions, as they can perform transitions to $\unsafe$. In
addition to what SESES does, it prevents BTB speculation by rewriting
ordinary call instructions with non-speculating call instructions.
With $\fencetrans$, any ongoing speculation is stopped before
executing potentially unsafe operations (including non-speculative
call instructions), and their transient execution is prevented, yet
leaving the program's semantics unaltered at the architectural level.

Finally, note that the transformation $\fencetrans$, does not
stop speculation completely, instead it only prevents the transient
execution of those instructions that can perform unsafe operations.
For instance, conditional instructions can still execute speculatively
if their branches contain no potentially unsafe operations. This
is not in contrast with \Cref{def:cbu} because even in transient
execution, a conditional instruction cannot perform any safety
violation by itself. However, as their branches can contain unsafe
operations, the transformation visits them to protect any unsafe
operation therein.
By observing that $\fencetrans$ enjoys both the properties in
\Cref{def:sempres,def:cbu}, we can draw the following conclusion:

\begin{lemma}
  \label{thm:mitigation}
  If a system $\system$ is \emph{kernel-safe}, then
  $\fencetrans(\system)$ is \emph{speculative kernel safe}
  and semantically equivalent to $\system$.
\end{lemma}
\begin{proof}[Proof Sketch]
  Fix $\fencetrans(\system) = (\rfs, \syss, \caps)$. By induction on
  $\nat$, we show that when a system call reaches $\unsafe$ with the
  speculative semantics, i.e., when
  \[
    \lay \red[\fencetrans(\system)] \conf{\frame{\syss(\syscall)}{
        \regmap}{\km[\syscall]}, \bm\buf{(\lay \lcomp \rfs')}, \boolms}
    \sto{\Ds}{\Os}^n {\conf{\st, \bm {\buf'} {\mem}, \boolms'} \cons
      \cfstack} \sto \dir \obs \unsafe,
  \]
  $\boolms'$ is
  $\bot$. Moreover, we can assume without loss of generality that
  $\Ds$ does not contain $\dbt$ directives. By induction on the
  sequence of steps, and by inspecting the transition rules not
  involving $\dbt$, we observe that at each step the mis-speculation
  flag is $\bot$. Therefore, we deduce that the reduction can be
  mimicked by the non-speculative semantics, i.e.,
  \( \lay \red[\fencetrans(\system)] \conf{\frame{\syss(\syscall)}{
      \regmap}{\km[\syscall]}, \overline {\bm\buf{(\lay \lcomp
        \rfs')}}} \to^n {\conf{\st, \overline {\bm {\buf'} {\mem'}}}}\to \unsafe \).
  This proves that
  $\fencetrans$ imposes \emph{speculative kernel
    safety}. \emph{System call semantics preservation} follows from \Cref{rem:seq}.
\end{proof}

\begin{figure*}[t]
  \small
  \begin{align*}
    \optfencetrans[m][e]( \vx \ass \expr)&\defsym (\vx \ass \expr, m, e) \\
      \optfencetrans[m][e]( \cskip)&\defsym (\cskip, m, e)\\
      \optfencetrans[m][e]( \cmemread \vx\expr)&\defsym
                                               \begin{cases}
                                                 (\cmemread \vx\expr, \lnot e, e)& \text{if }m = \bot\\
                                                 (\cfence\sep\cmemread \vx\expr, \bot, \top)& \text{otherwise }
                                               \end{cases}\\ 
    \optfencetrans[m][e]( \cmemass \expr\exprtwo)&\defsym
                                                \begin{cases}
                                                  (\cmemass \expr\exprtwo, \bot, \bot) & \text{if }m = \bot\\
                                                  (\cfence; \cmemass \expr\exprtwo, \bot, \bot) & \text{otherwise }
                                                \end{cases} \\
    \optfencetrans[m][e]( \cwhile{\expr}{\cmd})& \defsym  (\cwhile{\expr}{{\optfencetrans[\top][\bot]( \cmd)}}, \top, \bot) \\
    \optfencetrans[m][e]( \cif{\expr}{\cmd}{\cmdtwo})& \defsym (\cif{\expr}{{\optfencetrans[\top][\bot]( \cmd) }}{{\optfencetrans[\top][\bot]( \cmdtwo) }}, \top, \bot) \\
          \optfencetrans[m][e]( \csyscall \syscall {\vec\expr})&\defsym  (\csyscall \syscall {\vec\expr}, \top, \bot) \\
      \optfencetrans[m][e]( \ccall \expr {\vec\exprtwo})&\defsym  (\cfence\sep \cscall \expr {\vec\exprtwo}, \top, \bot) \\
      \optfencetrans[m][e]( \cnil)&\defsym (\cnil,m, e)\\
    \optfencetrans[m][e]( \stat\sep \cmd)&\defsym \parbox[t]{18em}{ $(\stat';\cmd', m'', e'')$\ \emph{where} $(\stat', m', e')=  {\optfencetrans[m][e](\stat)}$,\\ \emph{and} $(\cmd', m'', e'') = \optfencetrans[m'][e'](\cmd)$.}
  \end{align*}
  \begin{align*}
    \optfencetrans (\syss) &\defsym \syscall \mapsto \optfencetrans[\top][\bot] (\syss(\syscall))&
    \optfencetrans (\rfs) &\defsym  \id \mapsto
    \begin{cases}
      \rfs(\id) & \text{if } \id \in \Arr \cup \Funu\\
      \optfencetrans[\top][\bot] (\rfs(\id)) & \text{otherwise }
    \end{cases}&
    \optfencetrans((\rfs, \syss, \caps)) &\defsym (\optfencetrans (\rfs), \optfencetrans(\syss), \caps).
  \end{align*}
  \vspace{-4mm}
  \caption{Optimized fencing transformation. When there is no risk of ambiguity, we use the notation $\optfencetrans (\rfs(\id))$ to represent the resulting command alone, rather than the full triple.}
  \label{fig:optfencetrans}
\end{figure*}

\paragraph{Optimized fencing transformation}
\label{sec:optfencetrans}
The transformation $\fencetrans$ of \Cref{fig:fencetrans} can be
enhanced by means of a simple static analysis to determine whether the
current instruction can be reached in transient execution or not.
The resulting transformation $\optfencetrans$ is described in
\Cref{fig:optfencetrans}. Specifically, while analyzing the AST of a
program, $\optfencetrans$ keeps track of two additional flags: $m$ and
$e$. The flag $m$ is $\top$ if the instruction may be reached after a
mis-speculation, and the flag $e$ is $\top$ if the instruction can
only be reached when the write-buffer is empty.  Therefore, when the
transformation encounters an instruction that may lead to the
$\unsafe$ state, but the flag $m$ is $\bot$, the $\cfence$ instruction
can be avoided.

For instance, consider the following snippet of code:
\[
  \cmd \defsym \cmemread \vx \expr ; \cmemread \vy \exprtwo ; \cmemass {\vz} {\mathtt{G}}. \qquad \text{with} \qquad
  \optfencetrans[\top][\bot](\cmd) =   (\cfence ; \cmemread \vx \expr ;  \cmemread \vy \exprtwo;  \cmemass \vz {\mathtt{G}}, \bot, \bot).
\]
when $\optfencetrans$ starts its execution, the initial value of $m$
is $\top$ and that of $e$ is $\bot$, meaning that the system may be
mis-speculating and the write buffer is not necessarily
empty. Therefore, $\optfencetrans$ protects the first load instruction
in $\cmd$ with a $\cfence$ instruction and sets $m$ to $\bot$ and $e$
to $\top$. When it instruments the instruction
$\cmemread \vy \exprtwo$, no $\cfence$ instruction is placed because
$m=\bot$. The value of $m$ is then set to $\lnot e = \bot$ since,
due to the empty buffer, no mis-speculation can occur (see
\Cref{fig:optfencetrans}, case of load operations with $m = \bot$ and
$e = \top$).  Notice that after analyzing the last store instruction
the transformation sets the flag $e$ to $\bot$, signaling that the
write buffer may be not empty anymore, and that subsequent unprotected
load instruction may be mis-speculating (see \Cref{fig:optfencetrans},
case of load operations with $m=\bot$).
Besides the load and store instruction we just analyzed, the
transformation sets the flag $m$ to $\top$ when it explores the
branches of $\kwd{if}$ or $\kwd{while}$ constructs to capture PHT
speculation. Finally, the flag $m$ is set to $\top$ and $e$ to $\bot$
when returning from an indirect call, due to the lack of
information on the called function.

Our optimizations do not impact the security of the transformation $\optfencetrans$,
as expressed by the following lemma:
\begin{lemma}
  \label{thm:optmitigation}
  If a system $\system$ is \emph{kernel-safe}, then
  $\optfencetrans(\system)$ is \emph{speculative kernel safe}
  and $\optfencetrans(\system)$ semantically equivalent to $\system$.
\end{lemma}
\begin{proof}[Proof Sketch] 
The proof follows a similar structure to that of \Cref{thm:mitigation}. However, to establish that whenever an unsafe state is reached, the mis-speculation flag in the predecessor state is $\bot$, we rely on a form of subject reduction and of soundness of our static analysis. Specifically,
    whenever the execution reaches a command $\cmd$, either $\cmd$ is instrumented, namely:    
    \(
      \cmd = \optfencetrans[m][e](\cmd'); \optfencetrans[\top][\bot](\cmdtwo_1);\ldots; \optfencetrans[\top][\bot](\cmdtwo_k),
    \) 
    for some $\cmd', \cmdtwo_1, \ldots, \cmdtwo_k$, $m$, and $e$, or $\cmd$ reaches an instrumented command in one step. In both cases, if $m$ is $\bot$, the mis-speculation flag is also $\bot$ and if $e$ is $\top$, the write buffer is empty.


This property can be established by induction on $\nat$, with the help of other technical invariant properties.
\end{proof}

\begin{figure*}[t]
  \small
  \begin{align*}
    \nospec( \vx \ass \expr)&\defsym \vx \ass \expr \\
    \nospec( \cskip)&\defsym \cskip \\
    \nospec( \cmemass \expr\exprtwo)&\defsym \cmemass \expr\exprtwo; \cfence\\
    \nospec( \cmemread \vx\expr)&\defsym \cmemread \vx\expr\\ 
    \nospec( \cif{\expr}{\cmd}{\cmdtwo}, m, e)& \defsym \cif{\expr}{\cfence ; \cmd}{\cfence ; \cmdtwo}\\
    \nospec( \ccall \expr {\vec\exprtwo})&\defsym\cscall \expr {\vec\exprtwo}\\
    \nospec( \csyscall \syscall {\vec \exprtwo})&\defsym\csyscall \syscall {\vec\exprtwo}\\
    \nospec( \cwhile{\expr}{\cmd})& \defsym \cwhile{\expr}{ \cfence; \cmd };\cfence\\
    \nospec( \cmd\sep \cmdtwo)&\defsym \nospec(\cmd) ; \nospec(\cmdtwo)\\
    \nospec( \cnil)&\defsym \cnil.
  \end{align*}
  \vspace{-1em}
  \begin{align*}
    \nospec (\syss) &\defsym \syscall \mapsto \cfence\sep \nospec (\syss(\syscall)) &
    \nospec (\rfs) &\defsym  \id \mapsto
    \begin{cases}
      \rfs(\id) & \text{if } \id \in \Arr \cup \Funu\\
      \nospec (\rfs(\id)) & \text{otherwise }
    \end{cases} &
    \nospec((\rfs, \syss, \caps)) &\defsym (\nospec (\rfs), \nospec(\syss), \caps).
  \end{align*}
  \vspace{-4mm}
  \caption{Speculation-blocking transformation.}
  \label{fig:nospec}
\end{figure*}

\paragraph{Speculation-blocking transformation}
\label{sec:nospecfencetrans}
Notice that the transformations proposed in \Cref{fig:fencetrans,fig:optfencetrans}
prevent the speculative execution of unsafe commands \emph{lazily},
as they do not completely stop speculative execution, but they just prevent
the transient execution of some instructions that may be unsafe. For instance,
if in the conditional instruction  $\cif \expr \cmd \cmdtwo$ neither $\cmd$ nor $\cmdtwo$ contain
potentially unsafe operations---e.g., if they only contain register assignments---%
$\fencetrans$ and $\optfencetrans$ leave those commands unchanged and the CPU can
still speculate over the value of $\expr$.

In line with Intel${}^{\text{\textregistered}}$'s guidelines for specualtive attacks'
mitigation~\cite{IntelGuidelines4}, another approach is to systematically prevent
any form of speculation by placing speculation barriers after
speculation sources and using safe jump instructions. In our model, this solution
approach down to:
\begin{itemize}
\item placing a $\cfence$ instruction after every direct branch to stop PHT speculation,
\item replacing every ordinary call instruction with a non-speculative one to stop BTB speculation,
\item placing a $\cfence$ instruction after every store operation to stop STL speculation. 
\end{itemize}

The transformation we outlined above is formally defined in
\Cref{fig:nospec}. It is quite easy to observe that $\nospec$
preserves the semantics of the system it is applied to and that no
potentially unsafe instruction can be reached during transient
execution.
As a consequence, we conclude that $\nospec$ enforces \emph{speculative kernel safety} on \emph{kernel safe} systems: 

\begin{lemma}
  \label{thm:nospecmitigation}
  If a system $\system$ is \emph{kernel-safe}, then
  $\nospec(\system)$ is \emph{speculative kernel safe}
  and $\nospec(\system)$ semantically equivalent to $\system$.
\end{lemma}

\begin{proof}[Proof Sketch]
  The proof of this result is analogous to that of \Cref{thm:mitigation}.
\end{proof}

\paragraph{Covering additional speculation mechanisms} In this work,
we focused on speculation related to Spectre v1, v2, and v4 in order
to keep our model abstract enough to allow for formal reasoning, while
covering the most significant speculation mechanisms. However, our
model and results can be extended to cover additional forms of
speculation, such as Load Value Injection (LVI) and return value
speculation. For instance, LVI can be modeled by allowing attackers to
control the loaded value with specific directives. Return value
speculation can be addressed by introducing explicit returns and
directives controlling predicted return addresses. The notions of
semantic preservation and speculative safety imposition naturally
extend to additional speculation mechanisms, ensuring the validity of
Theorem 3 in the enriched model. However, the notion of \emph{safety
  imposition} would become more stringent. Therefore, our
transformations would need to be modified to take into account the
additional sources of speculation.  For instance, the $\optfencetrans$
transformation could be modified to address LVI and return speculation
by replacing return instructions with specific non-speculative return
instructions, similar to our $\kwd{scall}$ instruction, and by
treating loaded values as potential sources of speculation, i.e., by
setting $m=\top$ after each load.


\paragraph{Final remarks} \Cref{prop:mitigation} establishes two
sufficient conditions under which a program transformation turns a
\emph{kernel safe} system into a semantically equivalent
\emph{speculative kernel safe} one, and we provided three
transformations that meet these conditions.  Moreover, as reported in
\Cref{tab:mitigations}, our transformations are capable of stopping
speculative probing on the attacks described in \Cref{sec:discussion}.


\section{Experimental Evaluation}
\label{sec:experiment}
\newcommand{\fimpl}{name\_of\_impl\xspace}

In this section, we evaluate the performance overhead of the
program transformations presented in the previous section.
Specifically, our goal is to measure the impact of our
transformations to \emph{kernel-space} execution, \emph{user-} space
execution, and \emph{user-space} I/O-bound workloads.





\subsection{Overview of the Implementation}
\label{sec:impl}

We implemented the transformations of
\Cref{fig:fencetrans,fig:optfencetrans,fig:nospec} as parts of the
LLVM/Clang compiler infrastructure~\cite{LLVM} version 18. More
precisely, we defined a novel \texttt{MachineFunctionPass} called
\texttt{X86SpeculativeSafetyPass}, implementing our transformations,
which can be activated by specific command-line flags. When a
transformation is activated, our pass iterates over every block and
every instruction within the block, protecting each instruction as
required by the corresponding definition in
\Cref{fig:fencetrans,fig:optfencetrans,fig:nospec}. We were able to
generalize the transformations from our language to the more
expressive x86 ISA by relying on the LLVM API, which allows us to
determine if a certain instruction may load, store, or
jump. Instructions showing these behaviors are protected in the same
way the selected transformation protects instructions
$\cmemread \vx \expr$, $\cmemass \expr \exprtwo$, and
$\ccall \expr {\vec\exprtwo}$ respectively, e.g., by replacing
$\kwd{call}$ instructions with $\cfence; \cscall \expr {\vec\exprtwo}$
sequence.  As we mentioned in \Cref{sec:sksenforcement}, the semantics
of $\cfence$ reproduces that of the x86 instruction \texttt{lfence},
therefore our pass uses \texttt{lfence} to implement the
$\cfence$ instruction. Finally, our LLVM pass supports two ways of
implementing the non-speculative call instruction
$\cscall \expr {\vec \exprtwo}$:
\begin{itemize}
\item One way is to exchange indirect jump and call instructions  with \emph{retpoline} thunks, aimed at preventing branch target speculation. Specifically, we replace indirect jumps and calls with Linux's retpoline thunks at boot time by passing the specific boot command-line parameter \texttt{spectre\_v2=retpoline}.
\item A faster alternative is to rely on Intel${}^{\text{\textregistered}}$ eIBRS~\cite{eIBRS}, leaving branch instructions unchanged. This measure prevents the target of any indirect branch that is executed in \emph{kernel-space} to be predicted depending on \emph{user-space} execution.     
\end{itemize}

We chose to benchmark both the mechanisms as in our attack model,
where RSB speculation is not taken into account, retpoline offers
stronger security guarantees than eIBRS: contrarily to eIBRS retpoline
is not subject to confused-deputy-attacks~\cite{BHInjection,eIBRS},
although it has larger performance overheads in
practice~\cite{MitigationOverhead}.

Like similar passes such as SLH and SESES, the
\texttt{X86SpeculativeSafetyPass} runs just before register allocation
to ensure that subsequent passes do not alter the program's control
flow, which could compromise its protections. The implementation of
our LLVM pass, along with additional details on its behavior, is
available in~\cite{repo}.

      

\subsection{Experimental Evaluation}
\label{sec:expeval}

\paragraph{Performance evaluation goals} Our experimental evaluation aims to determine the performance overhead induced by the transformations of \Cref{fig:fencetrans,fig:optfencetrans,fig:nospec} on \emph{kernel-space} execution, user-space programs that combine \emph{kernel-} and \emph{user-}space execution, and on user-space programs that perform intensive I/O workloads.

\paragraph{Benchmark choice}  We chose to evaluate the overhead on heavy \emph{kernel-space} tasks with the UnixBench Benchmark, which is meant to evaluate the performance of tasks with high pressure on system calls like \texttt{execl}, \texttt{getpid}, \texttt{fork}, or by executing patterns like file copy and inter-process communication. To benchmark real-world computationally heavy \emph{user-space} applications, we used the SPEC${}^{\text{\textregistered}}$ benchmark suite, which includes tasks like compilation, interpretation, compression, simulation, and other  that engage \emph{kernel-space} execution only sporadically.
To measure the overhead on I/O-bound workloads, we used the following benchmarks:
\begin{itemize}
\item The ab~\cite{ab} benchmark, for measuring the latency of widely used web servers. Specifically, we measured the latency of 1000 HTTP GET requests for a 1GB file.
\item The crossdb~\cite{crossdb} benchmark, to measure the response speed of the \texttt{sqlite} and \texttt{crossdb} DBMSs to 200.000 and 2.000.000 queries, respectively.
\item The cryptsetup~\cite{cryptsetup} benchmark, for measuring the throughput of kernel side cryptographic applications. Specifically, we measured the encryption and decryption throughput of AES-EBC
\item The ugrep~\cite{ugrep} benchmark, for assessing the speed of GNU's \texttt{grep} tool on different patterns across multiple files and directories. 
\item The vbench~\cite{vbench} benchmark to measure the transcoding speed of \texttt{ffmpeg} on large videos.
\end{itemize}

\paragraph{Performance evaluation methodology} All the benchmarks were executed on a machine running the Debian 12 distribution using the following Linux 6.10 kernels: 

\begin{varenumerate}
\item\label{k:vanilla} \textbf{Off-the-shelf kernel}. This kernel was compiled from the Linux source using Clang with Debian 12's default settings.
\item\label{k:instrumented} \textbf{Baseline transformed kernel}. This kernel is protected with the $\fencetrans$ transformation (\Cref{fig:fencetrans}), placing $\mathtt{lfence}$ instructions before instructions that interact with the memory and blocking BTB speculation via retpoline or eIBRS.
\item\label{k:instrumented} \textbf{Optimized transformed kernel}. This kernel is protected with the implementation of the $\optfencetrans$ transformation (\Cref{fig:optfencetrans}), selectively placing $\mathtt{lfence}$ instructions before instructions that interact with the memory, and blocking BTB speculation via retpoline or eIBRS.
\item\label{k:nospec} \textbf{Intel${}^{\text{\textregistered}}$ Guidelines' kernel}. This kernel is obtained by transforming the kernel with the implementation of the $\nospec$ transformation (\Cref{fig:nospec}) that applies Intel${}^{\text{\textregistered}}$ guidelines to block \emph{kernel-space} PHT, BTB and STL speculation.
\end{varenumerate}

\ifdefined\arxiv{
\begin{table}[t]
  \resizebox{\textwidth}{!}{%
    \begin{tabular}{lllll}
      \toprule
      & Off-the-shelf & Optimized transformation & Baseline transformation & Intel's guidelines \\
      \midrule
      {perlbench}
      & $\begin{array}[h]{r@{\,\,=\,\,}l} t & 188.04\ s\\\sigma & 0.82\ s\end{array}$
      & $\begin{array}[h]{r@{\,\,=\,\,}l} t & 190.42\ s\\\sigma & 0.92\ s\end{array}$
      & $\begin{array}[h]{r@{\,\,=\,\,}l} t & 190.98\ s\\\sigma & 1.28\ s\end{array}$
      & $\begin{array}[h]{r@{\,\,=\,\,}l} t & 190.07\ s\\\sigma & 1.18\ s\end{array}$ \\[3.5mm]
      {gcc}
      & $\begin{array}[h]{r@{\,\,=\,\,}l} t & 319.08\ s\\\sigma & 0.64\ s\end{array}$
      & $\begin{array}[h]{r@{\,\,=\,\,}l} t & 333.46\ s\\\sigma & 0.28\ s\end{array}$
      & $\begin{array}[h]{r@{\,\,=\,\,}l} t & 347.04\ s\\\sigma & 0.56\ s\end{array}$
      & $\begin{array}[h]{r@{\,\,=\,\,}l} t & 343.77\ s\\\sigma & 0.55\ s\end{array}$ \\[3.5mm]
      {mcf}
      & $\begin{array}[h]{r@{\,\,=\,\,}l} t & 436.33\ s\\\sigma & 5.85\ s\end{array}$
      & $\begin{array}[h]{r@{\,\,=\,\,}l} t & 437.87\ s\\\sigma & 0.63\ s\end{array}$
      & $\begin{array}[h]{r@{\,\,=\,\,}l} t & 439.33\ s\\\sigma & 0.51\ s\end{array}$
      & $\begin{array}[h]{r@{\,\,=\,\,}l} t & 437.80\ s\\\sigma & 0.68\ s\end{array}$ \\[3.5mm]
      {omnetpp}
      & $\begin{array}[h]{r@{\,\,=\,\,}l} t & 287.51\ s\\\sigma & 6.54\ s\end{array}$
      & $\begin{array}[h]{r@{\,\,=\,\,}l} t & 288.24\ s\\\sigma & 4.82\ s\end{array}$
      & $\begin{array}[h]{r@{\,\,=\,\,}l} t & 292.42\ s\\\sigma & 8.20\ s\end{array}$
      & $\begin{array}[h]{r@{\,\,=\,\,}l} t & 288.79\ s\\\sigma & 6.53\ s\end{array}$ \\[3.5mm]
      {xalancbmk}
      & $\begin{array}[h]{r@{\,\,=\,\,}l} t & 154.04\ s\\\sigma & 0.42\ s\end{array}$
      & $\begin{array}[h]{r@{\,\,=\,\,}l} t & 156.93\ s\\\sigma & 0.29\ s\end{array}$
      & $\begin{array}[h]{r@{\,\,=\,\,}l} t & 156.50\ s\\\sigma & 0.41\ s\end{array}$
      & $\begin{array}[h]{r@{\,\,=\,\,}l} t & 156.46\ s\\\sigma & 0.39\ s\end{array}$ \\[3.5mm]
      {x264}
      & $\begin{array}[h]{r@{\,\,=\,\,}l} t & 128.38\ s\\\sigma & 0.49\ s\end{array}$
      & $\begin{array}[h]{r@{\,\,=\,\,}l} t & 129.98\ s\\\sigma & 0.48\ s\end{array}$
      & $\begin{array}[h]{r@{\,\,=\,\,}l} t & 132.10\ s\\\sigma & 0.62\ s\end{array}$
      & $\begin{array}[h]{r@{\,\,=\,\,}l} t & 131.24\ s\\\sigma & 0.63\ s\end{array}$ \\[3.5mm]
      {deepsjeng}
      & $\begin{array}[h]{r@{\,\,=\,\,}l} t & 233.84\ s\\\sigma & 0.21\ s\end{array}$
      & $\begin{array}[h]{r@{\,\,=\,\,}l} t & 234.72\ s\\\sigma & 0.26\ s\end{array}$
      & $\begin{array}[h]{r@{\,\,=\,\,}l} t & 236.87\ s\\\sigma & 0.23\ s\end{array}$
      & $\begin{array}[h]{r@{\,\,=\,\,}l} t & 234.70\ s\\\sigma & 0.21\ s\end{array}$ \\[3.5mm]
      {leela}
      & $\begin{array}[h]{r@{\,\,=\,\,}l} t & 295.56\ s\\\sigma & 0.96\ s\end{array}$
      & $\begin{array}[h]{r@{\,\,=\,\,}l} t & 296.98\ s\\\sigma & 0.92\ s\end{array}$
      & $\begin{array}[h]{r@{\,\,=\,\,}l} t & 298.77\ s\\\sigma & 0.75\ s\end{array}$
      & $\begin{array}[h]{r@{\,\,=\,\,}l} t & 297.29\ s\\\sigma & 0.86\ s\end{array}$ \\[3.5mm]
      {exchange2}
      & $\begin{array}[h]{r@{\,\,=\,\,}l} t & 119.53\ s\\\sigma & 0.27\ s\end{array}$
      & $\begin{array}[h]{r@{\,\,=\,\,}l} t & 120.17\ s\\\sigma & 0.19\ s\end{array}$
      & $\begin{array}[h]{r@{\,\,=\,\,}l} t & 120.68\ s\\\sigma & 0.19\ s\end{array}$
      & $\begin{array}[h]{r@{\,\,=\,\,}l} t & 120.25\ s\\\sigma & 0.21\ s\end{array}$ \\[3.5mm]
      {xz}
      & $\begin{array}[h]{r@{\,\,=\,\,}l} t & 909.59\ s\\\sigma & 3.62\ s\end{array}$
      & $\begin{array}[h]{r@{\,\,=\,\,}l} t & 923.79\ s\\\sigma & 2.94\ s\end{array}$
      & $\begin{array}[h]{r@{\,\,=\,\,}l} t & 940.53\ s\\\sigma & 2.83\ s\end{array}$
      & $\begin{array}[h]{r@{\,\,=\,\,}l} t & 933.27\ s\\\sigma & 4.55\ s\end{array}$ \\
      \bottomrule
    \end{tabular}
  }
  \caption{Measurements of the average run time $t$ and standard deviation $\sigma$ over 40 runs of the SPEC${}^{\text{\textregistered}}$ CPU 2017 Benchmark with eIBRS active.}
  \label{tab:data2a}
\end{table}

\begin{table}[t]
  \resizebox{\textwidth}{!}{%
    \begin{tabular}{lllll}
      \toprule
      & Off-the-shelf & Optimized transformation & Baseline transformation & Intel's guidelines \\
      \midrule
      {perlbench}
      &  $\begin{array}[h]{r@{\,\,=\,\,}l} t & 188.65\ s\\\sigma & 0.84\ s\end{array}$
      & $\begin{array}[h]{r@{\,\,=\,\,}l} t & 189.02\ s\\\sigma & 0.83\ s\end{array}$
      & $\begin{array}[h]{r@{\,\,=\,\,}l} t & 190.12\ s\\\sigma & 1.16\ s\end{array}$
      & $\begin{array}[h]{r@{\,\,=\,\,}l} t & 190.34\ s\\\sigma & 2.54\ s\end{array}$ \\[3.5mm]
      {gcc}
      & $\begin{array}[h]{r@{\,\,=\,\,}l} t & 319.62\ s\\\sigma & 0.69\ s\end{array}$
      & $\begin{array}[h]{r@{\,\,=\,\,}l} t & 333.13\ s\\\sigma & 0.47\ s\end{array}$
      & $\begin{array}[h]{r@{\,\,=\,\,}l} t & 345.71\ s\\\sigma & 0.47\ s\end{array}$
      & $\begin{array}[h]{r@{\,\,=\,\,}l} t & 344.00\ s\\\sigma & 1.02\ s\end{array}$ \\[3.5mm]
      {mcf}
      &  $\begin{array}[h]{r@{\,\,=\,\,}l} t & 436.55\ s\\\sigma & 2.54\ s\end{array}$
      & $\begin{array}[h]{r@{\,\,=\,\,}l} t & 437.10\ s\\\sigma & 0.61\ s\end{array}$
      & $\begin{array}[h]{r@{\,\,=\,\,}l} t & 438.05\ s\\\sigma & 0.81\ s\end{array}$
      & $\begin{array}[h]{r@{\,\,=\,\,}l} t & 437.81\ s\\\sigma & 1.05\ s\end{array}$ \\[3.5mm]
      {omnetpp}
      & $\begin{array}[h]{r@{\,\,=\,\,}l} t & 286.75\ s\\\sigma & 3.35\ s\end{array}$
      & $\begin{array}[h]{r@{\,\,=\,\,}l} t & 290.81\ s\\\sigma & 7.47\ s\end{array}$
      & $\begin{array}[h]{r@{\,\,=\,\,}l} t & 289.78\ s\\\sigma & 6.98\ s\end{array}$
      & $\begin{array}[h]{r@{\,\,=\,\,}l} t & 288.40\ s\\\sigma & 3.53\ s\end{array}$ \\[3.5mm]
      {xalancbmk}
      & $\begin{array}[h]{r@{\,\,=\,\,}l} t & 155.46\ s\\\sigma & 0.25\ s\end{array}$
      & $\begin{array}[h]{r@{\,\,=\,\,}l} t & 156.04\ s\\\sigma & 0.42\ s\end{array}$
      & $\begin{array}[h]{r@{\,\,=\,\,}l} t & 157.30\ s\\\sigma & 0.70\ s\end{array}$
      & $\begin{array}[h]{r@{\,\,=\,\,}l} t & 156.59\ s\\\sigma & 0.50\ s\end{array}$ \\[3.5mm]
      {x264}
      & $\begin{array}[h]{r@{\,\,=\,\,}l} t & 128.66\ s\\\sigma & 0.46\ s\end{array}$
      & $\begin{array}[h]{r@{\,\,=\,\,}l} t & 129.99\ s\\\sigma & 0.59\ s\end{array}$
      & $\begin{array}[h]{r@{\,\,=\,\,}l} t & 131.58\ s\\\sigma & 0.65\ s\end{array}$
      & $\begin{array}[h]{r@{\,\,=\,\,}l} t & 131.11\ s\\\sigma & 0.66\ s\end{array}$ \\[3.5mm]
      {deepsjeng}
      & $\begin{array}[h]{r@{\,\,=\,\,}l} t & 234.00\ s\\\sigma & 0.24\ s\end{array}$
      & $\begin{array}[h]{r@{\,\,=\,\,}l} t & 235.61\ s\\\sigma & 0.16\ s\end{array}$
      & $\begin{array}[h]{r@{\,\,=\,\,}l} t & 235.27\ s\\\sigma & 0.86\ s\end{array}$
      & $\begin{array}[h]{r@{\,\,=\,\,}l} t & 235.88\ s\\\sigma & 0.26\ s\end{array}$ \\[3.5mm]
      {leela}
      & $\begin{array}[h]{r@{\,\,=\,\,}l} t & 296.35\ s\\\sigma & 0.86\ s\end{array}$
      & $\begin{array}[h]{r@{\,\,=\,\,}l} t & 296.87\ s\\\sigma & 0.80\ s\end{array}$
      & $\begin{array}[h]{r@{\,\,=\,\,}l} t & 296.95\ s\\\sigma & 0.88\ s\end{array}$
      & $\begin{array}[h]{r@{\,\,=\,\,}l} t & 297.33\ s\\\sigma & 0.79\ s\end{array}$ \\[3.5mm]
      {exchange2}
      & $\begin{array}[h]{r@{\,\,=\,\,}l} t & 120.00\ s\\\sigma & 0.26\ s\end{array}$
      & $\begin{array}[h]{r@{\,\,=\,\,}l} t & 120.15\ s\\\sigma & 0.22\ s\end{array}$
      & $\begin{array}[h]{r@{\,\,=\,\,}l} t & 120.21\ s\\\sigma & 0.20\ s\end{array}$
      & $\begin{array}[h]{r@{\,\,=\,\,}l} t & 120.29\ s\\\sigma & 0.28\ s\end{array}$ \\[3.5mm]
      {xz}
      &  $\begin{array}[h]{r@{\,\,=\,\,}l} t & 910.89\ s\\\sigma & 2.19\ s\end{array}$
      & $\begin{array}[h]{r@{\,\,=\,\,}l} t & 930.33\ s\\\sigma & 3.04\ s\end{array}$
      & $\begin{array}[h]{r@{\,\,=\,\,}l} t & 936.67\ s\\\sigma & 4.14\ s\end{array}$
      & $\begin{array}[h]{r@{\,\,=\,\,}l} t & 937.97\ s\\\sigma & 3.19\ s\end{array}$ \\[3.5mm]
      \bottomrule
    \end{tabular}
  }
  \caption{Measurements of the average run time $t$ and standard deviation $\sigma$ over 40 runs of the SPEC${}^{\text{\textregistered}}$ CPU 2017 Benchmark with retpoline active.}
  \label{tab:data2b}
\end{table}

\begin{table}[t]
  \centering
  \small
  \begin{subtable}[t]{\textwidth}
  \resizebox{\textwidth}{!}{%
    \begin{tabular}[t]{lllll}
      \toprule
      & Off-the-shelf & Optimized transformation & Baseline transformation & Intel's guidelines \\
      \midrule
      Execl Throughput  & $4455~\mathit{lps}$ & $729.1~\mathit{lps}$ & $405.8~\mathit{lps}$ & $425.7~\mathit{lps}$\\
      File Copy  & $4.658\cdot 10^6~\mathit{KB/s}$ & $1.16\cdot 10^6~\mathit{KB/s}$ & $6.524\cdot 10^5~\mathit{KB/s}$ & $7.383\cdot 10^5~\mathit{KB/s}$\\
      Pipe-based IPC  & $2.396\cdot 10^5~\mathit{lps}$ & $5.54\cdot 10^4~\mathit{lps}$ & $2.813\cdot 10^4~\mathit{lps}$ & $3.379\cdot 10^4~\mathit{lps}$\\
      Process Creation  & $9622~\mathit{lps}$ & $985.8~\mathit{lps}$ & $608.4~\mathit{lps}$ & $620~\mathit{lps}$\\
      Shell Scripts (8 concurrent)  & $5826~\mathit{lpm}$ & $2091~\mathit{lpm}$ & $1443~\mathit{lpm}$ & $1534~\mathit{lpm}$\\
      \texttt{getpid} overhead  & $3.234\cdot 10^6~\mathit{lps}$ & $7.181\cdot 10^5~\mathit{lps}$ & $4.56\cdot 10^5~\mathit{lps}$ & $4.706\cdot 10^5~\mathit{lps}$\\
      \bottomrule
    \end{tabular}
  }
  \caption{Results with eIBRS}
\end{subtable}

\vspace{2mm}

\begin{subtable}[t]{\textwidth}
  \resizebox{\textwidth}{!}{%
    \begin{tabular}[t]{lllll}
      \toprule
      & Off-the-shelf & Optimized transformation & Baseline transformation & Intel's guidelines \\
      \midrule
      Execl Throughput & $4354~\mathit{lps}$ & $721~\mathit{lps}$ & $406.4~\mathit{lps}$ & $426.6~\mathit{lps}$  \\
      File Copy & $4.483\cdot 10^6~\mathit{KB/s}$ & $1.136\cdot 10^6~\mathit{KB/s}$ & $6.501\cdot 10^5~\mathit{KB/s}$ & $7.372\cdot 10^5~\mathit{KB/s}$  \\
      Pipe-based IPC & $2.325\cdot 10^5~\mathit{lps}$ & $5.395\cdot 10^4~\mathit{lps}$ & $2.9\cdot 10^4~\mathit{lps}$ & $3.289\cdot 10^4~\mathit{lps}$ \\
      Process Creation & $9344~\mathit{lps}$ & $991.4~\mathit{lps}$ & $593.9~\mathit{lps}$ & $626.5~\mathit{lps}$ \\
      Shell Scripts (8 concurrent) & $5786~\mathit{lpm}$ & $2087~\mathit{lpm}$ & $1344~\mathit{lpm}$ & $1425~\mathit{lpm}$  \\
      \texttt{getpid} overhead & $3.234\cdot 10^6~\mathit{lps}$ & $7.188\cdot 10^5~\mathit{lps}$ & $4.566\cdot 10^5~\mathit{lps}$ & $4.709\cdot 10^5~\mathit{lps}$  \\
      \bottomrule
    \end{tabular}
  }
  \caption{Results with retpoline}
\end{subtable}

  \caption{Throughput of the UnixBench Benchmark. The unit $\mathit{lps}$ stands for ``loops per second'', indicating how many times a specific task has completed within one second. Similarly $\mathit{lps}$ stands for ``loops per minute''.}
  \label{tab:data1}
\end{table}

\begin{table}[t]
  \centering
  \resizebox{\textwidth}{!}{
    \begin{tabular}{lllll}
      \toprule
      & Off-the-shelf & Optimized transformation & Baseline transformation & Intel's guidelines \\
      \midrule
      {Apache}
      & $\begin{array}[h]{r@{\,\,=\,\,}l} t & 169\ ms\\\sigma & 5.2\ ms\end{array}$
      & $\begin{array}[h]{r@{\,\,=\,\,}l} t & 592\ ms\\\sigma & 7.2\ ms\end{array}$
      & $\begin{array}[h]{r@{\,\,=\,\,}l} t & 954\ ms\\\sigma & 10.8\ ms\end{array}$
      & $\begin{array}[h]{r@{\,\,=\,\,}l} t & 858\ ms\\\sigma & 11.2\ ms\end{array}$ \\[3.5mm]
      {ngninx}
      & $\begin{array}[h]{r@{\,\,=\,\,}l} t & 140\ ms\\\sigma & 9.4\ ms\end{array}$
      & $\begin{array}[h]{r@{\,\,=\,\,}l} t & 578\ ms\\\sigma & 13.6\ ms\end{array}$
      & $\begin{array}[h]{r@{\,\,=\,\,}l} t & 999\ ms\\\sigma & 19.4\ ms\end{array}$
      & $\begin{array}[h]{r@{\,\,=\,\,}l} t & 865\ ms\\\sigma & 18\ ms\end{array}$ \\[3.5mm]
      {AES ECB encryption}
      & $\begin{array}[h]{r@{\,\,=\,\,}l} th & 5711.0\ \mathit{MiB/s}\end{array}$
                      & $\begin{array}[h]{r@{\,\,=\,\,}l} th & 1724.7\ \mathit{MiB/s}\end{array}$
                                                 & $\begin{array}[h]{r@{\,\,=\,\,}l} th  & 1023.2\ \mathit{MiB/s}\end{array}$
                                                                           & $\begin{array}[h]{r@{\,\,=\,\,}l} th & 1104.9\ \mathit{MiB/s}\end{array}$ \\
      {AES ECB decryption}
      & $\begin{array}[h]{r@{\,\,=\,\,}l} th & 5700.2\ \mathit{MiB/s}\end{array}$
                      & $\begin{array}[h]{r@{\,\,=\,\,}l} th & 1726.3\ \mathit{MiB/s}\end{array}$
                                                 & $\begin{array}[h]{r@{\,\,=\,\,}l} th & 1024.5\ \mathit{MiB/s}\end{array}$
                                                                           & $\begin{array}[h]{r@{\,\,=\,\,}l} th & 1112.4\ \mathit{MiB/s}\end{array}$ \\[3.5mm]
      {SQLite}
      & $\begin{array}[h]{r@{\,\,=\,\,}l} t & 1257.15\ ms\\\sigma & 0.36\ ms\end{array}$
      & $\begin{array}[h]{r@{\,\,=\,\,}l} t & 1235.43\ ms\\\sigma & 2.59\ ms\end{array}$
      & $\begin{array}[h]{r@{\,\,=\,\,}l} t & 1225.62\ ms\\\sigma & 3.95\ ms\end{array}$
      & $\begin{array}[h]{r@{\,\,=\,\,}l} t & 1235.55\ ms\\\sigma & 2.03\ ms\end{array}$ \\[3.5mm]
      {CrossDB}
      & $\begin{array}[h]{r@{\,\,=\,\,}l} t & 517.34\ ms\\\sigma & 5.36\ ms\end{array}$
      & $\begin{array}[h]{r@{\,\,=\,\,}l} t & 517.58\ ms\\\sigma & 3.30\ ms\end{array}$
      & $\begin{array}[h]{r@{\,\,=\,\,}l} t & 529.10\ ms\\\sigma & 2.16\ ms\end{array}$
      & $\begin{array}[h]{r@{\,\,=\,\,}l} t & 521.10\ ms\\\sigma & 4.49\ ms\end{array}$ \\[3.5mm]
      {grep}
      & $\begin{array}[h]{r@{\,\,=\,\,}l} t & 102.290\ s\\\sigma & 0.300\ s\end{array}$
      & $\begin{array}[h]{r@{\,\,=\,\,}l} t & 105.921\ s\\\sigma & 0.371\ s\end{array}$
      & $\begin{array}[h]{r@{\,\,=\,\,}l} t & 108.932\ s\\\sigma & 0.210\ s\end{array}$
      & $\begin{array}[h]{r@{\,\,=\,\,}l} t & 108.472\ s\\\sigma & 0.328\ s\end{array}$ \\[3.5mm]
      {vbench}
      & $\begin{array}[h]{r@{\,\,=\,\,}l} t & 149.624\ s\\\sigma & 0.142\ s\end{array}$
      & $\begin{array}[h]{r@{\,\,=\,\,}l} t & 153.357\ s\\\sigma & 0.190\ s\end{array}$
      & $\begin{array}[h]{r@{\,\,=\,\,}l} t & 157.963\ s\\\sigma & 0.214\ s\end{array}$
      & $\begin{array}[h]{r@{\,\,=\,\,}l} t & 157.339\ s\\\sigma & 0.268\ s\end{array}$ \\
      \bottomrule
    \end{tabular}
  }
  \caption{ Measurements of the average time $t$, standard deviation $\sigma$ and
    throughput $th$ of the I/O-bound workloads with eIBRS active.}
  \label{tab:data3a}
\end{table}
\begin{table}[t]
  \resizebox{\textwidth}{!}{%
    \begin{tabular}{lllll}
      \toprule
      & Off-the-shelf & Optimized transformation & Baseline transformation & Intel's guidelines \\
      \midrule
      {Apache}
      & $\begin{array}[h]{r@{\,\,=\,\,}l} t & 175\ ms\\\sigma & 5.4\ ms\end{array}$
      & $\begin{array}[h]{r@{\,\,=\,\,}l} t & 597\ ms\\\sigma & 6.6\ ms\end{array}$
      & $\begin{array}[h]{r@{\,\,=\,\,}l} t & 967\ ms\\\sigma & 8.1\ ms\end{array}$
      & $\begin{array}[h]{r@{\,\,=\,\,}l} t & 863\ ms\\\sigma & 8.1\ ms\end{array}$ \\[3.5mm]
      {ngninx}
      & $\begin{array}[h]{r@{\,\,=\,\,}l} t & 154\ ms\\\sigma & 9.5\ ms\end{array}$
      & $\begin{array}[h]{r@{\,\,=\,\,}l} t & 580\ ms\\\sigma & 13.2\ ms\end{array}$
      & $\begin{array}[h]{r@{\,\,=\,\,}l} t & 1010\ ms\\\sigma & 26.2\ ms\end{array}$
      & $\begin{array}[h]{r@{\,\,=\,\,}l} t & 874\ ms\\\sigma & 22.9\ ms\end{array}$ \\[3.5mm]
      {AES ECB encryption}
      & $\begin{array}[h]{r@{\,\,=\,\,}l} th & 5692.8\ \mathit{MiB/s}\end{array}$
                      & $\begin{array}[h]{r@{\,\,=\,\,}l} th & 1708.0\ \mathit{MiB/s}\end{array}$
      & $\begin{array}[h]{r@{\,\,=\,\,}l} th  & 1021.6\ \mathit{MiB/s}\end{array}$
                                                           & $\begin{array}[h]{r@{\,\,=\,\,}l} th & 1109.1\ \mathit{MiB/s}\end{array}$ \\
      {AES ECB decryption}
      & $\begin{array}[h]{r@{\,\,=\,\,}l} th & 5610.9\ \mathit{MiB/s}\end{array}$
                      & $\begin{array}[h]{r@{\,\,=\,\,}l} th & 1711.7\ \mathit{MiB/s}\end{array}$
      & $\begin{array}[h]{r@{\,\,=\,\,}l} th & 1023.7\ \mathit{MiB/s}\end{array}$
                                                           & $\begin{array}[h]{r@{\,\,=\,\,}l} th & 1108.8\ \mathit{MiB/s}\end{array}$ \\[3.5mm]
      {SQLite}
      & $\begin{array}[h]{r@{\,\,=\,\,}l} t & 1262.09\ ms\\\sigma & 2.63\ ms\end{array}$
      & $\begin{array}[h]{r@{\,\,=\,\,}l} t & 1234.22\ ms\\\sigma & 4.83\ ms\end{array}$
      & $\begin{array}[h]{r@{\,\,=\,\,}l} t & 1248.22\ ms\\\sigma & 2.99\ ms\end{array}$
      & $\begin{array}[h]{r@{\,\,=\,\,}l} t & 1234.80\ ms\\\sigma & 3.84\ ms\end{array}$ \\[3.5mm]
      {CrossDB}
      & $\begin{array}[h]{r@{\,\,=\,\,}l} t & 508.92\ ms\\\sigma & 3.20\ ms\end{array}$
      & $\begin{array}[h]{r@{\,\,=\,\,}l} t & 524.98\ ms\\\sigma & 0.95\ ms\end{array}$
      & $\begin{array}[h]{r@{\,\,=\,\,}l} t & 521.85\ ms\\\sigma & 8.76\ ms\end{array}$
      & $\begin{array}[h]{r@{\,\,=\,\,}l} t & 513.70\ ms\\\sigma & 3.46\ ms\end{array}$ \\[3.5mm]
      {grep}
      & $\begin{array}[h]{r@{\,\,=\,\,}l} t & 102.122\ s\\\sigma & 0.473\ s\end{array}$
      & $\begin{array}[h]{r@{\,\,=\,\,}l} t & 105.815\ s\\\sigma & 0.324\ s\end{array}$
      & $\begin{array}[h]{r@{\,\,=\,\,}l} t & 109.612\ s\\\sigma & 0.303\ s\end{array}$
      & $\begin{array}[h]{r@{\,\,=\,\,}l} t & 111.164\ s\\\sigma & 5.42\ s\end{array}$ \\[3.5mm]
      {vbench}
      & $\begin{array}[h]{r@{\,\,=\,\,}l} t & 149.243\ s\\\sigma & 0.194\ s\end{array}$
      & $\begin{array}[h]{r@{\,\,=\,\,}l} t & 153.045\ s\\\sigma & 0.158\ s\end{array}$
      & $\begin{array}[h]{r@{\,\,=\,\,}l} t & 158.53\ s\\\sigma & 0.248\ s\end{array}$
      & $\begin{array}[h]{r@{\,\,=\,\,}l} t & 167.045\ s\\\sigma & 0.382\ s\end{array}$ \\[3.5mm]
      \bottomrule
    \end{tabular}
  }
  \caption{ Measurements of the average time $t$, standard deviation $\sigma$ and
    throughput $th$ of the I/O-bound workloads with retpoline active.}
  \label{tab:data3b}
\end{table}

}\fi






\begin{figure}[p]
  \centering
  \begin{subfigure}[l]{0.48\textwidth}
    \includegraphics[width=\textwidth]{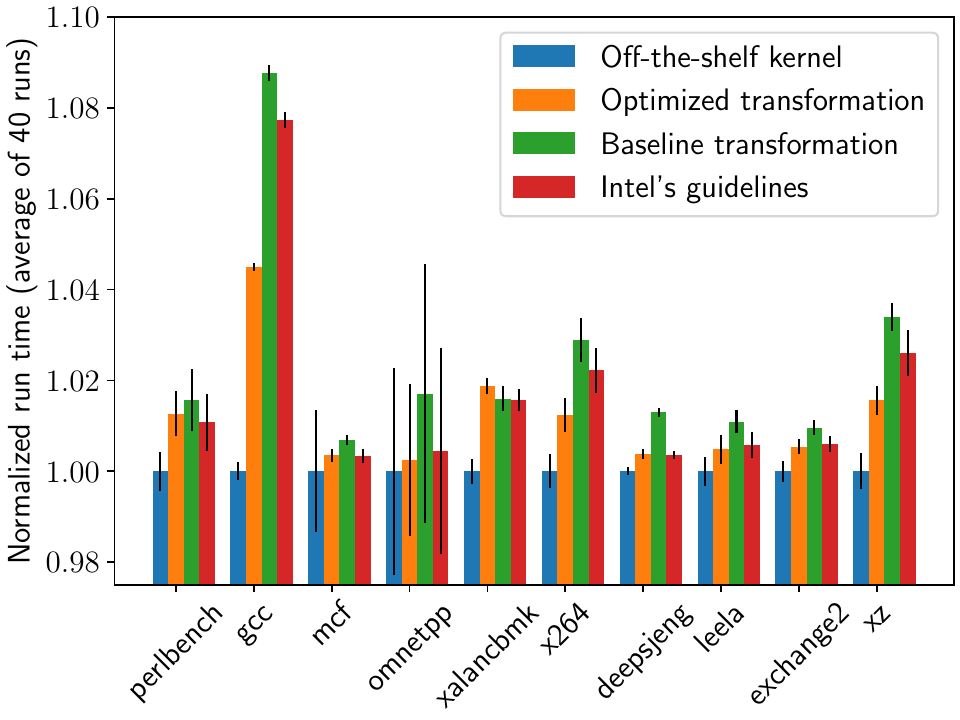}
    \caption{Overhead on the SPEC${}^{\text{\textregistered}}$ CPU 2017 Benchmark with eIBRS.}
    \label{fig:speceibrs}
  \end{subfigure}
  \hfill
  \begin{subfigure}[l]{0.48\textwidth}
    \includegraphics[width=\textwidth]{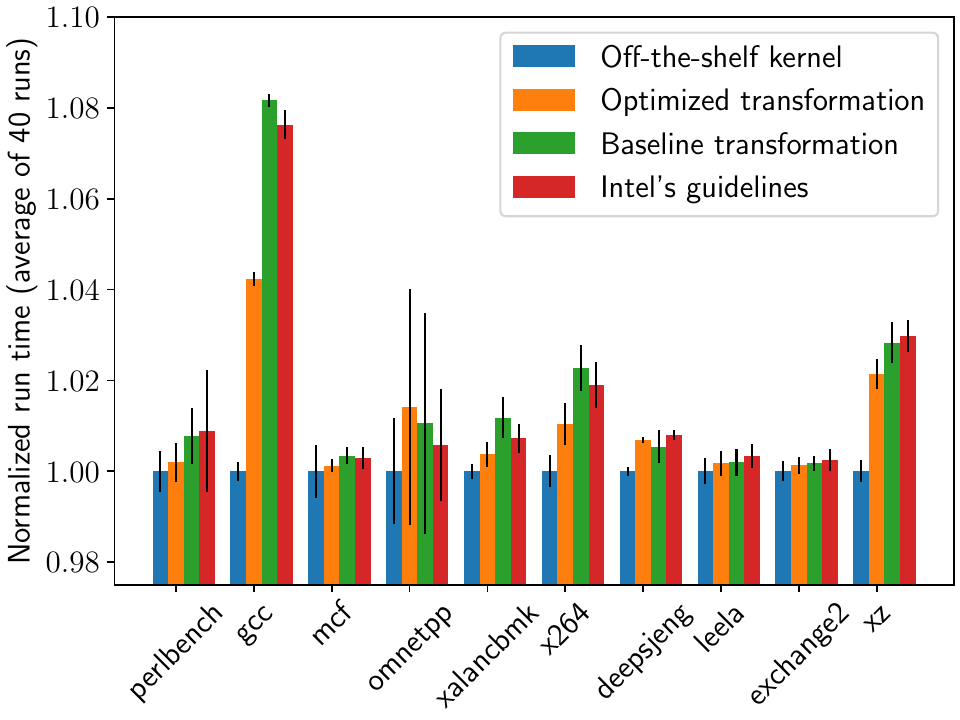}
    \caption{Overhead on the SPEC${}^{\text{\textregistered}}$ CPU 2017 Benchmark with retpoline.}
    \label{fig:specretpoline}
  \end{subfigure}
  \hfill
  \begin{subfigure}[l]{0.48\textwidth}
    \includegraphics[width=\textwidth]{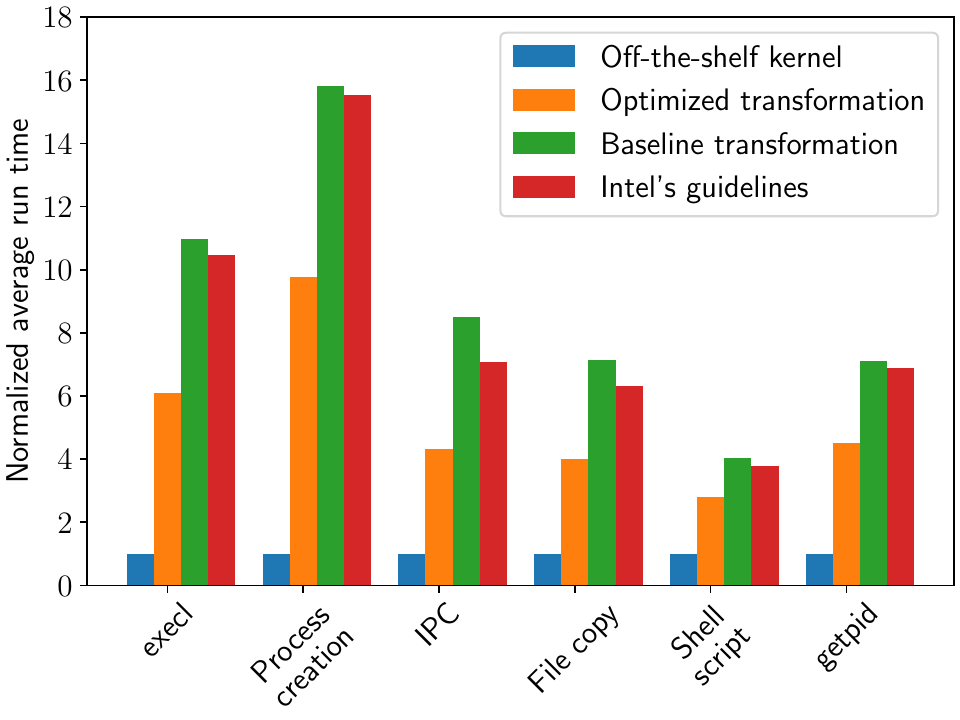}
    \caption{Overhead on the UnixBench Benchmark with eIBRS.}
    \label{fig:unixbencheibrs}
  \end{subfigure}
  \hfill
  \begin{subfigure}[l]{0.48\textwidth}
    \includegraphics[width=\textwidth]{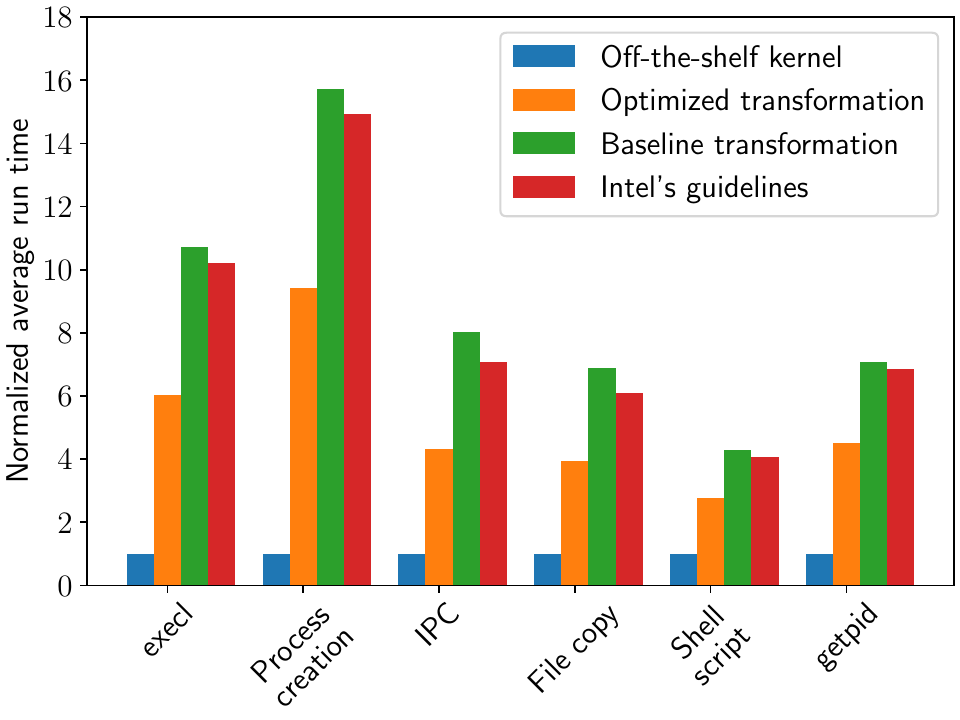}
    \caption{Overhead on the UnixBench Benchmark with retpoline.}
    \label{fig:unixbenchretpoline}
  \end{subfigure}
  \hfill
  \begin{subfigure}[l]{0.48\textwidth}
    \includegraphics[width=\textwidth]{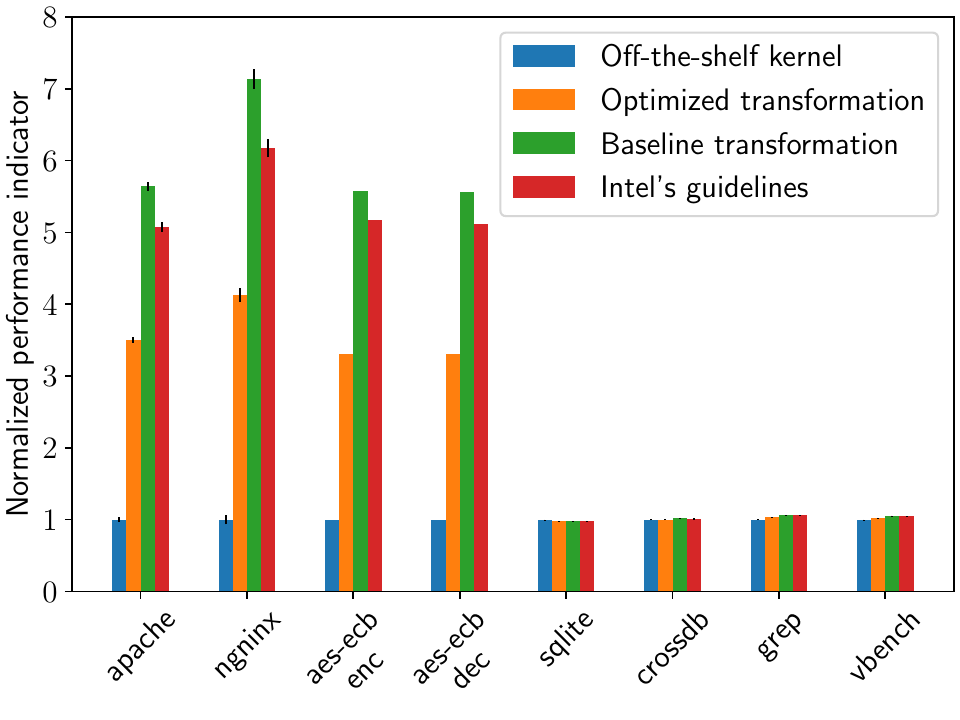}
    \caption{Overhead on the I/O benchmarks with eIBRS.}
    \label{fig:ioeibrs}
  \end{subfigure}
  \hfill
  \begin{subfigure}[l]{0.48\textwidth}
    \includegraphics[width=\textwidth]{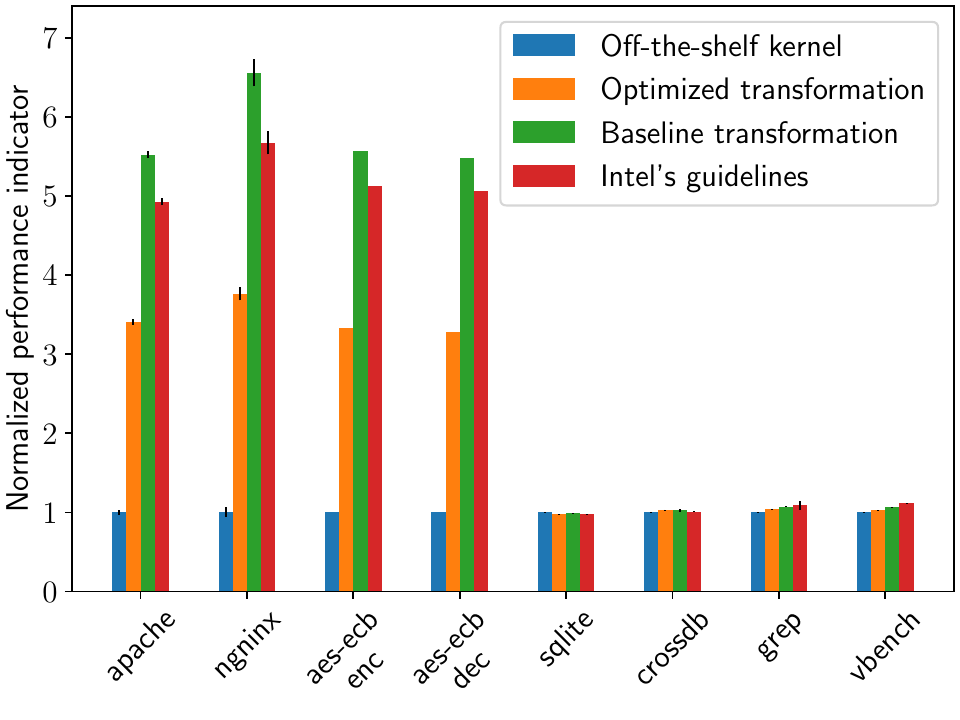}
    \caption{Overhead on the I/O benchmarks with retpoline.}
    \label{fig:ioretpoline}
  \end{subfigure}
  \caption{Overheads of the different transformations of the
    kernel. In the figures, the scale of the $y$-axis is linear. To
    better highlight variations in overhead, the $y$-axis origin is
    set to 0.975 in \Cref{fig:speceibrs,fig:specretpoline}.
  }
  \label{fig:results}
\end{figure}

All kernels were compiled with the default settings against speculative attacks. In particular, these settings include support for eIBRS~\cite{eIBRS}, retpoline~\cite{Retpoline}, and BHI\_DIS\_S~\cite{IntelBHIDISS}, which is an Intel${}^{\text{\textregistered}}$ mitigation for Branch History Injection~\cite{BHInjection}. The settings also include measures against \emph{Return Stack Buffer} (RSB) speculation~\cite{Retbleed}---out of scope for this work---including support for call depth tracking~\cite{CDT} and untrained return thunks, implementing AMD's JMP2RET~\cite{AMDJMP2RET}.
All the experiments were run on a machine equipped with an
Intel${}^{\text{\textregistered}}$ Core i5-1345U processor with 12
logical cores and 16GB of DDR4 memory. Notably, on such
processor, retpoline and eIBRS prevent indirect jump target
speculation~\cite{AffectedIntel,skylake1}.

\paragraph{Data collection and results}
\ifdefined\conference{The results of our experimental evaluation are
    summarized in \Cref{fig:results}.}\fi \ifdefined\arxiv{The results
    of our experimental evaluation are shown in
    \Cref{tab:data1,tab:data2a,tab:data2b,tab:data3a,tab:data3b} and summarized in
    \Cref{fig:results}.}\fi\xspace Most of our benchmarks reported
execution time measurements. For these benchmarks, the height of the
bar represents the average execution times recorded for each of the
kernels, divided by the average required by the off-the shelf kernel
\eqref{k:vanilla}. The black line displays the standard deviation of
our sample. As UnixBench and cryptsetup measure how many times a
specific task can be completed in a fixed unit of time, they provide a
throughput measurement and no indicator of the standard deviation. By
inverting the throughput measurement, we obtained an indicator of the
average execution time of each task, we estimated the overhead by
dividing the average indicator by the estimation of the same indicator
for the off-the-shelf
kernel~\eqref{k:vanilla}.


\paragraph{Final remarks} Our $\optfencetrans$  transformation outperforms others transformations and incurs in negligible overhead in the SPEC${}^{\text{\textregistered}}$ benchmark and some I/O-bound workloads, where its overhead is usually below 1\% and never exceeds 5\%. In UnixBench, and the I/O-bound workloads which put more pressure on store and load operations, it incurs a 3x-10x overhead compared to the non-transformed kernel. While this overhead is significant, future hardware advancements could help mitigate this cost, making such protections more practical for a wider range of systems.

\section{Related Work}
\label{sec:relatedwork}
\paragraph{On Layout Randomization. }
The first work that provided  a formal account of layout randomization was by \Citet{Abadi}, later extended in \cite{Abadi2,Abadi3}. In these works, the authors show that layout randomization prevents, with high probability, malicious programs from accessing the memory of a victim in an execution context with  shared address space. We have already discussed this in the body of the
paper how these results do not
model speculative execution or side-channel observations.


\paragraph{Spatial Memory Safety and Non-Interference} Spatial memory safety is typically defined by associating a software component with a memory area and requiring that, at runtime, the component only accesses that area~\cite{HighAssurance,SoftBound,MSWasm}. \Citet{PierceMS} demonstrated that memory safety can be expressed in terms of non-interference; this property, in turn, stipulates that the final output of a computation is not influenced by secret data that a program must keep confidential~\cite{NonInt}. Both of these properties have been extended to the speculative model. The definition of \emph{speculative memory safety} from~\cite{HighAssurance} closely aligns with our notion of \emph{speculative kernel safety}, with the difference that the latter also imposes some restriction on the victim's control flow. \emph{Speculative non-interference} was initially introduced in the context of the \textsc{Spectector} symbolic analyzer~\cite{Spectector}. \textsc{Spectector}'s property captures information flows to side-channels that occur with transient execution but not in sequential execution. In contrast to \textsc{Spectector}'s approach, our definition aligns with \emph{speculative constant-time}~\cite{CTFundations}, as it targets information leaks that occur during normal and transient execution.



\paragraph{Formal Analysis of Security Properties of Privileged Execution Environments. } \Citet{SLNonInt} deploy a model with side-channel leaks and privileged execution mode, without specualtive execution. In particular, they are interested in studying the preservation of constant-time in virtualization platforms. They also model privilege-raising procedures \emph{hypercalls}, similar to our system calls. They show that if one of the hosts is constant-time then the system enjoys a form of non-interference with respect to that host's secret memory. For this reason, although the two models are similar, the purposes of \Citet{SLNonInt} and our work are different: in~\cite{SLNonInt} the victim and the attacker have the same levels of privilege and the role of the hypervisor is to ensure their separation whilst, in our work, the privileged code base is itself the victim.  In addition, \Citet{SLNonInt} study constant-time, while we focus on memory safety and control flow integrity. 

\paragraph{Attacks to Kernel Layout Randomization} Attacks that aim at leaking information about the kernel's layout are very popular and can rely on implementation bugs that reveal information the kernel's layout~\cite{Uncontained,KernelBugs,KernelElasticObj} or on side-channel info-leaks~\cite{cacheKASLR,TagBleed, EntryBleed, EchoLoad}. In particular attacks such as EchoLoad, TagBleed and EntryBleed~\cite{TagBleed, EntryBleed, EchoLoad} are successful even in presence of state-of-art mitigations such as Intel${}^{\text{\textregistered}}$'s Page Table Isolation (PTI)~\cite{PTI}. These attacks motivate our decision to take into account side-channel info-leaks. Due to address-space separation between kernel and user space programs, an attacker cannot easily use a pointer to a kernel address to access the victim's memory. So, in general, if the attacker does not control the value of a pointer that is used by the victim, this kind of leak is not harmful.

The Meltdown attack~\cite{Meltdown} uses speculative execution to overcome address-space separation on operating systems running on Intel${}^{\text{\textregistered}}$ processors that do not adopt KAISER~\cite{Kaiser} or PTI~\cite{PTI}. In particular, the hardware can speculatively access an address before checking its permissions. The attack uses this small time window to access kernel memory content and leak it by using a side-channel info-leak gadget. These attacks can also be used to leak information on the layout: by dereferencing kernel addresses during transient execution, the kernel's address space can be probed without crashing the system. Due to the adoption of PTI~\cite{PTI}, this kind of attack is mitigated by removing most of the kernel-space addresses from the page tables of user-space programs. The BlindSide attack~\cite{BlindSide} overcomes this issue by probing directly from kernel-space. Similar attacks can be mounted by triggering different forms of mispredictions~\cite{SPEAR}.

Branch target buffer (BTB) speculation---related to Spectre v2---can be used by attackers to defeat kernel's layout randomization. \Citet{JumpOverASLR} were able to show that BTB mis-speculations reveal information on victim's layout. \Citet{BHInjection} showed how, in the presence of BTB speculation, an attacker can steer the control flow of kernel's indirect branches close to the context switch, even in the presence of KASLR and eIBRS. As shown by \Citet{InSpectre}, currently the Linux kernel contains hundreds of such exploitable indirect jumps. The form of speculation exploited by this attacks is taken into account in our model by allowing arbitrary speculation on unsafe jumps.
Akin to BTB speculation, Return Stack Buffer (RSB) speculation can be used to transiently divert the victim's control flow to arbitrary locations when it executes a return instruction. \Citet{Retbleed} used this vulnerability to break Layout Randomization and to leak kernel's memory. We did not consider RSB speculation in our model. Therefore, systems that are protected with our mitigations may still be vulnerable to attacks relying on RSB speculation.

However, attacks such as Retbleed~\cite{Retbleed} and BHI~\cite{BHInjection} do not compromise our reliance on eIBRS and retpoline since such measures are claimed to be effective on modern Intel${}^{\text{\textregistered}}$ processors, including the one used for our benchmarks~\cite{AffectedIntel,skylake1}.


\paragraph{Comparison with Eclipse~\cite{Eclipse}}
Eclipse~\cite{Eclipse} is a software-level protection measure that
obfuscates the operands of instructions used to perform speculative
probing by attacks like BlindSide~\cite{BlindSide} or
PACMAN~\cite{PACMAN} by inserting artificial data dependencies. The
performance evaluation of Eclipse shows outstanding results on
\emph{kernel-space} execution with maximum overheads smaller than 8\%,
i.e., more than 100 times smaller than our best performing
transformation $\optfencetrans$. This phenomenon is not surprising,
and is due to with the high specificity of Eclipse. As Eclipse only
considers PHT speculation, it can prevent speculative safety
violations by inserting artificial data dependencies with the values
of the guards, in a similar manner as SLH, obtaining good
performance. 
However, it is well known that SLH is not effective when attackers can
control other forms of speculations like BTB speculation~\cite{SLH}.
In contrast, the threat model that we are considering is stronger than
the one adopted by Eclipse~\cite{Eclipse}, as in our case attackers
can also control STL and BTB speculation. As a consequence, our
transformation cannot rely on inserting artificial data dependencies
to prevent the speculative execution of vulnerable
instructions. Another consequence is that the set of instructions that
can be executed speculatively according to Eclipse's criterion---and
therefore need protection---is smaller than for our transformations.
In addition, among the speculatively executable instructions, Eclipse
is only protecting a narrow class of instructions, e.g. indirect
branches, whilst our transformations protect indirect branches, loads
and stores--- i.e. all those instructions that can be used by an
attacker to violate speculative kernel safety.

\section{Conclusion and Future Work}
\label{sec:conclusion}

We have formally demonstrated that the kernel's
layout randomization probabilistically ensures kernel safety for a
classic model, where an attacker cannot compromise the system via
speculative execution or side channels. In this model, users of an
operating system execute without privileges, but victims can feature
pointer arithmetic, introspection, and indirect jumps.
We have also shown that the protection offered by layout randomization
does not naturally scale against attackers that can control
side-channels and speculative execution related to Spectre v1, v2 and
v4.  We stipulate a sufficient condition to enforce kernel safety in
the Spectre era and we proposed mechanisms based on program
transformations that provably enforce {speculative kernel safety} on a
system, provided that this system already enjoys {kernel safety} in
the classic model.  
We proved the soundness of three such transformation, we implemented
them as part of the LLVM/Clang compiler suite, and we established
their performance overhead.  To the best of our knowledge, our work is
the first to formally investigate and provide methods to achieve
kernel safety in the presence of speculative and side-channel
vulnerabilities, while also evaluating the performance of the proposed
mitigations.

\paragraph{Future work} This work prepares the ground for future
developments such as the evaluation of other probabilistic techniques
for the enforcement of safety properties such as Arm's
PA~\cite{ArmpaBTI}, and hardware backed capability machines like
CHERI~\cite{CHERI}.
We also consider the possibility to enhance our model with other
speculative vulnerabilities including LVI, RSB or Phantom speculation,
to develop software-level protection measures enforcing speculative
kernel safety in this stronger attacker model.  


\section*{Acknowledgments}
\label{sec:acks}

  We are grateful to Gilles Barthe, M\'arton Bogn\'ar, Ugo Dal Lago, Lesly-Ann Daniel, Benjamin Gr\'egoire, Jean-Pierre Lozi, Frank Piessens, Aurore Poirier, and Manuel Serrano for their comments on an early draft of the paper. This work was partially supported through the projects PPS ANR-19-C48-0014 and UCA DS4H ANR-17-EURE-0004, and by the Wallenberg AI, Autonomous Systems and Software Program (WASP) funded by the Knut and Alice Wallenberg Foundation.
  
\bibliographystyle{plainnat}
\bibliography{bibliography}

\begin{thebibliography}{86}
\providecommand{\natexlab}[1]{#1}
\providecommand{\url}[1]{\texttt{#1}}
\expandafter\ifx\csname urlstyle\endcsname\relax
  \providecommand{\doi}[1]{doi: #1}\else
  \providecommand{\doi}{doi: \begingroup \urlstyle{rm}\Url}\fi

\bibitem[Abadi and Planul(2013)]{Abadi2}
Mart{\'i}n Abadi and J{\'e}r{\'e}my Planul.
\newblock On layout randomization for arrays and functions.
\newblock In \emph{Principles of Security and Trust}, pages 167--185, Berlin,
  Heidelberg, 2013. Springer.
\newblock ISBN 978-3-642-36830-1.

\bibitem[Abadi and Plotkin(2012)]{Abadi}
Mart\'{\i}n Abadi and Gordon~D. Plotkin.
\newblock On protection by layout randomization.
\newblock \emph{ACM Trans. Inf. Syst. Secur.}, 15\penalty0 (2), jul 2012.
\newblock ISSN 1094-9224.

\bibitem[Abadi et~al.(2005)Abadi, Budiu, Erlingsson, and Ligatti]{CFI}
Mart\'{\i}n Abadi, Mihai Budiu, \'{U}lfar Erlingsson, and Jay Ligatti.
\newblock Control-flow integrity.
\newblock In \emph{Proceedings of the 12th ACM Conference on Computer and
  Communications Security}, page 340–353, New York, NY, USA, 2005. ACM.
\newblock ISBN 1595932267.

\bibitem[Abadi et~al.(2014)Abadi, Planul, and Plotkin]{Abadi3}
Mart{\'i}n Abadi, J{\'e}r{\'e}my Planul, and Gordon~D. Plotkin.
\newblock \emph{Layout Randomization and Nondeterminism}, pages 1--39.
\newblock Springer International Publishing, Berlin, Heidelberg, 2014.

\bibitem[AMD(2022)]{AMDJMP2RET}
AMD.
\newblock Technical guidance for mitigating branch type confusion.
\newblock Technical report, AMD, November 2022.
\newblock URL
  \url{https://www.amd.com/content/dam/amd/en/documents/resources/technical-guidance-for-mitigating-branch-type-confusion.pdf}.
\newblock White Paper.

\bibitem[Ammanaghatta~Shivakumar et~al.(2022)Ammanaghatta~Shivakumar, Barthe,
  Gr\'{e}goire, Laporte, and Priya]{CTPolicies}
Basavesh Ammanaghatta~Shivakumar, Gilles Barthe, Benjamin Gr\'{e}goire, Vincent
  Laporte, and Swarn Priya.
\newblock Enforcing fine-grained constant-time policies.
\newblock In \emph{Proceedings of the 2022 ACM SIGSAC Conference on Computer
  and Communications Security}, page 83–96, New York, NY, USA, 2022. ACM.
\newblock ISBN 9781450394505.

\bibitem[Azevedo~de Amorim et~al.(2018)Azevedo~de Amorim, Hri{\c{t}}cu, and
  Pierce]{PierceMS}
Arthur Azevedo~de Amorim, C{\u{a}}t{\u{a}}lin Hri{\c{t}}cu, and Benjamin~C
  Pierce.
\newblock The meaning of memory safety.
\newblock In \emph{Proceedings of Principles of Security and Trust: 7th
  International Conference}, pages 79--105, Springer, 2018. Springer Berlin,
  Heidelberg.

\bibitem[Barberis et~al.(2022)Barberis, Frigo, Muench, Bos, and
  Giuffrida]{BHInjection}
Enrico Barberis, Pietro Frigo, Marius Muench, Herbert Bos, and Cristiano
  Giuffrida.
\newblock Branch history injection: On the effectiveness of hardware
  mitigations against {Cross-Privilege} spectre-v2 attacks.
\newblock In \emph{31st USENIX Security Symposium (USENIX Security'22)}, pages
  971--988, Boston, MA, August 2022. USENIX Association.

\bibitem[Barthe et~al.(2014)Barthe, Betarte, Campo, Luna, and
  Pichardie]{SLNonInt}
Gilles Barthe, Gustavo Betarte, Juan Campo, Carlos Luna, and David Pichardie.
\newblock System-level non-interference for constant-time cryptography.
\newblock In \emph{Proceedings of the 2014 ACM SIGSAC Conference on Computer
  and Communications Security}, page 1267–1279, New York, NY, USA, 2014. ACM.

\bibitem[Barthe et~al.(2021)Barthe, Cauligi, Grégoire, Koutsos, Liao,
  Oliveira, Priya, Rezk, and Schwabe]{HighAssurance}
Gilles Barthe, Sunjay Cauligi, Benjamin Grégoire, Adrien Koutsos, Kevin Liao,
  Tiago Oliveira, Swarn Priya, Tamara Rezk, and Peter Schwabe.
\newblock High-assurance cryptography in the spectre era.
\newblock In \emph{2021 IEEE Symposium on Security and Privacy (SP)}, pages
  1884--1901, New York, NY, USA, 2021. IEEE.

\bibitem[Behrens et~al.(2022)Behrens, Belay, and Kaashoek]{MitigationOverhead}
Jonathan Behrens, Adam Belay, and M.~Frans Kaashoek.
\newblock Performance evolution of mitigating transient execution attacks.
\newblock In \emph{Proceedings of the Seventeenth European Conference on
  Computer Systems}, page 251–265, New York, NY, USA, 2022. ACM.
\newblock ISBN 9781450391627.

\bibitem[Berger and Zorn(2006)]{DieHard}
Emery~D. Berger and Benjamin~G. Zorn.
\newblock Diehard: Probabilistic memory safety for unsafe languages.
\newblock In \emph{Proceedings of the 27th ACM SIGPLAN Conference on
  Programming Language Design and Implementation}, page 158–168, New York,
  NY, USA, 2006. ACM.
\newblock ISBN 1595933204.

\bibitem[Bovet and Cesati(2005)]{UnderstandingLinux}
Daniel~P Bovet and Marco Cesati.
\newblock \emph{Understanding the Linux Kernel: from I/O ports to process
  management}.
\newblock " O'Reilly Media, Inc.", 2005.

\bibitem[Canella et~al.(2020)Canella, Schwarz, Haubenwallner, Schwarzl, and
  Gruss]{EchoLoad}
Claudio Canella, Michael Schwarz, Martin Haubenwallner, Martin Schwarzl, and
  Daniel Gruss.
\newblock Kaslr: Break it, fix it, repeat.
\newblock In \emph{Proceedings of the 15th ACM Asia Conference on Computer and
  Communications Security}, page 481–493, New York, NY, USA, 2020. ACM.
\newblock ISBN 9781450367509.

\bibitem[Carruth(2018)]{SLH}
Chandler Carruth.
\newblock Speculative load hardening, Sep 2018.
\newblock URL \url{https://llvm.org/docs/SpeculativeLoadHardening.html}.

\bibitem[Cauligi et~al.(2020)Cauligi, Disselkoen, Gleissenthall, Tullsen,
  Stefan, Rezk, and Barthe]{CTFundations}
Sunjay Cauligi, Craig Disselkoen, Klaus~v. Gleissenthall, Dean Tullsen, Deian
  Stefan, Tamara Rezk, and Gilles Barthe.
\newblock Constant-time foundations for the new spectre era.
\newblock In \emph{Proceedings of the 41st ACM SIGPLAN Conference on
  Programming Language Design and Implementation}, page 913–926, New York,
  NY, USA, 2020. ACM.

\bibitem[Chen et~al.(2020)Chen, Lin, and Xing]{KernelElasticObj}
Yueqi Chen, Zhenpeng Lin, and Xinyu Xing.
\newblock A systematic study of elastic objects in kernel exploitation.
\newblock In \emph{Proceedings of the 2020 ACM SIGSAC Conference on Computer
  and Communications Security}, page 1165–1184, New York, NY, USA, 2020. ACM.
\newblock ISBN 9781450370899.

\bibitem[Choudhary et~al.(2021)Choudhary, Yu, Fletcher, and Morrison]{SPT}
Rutvik Choudhary, Jiyong Yu, Christopher Fletcher, and Adam Morrison.
\newblock Speculative privacy tracking (spt): Leaking information from
  speculative execution without compromising privacy.
\newblock In \emph{MICRO-54: 54th Annual IEEE/ACM International Symposium on
  Microarchitecture}, page 607–622, New York, NY, USA, 2021. ACM.
\newblock ISBN 9781450385572.

\bibitem[Christou et~al.(2024)Christou, Gaidis, Atlidakis, and
  Kemerlis]{Eclipse}
Neophytos Christou, Alexander~J. Gaidis, Vaggelis Atlidakis, and Vasileios~P.
  Kemerlis.
\newblock Eclipse: Preventing speculative memory-error abuse with artificial
  data dependencies.
\newblock In \emph{Proceedings of the 2024 ACM SIGSAC Conference on Computer
  and Communications Security}, CCS '24, New York, NY, USA, 2024. ACM.

\bibitem[Corbet(2012)]{LWNSMAP}
Jonathan Corbet.
\newblock Supervisor mode access prevention, 2012.
\newblock URL \url{https://lwn.net/Articles/517475/}.

\bibitem[Corporation(2018{\natexlab{a}})]{IntelGuidelines4}
Intel Corporation.
\newblock Intel analysis of speculative execution side channels.
\newblock Technical report, Intel Corporation, May 2018{\natexlab{a}}.
\newblock URL
  \url{https://www.intel.com/content/www/us/en/content-details/671163/intel-analysis-of-speculative-execution-side-channels.html}.
\newblock White Paper.

\bibitem[Corporation(2018{\natexlab{b}})]{eIBRS}
Intel Corporation.
\newblock Indirect branch restricted speculation, 2018{\natexlab{b}}.
\newblock URL
  \url{https://www.intel.com/content/www/us/en/developer/articles/technical/software-security-guidance/technical-documentation/indirect-branch-restricted-speculation.html}.

\bibitem[Corporation(2018{\natexlab{c}})]{skylake1}
Intel Corporation.
\newblock Retpoline: A branch target injection mitigation, 2018{\natexlab{c}}.
\newblock URL
  \url{https://www.intel.com/content/dam/develop/external/us/en/documents/retpoline-a-branch-target-injection-mitigation.pdf}.

\bibitem[Corporation(2022)]{IntelBHIDISS}
Intel Corporation.
\newblock Branch history injection and intra-mode branch target injection /
  cve-2022-0001, cve-2022-0002 / intel-sa-00598.
\newblock Technical report, Intel Corporation, May 2022.
\newblock URL
  \url{https://www.intel.com/content/www/us/en/developer/articles/technical/software-security-guidance/technical-documentation/branch-history-injection.html}.
\newblock Technical Documentation.

\bibitem[Corporation(2025)]{AffectedIntel}
Intel Corporation.
\newblock Affected processors: Guidance for security issues on
  intel${}^{\text{\textregistered}}$ processors, 2025.
\newblock URL
  \url{https://www.intel.com/content/www/us/en/developer/topic-technology/software-security-guidance/processors-affected-consolidated-product-cpu-model.html}.

\bibitem[CrossDB(2024)]{crossdb}
CrossDB.
\newblock Bench test, 2024.
\newblock URL \url{https://crossdb.org/get-started/bench/}.

\bibitem[cryptsetup Group(2025)]{cryptsetup}
cryptsetup Group.
\newblock cryptsetup, 2025.
\newblock URL \url{https://gitlab.com/cryptsetup/cryptsetup}.

\bibitem[Daniel et~al.(2023)Daniel, Bognar, Noorman, Bardin, Rezk, and
  Piessens]{ProSpeCT}
Lesly-Ann Daniel, Marton Bognar, Job Noorman, S{\'e}bastien Bardin, Tamara
  Rezk, and Frank Piessens.
\newblock {ProSpeCT}: Provably secure speculation for the {Constant-Time}
  policy.
\newblock In \emph{32nd USENIX Security Symposium (USENIX Security 23)}, pages
  7161--7178, Anaheim, CA, August 2023. USENIX Association.
\newblock ISBN 978-1-939133-37-3.

\bibitem[Davoli(2024)]{repo}
Davide Davoli.
\newblock Comprehensive kernel safety in the spectre era, 2024.
\newblock URL
  \url{https://gitlab.inria.fr/ddavoli/comprehensive-kernel-safety-in-the-spectre-era}.

\bibitem[Davoli et~al.(2024)Davoli, Avanzini, and Rezk]{ConferencePaper}
Davide Davoli, Martin Avanzini, and Tamara Rezk.
\newblock On kernel's safety in the spectre era (and kaslr is formally dead).
\newblock In \emph{Proceedings of the 2024 {ACM} {SIGSAC} Conference on
  Computer and Communications Security, {CCS} 2024, USA, October 14-18, 2023},
  New York, NY, USA, 2024. {ACM}.

\bibitem[de~Raadt(2017)]{OpenBSD6.3}
Theo de~Raadt.
\newblock Openbsd 6.3, Oct 2017.
\newblock URL \url{https://www.openbsd.org/33.html}.

\bibitem[Edge(2013)]{LWNKASLR}
Jake Edge.
\newblock Kernel address space layout randomization, 2013.
\newblock URL \url{https://lwn.net/Articles/569635/}.

\bibitem[Evtyushkin et~al.(2016)Evtyushkin, Ponomarev, and
  Abu-Ghazaleh]{JumpOverASLR}
Dmitry Evtyushkin, Dmitry Ponomarev, and Nael Abu-Ghazaleh.
\newblock Jump over aslr: attacking branch predictors to bypass aslr.
\newblock In \emph{The 49th Annual IEEE/ACM International Symposium on
  Microarchitecture}, New York, NY, USA, 2016. IEEE.

\bibitem[Fischer(2011)]{SMEP}
Stephen Fischer.
\newblock Supervisor mode execution protection, 2011.
\newblock URL
  \url{https://www.ncsi.com/nsatc11/presentations/wednesday/emerging_technologies/fischer.pdf}.

\bibitem[Foundation(2025)]{ab}
The Apache~Software Foundation.
\newblock ab -- apache http server benchmarking tool, 2025.
\newblock URL \url{https://httpd.apache.org/docs/2.4/programs/ab.html}.

\bibitem[Garnier(2016)]{FreelistRandomization}
Thomas Garnier.
\newblock Randomizing the linux kernel heap freelists, Sep 2016.
\newblock URL
  \url{https://mxatone.medium.com/randomizing-the-linux-kernel-heap-freelists-b899bb99c767}.

\bibitem[Ge et~al.(2016)Ge, Talele, Payer, and Jaeger]{FineGrainedkCFI}
Xinyang Ge, Nirupama Talele, Mathias Payer, and Trent Jaeger.
\newblock Fine-grained control-flow integrity for kernel software.
\newblock In \emph{2016 IEEE European Symposium on Security and Privacy
  (EuroS{\&}P)}, pages 179--194, New York, NY, USA, 2016. IEEE.
\newblock \doi{10.1109/EuroSP.2016.24}.

\bibitem[Gleixner(2022)]{CDT}
Thomas Gleixner.
\newblock Fix rsb fill on context switch for serialize.
\newblock
  \url{https://lore.kernel.org/all/20220716230344.239749011@linutronix.de/},
  July 2022.
\newblock Linux Kernel Mailing List.

\bibitem[Goguen and Meseguer(1982)]{NonInt}
J.~A. Goguen and J.~Meseguer.
\newblock Security policies and security models.
\newblock In \emph{1982 IEEE Symposium on Security and Privacy}, New York, NY,
  USA, 1982. IEEE.

\bibitem[G\"{o}ktas et~al.(2020)G\"{o}ktas, Razavi, Portokalidis, Bos, and
  Giuffrida]{BlindSide}
Enes G\"{o}ktas, Kaveh Razavi, Georgios Portokalidis, Herbert Bos, and
  Cristiano Giuffrida.
\newblock Speculative probing: Hacking blind in the spectre era.
\newblock In \emph{Proceedings of the 2020 ACM SIGSAC Conference on Computer
  and Communications Security}, page 1871–1885, New York, NY, USA, 2020. ACM.

\bibitem[Gruss et~al.(2016)Gruss, Maurice, Fogh, Lipp, and Mangard]{cacheKASLR}
Daniel Gruss, Cl\'{e}mentine Maurice, Anders Fogh, Moritz Lipp, and Stefan
  Mangard.
\newblock Prefetch side-channel attacks: Bypassing smap and kernel aslr.
\newblock In \emph{Proceedings of the 2016 ACM SIGSAC Conference on Computer
  and Communications Security}, page 368–379, New York, NY, USA, 2016. ACM.
\newblock ISBN 9781450341394.

\bibitem[Gruss et~al.(2017)Gruss, Lipp, Schwarz, Fellner, Maurice, and
  Mangard]{Kaiser}
Daniel Gruss, Moritz Lipp, Michael Schwarz, Richard Fellner, Cl{\'e}mentine
  Maurice, and Stefan Mangard.
\newblock Kaslr is dead: Long live kaslr.
\newblock In \emph{Engineering Secure Software and Systems}, pages 161--176,
  Berlin, Heidelberg, 2017. Springer International Publishing.
\newblock ISBN 978-3-319-62105-0.

\bibitem[Guarnieri et~al.(2020)Guarnieri, Köpf, Morales, Reineke, and
  Sánchez]{Spectector}
Marco Guarnieri, Boris Köpf, José~F. Morales, Jan Reineke, and Andrés
  Sánchez.
\newblock Spectector: Principled detection of speculative information flows.
\newblock In \emph{2020 IEEE Symposium on Security and Privacy (SP)}, pages
  1--19, New York, NY, USA, 2020. IEEE.

\bibitem[Guarnieri et~al.(2021)Guarnieri, Köpf, Reineke, and Vila]{Contracts}
Marco Guarnieri, Boris Köpf, Jan Reineke, and Pepe Vila.
\newblock Hardware-software contracts for secure speculation.
\newblock In \emph{2021 IEEE Symposium on Security and Privacy (SP)}, pages
  1868--1883, New York, NY, USA, 2021. IEEE.

\bibitem[Hund et~al.(2013)Hund, Willems, and Holz]{TLBKASLR}
Ralf Hund, Carsten Willems, and Thorsten Holz.
\newblock Practical timing side channel attacks against kernel space aslr.
\newblock In \emph{2013 IEEE Symposium on Security and Privacy}, pages
  191--205, New York, NY, USA, 2013. IEEE.

\bibitem[Inc.(2011)]{MacOSXASLR}
Apple Inc.
\newblock Mac os x has you covered, May 2011.
\newblock URL \url{http://www.apple.com/macosx/security/}.

\bibitem[Inc(2025)]{ugrep}
Genivia Inc.
\newblock ugrep, 2025.
\newblock URL \url{https://github.com/Genivia/ugrep-benchmarks}.

\bibitem[Int(2023)]{IntelManual}
\emph{Intel \textregistered 64 and IA-32 Architectures Software Developer's
  Manualx}.
\newblock Intel Corporation, September 2023.

\bibitem[kernel~development community(2023)]{PTI}
The kernel~development community.
\newblock Page table isolation (pti), 2023.
\newblock URL \url{https://www.kernel.org/doc/html/next/x86/pti.html}.

\bibitem[Kocher et~al.(2019)Kocher, Horn, Fogh, Genkin, Gruss, Haas, Hamburg,
  Lipp, Mangard, Prescher, Schwarz, and Yarom]{Spectre}
Paul Kocher, Jann Horn, Anders Fogh, Daniel Genkin, Daniel Gruss, Werner Haas,
  Mike Hamburg, Moritz Lipp, Stefan Mangard, Thomas Prescher, Michael Schwarz,
  and Yuval Yarom.
\newblock Spectre attacks: Exploiting speculative execution.
\newblock In \emph{2019 IEEE Symposium on Security and Privacy (SP)}, pages
  1--19, New York, NY, USA, 2019. IEEE.

\bibitem[Koschel et~al.(2020)Koschel, Giuffrida, Bos, and Razavi]{TagBleed}
Jakob Koschel, Cristiano Giuffrida, Herbert Bos, and Kaveh Razavi.
\newblock Tagbleed: Breaking kaslr on the isolated kernel address space using
  tagged tlbs.
\newblock In \emph{2020 IEEE European Symposium on Security and Privacy
  (EuroS\&P)}, pages 309--321, New York, NY, USA, 2020. IEEE.

\bibitem[Koschel et~al.(2023)Koschel, Borrello, D{\textquoteright}Elia, Bos,
  and Giuffrida]{Uncontained}
Jakob Koschel, Pietro Borrello, Daniele~Cono D{\textquoteright}Elia, Herbert
  Bos, and Cristiano Giuffrida.
\newblock Uncontained: Uncovering container confusion in the linux kernel.
\newblock In \emph{32nd USENIX Security Symposium (USENIX Security 23)}, pages
  5055--5072, Anaheim, CA, August 2023. USENIX Association.
\newblock ISBN 978-1-939133-37-3.

\bibitem[Lattner and Adve(2004)]{LLVM}
C.~Lattner and V.~Adve.
\newblock Llvm: a compilation framework for lifelong program analysis \&
  transformation.
\newblock In \emph{International Symposium on Code Generation and Optimization,
  2004. CGO 2004.}, pages 75--86, New York, NY, USA, 2004. IEEE.

\bibitem[Li et~al.(2011)Li, Wang, Bletsch, Srinivasan, Grace, and
  Jiang]{CompProtKer}
Jinku Li, Zhi Wang, Tyler Bletsch, Deepa Srinivasan, Michael Grace, and Xuxian
  Jiang.
\newblock Comprehensive and efficient protection of kernel control data.
\newblock \emph{IEEE Transactions on Information Forensics and Security},
  6\penalty0 (4):\penalty0 1404--1417, 2011.

\bibitem[Limited(2022)]{ArmpaBTI}
Arm Limited.
\newblock Learn the architecture -- providing protection for complex software,
  2022.
\newblock URL \url{https://developer.arm.com/documentation/102433/0100}.

\bibitem[Lipp et~al.(2018)Lipp, Schwarz, Gruss, Prescher, Haas, Fogh, Horn,
  Mangard, Kocher, Genkin, Yarom, and Hamburg]{Meltdown}
Moritz Lipp, Michael Schwarz, Daniel Gruss, Thomas Prescher, Werner Haas,
  Anders Fogh, Jann Horn, Stefan Mangard, Paul Kocher, Daniel Genkin, Yuval
  Yarom, and Mike Hamburg.
\newblock Meltdown: Reading kernel memory from user space.
\newblock In \emph{27th USENIX Security Symposium (USENIX Security 18)}, pages
  973--990, Baltimore, MD, August 2018. USENIX Association.
\newblock ISBN 978-1-939133-04-5.

\bibitem[Liu et~al.(2023{\natexlab{a}})Liu, Ravichandran, and Yan]{EntryBleed}
William Liu, Joseph Ravichandran, and Mengjia Yan.
\newblock Entrybleed: A universal kaslr bypass against kpti on linux.
\newblock In \emph{Proceedings of the 12th International Workshop on Hardware
  and Architectural Support for Security and Privacy}, HASP '23, page 10–18,
  New York, NY, USA, 2023{\natexlab{a}}. ACM.
\newblock ISBN 9798400716232.

\bibitem[Liu et~al.(2023{\natexlab{b}})Liu, Lin, Chen, Wu, Zou, Mu, and
  Xing]{KernelBugs}
Ziqin Liu, Zhenpeng Lin, Yueqi Chen, Yuhang Wu, Yalong Zou, Dongliang Mu, and
  Xinyu Xing.
\newblock Towards unveiling exploitation potential with multiple error
  behaviors for kernel bugs.
\newblock \emph{IEEE Transactions on Dependable and Secure Computing},
  21\penalty0 (1):\penalty0 1--18, 2023{\natexlab{b}}.

\bibitem[Lottarini et~al.(2018)Lottarini, Ramirez, Coburn, Kim, Ranganathan,
  Stodolsky, and Wachsler]{vbench}
Andrea Lottarini, Alex Ramirez, Joel Coburn, Martha~A. Kim, Parthasarathy
  Ranganathan, Daniel Stodolsky, and Mark Wachsler.
\newblock vbench: Benchmarking video transcoding in the cloud.
\newblock In \emph{Proceedings of the Twenty-Third International Conference on
  Architectural Support for Programming Languages and Operating Systems},
  ASPLOS '18, page 797–809, New York, NY, USA, 2018. Association for
  Computing Machinery.
\newblock ISBN 9781450349116.
\newblock \doi{10.1145/3173162.3173207}.
\newblock URL \url{https://doi.org/10.1145/3173162.3173207}.

\bibitem[Mambretti et~al.(2021)Mambretti, Sandulescu, Sorniotti, Robertson,
  Kirda, and Kurmus]{SPEAR}
A.~Mambretti, A.~Sandulescu, A.~Sorniotti, W.~Robertson, E.~Kirda, and
  A.~Kurmus.
\newblock Bypassing memory safety mechanisms through speculative control flow
  hijacks.
\newblock In \emph{2021 IEEE European Symposium on Security and Privacy
  (EuroS{\&}P)}, pages 633--649, Los Alamitos, CA, USA, sep 2021. IEEE Computer
  Society.

\bibitem[Mandt(2013)]{iOSASLR}
Tarjei Mandt.
\newblock Attacking the ios kernel: A look at ‘evasi0n’, March 2013.
\newblock URL \url{https://papers.put.as/papers/ios/2013/NISlecture201303.pdf}.

\bibitem[Maste(2023)]{FreeBSDASLR}
Ed~Maste.
\newblock Address space layout randomization (aslr), July 2023.
\newblock URL \url{https://wiki.freebsd.org/AddressSpaceLayoutRandomization}.

\bibitem[Michael et~al.(2023)Michael, Gollamudi, Bosamiya, Johnson, Denlinger,
  Disselkoen, Watt, Parno, Patrignani, Vassena, and Stefan]{MSWasm}
Alexandra~E. Michael, Anitha Gollamudi, Jay Bosamiya, Evan Johnson, Aidan
  Denlinger, Craig Disselkoen, Conrad Watt, Bryan Parno, Marco Patrignani,
  Marco Vassena, and Deian Stefan.
\newblock Mswasm: Soundly enforcing memory-safe execution of unsafe code.
\newblock In \emph{Proceedings of the 50th ACM SIGACT-SIGPLAN Symposium on
  Principles of Programming Languages}, volume~7 of \emph{POPL '23}, New York,
  NY, USA, jan 2023. ACM.

\bibitem[Moreira et~al.(2017)Moreira, Rigo, Polychronakis, and Kemerlis]{kCFI}
Jo{\~a}o Moreira, Sandro Rigo, Michalis Polychronakis, and Vasileios~P
  Kemerlis.
\newblock Drop the rop fine-grained control-flow integrity for the linux
  kernel, 2017.

\bibitem[Nagarakatte et~al.(2009)Nagarakatte, Zhao, Martin, and
  Zdancewic]{SoftBound}
Santosh Nagarakatte, Jianzhou Zhao, Milo~M.K. Martin, and Steve Zdancewic.
\newblock Softbound: Highly compatible and complete spatial memory safety for
  c.
\newblock \emph{SIGPLAN Not.}, 44\penalty0 (6):\penalty0 245–258, 2009.

\bibitem[Patrignani and Guarnieri(2021)]{Exorcising}
Marco Patrignani and Marco Guarnieri.
\newblock Exorcising spectres with secure compilers.
\newblock In \emph{Proceedings of the 2021 ACM SIGSAC Conference on Computer
  and Communications Security}, page 445–461, New York, NY, USA, 2021. ACM.
\newblock ISBN 9781450384544.

\bibitem[Project(2022)]{AndroidASLR}
Android Open~Source Project.
\newblock Kernel hardening, August 2022.
\newblock URL
  \url{https://source.android.com/docs/core/architecture/kernel/hardening}.

\bibitem[Project(2025)]{SESES}
The~LLVM Project.
\newblock X86speculativeexecutionsideeffectsuppression.cpp file reference,
  2025.
\newblock URL
  \url{https://www.llvm.org/doxygen/X86SpeculativeExecutionSideEffectSuppression_8cpp.html}.

\bibitem[Proven(2022)]{RustInLinux}
Liam Proven.
\newblock Linux 6.1: Rust to hit mainline kernel, October 2022.
\newblock URL
  \url{https://www.theregister.com/2022/10/05/rust_kernel_pull_request_pulled/}.

\bibitem[Ravichandran et~al.(2022)Ravichandran, Na, Lang, and Yan]{PACMAN}
Joseph Ravichandran, Weon~Taek Na, Jay Lang, and Mengjia Yan.
\newblock Pacman: Attacking arm pointer authentication with speculative
  execution.
\newblock In \emph{Proceedings of the 49th Annual International Symposium on
  Computer Architecture}, pages 685--698, New York, NY, USA, 2022. ACM.

\bibitem[Reshetova et~al.(2018)Reshetova, Liljestrand, Paverd, and
  Asokan]{towardkernelsafety}
Elena Reshetova, Hans Liljestrand, Andrew Paverd, and N~Asokan.
\newblock Toward linux kernel memory safety.
\newblock \emph{Software: Practice and Experience}, 48\penalty0 (12):\penalty0
  2237--2256, 2018.

\bibitem[S and Nikolenko(2022)]{HeapFengShui2022}
Michael S and Vitaly Nikolenko.
\newblock Linux kernel heap feng shui in 2022, May 2022.
\newblock URL \url{https://duasynt.com/blog/linux-kernel-heap-feng-shui-2022}.

\bibitem[SecurityScorecard(2022)]{ThreatOverview}
SecurityScorecard.
\newblock Threat overview for linux kernel, November 2022.
\newblock URL
  \url{https://www.cvedetails.com/product/47/Linux-Linux-Kernel.html}.

\bibitem[Shacham et~al.(2004)Shacham, Page, Pfaff, Goh, Modadugu, and
  Boneh]{ApacheAttack}
Hovav Shacham, Matthew Page, Ben Pfaff, Eu-Jin Goh, Nagendra Modadugu, and Dan
  Boneh.
\newblock On the effectiveness of address-space randomization.
\newblock In \emph{Proceedings of the 11th ACM Conference on Computer and
  Communications Security}, page 298–307, New York, NY, USA, 2004. ACM.
\newblock ISBN 1581139616.

\bibitem[Tanenbaum and Bos(2014)]{Tanenbaum}
Andrew~S. Tanenbaum and Herbert Bos.
\newblock \emph{Modern Operating Systems}.
\newblock Prentice Hall Press, USA, 4th edition, 2014.
\newblock ISBN 013359162X.

\bibitem[Team(2003)]{PaXASLR}
PaX Team.
\newblock Documentation for the pax project, 2003.
\newblock URL \url{https://pax.grsecurity.net/docs/}.

\bibitem[Trujillo et~al.(2023)Trujillo, Wikner, and Razavi]{Inception}
Dani{\"e}l Trujillo, Johannes Wikner, and Kaveh Razavi.
\newblock Inception: Exposing new attack surfaces with training in transient
  execution.
\newblock In \emph{32nd USENIX Security Symposium (USENIX Security 23)}, pages
  7303--7320, Anaheim, CA, August 2023. USENIX Association.
\newblock ISBN 978-1-939133-37-3.
\newblock URL
  \url{https://www.usenix.org/conference/usenixsecurity23/presentation/trujillo}.

\bibitem[Turner(2018)]{Retpoline}
Paul Turner.
\newblock Retpoline: a software construct for preventing
  branch-target-injection, 2018.
\newblock URL \url{https://support.google.com/faqs/answer/7625886}.

\bibitem[Wang and Jiang(2010)]{HyperSafe}
Zhi Wang and Xuxian Jiang.
\newblock Hypersafe: A lightweight approach to provide lifetime hypervisor
  control-flow integrity.
\newblock In \emph{2010 IEEE Symposium on Security and Privacy}, pages
  380--395, New York, NY, USA, 2010. IEEE.
\newblock \doi{10.1109/SP.2010.30}.

\bibitem[Watson et~al.(2015)Watson, Woodruff, Neumann, Moore, Anderson,
  Chisnall, Dave, Davis, Gudka, Laurie, Murdoch, Norton, Roe, Son, and
  Vadera]{CHERI}
Robert~N.M. Watson, Jonathan Woodruff, Peter~G. Neumann, Simon~W. Moore,
  Jonathan Anderson, David Chisnall, Nirav Dave, Brooks Davis, Khilan Gudka,
  Ben Laurie, Steven~J. Murdoch, Robert Norton, Michael Roe, Stacey Son, and
  Munraj Vadera.
\newblock Cheri: A hybrid capability-system architecture for scalable software
  compartmentalization.
\newblock In \emph{2015 IEEE Symposium on Security and Privacy}, pages 20--37,
  New York, NY, USA, 2015. IEEE.

\bibitem[Weisse et~al.(2019)Weisse, Neal, Loughlin, Wenisch, and Kasikci]{NDA}
Ofir Weisse, Ian Neal, Kevin Loughlin, Thomas~F. Wenisch, and Baris Kasikci.
\newblock Nda: Preventing speculative execution attacks at their source.
\newblock In \emph{Proceedings of the 52nd Annual IEEE/ACM International
  Symposium on Microarchitecture}, page 572–586, New York, NY, USA, 2019.
  ACM.
\newblock ISBN 9781450369381.

\bibitem[Wiebing et~al.(2024)Wiebing, de~Faveri~Tron, Bos, and
  Giuffrida]{InSpectre}
Sander Wiebing, Alvise de~Faveri~Tron, Herbert Bos, and Cristiano Giuffrida.
\newblock {InSpectre} gadget: Inspecting the residual attack surface of
  cross-privilege spectre v2.
\newblock In \emph{33rd USENIX Security Symposium (USENIX Security 24)}, pages
  577--594, Philadelphia, PA, August 2024. USENIX Association.
\newblock ISBN 978-1-939133-44-1.

\bibitem[Wikner and Razavi(2022)]{Retbleed}
Johannes Wikner and Kaveh Razavi.
\newblock {RETBLEED}: Arbitrary speculative code execution with return
  instructions.
\newblock In \emph{31st USENIX Security Symposium (USENIX Security 22)}, pages
  3825--3842, Boston, MA, August 2022. USENIX Association.
\newblock ISBN 978-1-939133-31-1.

\bibitem[Wikner et~al.(2023)Wikner, Trujillo, and Razavi]{Phantom}
Johannes Wikner, Dani\"{e}l Trujillo, and Kaveh Razavi.
\newblock Phantom: Exploiting decoder-detectable mispredictions.
\newblock In \emph{Proceedings of the 56th Annual IEEE/ACM International
  Symposium on Microarchitecture}, MICRO '23, page 49–61, New York, NY, USA,
  2023. Association for Computing Machinery.
\newblock ISBN 9798400703294.
\newblock \doi{10.1145/3613424.3614275}.
\newblock URL \url{https://doi.org/10.1145/3613424.3614275}.

\bibitem[Yu et~al.(2019)Yu, Yan, Khyzha, Morrison, Torrellas, and
  Fletcher]{STT}
Jiyong Yu, Mengjia Yan, Artem Khyzha, Adam Morrison, Josep Torrellas, and
  Christopher~W. Fletcher.
\newblock Speculative taint tracking (stt): A comprehensive protection for
  speculatively accessed data.
\newblock In \emph{Proceedings of the 52nd Annual IEEE/ACM International
  Symposium on Microarchitecture}, page 954–968, New York, NY, USA, 2019.
  ACM.
\newblock ISBN 9781450369381.

\bibitem[Zhang et~al.(2023)Zhang, Barthe, Chuengsatiansup, Schwabe, and
  Yarom]{UltimateSLH}
Zhiyuan Zhang, Gilles Barthe, Chitchanok Chuengsatiansup, Peter Schwabe, and
  Yuval Yarom.
\newblock Ultimate slh: taking speculative load hardening to the next level.
\newblock In \emph{Proceedings of the 32nd USENIX Conference on Security
  Symposium}, USA, 2023. USENIX Association.
\newblock ISBN 978-1-939133-37-3.

\end{thebibliography}

\pagebreak

\appendix

\section{Proofs}
\label{sec:appendix}

\subsection{Appendix for \Cref{sec:safety1}}
\label{sec:appsafety1}

\subsubsection{Omitted Proofs and Results}
\label{sec:proofs1}

We begin by introducing additional notation that will be used throughout this section. We then present the proof of \Cref{thm:scenario1}, followed by the proofs of the intermediate results upon which this proof relies.

\paragraph*{Semantics' Notation}
We write $\update{{\rfs}}{(\ar, i)} \val$ to denote the store that is identical to $\rfs$ except for the array $\ar$, which remains pointwise equal to $\rfs(\ar)$ except at index $i$, where it is updated to $\val$.

We use the notation $\nf \confone{}$ when $\confone$ does not reduce
in system $\system$ and using layout $\lay$.
$\nstep{!} \confone \conftwo$ to represent the formula
$\exists \nat'. \nstep{\nat'} \confone \conftwo \land \nf \conftwo {}$. Similarly, we use $\nstep{!\nat} \confone \conftwo$ to
denote the formula
$\exists \nat' \leq \nat. \nstep{\nat'} \confone \conftwo \land
\nf \conftwo {}$.

Additionally, we generalize the function $\Eval\cdot$ to take full configurations as arguments. This allows us to use it in the following way:
\[
  \Eval {\conf{\st, \lay\lcomp\rfs}} \defsym
  \begin{cases}
    \Omega
    & \text{if $\diverge{\conf{\st, \lay\lcomp\rfs}}$},\\
    (\val,\store')
    & \text{if } \lay \red \conf{\st, \lay\lcomp\rfs}  \to^* \conf{\frame{\cnil}{\regmap[\ret \mapsto \val]}{\opt}, \lay \lcomp \store'}, \\
    \err
    & \text{if } \lay \red \conf{\st, \lay\lcomp\rfs} \to^* \err, \\
    \unsafe
    & \text{if } \lay \red \conf{\st, \lay\lcomp\rfs} \to^* \unsafe.
  \end{cases}
\]
\noindent
We also extend the relation $\evaleq$ by stating that:
\[
(\val,\store) \evaleq \conf{\frame{\cnil}{\regmap[\ret \mapsto \val]}{\opt}, \lay \lcomp \store}
\]
holds for every layout $\lay$, register map $\regmap$, and every flag $\opt$. Furthermore, we take its symmetric and transitive closure.

\paragraph*{User- and Kernel-mode stack}
We write $\um(\cmd)$ as a shorthand for $\ids(\cmd) \subseteq \Idu$. The predicate  $\ids(\st)$ is defined inductively as follows:
\[
  \um(\nil)\defsym \top \quad \um(\fr:\st)\defsym \um(\fr) \land \um(\st)\quad \um(\frame{\cmd}{\regmap}{\opt}) \defsym \um(\cmd) \land \opt=\um. 
\]
The predicate $\km[\syscall]$ is defined analogously:
\[
  \km[\syscall](\nil)\defsym \top \quad \km[\syscall](\fr:\st)\defsym \km[\syscalltwo](\fr) \land \km[\syscall](\st)\quad \km[\syscall](\frame{\cmd}{\regmap}{\opt}) \defsym \km(\cmd) \land \opt=\km[\syscall]. 
\]
We write $\km(\cdot)$ as a shorthand for $\exists \syscall \in \Sys. \km[\syscall](\cdot)$.

\begin{proof}[Proof of \Cref{thm:scenario1}]
  Let $\system=(\rfs, \syss, \caps)$ be a system. We aim to show that for every \emph{unprivileged} command $\cmd$,
register map $\regmap$, and layout distribution $\mu$, the following holds:
\[
    \Pr_{\lay \leftarrow\mu}\left[\Eval {\conf{\frame{\cmd}{\regmap}{\um}, \lay \lcomp \rfs}} = \unsafe \right]=
      \Pr_{\lay \leftarrow\mu}\left[\nstep! {\conf{\frame{\cmd}{\regmap}{\um}, \lay \lcomp \rfs}} \unsafe \right] \le
    1-\delta_\mu.
  \]
  If the probability is not 0, by \Cref{lemma:premiseApreserv} (that we can apply because all the system calls in $\system$ are \emph{layout non-interferent}, and because all the layouts are identical with respect to user-space identifiers), we deduce that:
  \[
    \forall \lay. \Eval[\system][\lay]{\conf{\frame{\cmd}{\regmap}{\um}, \lay \lcomp \rfs}} \evaleq \unsafe.
  \]
  This proposition can be rewritten as:
  \[
    \forall \lay. \Eval[\system][\lay]{\conf{\frame{\cmd}{\regmap}{\um}, \lay \lcomp \rfs}} \in \{\unsafe, \err\},
  \]
  which allows us to deduce that:
  \[
    \forall \lay. \exists \conftwo_\lay, \nat_\lay. \nstep{!\nat_\lay} {\conf{\frame{\cmd}{\regmap}{\um}, \lay \lcomp \rfs}} {\conftwo_\lay}.
  \]
  Since memory size is finite, also $\Lay$ is finite, so there is a natural number $\overline \nat$ such that $\overline \nat \ge \max_{\lay \in \Lay}{\nat_\lay}$.
  Thus, we obtain:
  \[
    \forall \lay. \exists \conftwo_\lay. \nstep{!\overline \nat} {\conf{\frame{\cmd}{\regmap}{\um}, \lay \lcomp \rfs}} {\conftwo_\lay}.
  \]
  Applying \Cref{lemma:mainlemmapms} to the system $\system$, the
  bound $\overline \nat$, the unprivileged attacker $\cmd$, the
  initial register map $\regmap$, the distribution $\mu$, and to the last
  intermediate claim, we deduce that one of the following statements
  holds:
  \begin{varitemize}
  \item There are a register map $\regmap'$ and a store $\rfs'\eqon{\Fun}\rfs$
    such that for every layout $\lay$:
    \[
      \nstep {!\overline \nat}{\conf{\frame{\cmd}{\regmap}{\um}, \lay \lcomp \rfs}}{\conf{\frame{\cnil}{\regmap}{\um}, \lay \lcomp \rfs'}},
    \]
  \item
    $
      \Pr_{\lay \leftarrow \mu}\left[\nstep {!\overline \nat} {\conf{\frame{\cmd}{\regmap}{\um}, \lay \lcomp \rfs}}\err \right] \ge \delta_\mu.
    $
  \end{varitemize}
  We go by cases on these two statements. If the first one holds,
  then for every layout $\lay$: 
  \[
    \nstep {!}{\conf{\frame{\cmd}{\regmap}{\um}, \lay \lcomp \rfs}}{\conf{\frame{\cnil}{\regmap'}{\um}, \lay \lcomp \rfs'}}.
  \]
  Consequently, we have
  \[
    \Pr_{\lay \leftarrow \mu}\left[    \nstep {!}{\conf{\frame{\cmd}{\regmap}{\um}, \lay \lcomp \rfs}}{\conf{\frame{\cnil}{\regmap'}{\um}, \lay \lcomp \rfs'}} \right] = 1.
  \]
  Since the final configuration cannot be contemporary $\unsafe$
  and ${\conf{\frame{\cnil}{\regmap'}{\um}, \lay \lcomp \rfs'}}$, we conclude:
  \begin{equation*}
    \Pr_{\lay \leftarrow \mu}\left[    \nstep {!}{\conf{\frame{\cmd}{\regmap}{\um}, \lay \lcomp \rfs}}{\unsafe} \right] = 0.
  \end{equation*}
  This established the claim. If the second proposition holds then:
  \begin{equation*}
    \Pr_{\lay \leftarrow \mu}\left[\nstep {!} {\conf{\frame{\cmd}{\regmap}{\um}, \lay \lcomp \rfs}}\err \right] \ge\\
    \Pr_{\lay \leftarrow \mu}\left[\nstep {!\overline \nat} {\conf{\frame{\cmd}{\regmap}{\um}, \lay \lcomp \rfs}}\err \right] \ge \delta_\mu,
  \end{equation*}
  and therefore we conclude that:
  \begin{equation*}
    \Pr_{\lay \leftarrow \mu}\left[\nstep {!} {\conf{\frame{\cmd}{\regmap}{\um}, \lay \lcomp \rfs}}\unsafe \right] \le\\
    1- \Pr_{\lay \leftarrow \mu}\left[\nstep {!} {\conf{\frame{\cmd}{\regmap}{\um}, \lay \lcomp \rfs}}\err \right] \le 1-\delta_\mu.
  \end{equation*}
\end{proof}

\begin{lemma}[Preservation of \emph{layout non-interference}]
  \label{lemma:premiseApreserv}
  If every system call in $\system$ is \emph{layout non-interferent},
  then the entire system is also \emph{layout non-interferent} with
  respect to \emph{unprivileged} attackers, in the following sense:
  for every \emph{unprivileged} attacker $\cmd$, register map
  $\regmap$, pair of layouts $\lay_1, \lay_2$, and configuration
  $\confone_1$ we have that:
  \[
    \Eval[\system][\lay_1]{\conf{\frame{\cmd}{\regmap}{\um}, \lay_1\lcomp \rfs}} \evaleq \Eval[\system][\lay_2]{\conf{\frame{\cmd}{\regmap}{\um}, \lay_2\lcomp \rfs}}
  \]  
\end{lemma}
\begin{proof}
  The main claim is a consequence of an auxiliary claim, namely:
    \begin{equation*}
    \nstep[\system][\lay_1]!{\conf{\frame{\cmd}{\regmap}{\um}, \lay_1\lcomp \rfs}}{\confone_1} \Rightarrow
    \Eval[\system][\lay_2]{\conf{\frame{\cmd}{\regmap}{\um}, \lay_2\lcomp \rfs}} \evaleq \confone_1
    \tag{C}
  \end{equation*}
  To prove (C), we establish a slightly stronger statement: instead of
  quantifying over $\cmd$ and $\regmap$, we generalize to a non-empty
  stack $\st$ such that $\um(\st)$, and we rewrite the claim in a more
  convenient shape, where all the meta-variables that are not
  quantified explicitly are quantified universally.
  \begin{equation*}
    \forall \nat. \forall \rfs'\eqon{\Fun} \rfs. \left(\nstep[\system][\lay_1]\nat{\conf{\st, \lay_1\lcomp \rfs'}}{\confone_1} \land \final[\system][\lay_1] {\confone_1}\right) \Rightarrow\\
    \Eval[\system][\lay_2]{\conf{\st, \lay_2\lcomp \rfs'}} \evaleq \confone_1
    \tag{C'}
  \end{equation*}
  Once (C') is established, we derive (C) by instantiating $\rfs'$ as
  $\rfs$ and $\st$ as
  ${\conf{\frame{\cmd}{\regmap}{\um}, \lay_1\lcomp \rfs}}$. We fix
  $\nat$ as the number of steps taken by the reduction from the
  initial configuration to the terminal configuration
  $\confone_1$. From (C), we can deduce the main claim by case
  analysis:
  \begin{proofcases}
    \proofcase{$\Eval[\system][\lay_1]{\conf{\frame{\cmd}{\regmap}{\um},
          \lay_1\lcomp \rfs}} = \Omega$} In this case, the value of
    $\Eval[\system][\lay_2]{\conf{\frame{\cmd}{\regmap}{\um},
        \lay_2\lcomp \rfs}}$ must also be $\Omega$. If this was not
    the case, we could apply (C) with $\lay_1$ and $\lay_2$ swapped to
    prove that
    $\Eval[\system][\lay_1]{\conf{\frame{\cmd}{\regmap}{\um},
        \lay_1\lcomp \rfs}}\neq \Omega$.
    \proofcase{$\Eval[\system][\lay_1]{\conf{\frame{\cmd}{\regmap}{\um},
          \lay_1\lcomp \rfs}} \neq \Omega$} In this case, the
    conclusion is a direct consequence of (C).
  \end{proofcases}

  Now that we showed how the main claim can be deduced from (C'), we
  can focus on the proof of (C').  In the proof we will extensively
  use the assumption on the layouts:
  \[
    \forall \id \in \Idu.\lay_1(\id)=\lay_2(\id),
    \tag{\dag}
  \]
  For this reason, we fix it on top and name it (\dag). The proof goes by induction on $\nat$.
  \begin{proofcases}
    \proofcase{0} We assume
    \[
      \nstep[\system][\lay_1]0{\conf{\st, \lay_1\lcomp \rfs'}}{\confone_1} \land \final[\system][\lay_1] {\confone_1}.
    \]
    This implies that $\confone_1 = \conf{\st, \lay_1\lcomp
      \rfs'}$. From this last observation, and since $\confone_1$ is
    terminal, we deduce that $\st = \frame{\cnil}{\regmap'}{\um}$.  By
    examining the semantics, we also establish that
    \[
      \final[\system][\lay_2]{\conf{\st, \lay_2\lcomp \rfs'}},
    \]
    so we conclude that
    $\Eval[\system][\lay_2]{\conf{\st, \lay_2\lcomp \rfs'}} =
    (\regmap'(\ret),\rfs') \evaleq \confone_1$.
    \proofcase{$\nat+1$}
    Observe that $\st \neq \nil$, as it would contradict the assumption
    \[
      \nstep[\system][\lay_1]{\nat+1}{\conf{\st, \lay_1\lcomp \rfs'}}{\confone_1} \land \final[\system][\lay_1] {\confone_1}.
      \tag{H}
    \]
    Thus, we assume $\st = \frame \cmd \regmap \um :\st'$, and proceed by case analysis on $\cmd$.  
    Note that the induction hypothesis coincides syntactically with (C').
    \begin{proofcases}
      \proofcase{$\cnil$} In this case, we rewrite the first part of (H) as follows:
      \begin{equation*}
        \step[\system][\lay_1] {\conf{\frame{\cnil}{ \regmap}{\um}:\st', \lay_1\lcomp \rfs'}}{}
        \conf{\st'', \lay_1\lcomp \rfs'}\to^\nat{\confone_1},
      \end{equation*}
      where $\st''$ is obtained by updating the topmost register map
      in $\st'$ with the return value from $\regmap$.  Since
      $\um(\st')$ holds, also $\um(st'')$ does, so we can apply the IH
      and conclude that
      $ \Eval[\system][\lay_2]{\conf{\st'', \lay_2\lcomp \rfs'}}
      \evaleq \confone_1 $.  Thus, to conclude the proof, it suffices
      to observe that
      $ \step[\system][\lay_2]{\conf{\frame{\cnil}{\regmap}{\um}:\st',
          \lay_2\lcomp \rfs'}}{} {\conf{\st'', \lay_2\lcomp \rfs'}}, $
      which follows by introspection of the semantics.
      \proofcase{$\vx \ass\expr\sep \cmdtwo$} In this case, we rewrite
      the first part of (H) as follows:
      \begin{equation*}
        \step[\system][\lay_1] {\conf{\frame{\vx \ass\expr\sep\cmdtwo}{ \regmap}{\um}:\st', \lay_1\lcomp \rfs'}}{}\\
        {\conf{\frame{\cmdtwo}{\update \regmap\vx {\sem \expr_{\regmap, \lay_1}}}{\um}:\st', \lay_1\lcomp \rfs'}}\to^\nat{\confone_1}
      \end{equation*}
      observe that
      $\um(\frame{\cmdtwo}{\update \regmap\vx {\sem \expr_{\regmap,
            \lay_1}}}{\um}:\st')$ holds, so we can apply the IH and
      conclude that
      \[
        \Eval[\system][\lay_2]{\conf{\frame{\cmdtwo}{\update \regmap\vx {\sem \expr_{\regmap, \lay_1}}, }{\um}:\st', \lay_2\lcomp \rfs'}} \evaleq \confone_1
      \]
      This means that, in order to conclude the proof, it suffices to observe that
      \begin{equation*}
        \step[\system][\lay_2]{\conf{\frame{\vx \ass\expr\sep\cmdtwo}{\regmap}{\um}:\st', \lay_2\lcomp \rfs'}}{}\\
        {\conf{\frame{\cmdtwo}{\update \regmap\vx {\sem \expr_{\regmap, \lay_1}}}{\um}:\st', \lay_2\lcomp \rfs'}},
      \end{equation*}
      which, in turn, reduces to showing that
      $\sem \expr_{\regmap, \lay_1}=\sem \expr_{\regmap, \lay_2}$. By
      definition of $\um(\st)$, we have that
      $\vx \ass\expr\sep\cmdtwo$ is an unprivileged command, so, in
      particular all identifiers in $\expr$ belong to $\Idu$, hence,
      the $\sem \expr_{\regmap, \lay_1}=\sem \expr_{\regmap, \lay_2}$
      follows from \Cref{rem:expreval}, and (\dag).
      \proofcase{$\cskip\sep \cmdtwo$} Analogous to the case of
      assignments.
      \proofcase{$\cif \expr {\cmd_\bot}{\cmd_\top}\sep\cmdtwo$}
      Analogous to the case of assignments.
      \proofcase{$\cwhile \expr {\cmd}\sep\cmdtwo$} Analogous to the
      case of assignments.
      \proofcase{$\cmemass \expr \exprtwo\sep \cmdtwo$} In this case,
      we start by observing that
      $\sem \expr_{\regmap, \lay_1}=\sem \expr_{\regmap, \lay_2}$ by
      \Cref{rem:expreval}. This also implies
      $\toAdd{\sem \expr_{\regmap, \lay_1}}=\toAdd{\sem
        \expr_{\regmap, \lay_2}}$. Let
      $\toAdd{\sem \expr_{\regmap, \lay_1}}= \add$. By
      \Cref{rem:expreval}, we also know that
      $\sem \exprtwo_{\regmap, \lay_1}=\sem \exprtwo_{\regmap,
        \lay_2}$, and we denote this value $\val$.
      From the assumption (\dag) we deduce that
      $\underline \lay_1(\Ar[\um])=\underline \lay_2(\Ar[\um])$, and
      we denote this set $P$.  We go by cases on $\add \in P$.
      \begin{proofcases}
        \proofcase{$\add \in P$} In this case, the rule \ref{WL:Store}
        applies to both the configurations
        \[
        {\conf{\frame{\cmemass  \expr \exprtwo\sep \cmdtwo}{\regmap}{\um}:\st', \lay_1\lcomp \rfs'}}
        \]
        and
        \[
        {\conf{\frame{\cmemass  \expr \exprtwo\sep \cmdtwo}{\regmap}{\um}:\st', \lay_2\lcomp \rfs'}}
        \]
        under the layouts $\lay_1$ and $\lay_2$ respectively, obtaining:
        \[
          {\conf{\frame{\cmdtwo}{\regmap}{\um}:\st', \update{\lay_1\lcomp \rfs'}{\add}{\val}}}
        \]
        and
        \[
          {\conf{\frame{\cmdtwo}{\regmap}{\um}:\st', \update{\lay_2\lcomp \rfs'}{\add}{\val}}}.
        \]
        Since $\add \in \underline \lay_1(\Ar[\um])$
        and by definition of $\underline \lay_1$, we deduce that
        there exist $\ar \in \Ar[\um]$
        and $0\le i\le\size\ar$ such that $\lay_1(\ar)+i = \add$.
        By (\dag), we also deduce that $\lay_2(\ar)+i = \add$,
        so we can apply \Cref{rem:memupdtostupd}
        in order to show that:
        \[
          \update{\lay_1\lcomp \rfs'}{\add}{\val} =
          \lay_1\lcomp \update{\rfs'}{(\ar, i)}{\val}
        \]
        and that 
        \[
          \update{\lay_2\lcomp \rfs'}{\add}{\val} =
          \lay_1\lcomp \update{\rfs'}{(\ar, i)}{\val}.
        \]
        Finally, since $\um({\frame{\cmdtwo}{\regmap}{\um}:\st'})$ holds
        and we have:
        \[
          \update{\rfs'}{(\ar, i)}\val\eqon{\Fn}{\rfs'}\eqon{\Fn}\rfs,
        \]
        we can apply the IH, and conclude the proof of this
        sub-derivation.  \proofcase{$\add \notin P$} In this case,
        from (\dag), we deduce that
        $\add \notin \underline \lay_2(\Ar[\um])$.  Therefore, the
        rule \ref{WL:Store-Error} applies, showing both:
        \[
          \step[\system][\lay_1]
          {\conf{\frame{\cmemass  \expr \exprtwo\sep \cmdtwo}{\regmap}{\um}:\st', \lay_1\lcomp \rfs'}}
          \err,
        \]
        and
        \[
          \step[\system][\lay_2]
          {\conf{\frame{\cmemass  \expr \exprtwo\sep \cmdtwo}{\regmap}{\um}:\st', \lay_2\lcomp \rfs'}}
          \err.
        \]
        This establishes the claim. 
      \end{proofcases}
      \proofcase{$\cmemread \vx \expr\sep \cmdtwo$} Analogous to the
      case of memory store operations.
      \proofcase{$\ccall \expr {\exprtwo_1, \dots, \exprtwo_k}\sep
        \cmdtwo$} This case is also similar to that of memory loads,
      but requires some non-trivial observations.  As before, we
      define $\add = \toAdd{\sem \expr_{\regmap,\lay_1}}$ and observe
      that $\add = \toAdd{\sem \expr_{\regmap,\lay_2}}$ because of
      \Cref{rem:expreval}. Similarly, we introduce the values
      $\val_1, \dots,\val_k$ which correspond to the semantics of
      $\exprtwo_1,\dots,\exprtwo_k$ evaluated under $\regmap$ and both
      layouts $\lay_1, \lay_2$. By (\dag), we have
      $\lay_1(\Fn[\um]) = \lay_2(\Fn[\um])$, so we denote this set
      $P_{\Fn[\um]}$, and we go by case analysis on
      $\add \in P_{\Fn[\um]}$.
      \begin{proofcases}
        \proofcase{$\add\in P_{\Fn[\um]}$} In this case, there exists
        $\fn \in \Fn[\um]$ such that $\lay_1(\fn)=\add$, and from
        (\dag) we deduce $\add=\lay_2(\fn)$. Using these intermediate
        conclusions and the definition of $\lcomp$, we deduce that
        $(\lay_1\lcomp \rfs')(\add) = \rfs'(\fn) = \rfs(\fn)$, (the
        last step follows from the assumption $\rfs'\eqon{\Fn}\rfs$),
        and similarly for $\lay_1\lcomp \rfs'(\add)$. Since
        $\add\in P_{\Fn[\um]}$, we deduce that rule \ref{WL:Call} can
        be applied, proving:
        \begin{equation*}
          \step[\system][\lay_1]
          {\conf{\frame{\ccall  \expr {\exprtwo_1, \dots, \exprtwo_k}\sep \cmdtwo}{\regmap}{\um}:\st', \lay_1\lcomp \rfs'}}{}\\
          {\conf{\frame{\rfs(\fn)}{\regmap_0'}{\um}:\frame{\cmdtwo}{\regmap}{\um}:\st', \lay_1\lcomp \rfs'}},
        \end{equation*}
        and
        \begin{equation*}
          \step[\system][\lay_2]
          {\conf{\frame{\ccall  \expr {\exprtwo_1, \dots, \exprtwo_k}\sep \cmdtwo}{\regmap}{\um}:\st', \lay_2\lcomp \rfs'}}{}\\
          {\conf{\frame{\rfs(\fn)}{\regmap_0'}{\um}:\frame{\cmdtwo}{\regmap}{\um}:\st', \lay_2\lcomp \rfs'}},
        \end{equation*}
        where $\regmap_0'= \regmap_0[\vx_1,\dots,\vx_k\upd \val_1,\dots,\val_k]$.
        Finally, by assumption on $\rfs$, we know that $\rfs(\fn)$ is an unprivileged program. So, in particular,
        \[
          \um({\frame{\rfs(\fn)}{\regmap_0'}{\um}:\frame{\cmdtwo}{\regmap}{\um}:\st'})
        \]
        holds. This allows us to apply the IH and conclude the proof.
        \proofcase{$p \notin P_{\Fn[\um]}$} Analogous to the corresponding case for store operations.
      \end{proofcases}
      \proofcase{$\csyscall \syscall{\expr_1, \ldots,
          \expr_k}\sep\cmdtwo$} This is perhaps the most interesting
      case of this proof. We begin by introducing the values
      $\val_1, \dots,\val_k$, which correspond to the
      semantics of $\exprtwo_1,\dots,\exprtwo_k$ evaluated under
      $\regmap$ as well as any of layouts $\lay_1, \lay_2$.These
      values do not depend on the layout because of \Cref{rem:expreval}. By
      introspection of the rule \ref{WL:SystemCall}, we deduce that
      both the following statements hold:
      \begin{equation*}
          \step[\system][\lay_1]
          {\conf{\frame{\csyscall  \syscall {\exprtwo_1, \dots, \exprtwo_k}\sep \cmdtwo}{\regmap}{\um}:\st', \lay_1\lcomp \rfs'}}{}\\
          {\conf{\frame{\syss(\syscall)}{\regmap_0'}{\km[\syscall]}:\frame{\cmdtwo}{\regmap}{\um}:\st', \lay_1\lcomp \rfs'}},
        \end{equation*}
        and
        \begin{equation*}
          \step[\system][\lay_2]
          {\conf{\frame{\csyscall \syscall {\exprtwo_1, \dots, \exprtwo_k}\sep \cmdtwo}{\regmap}{\um}:\st', \lay_2\lcomp \rfs'}}{}\\
          {\conf{\frame{\syss(\syscall)}{\regmap_0'}{\km[\syscall]}:\frame{\cmdtwo}{\regmap}{\um}:\st', \lay_2\lcomp \rfs'}},
        \end{equation*}
        where
        $\regmap_0'= \regmap_0[\vx_1,\dots,\vx_k\upd
        \val_1,\dots,\val_k]$. Unfortunately, we cannot apply the IH
        to the configuration
        \[
          {\conf{\frame{\syss(\syscall)}{\regmap_0'}{\km[\syscall]}:\frame{\cmdtwo}{\regmap}{\um}:\st', \lay_1\lcomp \rfs'}}
        \]
        because  $\um(\frame{\syss(\syscall)}{\regmap_0'}{\km[\syscall]}:\frame{\cmdtwo}{\regmap}{\um}:\st')$ does not hold. However, we can go by cases on
        \[
          \Eval[\system][\lay_1]{{{\syss(\syscall)},{\regmap_0'},{\km[\syscall]},\rfs'}}
        \]
        \begin{proofcases}
          \proofcase{$\Omega$} By \Cref{lemma:contextplugging}, we deduce that 
        \[
          \Eval[\system][\lay_1]{\conf{\frame{\syss(\syscall)}{\regmap_0'}{\km[\syscall]}:\frame{\cmdtwo}{\regmap}{\um}:\st', \lay_1\lcomp \rfs'}} = \Omega,
        \]
        which leads to a contradiction because we assumed
        \small
        \[
          \nstep[\system][\lay_1] {\nat+1} {\conf{\frame{\syss(\syscall)}{\regmap_0'}{\km[\syscall]}:\frame{\cmdtwo}{\regmap}{\um}:\st', \lay_1\lcomp \rfs'}}{\confone_1}\land \nf[\system][\lay_1] {\confone_1}.
        \]
        \normalsize
        \proofcase{$\unsafe$} From the \emph{layout non-intreference property} (\Cref{def:lni}), we deduce that
        \[
          \Eval[\system][\lay_2]{{\syss(\syscall), \regmap_0', \km[\syscall], \rfs'}} \in\{\err,\unsafe\}.
        \]
        We take $\err$ as an example. This means that 
        \[
          \exists h.\nstep[\system][\lay_2] h {\conf{\frame{\syss(\syscall)}{\regmap_0'}{\km[\syscall]}, \lay_2\lcomp \rfs'}}{\err}.
        \]
        By applying \Cref{lemma:contextplugging} to the reduction with $\lay_1$,
        we deduce that
        \[
          \nstep[\system][\lay_1] {\nat} {\conf{\frame{\syss(\syscall)}{\regmap_0'}{\km[\syscall]}:\frame{\cmdtwo}{\regmap}{\um}:\st', \lay_1\lcomp \rfs'}}{\unsafe}.
        \]
        By applying \Cref{lemma:contextplugging} to the
        reduction with $\lay_2$, we deduce that:
        \[
          \nstep[\system][\lay_2] {!} {\conf{\frame{\syss(\syscall)}{\regmap_0'}{\opt}:\frame{\cmdtwo}{\regmap}{\um}:\st', \lay_2\lcomp \rfs'}} {\err}.
        \]
        We conclude by observing that $\unsafe\evaleq \err$. The case with $\unsafe$ follows analogously.
        \proofcase{$\err$} Analogous to the previous case.
        \proofcase{$(\val, \mem)$}
        We start by observing that $\mem = \lay_1\lcomp \rfs''$ for some $\rfs''$ due to \Cref{rem:simispreserved}. Form the \emph{layout non-interference}
        property (\Cref{def:lni}),
        we deduce that
        \begin{equation*}
          \Eval[\system][\lay_2]{{\syss(\syscall)},{\regmap_0'},{\km[\syscall]},\rfs'} =
          {(\val,\rfs'')}.
        \end{equation*}
        By applying  \Cref{lemma:contextplugging} twice, we deduce that:
        \begin{multline*}
          \step[\system][\lay_1]
          {\conf{\frame{\csyscall  \syscall {\exprtwo_1, \dots, \exprtwo_k}\sep \cmdtwo}{\regmap}{\um}:\st', \lay_1\lcomp \rfs'}}{}\\
          {\conf{\frame{\syss(\syscall)}{\regmap_0'}{\km[\syscall]}:\frame{\cmdtwo}{\regmap}{\um}:\st', \lay_1\lcomp \rfs'}} \to^*\\
          {\conf{\frame{\cmdtwo}{\update\regmap\ret\val}{\um}:\st', \lay_1\lcomp \rfs''}}
        \end{multline*}
        and 
        \begin{multline*}
          \step[\system][\lay_2]
          {\conf{\frame{\csyscall  \syscall {\exprtwo_1, \dots, \exprtwo_k}\sep \cmdtwo}{\regmap}{\um}:\st', \lay_2\lcomp \rfs'}}{}\\
          {\conf{\frame{\syss(\syscall)}{\regmap_0'}{\km[\syscall]}:\frame{\cmdtwo}{\regmap}{\um}:\st', \lay_2\lcomp \rfs'}} \to^*\\
          {\conf{\frame{\cmdtwo}{\update\regmap\ret\val}{\um}:\st', \lay_2\lcomp \rfs''}}
        \end{multline*}
        Next, we observe that we assumed
        \begin{equation*}
          \nstep[\system][\lay_1]{n+1}
          {\conf{\frame{\csyscall  \syscall {\exprtwo_1, \dots, \exprtwo_k}\sep \cmdtwo}{\regmap}{\um}:\st', \lay_1\lcomp \rfs'}}{}
          {\confone_1},
        \end{equation*}
        and $\confone_1$ is terminal, so it must be the case that
        $\nstep[\system][\lay_1]!{\conf{\frame{\cmdtwo}{\update\regmap\ret\val}{\um}:\st',
            \lay_1\lcomp \rfs''}}{\confone_1}$. In particular, the
        reduction requires less than $\nat+1$ steps. From the
        assumption $\um(\st)$, we also deduce $\um(\st')$, and
        $\um(\frame{\cmdtwo}{\update\regmap\ret\val}{\um}:\st')$ and
        from \Cref{rem:simispreserved} we conclude that
        $\rfs''\eqon{\Fn}\rfs'\eqon{\Fn}\rfs$, so we can apply the IH
        to
        \[
          {\conf{\frame{\cmdtwo}{\update\regmap\ret\val}{\um}:\st', \lay_1\lcomp \rfs''}}
        \]
        and deduce that
        \[
          \Eval[\system][\lay_2]{\conf{\frame{\cnil}{\regmap'}{\opt_2}:\st', \lay_2\lcomp \rfs''}}\evaleq
          \confone_1,
        \]
        which concludes the proof.
      \end{proofcases}
    \end{proofcases}
  \end{proofcases}
\end{proof}

\begin{lemma}[Main Lemma]
  \label{lemma:mainlemmapms}
  Let $\system=(\rfs, \syss, \caps)$ be a \emph{layout non-interferent} system.
  For every $\nat \in \Nat$,
  \emph{unprivileged} program $\cmd$,
  register map $\regmap$, and distribution $\mu$,
  if
  \[
    \forall \lay\in \Lay. \exists \conftwo_\lay. \nstep{!\nat}{\conf{\frame{\cmd}{\regmap}{\um}, \lay \lcomp \rfs}} \conftwo_\lay,
  \]
  then one of the two following statements must hold: 
  \begin{varitemize}
  \item There are a register map $\regmap'$ and a store $\rfs'\eqon{\Fun}\rfs$
    such that, for every layout $\lay$, we have the reduction
    \[
      \nstep {! \nat}{\conf{\frame{\cmd}{\regmap}{\um}, \lay \lcomp \rfs}}{\conf{\frame{\cnil}{\regmap'}{\um}, \lay \lcomp \rfs'}}.
    \]
  \item
    $
      \Pr_{\lay \leftarrow \mu}\left[\nstep {! \nat} {\conf{\frame{\cmd}{\regmap}{\um}, \lay \lcomp \rfs}}\err \right] \ge \delta_\mu.
    $
  \end{varitemize}  
\end{lemma}

\begin{proof}
  To facilitate the induction, we prove a slightly stronger statement. Specifically, we replace the quantification over $\cmd$ and $\regmap$ with a quantification over a non-empty stack $\st$ such that $\um(\st)$ holds. Additionally, we quantify over a general $\rfs' \eqon{\Fn} \rfs$ instead of $\rfs$, and we rewrite the claim in a more convenient form. In the following claim, all meta-variables that are not explicitly quantified are globally and universally quantified. As a result, we obtain the following claim:
  if $\forall \lay\in \Lay. \exists \conftwo_\lay. \nstep{!\nat}{\conf{\st, \lay \lcomp \rfs'}} \conftwo_\lay$, then:
  \begin{varitemize}
  \item $\exists \regmap', \rfs''\eqon{\Fn}\rfs'. \forall \lay\in \Lay.\nstep {! \nat}{\conf{\st, \lay \lcomp \rfs'}}{} {\conf{\frame{\cnil}{\regmap'}{\um}, \lay \lcomp \rfs''}}$, or
  \item $\Pr_{\lay \leftarrow \mu}\left[\nstep {! \nat} {\conf{\st, \lay \lcomp \rfs'}}\err \right] \ge \delta_\mu$.
  \end{varitemize}
  We call this auxiliary claim (C') and we proceed by induction on $\nat$.
  \begin{proofcases}
    %
    %
    \proofcase{0} In this case, we assume
    \[
      \forall \lay\in \Lay. \exists \conftwo_\lay. \nstep{!0}{\conf{\st, \lay \lcomp \rfs'}} \conftwo_\lay,
    \]
    and from this premise, we deduce that
    $\conf{\st, \lay \lcomp \rfs'}$ is terminal. Using the definition
    of the notation $\to^!$, the fact that $\um(\st)$ holds, and the
    assumption on the non-emptiness of $\st$, we conclude that
    $\st = \frame \cnil \regmap \um : \nil$ for some $\regmap$. Hence, the claim
    holds trivially by choosing the register map $\regmap$ and the
    store $\rfs'$.
    %
    %
    \proofcase{$\nat+1$} In this case, we start by assuming
    \[
      \forall \lay\in \Lay. \exists \conftwo_\lay. \nstep{!\nat+1}{\conf{\st, \lay \lcomp \rfs'}} \conftwo_\lay
      \tag{H}
    \]
    and we go by cases on $\st$, excluding the case where $\st =\nil$ due to our assumption on the non-emptiness of $\st$. 
    Therefore, in the following, we assume that $\st = \frame \cmd \regmap \um :\st'$,
    and we proceed by cases on $\cmd$. Notice that the induction
    hypothesis coincides syntactically with (C').
    \begin{proofcases}
      \proofcase{$\vx \ass\expr\sep \cmdtwo$} In this case, we observe
      that there exists a value $\val$ such that, for every layout
      $\lay$, we have $\sem \expr_{\regmap, \lay} = \val$. More
      specifically, this follows from \Cref{rem:expreval}, since for
      every $\id \in \Idu$ and every pair of layouts $\lay_1, \lay_2$,
      it holds that $\lay_1(\id) = \lay_2(\id)$. In particular, all
      identifiers appearing in $\expr$ belong to $\Idu$.
      By introspection of the semantics, we deduce that for every layout $\lay$:
      \begin{equation*}
        \step {\conf{\frame{\vx \ass\expr\sep\cmdtwo}{\regmap}{\um}:\st', \lay\lcomp \rfs'}}{}\\
        {\conf{\frame{\cmdtwo}{\update \regmap\vx {\val}}{\um}:\st', \lay\lcomp \rfs'}}.
        \tag{$*$}
      \end{equation*}
      Notice that
      $\um(\frame{\cmdtwo}{\update \regmap\vx
        {\val}}{\um}:\st')$. Because of (H) and this observation, we
      can apply the IH to the stack in the target configuration and
      $\rfs'$ to conclude that one among (A) and (B) below holds.
      \begin{equation*}
        \exists \regmap', \rfs''\eqon{\Fn}\rfs'. \forall \lay\in \Lay.\\
        \nstep {! \nat}{\conf{\frame{\cmdtwo}{\update \regmap\vx {\val}}{\um}:\st', \lay \lcomp \rfs'}}{}\\
        {\conf{\frame{\cnil}{\regmap'}{\um}, \lay \lcomp \rfs''}}
        \tag{A}
      \end{equation*}
      \[
        \Pr_{\lay \leftarrow \mu}\left[\nstep {! \nat} {\conf{\frame{\cmdtwo}{\update \regmap\vx {\val}}{\um}:\st', \lay \lcomp \rfs'}}\err \right] \ge \delta_\mu.
        \tag{B}
      \]
      We go by cases on this disjunction.
      \begin{proofcases}
        \proofcase{A} In this case, we introduce $\regmap'$ and $\rfs''$ from the IH,
        we assume $\rfs''\eqon{\Fn}\rfs'$, we fix a layout $\lay$ and from
        ($*$) and (A), we conclude:
        \begin{equation*}
          \nstep{\nat+1} {\conf{\frame{\vx \ass\expr\sep\cmdtwo}{\regmap}{\um}:\st', \lay\lcomp \rfs'}}{}
          {\conf{\frame{\cnil}{\regmap'}{\um}, \lay \lcomp \rfs''}}.
        \end{equation*}
        Due to the generality of $\lay$, we can
        introduce universal quantification over all layouts, thereby
        completing the proof.

        \proofcase{B} We call $E$ the set of all the layouts $\overline \lay$ such that
        \[
          \nstep[\system][\overline \lay] {! \nat} {\conf{\frame{\cmdtwo}{\update \regmap\vx {\val}}{\um}:\st', \lay \lcomp \rfs'}}\err
        \]
        Notice that we can apply the assumption ($*$)
        to each of these layouts, showing that
        \[
          \forall \overline \lay\in E.\nstep[\system][\overline \lay]{!\nat+1} {\conf{\frame{\vx \ass\expr\sep\cmdtwo}{\regmap}{\um}:\st', \overline\lay\lcomp \rfs'}}{\err}.
        \]
        From (B), we know that the probability associated to the event $E$ is bigger than $\delta_\mu$,
        the last observation establishes the claim.
      \end{proofcases}
      \proofcase{$\cnil$} Analogous to the previous case.
      \proofcase{$\cskip\sep \cmdtwo$} Analogous to the case of
      assignments.
      \proofcase{$\cif \expr {\cmd_\bot}{\cmd_\top}\sep\cmdtwo$}
      Analogous to the case of assignments.
      \proofcase{$\cwhile \expr {\cmd}\sep\cmdtwo$} Analogous to the
      case of assignments.
      \proofcase{$\cmemass \expr \exprtwo\sep \cmdtwo$} In this case,
      we begin by observing that there exists a value $\val_\expr$
      such that for every $\lay \in \Lay$, we have
      $\sem\expr_{\regmap,\lay} = \val_\expr$. This follows from
      \Cref{rem:expreval}, since we assume that
      $\um(\cmemass \expr \exprtwo \sep \cmdtwo)$, which implies that
      all identifiers within $\expr$ belong to $\Idu$.
      For the same reason, there exists a unique value $\val_\exprtwo$
      such that for every layout $\lay$, we have
      $\sem\exprtwo_{\regmap,\lay} = \val_\exprtwo$. Consequently,
      there is a unique address $\add \in \Add$ such that
      $\add = \toAdd{\val}$.
      Similarly, we deduce that there exists a unique set of addresses
      $P_{\Aru}$ such that, for every layout $\lay$, we have
      $\underline{\lay}(\Aru) = P_{\Aru}$. Moreover, due to the
      assumptions on the set of layouts, we conclude that
      $P_{\Aru} \subseteq \Addu$.
      Finally, we proceed by cases on whether $\add \in P_{\Aru}$.
      \begin{proofcases}
        \proofcase{$\add \in P_{\Aru}$} The rule \ref{WL:Store}
        can be applied to each configuration
        \[
        {\conf{\frame{\cmemass  \expr \exprtwo\sep \cmdtwo}{\regmap}{\um}:\st', \lay\lcomp \rfs'}}
        \]
        for every layout $\lay$, obtaining:
        \[
          {\conf{\frame{\cmdtwo}{\regmap}{\um}:\st', \update{\lay\lcomp \rfs'}{\add}{\val_\expr}}};
        \]
        from the assumption $\add \in \underline \lay(\Ar[\um])$
        and the definition of $\underline \lay$, we deduce that
        there is $\ar \in \Ar[\um]$
        and $0\le i\le\size\ar$ such that $\lay(\ar)+i = \add$
        and that for every other layout $\overline \lay$
        it must hold that $\overline \lay(\ar)+i = \add$,
        so we can apply \Cref{rem:memupdtostupd}
        in order to show that
        \[
          \forall \lay \in \Lay. \update{\lay\lcomp \rfs'}{\add}{\val_\expr} =
          \lay\lcomp \update{\rfs'}{(\ar, i)}{\val_\expr}.
        \]
        Finally, we observe that $\um({\frame{\cmdtwo}{\regmap}{\um}:\st'})$ holds
        and that
        \[
          \update{\rfs'}{(\ar, i)}{\val_\expr}\eqon{\Fn}{\rfs'}\eqon{\Fn}\rfs
        \]
        holds as well, so we can apply the IH
        and conclude the proof of this sub-derivation as we did in the case of assignments.
        \proofcase{$p \notin P_{\Aru}$}
        In this case,
        we just observe that
        the rule \ref{WL:Store-Error}
        can be applied to show both
        \[
          \step[\system][\lay_1]
          {\conf{\frame{\cmemass  \expr \exprtwo\sep \cmdtwo}{\regmap}{\um}:\st', \lay_1\lcomp \rfs'}}
          \err,
        \]
        and
        \[
          \step[\system][\lay_2]
          {\conf{\frame{\cmemass  \expr \exprtwo\sep \cmdtwo}{\regmap}{\um}:\st', \lay_2\lcomp \rfs'}}
          \err.
        \]
        This establishes the claim. 
      \end{proofcases}
      \proofcase{$\cmemread \vx \expr\sep \cmdtwo$} Analogous to the case of store operations.
      \proofcase{$\ccall \expr {\exprtwo_1, \dots, \exprtwo_k}\sep \cmdtwo$} 
      This case is similar to that of stores but requires additional
considerations regarding the stacks of the target configurations.
      We begin by observing that there exists a unique address $\add$
such that for every $\lay \in \Lay$, we have $\toAdd{\sem
\expr_{\regmap,\lay}} = \add$. Similarly, we define the values
$\val_1, \dots, \val_k$, which correspond to the semantics of
$\exprtwo_1, \dots, \exprtwo_k$ evaluated under $\regmap$ for every
layout $\lay \in \Lay$.
      Next, we note the existence of a set $P_{\Fn[\um]}$ such that
for all $\lay \in \Lay$, we have $\underline{\lay} (\Fn[\um]) =
P_{\Fn[\um]}$.
      The proof proceeds by cases on $\add \in P_{\Fn[\um]}$.

      \begin{proofcases}
        \proofcase{$\add\in P_{\Fn[\um]}$} In this case, there exists
        a unique $\fn \in \Fn[\um]$ such that for every
        $\lay \in \Lay$, we have $\lay(\fn) = \add$. From this
        observation and the definition of $\lcomp$, we deduce that for
        every $\lay \in \Lay$, $(\lay\lcomp \rfs')(\add) = \rfs'(\fn) = \rfs(\fn)$,
        where the last step follows from the assumption $\rfs' \eqon{\Fn} \rfs$. Since
        $\add\in P_{\Fn[\um]}$, we deduce that:
        \begin{multline*}
          \forall \add\in P_{\Fn[\um]}.
          \step
          {\conf{\frame{\ccall  \expr {\exprtwo_1, \dots, \exprtwo_k}\sep \cmdtwo}{\regmap}{\um}:\st', \lay\lcomp \rfs'}}{}\\
          {\conf{\frame{\rfs(\fn)}{\regmap_0'}{\um}:\frame{\cmdtwo}{\regmap}{\um}:\st', \lay\lcomp \rfs'}},
        \end{multline*}
        where
        $\regmap_0'= \regmap_0[\vx_1,\dots,\vx_k\upd
        \val_1,\dots,\val_k]$.  Finally, we observe that
        $\um(\frame{\rfs(\fn)}{\regmap_0'}{\um}:\frame{\cmdtwo}{\regmap}{\um}:\st')$
        holds since $\rfs$ maps unprivileged commands to identifiers
        in $\Fnu$.  This allows us to apply the IH. The proof then
        follows the same structure as in the case of assignments.

        \proofcase{$p \notin P_{\Fn[\um]}$} Analogous to the
        corresponding case for store operations.
      \end{proofcases}
      \proofcase{$\csyscall \syscall{\expr_1, \ldots, \expr_k}\sep\cmdtwo$}
      We begin by introducing the values $\val_1, \dots, \val_k$, which correspond to the semantics of $\exprtwo_1, \dots, \exprtwo_k$ evaluated under $\regmap$ and all layouts $\lay \in \Lay$. By introspection of the rule \ref{WL:SystemCall}, we deduce that for every $\lay \in \Lay$, the following holds:
      \begin{equation*}
        \step
        {\conf{\frame{\csyscall  \syscall {\exprtwo_1, \dots, \exprtwo_k}\sep \cmdtwo}{\regmap}{\um}:\st', \lay\lcomp \rfs'}}{}\\
        {\conf{\frame{\syss(\syscall)}{\regmap_0'}{\km[\syscall]}:\frame{\cmdtwo}{\regmap}{\um}:\st', \lay\lcomp \rfs'}},
        \tag{\dag}
      \end{equation*}
      where $\regmap_0'= \regmap_0[\vx_1,\dots,\vx_k\upd \val_1,\dots,\val_k]$. However, we cannot apply the IH on the configuration
      \[
        {\conf{\frame{\syss(\syscall)}{\regmap_0'}{\km[\syscall]}:\frame{\cmdtwo}{\regmap}{\um}:\st', \lay\lcomp \rfs'}}
      \]
      because
      $\um(\frame{\syss(\syscall)}{\regmap_0'}{\km[\syscall]}:\frame{\cmdtwo}{\regmap}{\um}:\st')$
      does not hold. Nevertheless, we can apply
      \Cref{lemma:onsyscallterm} and deduce that
      one of the following statements must hold:
      \begin{equation*}
        \exists \overline \regmap, \overline \rfs\eqon{\Fn} \rfs'.\forall \lay \in \Lay.\\
        \nstep !{\conf{\frame{\syss(\syscall)}{\regmap_0'}{\km[\syscall]}, {\lay\lcomp \rfs'}}}  {\conf{\frame{\cnil}{\overline \regmap}{\km[\syscall]}, \lay\lcomp \rfs''}},
        \tag{A}
      \end{equation*}
      \[
        \Pr_{\lay\leftarrow\mu}\left[\nstep !{\conf{\frame{\syss(\syscall)}{\regmap_0'}{\km[\syscall]}, \lay\lcomp \rfs'}}\err\right] \ge \delta_\mu,
        \tag{B}
      \]
      \[
        \forall \lay \in \Lay.\Eval{\conf{\frame{\syss(\syscall)}{\regmap_0'}{\km[\syscall]}, \lay\lcomp \rfs'}} = \Omega.
        \tag{C}
      \]
      We go by cases on the valid statement.
      \begin{proofcases}
        \proofcase{A} From (A) and \Cref{lemma:contextplugging}, we conclude that
        there exist a register map $\overline \regmap$ and a store
        $\overline \rfs$ such that for every
        layout $\lay$, we have:
        \begin{equation*}
          \nstep !        {\conf{\frame{\syss(\syscall)}{\regmap_0'}{\km[\syscall]}:\frame{\cmdtwo}{\regmap}{\um}:\st', \lay\lcomp \rfs'}}{}\\
          {\conf{\frame{\cmdtwo}{\update\regmap \ret{\overline \regmap(\ret)}}{\um}:\st', \lay\lcomp \overline \rfs}}.
        \end{equation*}
        We also observe that $\rfs \eqon{\Fn} \rfs'$ by \Cref{rem:simispreserved}. Combining this with (\dag), we conclude:
        \begin{equation*}
          \nstep ! {\conf{\frame{\csyscall  \syscall {\exprtwo_1, \dots, \exprtwo_k}\sep \cmdtwo}{\regmap}{\um}:\st', \lay\lcomp \rfs'}}{}\\
          {\conf{\frame{\cmdtwo}{\update\regmap \ret{\overline \regmap(\ret)}}{\um}:\st', \lay\lcomp \overline \rfs}}
        \end{equation*}
        holds as well.
        From (H), we also deduce that the number of steps must be at most $\nat+1$, while from (\dag), we infer that it must be greater than 1. This allows us to apply the IH to 
        ${\conf{\frame{\cmdtwo}{\update\regmap \ret{\overline \regmap(\ret)}}{\um}:\st', \lay\lcomp \overline \rfs}}$. The proof then proceeds similarly to the case of assignments.
        \proofcase{B} Define $E$ as the set of all layouts $\overline \lay \in \Lay$ such that:
        \[
          \nstep[\system][\overline \lay] {!} {\conf{\frame{\syss(\syscall)}{\regmap_0'}{\km[\syscall]}, \overline \lay\lcomp \rfs'}}\err
        \]
        Since (\dag) holds for each of these layouts, it follows that
        for every $\overline \lay \in E$, we have:
        \[
          \nstep[\system][\overline \lay]{!} {\conf{\frame{\csyscall  \syscall {\exprtwo_1, \dots, \exprtwo_k}\sep \cmdtwo}{\regmap}{\um}:\st', \overline \lay\lcomp \rfs'}}{\err}.
        \]
        From (B), we know that the probability measure associated with
        this set is greater than $\delta_\mu$. Furthermore, from (H),
        we deduce:
        \[
          \nstep[\system][\overline \lay]{!\nat+1} {\conf{\frame{\csyscall  \syscall {\exprtwo_1, \dots, \exprtwo_k}\sep \cmdtwo}{\regmap, \um}:\st', \overline \lay\lcomp \rfs'}}{\err},
        \]
        which completes the claim.
        \proofcase{C} Let $\lay\in \Lay$ be any layout.
        From \Cref{lemma:contextplugging}, we deduce that 
        \[
          \Eval{\conf{\frame{\syss(\syscall)}{\regmap_0'}{\km[\syscall]}:\frame{\cmdtwo}{\regmap}{\um}:\st', \lay\lcomp \rfs'}} = \Omega,
        \]
        which contradicts our assumption (H). This completes the proof.
      \end{proofcases}
    \end{proofcases}
  \end{proofcases}
  The main claim is a particular case of C' where $\st$ is $\frame{\cmd}{\regmap}{\um}$, and $\rfs'$ is $\rfs$.
\end{proof}


\onsyscallterm*

\begin{proof}
  We begin by fixing all the universally quantified variables in the
  statement using the same meta-variable as before. We also fix some
  $\overline{\lay} \in \Lay$ and proceed by case analysis on
  \[
    \Eval[\system][\overline \lay]{\conf{\frame{\syss(\syscall)}{\regmap}{\km[\syscall]}, \lay\lcomp \rfs'}}.
  \]
  \begin{proofcases}
    \proofcase{$(\val, \rfs'')$}
    In this case, by the \emph{layout non-interference} property, we conclude that (A) holds for $\overline{\val} = \val$ and $\overline{\rfs} = \rfs''$. Furthermore, by \Cref{rem:simispreserved}, we know that $\rfs'' \eqon{\Fn} \rfs' \eqon{\Fn} \rfs$.
    \proofcase{$\unsafe, \err$}
    Here, by the \emph{layout non-interference} property, we deduce that for every $\lay \in \Lay$,
    \[
      \Eval{\conf{\frame{\syss(\syscall)}{\regmap}{\km[\syscall]}, \lay\lcomp \rfs'}} \in \{\err, \unsafe\}.
    \]
    Consequently, for every $\lay \in \Lay$, there exists some $\nat_\lay$ such that
    \[
      \exists \nat_\lay.\nstep {\nat_\lay} {\conf{\frame{\syss(\syscall)}{\regmap}{\km[\syscall]}, \lay\lcomp \rfs'}} \conftwo \in\{\err, \unsafe\}.
    \]
    Since $\Lay$ is finite, there exists a uniform bound $\overline{\nat}$ such that
    \[
      \forall \lay\in \Lay.\nstep {!\overline \nat} {\conf{\frame{\syss(\syscall)}{\regmap}{\km[\syscall]}, \lay\lcomp \rfs'}} \conftwo \in\{\err, \unsafe\}.
      \tag{\dag}
    \]
    From this, we conclude that  
    \[
      \Pr_{\lay \leftarrow \mu}\left[ \nstep {!\overline \nat} {\conf{\frame{\syss(\syscall)}{\regmap}{\km[\syscall]}, \lay\lcomp \rfs'}}\err\right]\ge \delta_\mu.
    \]
    This follows by considering an enumeration $\id_0, \dots, \id_h$ of the identifiers in $\refs_\system(\syss(\syscall))$ and noting that the probability above can be rewritten as:
    \begin{multline*}
      \sum_{\add_1, \dots, \add_h \in \Addk}
      \Pr_{\lay \leftarrow \mu}\Big[ \nstep {!\overline \nat} {\conf{\frame{\syss(\syscall)}{\regmap}{\km[\syscall]}, \lay\lcomp \rfs'}}\err \,\,\Big|\,\,
      \forall 1\le i\le h.\lay(\id_i)=\add_i\Big] \cdot\\
      \Pr_{\lay \leftarrow \mu}[\lay(\id_1)=\add_1,\dots, \lay(\id_h)=\add_h]
    \end{multline*}
    From (\dag) and \Cref{lemma:onsyscall}, we deduce that for every
    choice of $\add_1, \dots, \add_h \in \Addk$, the terms on the left
    in the expression above are bounded by $\delta_{\mu}$ making
    their convex combination bounded by $\delta_{\mu}$. moreover, since:
    \[
      \nstep {!\overline \nat} {\conf{\frame{\syss(\syscall)}{\regmap}{\km[\syscall]}, \lay\lcomp \rfs'}}\err
      \Rightarrow
      \nstep {!} {\conf{\frame{\syss(\syscall)}{\regmap}{\km[\syscall]}, \lay\lcomp \rfs'}}\err
    \]
    for the definition of $\to^{!m}$, and this means that 
    \[
      \Pr_{\lay \leftarrow \mu}\left[ \nstep {!} {\conf{\frame{\syss(\syscall)}{\regmap}{\km[\syscall]}, \lay\lcomp \rfs'}}\err\right]
      \ge
      \Pr_{\lay \leftarrow \mu}\left[ \nstep {!\overline \nat} {\conf{\frame{\syss(\syscall)}{\regmap}{\km[\syscall]}, \lay\lcomp \rfs'}}\err\right],
    \]
    which establishes the claim.
    \proofcase{$\Omega$} In this case, due to the \emph{layout non-interference} property, we conclude that (C) holds.
  \end{proofcases}
\end{proof}

For the next lemma, we extend $\refs_\sigma$ to frame stacks as follows:
\begin{align*}
  \refs_\sigma(\nil) &\defsym \emptyset&
  \refs_\sigma(f\cons \st) &\defsym \refs_\sigma(f) \cup \refs_\sigma(\st) &
  \refs_\sigma(\frame \cmd \regmap \opt) &\defsym \refs_\sigma(\cmd).
\end{align*}

\onsyscall*

\begin{proof}
  We proceed by induction on $\nat$.
  \begin{proofcases}
    \proofcase{0}
    In this case, we conclude that either (A) or (C) holds trivially, depending on whether $\syss(\syscall) \neq \cnil$ or not.
    \proofcase{$\nat+1$}
    We begin by applying the IH to the initial configuration. Three cases arise:
  \begin{proofcases}
    \proofcase{A} Let $\system$ be a system, $\syscall$ a system call
    and consider the initial configuration
    $ {\conf{\frame{\syss(\syscall)}{\regmap}{\km[\syscall]}, \lay
        \lcomp \rfs'}}$.  Suppose
    $\refs_\system(\syss(\syscall))\setminus \Sys= \{\id_1, \dots,
    \id_h\}$ and fix addresses $\add_1, \dots, \add_h$.  By the IH,
    there exists a stack $\st$ and a
    store $\rfs'' \eqon{\Fn} \rfs$ satisfying the following property:
    \[
      \forall \lay. \left( \forall 1 \leq i \leq h, \lay(\id_i) = \add_i\right)  \Rightarrow
      \nstep{\nat}
      {\conf{\frame{\syss(\syscall)}{\regmap}{\km[\syscall]}, \lay \lcomp \rfs'}}
      {\conf{\st, \lay \lcomp \rfs''}}.
    \]
    Since $\st$ is non-empty, we can assume that $\st = \frame \cmd {\regmap'} {\km[\syscall]}:\st'$.
    Furthermore, from the IH, we know that that $\refs_\system(\st)\subseteq \refs_\system(\syss(\syscall))$ (H). We now perform a case analysis on $\cmd$.
    \begin{proofcases}
      \proofcase{$\vx \ass\expr\sep \cmdtwo$} We first observe that
      there exists a value $\val$ such that, for every layout in our
      quantification $\lay$, $\sem \expr_{\regmap, \lay} = \val$.
      This follows from \Cref{rem:expreval}, given that for every
      $\id \in \refs_\system(\syss(\syscall))$ and any two layouts
      $\lay_1, \lay_2$, we assume that $\lay_1(\id) = \lay_2(\id)$. In
      particular, this applies to identifiers appearing in $\expr$
      that belong to $\refs_\system(\syss(\syscall))$ by (H).  For all
      layouts if our quantification layouts, we have:
      \begin{equation*}
        \step {\conf{\frame{\vx \ass\expr\sep\cmdtwo}{\regmap'}{\km[\syscall]}:\st', \lay\lcomp \rfs''}}{}\\
        {\conf{\frame{\cmdtwo}{\update {\regmap'}\vx {\val}}{\km}:\st', \lay\lcomp \rfs''}}.
        \tag{$*$}
      \end{equation*}
      From (H), we deduce
      \[
        \refs_\system(\frame{\cmdtwo}{\update {\regmap'}\vx {\val}}{\km}:\st')\subseteq\refs_\system(\syss(\syscall))
      \]
      Therefore, we conclude that (A) holds if either
      $\cmdtwo \neq \cnil$ or $\st'$ is non-empty. Otherwise, (C)
      holds.  \proofcase{$\cnil$} Analogous to the case above.
      \proofcase{$\cskip\sep \cmdtwo$} Analogous to the case of
      assignments.
      \proofcase{$\cif \expr {\cmd_\bot}{\cmd_\top}\sep\cmdtwo$}
      Analogous to the case of assignments.
      \proofcase{$\cwhile \expr {\cmd}\sep\cmdtwo$} Analogous to the
      case of assignments.
      \proofcase{$\cmemass \expr \exprtwo\sep \cmdtwo$} In this case,
      we start by observing that there exist a value $\val_\expr$ such
      that for each layout $\lay$ that satisfies the premise, we have
      $\sem\expr_{\regmap,\lay} = \val_\expr$. This follows from
      \Cref{rem:expreval}, (H2), and the definition of $\refs$, which
      ensures that all the identifiers within $\expr$ are in
      $\refs(\frame{\cmemass\expr\exprtwo\sep\cmdtwo}{\regmap'}{\km[\syscall]}:\st')$.
      For the same reason, there is a unique value $\val_\exprtwo$
      such that for each of the layout in our quantification, we have
      that $\sem\exprtwo_{\regmaptwo,\lay} = \val_\exprtwo$.
      Therefore, there is a unique $\add\in \Add = \toAdd{\val_\expr}$ which
      the store instruction attempts to write at.  Finally, we observe
      that there is a unique set $P$ such that, for every layout
      $\lay$ satisfying the assumption, we have
      $P = \underline \lay(\refs_\system(\syss(\syscall))\setminus \Sys)$.  We go by
      case analysis on $\add \in P$.
      \begin{proofcases}
        \proofcase{$\add\in P$} In this case, for every layout $\lay$
        that we are quantifying over, there exists some index $i$ such
        that $\add \in \underline \lay(\id_i)$.  If $\id_i$ is a
        procedure identifier, then we have
        $\add = \underline{\lay}(\id_i)$, and thus, independently of $\lay$, we have:
        \[
          \step
          {\conf{\frame{\cmemass\expr\exprtwo\sep\cmdtwo}{\regmap'}{\km[\syscall]}:\st',
              \lay\lcomp \rfs''}} {\err}.
        \]
        This shows that (B) holds with probability 1.  Otherwise,
        there exists a unique array $\ar$ and a unique index $j$ such
        that, independently of the layout, we have
        $\lay(\ar) + j = \add$. Notice that if $\ar$ and $j$ were not
        unique, then it would not be true that the layouts store
        $\id_1, \dots,\id_h$ respectively at $\add_1, \dots, \add_h$.
        Thus, each layout $\lay$ that satisfies the premises of rule
        \ref{WL:Store}, we apply the rule to show the following
        transition:
        \begin{equation*}
          \step {\conf{\frame{\cmemass\expr\exprtwo\sep\cmdtwo}{\regmap'}{\km[\syscall]}:\st', \lay\lcomp \rfs''}}{}\\
          {\conf{\frame{\cmdtwo}{\regmap'}{\km[\syscall]}:\st', \update{(\lay\lcomp \rfs'')}\add{\val_\exprtwo}}}.
        \end{equation*}
        Using \Cref{rem:memupdtostupd}, we further conclude:
        \begin{equation*}
          \step {\conf{\frame{\cmemass\expr\exprtwo\sep\cmdtwo}{\regmap'}{\km[\syscall]}:\st', \lay\lcomp \rfs''}}{}\\
          {\conf{\frame{\cmdtwo}{\regmap'}{\km[\syscall]}:\st', \lay\lcomp {(\update{\rfs''}{(\ar, j)}{\val_\exprtwo})}}}.
        \end{equation*}
        By observing the uniqueness of the target configuration modulo
        $\lay$, and given that
        $\refs_\system(\frame{\cmdtwo}{\regmap'}{\km[\syscall]}:\st')
        \subseteq \refs_\system(\syss(\syscall))$ follows from (H), we
        conclude that (A) holds. If $\cmdtwo = \cnil$ and $\st'$ is
        empty (C) holds, instead.
        \proofcase{$\add\notin P$}
        %
        Observe that, for every layout $\lay$ such that
        $\add \notin \underline{\lay}(\Idk)$,
        only rule \ref{WL:Store-Error} applies, which shows the following transition:
        \[
          \step {\conf{\frame{\cmemass\expr\exprtwo\sep\cmdtwo}{\regmap'}{\km[\syscall]}:\st', \lay\lcomp \rfs''}}\err.
        \]
        Thus, we observe that:
        \begin{equation*}
          \Pr_{\lay \leftarrow\mu}\big[\step {\conf{\frame{\cmemass\expr\exprtwo\sep\cmdtwo}{\regmap'}{\km[\syscall]}:\st', \lay\lcomp \rfs''}}\err\,\, \big|\\
          \forall 1\le i\le h.\lay(\id_i)=\add_i\big]
        \end{equation*}
        is greater than
        \[
          \Pr_{\lay \leftarrow\mu}\big[ \add \notin \underline\lay(\Idk) \mid
          \forall 1\le i\le h.\lay(\id_i)=\add_i\big]
        \]
        which, by definition, is greater than $\delta_{\mu}$.  
        This shows that (B) holds.
      \end{proofcases}
      \proofcase{$\cmemread \vx \expr\sep \cmdtwo$} Analogous to the
      case of store operations.
      \proofcase{$\ccall \expr {\exprtwo_1, \dots, \exprtwo_k}\sep
        \cmdtwo$} We start by observing that there exists a unique
      address $\add$ such that for every $\lay$ that satisfies the
      precondition, we have $\toAdd{\sem
        \expr_{\regmap,\lay}}=\add$. Similarly, we introduce the
      values $\val_1, \dots,\val_k$, which correspond to the semantics
      of $\exprtwo_1,\dots,\exprtwo_k$ evaluated under $\regmap$ and
      every layout that satisfies the precondition.  Then, we observe
      that there exists a set $P$ such that for each layout under
      consideration, it holds that
      $\underline{\lay} (\refs_\system(\syss(\syscall))\setminus
      \Sys)=P$. The proof proceeds by cases on whether $\add \in P$.
      \begin{proofcases}
        \proofcase{$\add\in P$} In this case, there is a unique identifier $\id_j$ such that for every layout $\lay$ that satisfies the precondition, we have $\add \in \underline{\lay}(\id_j)$. We analyze two cases based on whether $\id_j$ is a function identifier $ \fn$.
        \begin{proofcases}
          \proofcase{$\id_j= \fn$} In this case, from the definition
          of $\lcomp$, we deduce that for each of these layouts we
          have $\lay\lcomp \rfs''(\add) = \rfs''(\fn) = \rfs(\fn)$,
          where the last step follows from the assumption
          $\rfs''\eqon{\Fn}\rfs$.  Since $\add \in P$, we conclude
          that, independently of the specific layout, if the
          preconditions hold, then:
          \begin{equation*}
            \step
            {\conf{\frame{\ccall  \expr {\exprtwo_1, \dots, \exprtwo_k}\sep \cmdtwo}{\regmap}{\km[\syscall]}:\st', \lay\lcomp \rfs''}}{}\\
            {\conf{\frame{\rfs(\fn)}{\regmap_0'}{\km[\syscall]}:\frame{\cmdtwo}{\regmap}{\km[\syscall]}:\st', \lay\lcomp \rfs''}},
          \end{equation*}
          where $\regmap_0'= \regmap_0[\vx_1,\dots,\vx_k\upd \val_1,\dots,\val_k]$.
          By definition, $\refs_\system(\syss(\syscall))$
          contains all the identifiers within
          $\rfs(\fn) = \rfs(\id_j)$ because $\id_j \in \refs_\system(\syss(\syscall))$ and
          $\refs$ is closed under procedure calls.
          This shows that (A) holds.
          \proofcase{$\id_j =\ar$} In this case, since the set of
          array identifiers and that of functions are disjoint, we
          conclude that for every layout $\lay$ that satisfies the
          preconditions, we have that
          $\add \notin \underline \lay(\Fn[{\km}])$.  This means that
          for each of these layouts, we can show:
          \begin{equation*}
            \step
            {\conf{\frame{\ccall  \expr {\exprtwo_1, \dots, \exprtwo_k}\sep \cmdtwo}{\regmap}{\km[\syscall]}:\st', \lay\lcomp \rfs''}}
            {\err},
          \end{equation*}
          and this means that (B) holds with probability 1.
        \end{proofcases}
        \proofcase{$p \notin P$} Analogous to the corresponding case
        for store operations.
      \end{proofcases}
      \proofcase{$\csyscall \expr {\exprtwo_1, \dots, \exprtwo_k}\sep \cmdtwo$}
      Analogous to the previous case.
    \end{proofcases}
  \end{proofcases}
    \proofcase{B} From the definition of $\to^{!\nat}$, we observe that if
    \[
      \nstep {! \nat} {\conf{\frame{\syss(\syscall)}{\regmap}{\km[\syscall]}, \lay\lcomp \rfs'}}\err,
    \]
    then
    \[
      \nstep {! \nat+1} {\conf{\frame{\syss(\syscall)}{\regmap}{\km[\syscall]}, \lay\lcomp \rfs'}}\err.
    \]
    This implies that
    \[
      \Pr_{\lay\leftarrow \mu}\Big[ \nstep {! \nat+1} {\conf{\frame{\syss(\syscall)}{\regmap}{\km[\syscall]}, \lay\lcomp \rfs'}}\err \,\,\Big|\,\,
      \forall 1\le i\le h.\lay(\id_i)=\add_i\Big]
    \]
    is greater than
    \[
      \Pr_{\lay\leftarrow \mu}\Big[ \nstep {! \nat} {\conf{\frame{\syss(\syscall)}{\regmap}{\km[\syscall]}, \lay\lcomp \rfs'}}\err \,\,\Big|\,\,
      \forall 1\le i\le h.\lay(\id_i)=\add_i\Big],
    \]
    which proves the claim.
    \proofcase{C} Similar to the previous case.
  \end{proofcases}
\end{proof}

\begin{remark}
  \label{rem:deltanu}
  \[
    \delta_\nu \ge \min_{\syscall \in \Sys}\frac{\kappa_{\km}/\theta-|{\Idk}|} {\kappa_{\km}/\theta-|{\refs_\system(\syss(\syscall))\setminus \Sys}|}
  \]
\end{remark}

\begin{proof}
  We want to prove that
  \begin{multline*}
     \min\bigl \{
      \displaystyle{\Pr_{\lay\leftarrow \nu}}
      [ \add \notin \underline\lay(\Id) \mid \lay(\id_i^\syscall)=p_i,
      \text{ for } 1\le i \le h]  \mid \syscall \in \Sys, p,p_1,\dots,p_h \in \Addk \land {}\\[3pt]
      \phantom{{} \mid} p \notin \{p_i, \ldots, p_i +\size{\id_i^\syscall}-1\}, \text{ for } 1\le i \le h
      \bigr \}.
    \end{multline*}
    is greater than
    $\min_{\syscall \in \Sys}\frac{\kappa_{\km}/\theta-|{\Idk}|} {\kappa_{\km}/\theta-|{\refs_\system(\syss(\syscall))}\setminus \Sys|}$.
    To establish this, we fix a system call $\syscall \in \Sys$
    and addresses
    $p,p_1,\dots,p_h \in \Addk$ that achieve the minimum.
    In particular, we have:
  \[
    p \notin \{p_i, \ldots, p_i +\size{\id_i^\syscall}-1\}, \text{ for } 1\le i \le h.
    \tag{\dag}
  \]
  Given that the minimum exists, $p_1, \ldots, p_h$ are the starting
  addresses of the slots $s_1, \ldots, s_h$ where the references
  $\refs_\system(\syss(\syscall))\setminus \Sys$ of $\syscall$ are allocated.
  Now, assume that $\add$ is located within one of these slots, say
  $s_j$. From assumption (\dag), we deduce that $\add$ cannot be
  allocated, since no object other than $\id_j$ can be placed in that
  slot. As a result, we obtain:
  \[
    \Pr_{\lay\leftarrow \nu}
    [ \add \notin \underline\lay(\Id) \mid \lay(\id_i^\syscall)=p_i,
    \text{ for } 1\le i \le h] = 1,
  \]
  which establishes the claim.

  If, on the other hand, $\add$ is located in a different slot $s$, we observe that:
  \begin{multline*}
    \Pr_{\lay\leftarrow \nu}
    [ \add \notin \underline \lay(\Id) \mid \lay(\id_i^\syscall)=p_i,
    \text{ for } 1\le i \le h] \ge\\
    \Pr_{\lay\leftarrow \nu}
    [ \mathit{slotof}(\add) \notin \lay(\Id) \mid \lay(\id_i^\syscall)=p_i,
    \text{ for } 1\le i \le h]
    \tag{$*$}
  \end{multline*} 
  where $\mathit{slotof}$ associates each address with the starting
  address of its corresponding slot.
  Due to the definition of $\nu$, the slots are sampled uniformly and
  they are independent. Therefore, the probability to the right in
  ($*$) is given by the ratio of the free slots and those that are not
  occupied by elements of
  $\refs_\system(\syss(\syscall))\setminus \Sys$, because of the
assumption (\dag). In conclusion ($*$) is equal to:
  \[
    \frac{\kappa_{\km}/\theta-|{\Idk}|}
    {\kappa_{\km}/\theta-|{\refs_\system(\syss(\syscall))}\setminus \Sys|} \ge
    \min_{\syscall \in \Sys}\frac{\kappa_{\km}/\theta-|{\Idk}|}
    {\kappa_{\km}/\theta-|{\refs_\system(\syss(\syscall))}\setminus \Sys|}.
  \]
\end{proof}

\paragraph{Technical observations}

\begin{remark}
  \label{rem:expreval}
Let $ \{\id_1, \dots, \id_h\} \subseteq \Id $ be a set of identifiers, $ \{\add_1, \dots, \add_h\} $ a set of addresses, and $ W \subseteq \{\lay \in \Lay \mid \forall 1 \leq i \leq h, \, \lay(\id_i) = p_i \} $ a set of layouts. Given an expression $ \expr $ such that $ \ids(\expr) \subseteq \{\id_1, \dots, \id_h\} $ and a register map $ \regmap $, there exists a value $ \val \in \Val $ such that
\[
\forall \lay \in W, \, \sem{\expr}_{\regmap, \lay} = \val.
\]
\end{remark}
\begin{proof}
  We proceed by cases on the size of $ W $. If $ |W| = 0 $ or $ |W| = 1 $, the claim is trivial. Otherwise, we prove the following auxiliary claim:
  \[
    \forall \expr. \ids(\expr) \subseteq \{\id_1,\dots,\id_h\}\Rightarrow \forall \lay_1, \lay_2  \in \Lay. \forall 1\le i \le h. \lay_1(\id_i) = \lay_2(\id_i) \Rightarrow
    \sem \expr_{\regmap, \lay_1} =\sem \expr_{\regmap, \lay_2},
  \]
  which can be proved by induction on the syntax of $ \expr $. The main claim follows from the IH on $ W \setminus \{\lay\} $ for some layout $ \lay $, the application of the auxiliary claim to a pair of layouts $ \overline{\lay} \in W \setminus \{\lay\} $ and $ \lay $, and the transitivity of equality.
\end{proof}

\begin{remark}
  \label{rem:memupdtostupd}
  Let $ \add \in \Add $ be an address, $ \ar \in \Arr $ be an array, and $ \rfs $ be a store. For $ 0 \leq i < \size(\ar) $, if $ \add = \lay(\ar) + i $, then
  \[
    \update{(\lay \lcomp \ars)}{\add}{\val} = \lay \lcomp (\update{{\rfs}}{(\add, i)}{\val}).
  \]
\end{remark}

\begin{proof}
  The claim follows directly from unrolling the definitions.
\end{proof}

\begin{remark}
  \label{rem:simispreserved}
  For every layout $ \lay \in \Lay $, and pair of configurations $ \conf{\st, \lay \lcomp \rfs} $ and $ \conf{\st', \mem} $ such that $ \nstep*{\conf{\st, \lay \lcomp \rfs}}{\conf{\st', \mem}} $, it holds that $ \mem = \lay \lcomp \rfs' $ for some $ \rfs' \eqon{\Fn} \rfs $.
\end{remark}

\begin{proof}
  The proof proceeds by induction on the length of the reduction. The base case follows from the reflexivity of $ \eqon{\Fn} $. The inductive case follows from introspection of the rule that has been used for the last transition. The only non-trivial case occurs for the \ref{WL:Store} rule, where the premise $ \add \in \underline{\lay} (\Ar[\opt]) $ guarantees the existence of a pair $ (\ar, i) $ such that $ 0 \leq i < \size\ar $. The observation follows from \Cref{rem:memupdtostupd}.
\end{proof}

\begin{lemma}
  \label{lemma:contextpluggingtech}
  For every layout $\lay\in \Lay$, configuration $\conf{\frame{\cmd}{\regmap}{\opt}, \mem}$ and non-empty stack $\fr:\st$,  $\nat\in \Nat$, and configuration $\confone$:
  \begin{varitemize}
  \item if $\nstep \nat {\conf{\frame{\cmd}{\regmap}{\opt}, \lay\lcomp \rfs}} {\conf{\st', \mem'}}$, then:
    \[
      \nstep \nat {\conf{\frame{\cmd}{\regmap}{\opt}:\fr:\st, \mem}} {\conf{\st':\fr:\st, \mem'}}
    \]
  \item if $\nstep \nat {\conf{\frame{\cmd}{\regmap}{\opt}, \lay\lcomp \rfs}} {\err}$, then:
    \[
      \nstep \nat {\conf{\frame{\cmd}{\regmap}{\opt}:\fr:\st, \mem}} {\err}
    \] 
  \item if $\nstep \nat {\conf{\frame{\cmd}{\regmap}{\opt}, \lay\lcomp \rfs}} {\unsafe}$, then:
    \[
      \nstep \nat {\conf{\frame{\cmd}{\regmap}{\opt}:\fr:\st, \mem}} {\unsafe}
    \] 
  \end{varitemize}
\end{lemma}

\begin{proof}
  By induction on $\nat$. The base case is trivial.  The inductive
  case follows from the IH and by introspection of the rule that
  has been used for showing the last transition. In particular, it suffices to
  observe that every rule that can be applied to show
  $
  \step {\conf{\st', \mem'}}
    {\conf{\st'', \mem'}}
  $
  can also be applied to show
  $
  \step
  {\conf{\st':\fr:\st, \mem'}}
  {\conf{\st'':\fr:\st, \mem'}}.
  $
\end{proof}

\begin{lemma}
  \label{lemma:contextplugging}
  For every layout $\lay\in \Lay$, configuration $\conf{\frame{\cmd}{\regmap}{\opt}, \lay\lcomp\rfs}$ and non-empty stack $\frame \cmdtwo {\regmap'}{\opt}:\st$,  $\nat\in \Nat$, and configuration $\confone$:
  \begin{varitemize}
  \item if $\Eval {\conf{\frame{\cmd}{\regmap}{\opt}, \lay\lcomp \rfs}} = {\val, \rfs'}$, then:
    \begin{equation*}
      \nstep * {\conf{\frame{\cmd}{\regmap}{\opt}:\frame \cmdtwo {\regmap'}{\opt}:\st, \lay\lcomp\rfs}} {}\\
      {\conf{\frame \cmdtwo {\update{\regmap'}\ret{\val}}{\opt}:\st, \lay\lcomp\rfs'}}
    \end{equation*}
  \item if $\Eval {\conf{\frame{\cmd}{\regmap}{\opt}, \lay\lcomp \rfs}} = \err$, then:
    \[
      \nstep * {\conf{\frame{\cmd}{\regmap}{\opt}:\frame \cmdtwo {\regmap'}{\opt}:\st, \lay\lcomp\rfs}} {\err}
    \] 
  \item if $\Eval {\conf{\frame{\cmd}{\regmap}{\opt}, \lay\lcomp \rfs}} = \unsafe$, then:
    \[
      \nstep * {\conf{\frame{\cmd}{\regmap}{\opt}:\frame \cmdtwo {\regmap'}{\opt}:\st, \lay\lcomp\rfs}} {\unsafe}
    \] 
  \item if $\Eval {\conf{\frame{\cmd}{\regmap}{\opt}, \lay\lcomp \rfs}} = \Omega$, then:
    \[
      \Eval {\conf{\frame{\cmd}{\regmap}{\opt}:\frame \cmdtwo {\regmap'}{\opt}:\st, \lay\lcomp\rfs}} = \Omega
    \] 
  \end{varitemize}
\end{lemma}

\begin{proof}
  All cases follow directly from \Cref{lemma:contextpluggingtech}. The
  first three cases are omitted, as the most interesting one is the
  last.
  Expanding the definition of $ \Eval \cdot $, we obtain that $\Eval {\conf{\frame{\cmd}{\regmap}{\opt}, \lay\lcomp \rfs}} = \Omega$ is equivalent to stating that for every $\nat$,
  \[
    \exists \conftwo. \nstep \nat {\conf{\frame{\cmd}{\regmap}{\opt}, \lay\lcomp \rfs}} \conftwo.
    \tag{\dag}
  \]
  Our goal is to show that for every $\nat$:
  \[
    \exists \conftwo'. \nstep \nat {\conf{\frame{\cmd}{\regmap}{\opt}:\frame \cmdtwo {\regmap'}{\opt}:\st, \lay\lcomp \rfs}} {\conftwo'}
  \]
  Fix $ \nat $. The claim follows by applying (\dag) to $\nat$ and then using \Cref{lemma:contextpluggingtech} on the $\nat$-step reduction  $\nstep \nat {\conf{\frame{\cmd}{\regmap}{\opt}, \lay\lcomp \rfs}} \conftwo$.
\end{proof}

%
%
%


\subsection{Appendix for \Cref{sec:safety2}}
\label{sec:appsafety2}

\paragraph*{Kernel-mode stack}
In this section, the predicate $\km[\syscall]$ additionally requires that all the commands within the stack belong to $\Cmd$, i.e., that they do not contain $\cspec\cdot$, $\cpoison \cdot$, and $\vx\ass\eobs$ statements:
\[
  \km[\syscall](\nil)\defsym \top \quad \km[\syscall](\fr:\st)\defsym \km[\syscalltwo](\fr) \land \km[\syscall](\st)\quad \km[\syscall](\frame{\cmd}{\regmap}{\opt}) \defsym \km(\cmd) \land \opt=\km[\syscall] \land \cmd \in \Cmd. 
\]
We write $\km(\cdot)$ as a shorthand for
$\exists \syscall \in \Sys. \km[\syscall](\cdot)$. Moreover, we relax
the definition of $\um$ by not requiring that programs are
unprivileged:
\[
  \um(\nil)\defsym \top \quad \um(\fr:\st)\defsym \um(\fr) \land \um(\st)\quad \um(\frame{\cmd}{\regmap}{\opt}) \defsym \opt=\um. 
\]

\subsubsection{Speculative Semantics of $\Cmd$}

The speculative semantics of $\Cmd$ is in \Cref{fig:scen2sem1,fig:scen2sem1bis,fig:scen2sem1bisbis}.

\begin{figure*}
  \centering
  \columnwidth=\linewidth
  \small
  \[
    \Infer[SI][Load-Step][\textsc{SLoad-Step}]
    {
      \lay \red \sframe{\frame{\cmemread \vx \expr\sep\cmd}{\regmap}{\opt}\cons\st}{\bmvar}{\boolms}\cons\cfstack
      \sto{\dstep}{\omem \add} \sframe{\frame{\cmd}{\update \regmap x \val}{\opt} \cons \st}{\bmvar}{\boolms}\cons\cfstack
    }
    {\toAdd{\sem\expr_{\regmap, \lay}} = \add &
      \bufread {\bmvar} \add 0 =\val, \bot &
      \add \in \underline \lay(\Ar[\opt]) &
      \fbox{$\opt = \km[\syscall] \Rightarrow \add \in \underline \lay(\caps(\syscall))$}
    }
  \]
  \\[-3mm]
\resizebox{\textwidth}{!}{\(
    \Infer[SI][Load][\textsc{SLoad-Load}]
    {\sstep
      {\confone \cons\cfstack}
      {\sframe{\frame{\cmd}{\update \regmap x \val}{\opt}\cons\st}{\bmvar}{\boolms\lor f}\cons
        \confone \cons \cfstack}  {{{\dload[\lbl]{i}}}} {\omem \add}}
    {\confone = \sframe{\frame{\cmemread[\lbl] \vx  \expr\sep\cmd}{\regmap}{\opt}\cons \st} {\bmvar}{\boolms} &
      \toAdd{\sem\expr_{\regmap, \lay}} = \add &
      \bufread {\bmvar} \add i =(\val, f) &
      \add \in \underline \lay(\Ar[\opt]) &
      \fbox{$\opt = \km[\syscall] \Rightarrow \add \in \underline \lay(\caps(\syscall))$}
    }
    \)
  }\\[-3mm]
  \[
    \Infer[SI][Load-Err][\textsc{SLoad-Error}]
    {\sstep
      {\sframe{\frame{\cmemread[\lbl] \vx \expr\sep\cmd}{\regmap}{\opt}\cons \st}{\bmvar}{\boolms}\cons\cfstack}
      {\sconf{\err, \boolms}\cons \cfstack}  {\dir} {\onone}}
    { \toAdd{\sem\expr_{\regmap, \lay}} = \add &
      \add \notin \underline \lay(\Ar[\opt]) &
      \dir = \dstep \lor \dir = \dload[\lbl]{i} &
    }
  \]
  \\[-3mm]
  \[
    \Infer[SI][Load-Unsafe][\textsc{SLoad-Unsafe}]
    {\sstep
      {\sframe{\frame{\cmemread[\lbl] \vx \expr\sep\cmd}{\regmap}{\km[\syscall]}\cons \st}{\bmvar}{\boolms}\cons\cfstack}
      {{\unsafe}}   {\dir} {\omem \add}}
    {\toAdd{\sem\expr_{\regmap, \lay}} = \add &
      \add \in \underline \lay(\Ar[\km]) &
      \dir = \dstep \lor \dir = \dload[\lbl]{i} &
      \fbox{$\add \notin \underline \lay(\caps(\syscall))$}
    }
  \]
  \\[-3mm]
  \[
    \Infer[SI][Store][\textsc{SStore}]
    {
      \lay \red \sframe{\frame{\cmemass \expr \exprtwo\sep\cmd}{\regmap}{\opt}\cons\st}{\bm\buf\mem}{\boolms}\cons\cfstack
      \sto{\dstep}{\omem \add} \sframe{\frame{\cmd}{\regmap}{\opt} \cons \st}{\bm{\bitem \add {\sem \exprtwo_{\regmap,\lay}}\cons \buf}\mem}{\boolms}\cons\cfstack
    }
    {\toAdd{\sem\expr_{\regmap, \lay}} = \add &
      \add \in \underline \lay(\Ar[\opt]) &
      \fbox{$\opt = \km[\syscall] \Rightarrow \add \in \underline \lay(\caps(\syscall))$}
    }
  \]
  \\[-3mm]
  \[
    \Infer[SI][Store-Err][\textsc{SStore-Error}]
    {\lay \red \sframe{\frame{\cmemass \expr \exprtwo\sep\cmd}{\regmap}{\opt}\cons\st}{\bmvar}{\boolms}\cons\cfstack
      \sto{\dstep}{\onone} 
      {\sconf{\err, \boolms}\cons \cfstack}}
    { \toAdd{\sem\expr_{\regmap, \lay}} = \add &
      \add \notin \underline \lay(\Ar[\opt]) 
    }
  \]
  \\[-3mm]
  \[
    \Infer[SI][Store-Unsafe][\textsc{SStore-Unsafe}]
    {\lay \red \sframe{\frame{\cmemass \expr \exprtwo\sep\cmd}{\regmap}{\km[\syscall]}\cons\st}{\bmvar}{\boolms}\cons\cfstack
      \sto{\dstep}{\onone} {\unsafe}}
    {\toAdd{\sem\expr_{\regmap, \lay}} = \add &
      \add \in \underline \lay(\Ar[\km]) &
      \fbox{$\add \notin \underline \lay(\caps(\syscall))$}
    }
  \]
  \caption{Speculative rules for $\Cmd$ and a system $\system=(\rfs, \syss, \caps)$, Part I.}
  \label{fig:scen2sem1}
\end{figure*}

\begin{figure*}
  \centering
  \small
  \columnwidth=\linewidth
  \[
    \Infer[SI][Call-Step][\textsc{SCall-Step}]{
      \lay \red \sframe{\frame{\stat\sep\cmd}{\regmap}{\opt}\cons\st}{\bmem}{\boolms}\cons\cfstack
      \sto{\dstep}{\omem \add} \sframe{\frame{\mem(\add)} {\regmap_0[\vec\vx\upd\sem{\vec\exprtwo}_{\regmap,\lay}]}{\opt}\cons\frame{\cmd}{\regmap}{\opt} \cons \st}{\bmem}{\boolms}\cons\cfstack
    }
    {\stat \in \{ \ccall \expr {\vec\exprtwo},
      \cscall \expr {\vec\exprtwo}\} &
      \toAdd{\sem\expr_{\regmap, \lay}} = \add &
      \add \in \underline \lay(\Fn[\opt]) &
      \fbox{$\opt = \km[\syscall] \Rightarrow \add \in \underline \lay(\caps(\syscall))$}
    }
  \]
  \\[-3mm]
  \[
    \Infer[SI][Call-Step-Unsafe][\textsc{SCall-Step-Unsafe}]{\sstep
      {\sframe{\frame{\stat\sep\cmd}{\regmap}{\km[\syscall]}\cons\st}{\bmvar}{\boolms}\cons\cfstack}
      {\unsafe}
      {\dstep}
      {\ojump \add}}
    {
      \stat \in \{ \ccall \expr {\vec\exprtwo},
      \cscall \expr {\vec\exprtwo}\} &
      \toAdd{\sem\expr_{\regmap, \lay}} = \add &
      \add \in \underline \lay(\Fn[\opt]) &
      \fbox{$\add \notin \underline \lay(\caps(\syscall))$}
    }
  \]
  \\[-3mm]
  \[
    \Infer[SI][Call-Step-Err][\textsc{SCall-Step-Error}]{
      \lay \red \sframe{\frame{\stat\sep\cmd}{\regmap}{\opt}\cons\st}{\bmvar}{\boolms}\cons\cfstack
      \sto{\dstep}{\onone} \sconf{\err,\boolms}\cons\cfstack
    }
    {
      \stat \in \{ \ccall \expr {\vec\exprtwo},
      \cscall \expr {\vec\exprtwo}\} &
      \toAdd{\sem\expr_{\regmap, \lay}} = \add &
      \add \notin \underline \lay(\Fn[\opt])
    }
  \]
  \\[-3mm]
  \[
    \Infer[SI][Call][\textsc{SCall}]{
        \lay \red \confone\cons\cfstack
        \sto{\djump[\lbl] \add}{\ojump \add} \sframe{\frame{\mem(\add)} {\regmap_0[\vec \vx\upd\sem{\vec \exprtwo}_{\regmap,\lay}]}{\opt}\cons\frame{\cmd}{\regmap}{\opt} \cons \st}{\bm\buf \mem}{\boolms \lor (\add \neq \toAdd{\sem\expr_{\regmap, \lay}})}\cons \confone\cons \cfstack
      }
      {
        \confone = \sframe{\frame{\ccall[\lbl] \expr {\vec \exprtwo}\sep\cmd}{\regmap}{\opt}\cons\st}{\bm\buf \mem}{\boolms} &
        \add \in \underline \lay(\Fn[\opt]) &
        \fbox{$\opt = \km[\syscall] \Rightarrow \add \in \underline \lay(\caps(\syscall))$}
      }
  \]
  \\[-3mm]
  \[
    \Infer[SI][Call-Unsafe][\textsc{SCall-Unsafe}]{\sstep
        {\sframe{\frame{\ccall[\lbl] \expr{\vec \exprtwo}\sep\cmd}{\regmap}{\km[\syscall]}\cons\st}{\bmvar}{\boolms}\cons\cfstack}
        {\unsafe}
        {\djump[\lbl] \add}
        {\ojump \add}}
      { \add \in \underline \lay(\Fn[\km]) &
        \fbox{$\add \notin \underline \lay(\caps(\syscall))$}
      }
  \]
  \\[-3mm]
  \[
    \Infer[SI][Call-Err][\textsc{SCall-Error}]{
        \lay \red \confone\cons\cfstack
        \sto{\djump[\lbl] \add}{\onone} \sconf{\err,\boolms \lor (\add \neq \toAdd{\sem\expr_{\regmap, \lay}})}\cons \confone \cons \cfstack
      }
      {
        \confone = \sframe{\frame{\ccall[\lbl] \expr {\vec \exprtwo}\sep\cmd}{\regmap}{\opt}\cons\st}{\bmem}{\boolms} &
        \add \notin \underline \lay(\Fn[\opt])
      }
    \]
    \\[-3mm]
    \[
    \Infer[SI][System-Call][\textsc{SSystemCall}]{
      \lay \red \sframe{\frame{\csyscall \syscall {\vec \exprtwo}\sep\cmd}{\regmap}{\opt}\cons\st}{\bmvar}{\boolms}\cons\cfstack
      \sto{\dstep}{\onone} \sframe{\frame{\syss(\add)} {\regmap_0[\vec\vx\upd\sem{\vec\exprtwo}_{\regmap,\lay}]}{\opt'}\cons\frame{\cmd}{\regmap}{\opt} \cons \st}{\bmvar}{\boolms}\cons\cfstack
    }
    {
      \opt = \um \Rightarrow \opt' = \km[\syscall] &
      \opt = \km[\syscalltwo] \Rightarrow \opt' = \km[\syscalltwo] }
  \]
  \\[-3mm]
  \[
    \Infer[SI][Pop][\textsc{SPop}]{\sstep
      {\sframe{\frame{\cnil}{\regmap}{\opt}\cons\frame{\cmd}{\regmap'}{\opt'} \cons \st}{\bm\buf\mem}{\boolms}\cons\cfstack}
      {\sframe{\frame{\cmd}{\update{\regmap'}{\ret}{\regmap(\ret)}}{\opt'} \cons \st}{\bm\buf\mem}{\boolms}\cons\cfstack}
      {\dstep}
      {\onone}}
    {}
  \]
\caption{Speculative rules for $\Cmd$ and a system $\system=(\rfs, \syss, \caps)$, Part II.}
  \label{fig:scen2sem1bis}
\end{figure*}

\begin{figure*}
  \columnwidth=\linewidth
  \centering
  \small
  \[
    \Infer[SI][Op][\textsc{SOp}]{\sstep
      {\sframe{\frame{\vx \ass \expr\sep\cmd}{\regmap}{\opt}\cons\st}{\bmvar}{\boolms}\cons \cfstack}
      {\sframe{\frame{\cmd}{\update \regmap x {\sem \expr_{\regmap, \lay}}}{\opt}\cons\st}{\bmvar}{\boolms}\cons \cfstack} {\dstep}{\onone}}{}
  \]
  \\[-3mm]
  \[
    \Infer[SI][Skip][\textsc{SSkip}]{\sstep
      {\sframe{\frame{\cskip\sep\cmd}{\regmap}{\opt}\cons \st}{\bmvar}{\boolms}\cons \cfstack}
      {\sframe{\frame{\cmd}{\regmap}{\opt}\cons\st}{\bmvar}{\boolms}\cons \cfstack} {\dstep}{\onone}}{}
  \]

  ~\\[-3mm]  

  \resizebox{\textwidth}{!}{
    \(
    \Infer[SI][Loop-Step][\textsc{SLoop-Step}]
    {\sstep
      {\sframe{\frame{\cwhile[\lbl] \expr \cmdtwo\sep\cmd}{\regmap}{\opt}\cons\st}{\bm\buf\mem}{\boolms}\cons\cfstack}
      {\specconfone_{d} \cons\cfstack}
      {\dstep}{\obranch d}}
    {
      d = {\toBool{\sem\expr_{\regmap, \lay}}} &
      \specconfone_\ctrue =\sframe{\frame{\cmdtwo\sep\cwhile[\lbl] \expr \cmdtwo\sep\cmd}{\regmap}{\opt}\cons\st}{\bm{\buf}{\mem}}{\boolms} &
      \specconfone_\cfalse =\sframe{\frame{\cmd}{\regmap}{\opt}\cons\st}{\bm \buf \mem}{\boolms}
    }
    \)}
  ~\\[-3mm]
  \[
    \Infer[SI][Loop-Branch][\textsc{SLoop}]
    {\sstep
      {\sframe{\frame{\cwhile[\lbl] \expr \cmdtwo\sep\cmd}{\regmap}{\opt}\cons\st}{\bmvar}{\boolms}\cons\cfstack}
      {\specconfone_{d}\cons  \sframe{\frame{\cwhile[\lbl]\expr \cmdtwo\sep\cmd}{\regmap}{\opt}\cons\st}{\bmvar}{\boolms}\cons\cfstack}
      {\dbranch d}{\obranch {d}}
    }
    {
      \boolms' = \boolms\lor(d\neq \toBool{\sem\expr_{\regmap, \lay}})&
      \specconfone_\ctrue =\sframe{\frame{\cmdtwo\sep\cwhile[\lbl] \expr \cmdtwo\sep\cmd}{\regmap}{\opt}\cons\st}{\bmvar}{\boolms'} &
      \specconfone_\cfalse =\sframe{\frame{\cmd}{\regmap}{\opt}\cons\st}{\bmvar}{\boolms'}
    }
  \]
  \\[-3mm]
  \[
    \Infer[SI][If][\textsc{SIf-Step}]{\sstep
      {\sframe{\frame{\cif[\lbl] \expr {\cmd_\ctrue} {\cmd_\cfalse}\sep\cmd}{\regmap}{\opt}\cons\st}{\bmvar}{\boolms}\cons\cfstack}
      {\sframe{\frame{\cmd_{d}\sep\cmd}{\regmap}{\opt}\cons\st}{\bmvar}{ \boolms}\cons\cfstack} {{\dstep}}{\obranch{d}}}
    {d=\toBool{\sem \expr_{\regmap, \lay}}}
  \]
  \\[-3mm]
  \[
    \Infer[SI][If-Branch][\textsc{SIf}]{\sstep
      {\specconfone\cons\cfstack}
      {\sframe{\frame{\cmd_{d}\sep\cmd}{\regmap}{\opt}\cons\st}{\bmvar}{\boolms\lor (d\neq\sem \expr_{\regmap, \lay}^\cBool)}\cons
        \specconfone\cons\cfstack}
      {{{\dbranch[\lbl] {d}}}}
      {\obranch d}}
    {\specconfone=\sframe {\frame{\cif[\lbl] \expr {\cmd_\ctrue} {\cmd_\cfalse}\sep\cmd}{\regmap}{\opt}\cons\st}{\bmvar}{\boolms}}
  \]
  \\[-3mm]
  \[
    \Infer[SI][Backtrack-Top][\textsc{Bt}_{\top}]{
      \lay \red \specconfone \cons\cfstack \sto{\dbt}{\obt \top} \cfstack
    }{
      \specconfone = \sframe{\st}{\bmvar}{\top} \lor \specconfone = \sconf{\err,\top}
    }
    \quad
    \Infer[SI][Backtrack-Bot][\textsc{Bt}_{\bot}]{
      \lay \red \specconfone \cons\cfstack \sto{\dbt}{\obt \bot} \specconfone \cons \nil
    }{
      \specconfone = \sframe{\st}{\bmvar}{\bot} \lor \specconfone = \sconf{\err,\bot}
      & \cfstack \neq \nil
    }
  \]
  \\[-3mm]
  \[
    \Infer[SI][Fence]{\sstep
      {\sframe{\frame{\cfence\sep\cmd}{\regmap}{\opt}\cons\st}{\bmvar}{\bot}\cons \cfstack}
      {\sframe{\frame{\cmd}{\regmap}{\opt}\cons\st} {\bm\nil {\overline {\bmvar}}} \bot \cons\cfstack} {\dstep}{\onone}}{}
  \]
  \caption{Speculative rules for $\Cmd$ and a system $\system=(\rfs, \syss, \caps)$, Part III.}
  \label{fig:scen2sem1bisbis}
\end{figure*}

\subsubsection{Semantics of $\SpCmd$} The semantics of $\SpCmd$ is in
\Cref{fig:scen2sem2,fig:scen2sem2bis}.

\begin{figure*}[t]
  \small
  \centering
  \columnwidth=\linewidth    
    \[
      \Infer[AL][Poison]{
      \lay \red \aconf{\frame{\cpoison \dir\sep\speccmd}{\regmap}{\opt}\cons\st }{\mem}{\Ds}{\Os}
      \ato \aconf{\frame{\speccmd}{\regmap}{\opt}\cons\st}{\mem}{\dir:\Ds}{\Os}
    }{\strut}
    \]
    \\[-3mm]
    \[
      \Infer[AL][Obs][\textsc{Observe}]{
      \lay \red \aconf{\frame{\vx \ass \eobs\sep\speccmd}{\regmap}{\opt}\cons\st }{\mem}{\Ds}{\obs:\Os}
      \ato \aconf{\frame{\speccmd}{\update \regmap\vx \obs}{\opt}\cons\st}{\mem}{\Ds}{\Os}
      }{}
    \]
    \\[-3mm]
    \[
      \Infer[AL][Obs-End][\textsc{Observe-End}]{
        \lay \red \aconf{\frame{\vx \ass \eobs\sep\adversary}{\regmap}{\opt}\cons\st }{\mem}{\Ds}{\nil}
        \ato \aconf{\frame{\adversary}{\update \regmap\vx \cnull}{\opt}\cons\st}{\mem}{\Ds}{\nil}
      }{}
    \]
    \\[-3mm]
    \[
      \Infer[AL][Spec-Init]{
      \lay \red \aconf{\frame{\cspec\cmd\sep\speccmd}{\regmap}{\opt}\cons\st }{\mem}{\Ds}{\Os}
      \ato \hconf{\sframe{\frame{\cmd}{\regmap}{\opt}}{\bm{\nil} \mem}{\bot}}{\frame{\speccmd}{\regmap}{\opt}\cons\st}{\Ds}{\Os}
      }{}
    \]
    \\[-3mm]
    \[
      \Infer[AL][Spec-Dir][\textsc{Spec-D}]{
        \lay \red \hconf{\cfstack}{\st}{\dir{\cons}\Ds}{\Os} \ato \hconf{\cfstack'}{\st}{\Ds}{\obs{\cons}\Os}
      }{
        \lay \red \cfstack \sto{\dir}{\obs} \cfstack'
      }
      \quad
      \Infer[AL][Spec-Step][\textsc{Spec-S}]{
        \lay \red \hconf{\cfstack}{\st}{\Ds}{\Os} \ato \hconf{\cfstack'}{\st}{\Ds}{\obs{\cons}\Os}
      }{
        \nf \cfstack \Ds &
        \lay \red \cfstack \sto{\dstep}{\obs} \cfstack'
      }
    \]
    \\[-3mm]
    \[
      \Infer[AL][Spec-Bt][\textsc{Spec-BT}]{
        \lay \red \hconf{\cfstack}{\st}{\Ds}{\Os} \ato \hconf{\cfstack'}{\st}{\Ds}{\obs{\cons}\Os}
      }{
        \nf \cfstack \Ds &
        \nf \cfstack \dstep &
        \lay \red \cfstack \sto{\dbt}{\obs} \cfstack'
      }
    \]
    \\[-3mm]
    \[
      \Infer[AL][Spec-Term]{
      \lay \red
      \hconf{\sframe{\frame{\cnil}{\regmap}{\opt}}{\bmvar}{\bot}}
      {\frame{\speccmd}{\regmap'}{\opt'}\cons\st}{\Ds}{\Os}
      \ato \aconf{\frame{\speccmd}{\regmap}{\opt}\cons\st}{\overline{\bmvar}}{\Ds}{\Os}
    }{}
    \]
    \\[-3mm]
    \[
      \Infer[AL][Spec-Error]{\esstep
        {\hconf{\sconf{\err,\bot}} \st \Ds \Os}
        {\err}
        {}}{}
      \quad
      \Infer[AL][Spec-Unsafe]{\esstep 
        {\hconf\unsafe\st \Ds \Os}
        {\unsafe}
        {}}{}
    \]
  \caption{Semantics of the non-standard constructs of $\SpCmd$ for the system $\system =(\rfs, \syss, \caps)$.}
  \label{fig:scen2sem2}
\end{figure*}

\begin{figure*}
  \centering
  \columnwidth=\linewidth
  \[
      \Infer[AL][Load][\textsc{ALoad}]
      { \esstep
        {\ntce{\cmemread \vx \expr\sep\cmd}{\regmap}{\opt}{\st}{\mem}}
        {\ntce{\cmd}{\update{\regmap}{x}{\mem(\add)}}{\opt}{\st}{\mem}}
      }
      {
        \toAdd{\sem{\expr}_{\regmap, \lay}} = \add &
        \add \in \underline \lay(\Ar[\opt]) &
        \fbox{$\opt = \km[\syscall] \Rightarrow \add \in \underline \lay(\caps(\syscall))$}
      }
    \]
    \\[-3mm]
    \[
      \Infer[AL][Load-Err][\textsc{ALoad-Error}]
      {\esstep
        {\ntce{\cmemread \vx \expr\sep\cmd}{\regmap}{\opt}{\st}{\mem}}
        {\err}
      }
      {\toAdd{\sem\expr_{\regmap, \lay}} = \add &
        \add \notin \underline \lay(\Ar[\opt])
      }
    \]
    \\[-3mm]
    \[
      \Infer[AL][Load-Unsafe][\textsc{ALoad-Unsafe}]
      {\esstep
        {\ntce{\cmemread \vx \expr\sep\cmd}{\regmap}{\km[\syscall]}{\st}{\mem}}
        {\unsafe}
      }
      {
        \toAdd{\sem{\expr}_{\regmap, \lay}} = \add &
        \add \in \underline \lay(\Ar[\km]) &
        \fbox{$\add \not\in \underline \lay(\caps(\syscall))$}
      }
    \]
    \\[-3mm]
    \[
      \Infer[AL][Store][\textsc{AStore}]
      { \esstep
        {\ntce{\cmemass \expr \exprtwo\sep\cmd}{\regmap}{\opt}{\st}{\mem}}
        {\ntce{\cmd}{\regmap}{\opt}{\st}{\update{\mem}{\add}{\sem \exprtwo_{\regmap, \lay}}}}
      }
      {
        \toAdd{\sem{\expr}_{\regmap, \lay}} = \add &
        \add \in \underline \lay(\Ar[\opt]) &
        \fbox{$\opt = \km[\syscall] \Rightarrow \add \in \underline \lay(\caps(\syscall))$}
      }
    \]
    \\[-3mm]
    \[
      \Infer[AL][Store-Err][\textsc{AStore-Error}]
      {\esstep
        {\ntce{\cmemass \expr \exprtwo\sep\cmd}{\regmap}{\opt}{\st}{\mem}}
        {\err}
      }
      {\toAdd{\sem\expr_{\regmap, \lay}} = \add &
        \add \notin \underline \lay(\Ar[\opt])
      }
    \]
    \\[-3mm]
    \[
      \Infer[AL][Store-Unsafe][\textsc{AStore-Unsafe}]
      {\esstep
        {\ntce{\cmemass \expr \exprtwo\sep\cmd}{\regmap}{\km[\syscall]}{\st}{\mem}}
        {\unsafe}
      }
      {
        \toAdd{\sem{\expr}_{\regmap, \lay}} = \add &
        \add \in \underline \lay(\Ar[\km]) &
        \fbox{$\add \not\in \underline \lay(\caps(\syscall))$}
      }
    \]
    \\[-3mm]
    \[
    \Infer[AL][Call][\textsc{ACall}]{
      \esstep
      {\ntce{\stat\sep\cmd}{\regmap}{\opt}{\st}{\mem}}
      {
        \ntc
        {\mem(\add)}
        {\regmap_0[\vec \vx \upd \sem{\vec \exprtwo}_{\regmap,\lay}]}
        {\opt}
        {\frame{\cmd}{\regmap}{\opt} : \st}
        {\mem}
      }
    }
    {
      \stat \in \{ \ccall \expr {\vec \exprtwo},\cscall \expr {\vec \exprtwo}\} &
      \toAdd{\sem{\expr}_{\regmap, \lay}}=\add &
      \add \in \underline \lay(\Fn[\opt]) &
      \fbox{$\opt = \km[\syscall] \Rightarrow \add \in \underline \lay(\caps(\syscall))$}
      }
    \]
    \\[-3mm]
    \[
      \Infer[AL][Call-Unsafe][\textsc{ACall-Unsafe}]{
        \esstep
        {\ntce{\stat\sep\cmd}{\regmap}{\km[\syscall]}{\st}{\mem}}
        {\unsafe}
      }
      {
        \stat \in \{ \ccall \expr {\vec \exprtwo},\cscall \expr {\vec \exprtwo}\} &
        \toAdd{\sem{\expr}_{\regmap, \lay}} = \add &
        \add \in \underline \lay(\Fn[\km]) &
        \fbox{$\add \not\in \underline \lay(\caps(\syscall))$}
      }
    \]
    \\[-3mm]
    \[
      \Infer[AL][Call-Err][\textsc{ACall-Error}]
      {\esstep
        {\ntce{\stat\sep\cmd}{\regmap}{\opt}{\st}{\mem}}
        {\err}
      }
      {
      \stat \in \{ \ccall \expr {\vec \exprtwo},\cscall \expr {\vec \exprtwo}\} &
        \toAdd{\sem\expr_{\regmap, \lay}} = \add &
        \add \notin \underline \lay(\Fn[\opt])
      }
    \]
    \\[-3mm]
    \[
      \Infer[AL][System-Call][\textsc{ASC}]{
        \esstep
        {\ntce{\csyscall{\syscall}{\vec \exprtwo}\sep\cmd}{\regmap}{\opt}{\st}{\mem}}
        {
          \ntce
          {\syss(\syscall)}
          {\regmap_0[\vec \vx\upd \sem{\vec \exprtwo}_{\regmap,\lay}]}
          {b'}
          {\frame{\cmd}{\regmap}{\opt} : \st}
          {\mem}
        }
      }
      {
      b = \um \Rightarrow b' = \km[\syscall] &
      b = \km[\syscalltwo] \Rightarrow b' = \km[\syscalltwo] 
      }
    \]

  \caption{Semantics of standard construct of $\SpCmd$ for the system $\system =(\rfs, \syss, \caps)$, part I.}
  \label{fig:scen2sem2bis}
\end{figure*}

\begin{figure*}[t]
  \centering
  \columnwidth=\linewidth
  \[
    \Infer[AL][Pop][\textsc{APop}]{
      \esstep
      {\ntce{\cnil}{\regmap}{\opt}{\frame {\cmd} {\regmap'} {\opt'}\cons \st}{\mem}}
      {\ntce{\cmd}{\update{\regmap'}{\ret} {\regmap(\ret)}}{\opt'}{\st}{\mem}}
    }
    {}
    \]
    \\[-3mm]
    \[
      \Infer[AL][Skip][\textsc{ASkip}]
      { \esstep
        {\ntce{\cskip\sep\cmd}{\regmap}{\opt}{\st}{\mem}}
        {\ntce{\cmd}{\regmap}{\opt}{\st}{\mem}}
      }
      {
      }
    \]
    \\[-3mm]
    \[
      \Infer[AL][Fence][\textsc{AFence}]
      { \esstep
        {\ntce{\cfence\sep\cmd}{\regmap}{\opt}{\st}{\mem}}
        {\ntce{\cmd}{\regmap}{\opt}{\st}{\mem}}
      }
      {
      }
    \]
    \\[-3mm]
    \[
      \Infer[AL][Op][\textsc{AOp}]
      { \esstep
        {\ntce{\vx \ass \expr\sep\cmd}{\regmap}{\opt}{\st}{\mem}}
        {\ntce{\cmd}{\update{\regmap}{\vx}{\sem \expr_{\regmap, \lay}}}{\opt}{\st}{\mem}}
      }
      {
      }
    \]  
    \\[-3mm]
    \[
      \Infer[AL][If][\textsc{AIf}]
      { \esstep
        {\ntce{\cif \expr {\cmd_\ctrue} {\cmd_\cfalse}\sep\cmdtwo}{\regmap}{\opt}{\st}{\mem}}
        {\ntce{\cmd_{\toBool{\sem \expr_{\regmap, \lay}}}\sep \cmdtwo}{\regmap}{\opt}{\st}{\mem}}
      }
      {
      }
    \]
    \\[-3mm]
    \[
      \Infer[AL][While][\textsc{AWhile}]
      { \esstep
        {\ntce{\cwhile {\expr} {\cmd}\sep\cmdtwo}{\regmap}{\opt}{\st}{\mem}}
        {\confone_{\toBool{\sem \expr_{\regmap, \lay}}}}
      }
      {
        \confone_{\ctrue} = {\ntce{\cmd\sep\cwhile {\expr} {\cmd}\sep \cmdtwo}{\regmap}{\opt}{\st}{\mem}} &
        \confone_{\cfalse} = {\ntce{\cmdtwo}{\regmap}{\opt}{\st}{\mem}}     }
    \]
  \caption{Semantics of standard construct of $\SpCmd$ for the system $\system =(\rfs, \syss, \caps)$, part II.}
  \label{fig:scen2sem2bisbis}
\end{figure*}

\subsubsection{Buffered Memories}
\label{sec:bufmem}

\begin{remark}
  \label{remark:onbuflookup}
  If $\bufread{\bm\buf\mem} \add i= \val, \bot$, then $\bufread{\bm\buf\mem} \add 0= \val, \bot$.  
\end{remark}
\begin{proof}
  The claim is:
  \[
    \forall \bm \buf\mem.\forall i.\forall \val.\bufread{\bm\buf\mem} \add i= \val, \bot \to \bufread{\bm\buf\mem} \add 0= \val, \bot
  \]
  We proceed by induction on the length of the buffer $\buf$.
  \begin{proofcases}
    \proofcase{$\nil$} The claim follows directly from the definition of lookup.
    \proofcase{$\bitem {\add'} {\overline\val}\cons \buf$} By the IH, we have:
    \[
      \forall i, \val.\bufread{\bm\buf\mem}\add i=\val, \bot \to \bufread{\bm\buf\mem}\add i=\bufread{\bm\buf\mem}\add 0.
    \]
    The goal is to prove:
    \begin{equation*}
      \forall i, \val.\bufread{(\bitem {\add'} {\overline\val}\cons \mu, \mem)}\add i=\val, \bot \Rightarrow\\ \bufread{(\bitem {\add'} {\overline\val}\cons\mu, \mem)}\add i=\bufread{(\bitem {\add'} {\overline\val}\cons\mu, \mem)}\add 0.
    \end{equation*}
    We proceed by case analysis on $i$.
    \begin{proofcases}
      \proofcase{$0$} The claim simplifies to:
      \begin{equation*}
        \forall \val.\bufread{\bm{\bitem {\add'} {\overline\val}\cons \mu}\mem}\add 0=\val, \bot \Rightarrow\\
        \bufread{\bm{\bitem {\add'} {\overline\val}\cons\mu} \mem}\add 0=\bufread{\bm{\bitem {\add'} {\overline\val}\cons\mu} \mem}\add 0.
      \end{equation*}
      The conclusion is trivial.
      \proofcase{$i+1$} The claim now states:
      \begin{equation*}
        \forall i.\forall \val.\bufread{\bm{\bitem {\add'} {\overline\val}\cons \mu}{\mem}}\add {i+1}=\val, \bot \Rightarrow\\
        \bufread{\bm{\bitem {\add'} {\overline\val}\cons \mu}{\mem}}\add {i+1}=\bufread{\bm{\bitem {\add'} {\overline\val}\cons \mu}\mem}\add 0.
      \end{equation*}
      Fix $i, \val$ and assume $\bufread{\bm{\bitem {\add'} {\overline\val}\cons \mu} \mem}\add {i+1}=\val, \bot$, call this assumption (H). The claim simplifies to:
      \[
        \bufread{\bm{\bitem {\add'} {\overline\val}\cons\mu} \mem}\add {i+1}=\bufread{\bm{\bitem {\add'} {\overline\val}\cons\mu} \mem}\add 0.
      \]
      Observe that having $\add' = \add$ would lead to a contradiction, as (H) and the definition of lookup would imply $\bot = \top$, which is impossible. Hence, we must have $\add' \neq \add$, allowing us to deduce from the definition of lookup that:
      \begin{align*}
        \bufread{\bm{\bitem {\add'} {\overline\val}\cons \mu} \mem}\add {i+1} &= \bufread{\bm\buf\mem}\add {i+1} & 
        \bufread{\bm{\bitem {\add'} {\overline\val}\cons \mu} \mem}\add {0} &= \bufread{\bm\buf\mem}\add {0} 
      \end{align*}
      The claim is a consequence of these equations, the IH and (H).
    \end{proofcases}
  \end{proofcases}
\end{proof}

\begin{remark}
  \label{rem:bufreadoverline}
  For every buffered memory $\bm\buf\mem$, and every address $\add$ we have that
  $\bufread{\bm\buf\mem}\add 0=\overline {\bm\buf\mem}(\add)$.
\end{remark}
\begin{proof}
  The proof proceeds by induction on $\buf$.
  \begin{proofcases}
    \proofcase{$\buf=\nil$} If $\buf$ is empty, the claim follows
    trivially from the definition of lookup.

    \proofcase{$\buf=\bitem{\add'}{\val} \cons \buf$} In this case, the claim states:
    \[
      \bufread{\bm{\bitem{\add'}{\val} \cons \buf} \mem} \add 0 = \overline{\bm{\bitem{\add'}{\val} \cons \buf} \mem} (\add).
    \]
    Rewriting this proposition, we obtain:
    \[
      \bufread{\bm{\bitem{\add'}{\val} \cons \buf} \mem} \add 0 = \update{\overline{\bm{\buf} \mem}}{\add'}{\val} (\add).
    \]
    Now, we consider two cases:
    \begin{proofcases}
      \proofcase{$\add = \add'$} In this case, the claim follows directly from the definitions of lookup and memory update.
      \proofcase{$\add \neq \add'$} In this case, the claim follows
      from the definition of lookup and the IH.
    \end{proofcases}
  \end{proofcases}
\end{proof}

\begin{remark}
  \label{rem:overlinewrtdom}
  For every buffered memory $\bm \buf{(\lay\lcomp \rfs)}$
  if $\dom(\buf) \subseteq \lay(\Ar)$,
  then we have that:
  \[
    \overline {\bm \buf{(\lay\lcomp \rfs)}} = \lay \lcomp \rfs'
  \]
  for some $\rfs' \eqon{\Fn} \rfs$
\end{remark}
\begin{proof}
  The proof is by induction on the length of the buffer.
\end{proof}

\subsubsection{Omitted Proofs and Results}
\label{sec:proofs2}

Throughout the following, we assume---without loss of
generality---that for any configuration reached during the evaluation
of an initial configuration with memory $\lay \lcomp \rfs$, the memory
$\mem$ in that configuration satisfies \( \mem = \lay \lcomp \rfs' \)
for some $\rfs'$ such that $\rfs' \eqon{\Fn} \rfs$.  This assumption
can be proven by induction on the length of the reduction, similarly to \Cref{rem:simispreserved}.

\begin{proof}[Proof of \Cref{lemma:cttosafe}]
  Assume that a system call $\syscall$  
  of a system $\system = (\syss, \rfs, \caps)$ is not \emph{speculative kernel safe}.  
  This means that there exist some natural number $\nat$, layout $\lay$,  
  register map $\regmap$, buffered memory $(\buf, \lay \lcomp \rfs')$  
  with $\rfs' \eqon{\Fn} \rfs$, a sequence of directives $\Ds$,  
  a sequence of observations $\Os$, and a mis-speculation flag $\boolms$  
  such that  
  \[
    \nsstep \nat {\sframe{\frame{\syss(\syscall)}{\regmap}{\km[\syscall]}}{\bm{\nil}{\lay \lcomp \rfs'}}\boolms} \unsafe \Ds \Os.
  \]
  By analyzing the applicable rules, we deduce that the rule applied  
  must be one of \ref{SI:Load-Unsafe}, \ref{SI:Store-Unsafe},  
  \ref{SI:Call-Unsafe}, or \ref{SIE:CallJmp-Unsafe}.  
  In all these cases, the rightmost observation in $\Os$  
  must be either $\omem \add$ or $\ojump \add$ for some address $\add \in \underline{\lay}(\Idk)$.  
  More precisely, if the rule used was \ref{SI:Call-Unsafe},  
  then $\add \in \underline{\lay}(\Fnk)$; otherwise, $\add \in \underline{\lay}(\Ark)$.  
  In the following, we analyze the latter case,  
  as the proof for the former is analogous.
  From the definition of $\underline{\lay}(\Ark)$,  
  there exist $\ar \in \Ark$ and $0 \leq i < \size{\ar}$  
  such that $\lay(\ar) + i = \add$.  
  Our goal is now to construct a layout $\lay'$  
  such that $p \notin \underline{\lay'}(\Idk)$,  
  meaning that $\add$ is not allocated in $\lay'$.  
  To define $\lay'$, we analyze different cases based on $\add' = \lay(\ar)$.  
  \begin{proofcases}
    \proofcase{$\kappa_\um$} In this case, the array is stored at the
    beginning of the kernel-space address range.  Given the assumption
    on the size of this address space, there are at least
    $2\cdot\max_{\id \in \Idk} \size\id \geq \size{\ar}$ free
    addresses in the set
    $\{\kappa_\um + \size{\ar}, \dots, \kappa_\um + \kappa_{\km} -
    1\}$.  Therefore, $\ar$ can be relocated within this space,
    ensuring that $\add$ is no longer allocated.  We define $\lay'$ as
    such a layout.  \proofcase{$\kappa_{\km}-1-\size{\ar}$} Analogous
    to the case above.
    \proofcase{$\kappa_\um < \add' <\kappa_{\km}-1-\size{\ar}$} By the
    pigeonhole principle, at least one of the two address ranges
    $\{\kappa_\um, \dots, \add' -1\}$ and
    $\{\add'+\size \ar, \dots, \kappa_{\km}+\kappa_{\um}-1\kappa_{\um}\}$ contains at
    least $\max_{\id \in \Idk} \size\id\ge \size \ar$ free addresses,
    meaning that $\ar$ can be moved to one of those ranges leaving the
    range $\lay(\ar), \dots, \lay(\ar)+i$ (and in particular $\add$)
    free.  We call $\lay'$ one such layout.
  \end{proofcases}
  From the \emph{speculative side-channel layout-non-interference}
  assumption, we deduce that there exists $\cfstack'$ such that
  \[
    \nsstep[\system][\lay'] \nat {\sframe{\frame{\syss(\syscall)}{\regmap}{\km[\syscall]}}{\bm{\buf}{\lay' \lcomp \ars'}}{\boolms}} {\cfstack'} \Ds \Os,
  \]
  for some $\cfstack'$. Observe that this reduction, in particular, produces  
  the sequence of observations $\Os$.  
  Applying \Cref{rem:noaddobs},  
  we conclude that $\omem \add$ does not appear in $\Os$.  
  However, this contradicts our earlier conclusion  
  that the last observation in $\Os$ was precisely $\omem \add$.  
  This contradiction completes the proof.  
\end{proof}

\begin{proof}[Proof of \Cref{thm:scenario2}]
  We fix a system $\system = (\rfs, \syss, \caps)$ and proceed by contraposition.  
  Assume that there exist an unprivileged command $\speccmd \in \SpCmd$,  
  an initial register map $\regmap$, a number of steps $\nat$, and a layout $\lay$  
  such that:  
  \[
    \nesstep \nat {\conf{\frame \speccmd \regmap \um, \lay \lcomp \rfs, \nil, \nil}} \unsafe.
  \]
  First, we observe that $\nat \neq 0$.  
  By analyzing the rules of the semantics, we deduce that  
  the last applied rule must be one of \ref{AL:Load-Unsafe},  
  \ref{AL:Store-Unsafe}, \ref{AL:Call-Unsafe}, or \ref{AL:Spec-Unsafe}.  
  The proof is analogous for the first three rules,  
  so we consider \ref{AL:Load-Unsafe} as a representative case.  

\begin{proofcases}
  \proofcase{\ref{AL:Load-Unsafe}} By examining this rule, we deduce
  that there exists a configuration
    \[
      \conf{\frame {\cmemread \vx \expr} {\regmap'}{\km[\syscall]}:\st, \lay \lcomp\rfs', \Ds, \Os} 
    \]
    such that
    \begin{equation*}
      \nesstep \nat {\conf{\frame \speccmd \regmap \um, \lay \lcomp \rfs,\nil, \nil}}
      {\conf{\frame {\cmemread \vx \expr} {\regmap'}{\km[\syscall]}:\st, \lay \lcomp\rfs', \Ds, \Os}}\\ \ato \unsafe.
    \end{equation*}
    with $\rfs' \eqon{\Fn} \rfs$. By \Cref{rem:stackinv1}, we observe that $\st$  
    is the concatenation of a kernel-mode stack $\st'$ and a user-mode stack $\st''$.  
    Applying \Cref{lemma:thereisasyscall1}, we conclude that  
    there exists a configuration  
    \[
      {\conf{\frame {\syss(\syscall)} {\regmap_0[\vx_1, \dots, \vx_k\upd \val_1, \dots, \val_k]}{\km[\syscall]}, \lay \lcomp\rfs'', \Ds', \Os'}},
    \]
     with $\rfs'' \eqon{\Fn} \rfs$ and $\nat'\in \Nat$ such that:
    \begin{equation*}
      \lay \red {\conf{\frame {\syss(\syscall)} {\regmap_0[\vx_1, \dots, \vx_k\upd \val_1, \dots, \val_k]}{\km[\syscall]}, \lay \lcomp\rfs'', \Ds', \Os'}} \\ \ato^{\nat'}
      {\conf{\frame {\cmemread \vx \expr} {\regmap'}{\km[\syscall]}:\st', \lay \lcomp\rfs', \Ds, \Os}}
    \end{equation*}

    By analyzing the rule \ref{AL:Load-Unsafe}, we further deduce that:  
    \begin{varitemize}
    \item $\sem \expr_{\regmap', \lay} \in \underline \lay(\Ark)$.
    \item $\sem \expr_{\regmap', \lay} \notin \underline \lay(\caps(\syscall))$.
    \end{varitemize}
    This implies that the same rule applies to
    \[
      {\conf{\frame {\cmemread \vx \expr} {\regmap'}{\km[\syscall]}:\st', \lay \lcomp\rfs', \Ds, \Os}}.
    \]
    Thus, we conclude:
    \[
      \lay \red {\conf{\frame {\syss(\syscall)} {\regmap_0[\vx_1, \dots, \vx_k\upd \val_1, \dots, \val_k]}{\km[\syscall]}, \lay \lcomp \rfs'', \Ds', \Os'}} \ato^{\nat' + 1}\unsafe.
    \]
    Applying \Cref{lemma:ordinatytospec},  
    we deduce that the corresponding speculative configuration
    \[
      {\sframe{\frame {\syss(\syscall)} {\regmap_0[\vx_1, \dots, \vx_k\upd \val_1, \dots, \val_k]}{\km[\syscall]}}{\lay \lcomp \rfs'}{\bot}}
    \]
    also reduces in $\nat' + 1$ steps to $\unsafe$,  
    using the sequence of directives ${\dstep^{\nat' + 1}}$.  
    This contradicts \Cref{lemma:cttosafe} when applied to the system call $\syscall$.
    
    \proofcase{\ref{AL:Spec-Unsafe}} By introspection of the rule,
    we deduce that there exists a hybrid configuration
    \[
      {\hconf\unsafe {\frame{\adversary}{\regmap}{\opt}\cons\st}\Ds\Os}
    \]
    such that 
    \[
      \nesstep \nat  {\conf{\frame \speccmd \regmap \um, \lay \lcomp \rfs, \nil, \nil}}{} {\hconf\unsafe {\frame{\adversary}{\regmap}{\opt}\cons\st}\Ds\Os} \ato \unsafe.
    \]
    With an application of \Cref{rem:thereisaspecon}, we deduce that
    there exists a configuration
    \[
      {\sframe{\frame {\cmd} {\regmap'}{\um}}{\bm{\buf}{\lay \lcomp \rfs'}}{\bot}},
    \]
    where $\rfs' \eqon{\Fn}\rfs$, along with a sequence of directives
    $\Ds'$, a sequence of observations $\Os'$, and a natural number
    $\nat' \leq \nat$, such that
    \begin{equation*}
      \nsstep {\nat'}
      {\sframe{\frame {\cmd} {\regmap'}{\um}}{\bm{\buf}{\lay \lcomp \rfs'}}{\bot}}
      \unsafe {\Ds'}{\Os'}.
    \end{equation*} 
    By \Cref{lemma:nobt}, we can assume without loss of generality
    that $\Ds'$ does not contain any $\dbt$ directive.  By applying
    \Cref{lemma:thereisasyscall2} and examining the rules of the
    semantics, we deduce the existence of configurations
    \begin{equation*}
      {\sframe{\frame {\syss(\syscall)} {\regmap''}{\km[\syscall]}}{\bm{\buf'}{\lay \lcomp \rfs''}}{\bool}}
      \quad
      \text{and}
      \quad
      {\sframe{\frame {\cmdtwo} {\regmap'''}{{\km}_\syscall}:\st''}{\bm{\buf''}{\lay \lcomp \rfs'''}}{\boolms}}
    \end{equation*}
    along with a sequence of directives $\Ds''$,  
    a sequence of observations $\Os''$, and a store $\rfs''$,  
    such that  
    \begin{equation*}
      \nsstep {\nat''}
      {\sframe{\frame {\syss(\syscall)} {\regmap''}{\km[\syscall]}}{\bm{\buf'}{\lay \lcomp \rfs''}}{\bool}}
      {\sframe{\frame {\cmdtwo} {\regmap'''}{\km[\syscall]}:\st''}{\bm{\buf''}{\lay \lcomp \rfs'''}}{\boolms}}
      {\Ds''}{\Os''}\sto \dir \obs \unsafe.
    \end{equation*}
    This conclusion contradicts \Cref{lemma:cttosafe}
    applied to $\syscall$.
  \end{proofcases}
\end{proof}

\begin{lemma}
  \label{lemma:thereisasyscall1}
  For every system $\system = (\rfs, \syss, \caps)$,
  natural number $\nat$, configurations
  \[
    \confone = \conf{\frame {\speccmd}{\overline \regmap}\um, \lay \lcomp\rfs', \overline \Ds, \overline \Os}\]
  with $\rfs \eqon\Fn \rfs'$ and
  \[
    \conf{\frame{\cmd}{\regmap}{\km[\syscall]}:\st_{\km}:\st_\um, \lay \lcomp\rfs'', \Ds, \Os},
  \]
  where $\km(\st_{\km})$
  and $\um(\st_\um)$ hold, and
  \[
    \nesstep \nat \confone 
    {\conf{\frame{\cmd}{\regmap}{\km[\syscall]}:\st_{\km}:\st_\um,
          \lay \lcomp\rfs'', \Ds, \Os}},
  \]
  there exists a configuration
  \[
    \conf{\frame{\syss(\syscall)}{\regmap'}{\km[\syscall]}, \lay \lcomp\rfs''', \Ds, \Os},
  \]
  and a natural number $\nat'$ such that
  \[
    \nesstep {\nat'}  
    {\conf{\frame{\syss(\syscall)}{\regmap'}{\km[\syscall]}, \lay \lcomp\rfs''', \Ds, \Os}}
    {\conf{\frame{\cmd}{\regmap}{\km[\syscall]}:\st_{\km}, \lay \lcomp\rfs'', \Ds, \Os}}.
  \]
\end{lemma}

\begin{proof}
  The proof proceeds by induction on $\nat$. The base case holds trivially due to vacuity of the premise.  
  For the inductive step, we analyze the last transition using case analysis on the rule applied in the last transition.
  Most rules that produce have a uniform behavior and follow a
  similar argument, except for
  \ref{AL:Pop}, \ref{AL:Spec-Term}, and \ref{AL:System-Call}.  We illustrate the proofs for the representative case of rule
  \ref{AL:Load}, and the rules we mentioned above.
  \begin{proofcases}
    \proofcase{\ref{AL:Load}} In this case, our main assumption rewrites as follows:  
    \begin{equation*}
      \nesstep \nat \confone 
      {\conf{\frame{\cmemread \vx \expr\sep \cmd}{\regmap''}{\opt}:\st, \lay \lcomp\rfs'', \Ds, \Os}} \ato
      {\conf{\frame{\cmd}{\update{\regmap''}\vx{\sem\expr_{\regmap'', \lay}}}{\km[\syscall]}:\st_{\km}:\st_\um, \lay \lcomp\rfs'', \Ds, \Os}}.
    \end{equation*}
    By analyzing the rule \ref{AL:Load}, we deduce that the mode flag $\opt$ is $\km[\syscall]$. Thus, we can apply the IH, and we deduce the existence of a configuration
    \[
      \conf{\frame{\syss(\syscall)}{\regmap'}{\km[\syscall]}, \lay \lcomp\rfs''', \Ds, \Os},
    \]
    and a natural number $\nat'$ such that:
    \begin{equation*}
      \nesstep {\nat'}  
      {\conf{\frame{\syss(\syscall)}{\regmap'}{\km[\syscall]}, \lay \lcomp\rfs''', \Ds, \Os}} {}\\
      {\conf{\frame{\cmemread \vx \expr\sep \cmd}{\regmap''}{\km[\syscall]}:\st_{\km}, \lay \lcomp\rfs'', \Ds, \Os}}.
    \end{equation*}
    Finally, applying rule \ref{AL:Load}, we conclude:
    \[\lay \red
      {\conf{\frame{\cmemread \vx \expr\sep \cmd}{\regmap''}{\km[\syscall]}:\st_{\km}, \lay \lcomp\rfs'', \Ds, \Os}}\ato
      {\conf{\frame{\cmd}{\update{\regmap''}\vx{\sem\expr_{\regmap'', \lay}}}{\km[\syscall]}:\st_{\km}, \lay \lcomp\rfs'', \Ds, \Os}}
    \]
    and this shows the claim.
    \proofcase{\ref{AL:System-Call}} In this case, the assumption is that  
    \begin{multline*}
      \nesstep \nat \confone 
      {\conf{\frame{\csyscall \syscalltwo {\expr_1, \dots, \expr_k}\sep \cmd}{\regmap''}{\opt}:\st, \lay \lcomp\rfs', \Ds, \Os}} \ato\\
      {\conf{\frame{\syss(\syscalltwo)}{\regmap}{\km[\syscalltwo]}: \frame{\cmd}{\regmap''}{\opt}:\st_{\km}:\st_\um, \lay \lcomp\rfs'', \Ds, \Os}}
    \end{multline*}
    From \Cref{rem:stackinv1}, we deduce that $\st = \st_{\km}\cons \st_\um$.
    We proceed by case analysis on $b$.
    
    \begin{proofcases}
      \proofcase{$b = \um$}
      In this case, we have  $\st_{\km}= \nil$ and
      $\st_{\um}= \st$, so we can verify that:
      \begin{equation*}
        \esstep
        {\conf{\frame{\csyscall \syscalltwo {\expr_1, \dots, \expr_k}\sep \cmd}{\regmap''}{\opt}, \lay \lcomp\rfs', \Ds, \Os}} {}\\
        {\conf{\frame{\syss(\syscalltwo)}{\regmap}{\km[\syscalltwo]}: \frame{\cmd}{\regmap''}{\opt}, \lay \lcomp\rfs'', \Ds, \Os}}.
      \end{equation*}
      Thus, the claim is valid for $\nat' = 0$ and $\syscalltwo=\syscall$.
      \proofcase{$b = \km[\syscall]$} In this case, we have a nested system call, and we can apply the IH to the reduction:
      \begin{equation*}
        \nesstep \nat \confone 
        {\conf{\frame{\csyscall \syscalltwo {\expr_1, \dots, \expr_k}\sep \cmd}{\regmap''}{\km[\syscall]}:\st_{\km}\cons \st_\um, \lay \lcomp\rfs'', \Ds, \Os}}.
    \end{equation*}
    By the IH there exist a configuration
    \[
      \conf{\frame{\syss(\syscall)}{\regmap'}{\km[\syscall]}, \lay \lcomp\rfs''', \Ds, \Os},
    \]
    and a natural number $\nat'$ such that:
    \begin{equation*}
      \nesstep {\nat'}  
      {\conf{\frame{\syss(\syscall)}{\regmap'}{\km[\syscall]}, \lay \lcomp\rfs''', \Ds, \Os}} {}\\
      {\conf{\frame{\csyscall \syscalltwo {\expr_1, \dots, \expr_k}\sep \cmd}{\regmap''}{\km[\syscall]}:\st_{\km}, \lay \lcomp\rfs'', \Ds, \Os}}.
    \end{equation*}
    Finally, applying rule \ref{AL:System-Call}, we conclude:
    \begin{equation*}
      \nesstep {\nat'+1}  
      {\conf{\frame{\syss(\syscall)}{\regmap'}{\km[\syscall]}, \lay \lcomp\rfs''', \Ds, \Os}} {}\\
      {\conf{\frame {\syss(\syscalltwo)} {\regmap}{ \km[\syscall]}\cons \frame{ \cmd}{\regmap''}{\km[\syscall]}:\st_{\km}, \lay \lcomp\rfs'', \Ds, \Os}},
    \end{equation*}
    which establishes the claim.
    \end{proofcases}
    \proofcase{\ref{AL:Pop}} In this case, the assumption is that  
    \begin{equation*}
      \nstep \nat \confone 
      {\conf{\frame \cnil {\regmap''}{\opt}:\frame{\cmd}{\regmap}{\km[\syscall]}:\st, \lay \lcomp\rfs', \Ds, \Os}} \ato\\
      {\conf{\frame{\cmd}{\update{\regmap'}{\ret}{\regmap''(\ret)}}{\km[\syscall]}:\st, \lay \lcomp\rfs', \Ds, \Os}},
    \end{equation*}
    Applying \Cref{rem:stackinv1}, we rewrite $\st$ as $\st_{\km}:\st_\um$ and deduce that $\opt = \km[\syscall]$,  allowing us to apply the IH.  By the IH, there exists a configuration:
    \[
      \conf{\frame{\syss(\syscall)}{\regmap'}{\km[\syscall]}, \lay \lcomp\rfs''', \Ds, \Os},
    \]
    a natural number $\nat'$ such that:
    \begin{equation*}
      \nstep {\nat'}  
      {\conf{\frame{\syss(\syscall)}{\regmap'}{\km[\syscall]}, \lay \lcomp\rfs''', \Ds, \Os}} {}\\
      {\conf{\frame \cnil {\regmap''}{\km[\syscall]}:\frame{\cmd}{\regmap}{\km[\syscall]}:\st_{\km}, \lay \lcomp\rfs'}}.
    \end{equation*}
    To conclude the proof it suffices to verify that:
    \begin{multline*}
      \esstep  
      {\conf{\frame \cmdtwo {\regmap''}{\opt}:\frame{\cmd}{\regmap}{\km[\syscall]}:\st_{\km}, \lay \lcomp\rfs', \Ds, \Os}} {} \\
      {\conf{\frame{\cmd}{\update{\regmap'}{\ret}{\regmap''(\ret)}}{\km[\syscall]}:\frame{\cmd}{\regmap}{\km[\syscall]}:\st_{\km}, \lay \lcomp\rfs', \Ds, \Os}}.
    \end{multline*}
    \proofcase{\ref{AL:Spec-Term}} By \Cref{rem:stackinv1} and by introspection of the rules, we conclude that this case is absurd.
  \end{proofcases}
\end{proof}

\begin{lemma}
  \label{lemma:thereisasyscall0}
  For every system $\system = (\rfs, \syss, \caps)$, and
  natural number $\nat$, configurations
  \(
    \confone = \conf{\frame {\cmd}{\overline \regmap}\um, \lay \lcomp\rfs'}
  \)
  with $\rfs \eqon\Fn \rfs'$ and
  \(
    \conf{\frame{\cmdtwo}{\regmap}{\km[\syscall]}:\st, \lay \lcomp\rfs''},
  \)
  if 
  \[
    \nstep \nat {\conf{\frame {\cmd}{\overline \regmap}\um, \lay \lcomp\rfs'}} 
    \conf{\frame{\cmdtwo}{\regmap}{\km[\syscall]}:\st, \lay \lcomp\rfs''},
  \]
  then there exist a configuration
  \[
    {\conf{\frame{\csyscall \syscall {\expr_1, \ldots, \expr_k}\sep \cmd}{\regmap'}{\um}:\st', \lay \lcomp\rfs'''}},
  \]
  and a natural number $\nat'<\nat$ such that
  \[
    \nstep {\nat'}
    \confone
    {\conf{\frame{\csyscall \syscall {\expr_1, \ldots, \expr_k}\sep \cmd}{\regmap'}{\um}:\st', \lay \lcomp\rfs'''}}
    \to^{\nat-\nat'}
    \conf{\frame{\cmdtwo}{\regmap}{\km[\syscall]}:\st, \lay \lcomp\rfs''},
  \]
  and
  \[
    \nstep
    {\nat-\nat'-1}
    {\conf{\frame{\syss(\syscall) }{\update{\regmap_0} {\vx_1, \ldots, \vx_k}{\sem{\expr_1}_{\regmap', \lay}, \ldots, \sem{\expr_k}_{\regmap', \lay}}}{\km[\syscall]}, \lay \lcomp\rfs'''}}
    \conf{\frame{\cmdtwo}{\regmap}{\km[\syscall]}\cons\st'', \lay \lcomp\rfs''},
  \]
  where $\km[\syscall](\st'')$ and $\st''\cons \frame{\cmd}{\regmap'}{\um}\cons\st'=\st$.  
\end{lemma}

\begin{proof}
  The proof proceeds by induction on $\nat$. The base case holds
  trivially since the premise is vacuous.
  For the inductive step, we analyze the last transition using case
  analysis on the applied rule.
  Most rules follow a uniform argument, except for \ref{WL:Pop} and
  \ref{WL:SystemCall}. We illustrate the proof for the representative
  case of rule \ref{WL:Load}, along with the exceptional cases.
  \begin{proofcases}
    \proofcase{\ref{WL:Load}} In this case, we assume:  
    \begin{equation*}
      \nstep \nat \confone 
      {\conf{\frame{\cmemread \vx \expr\sep \cmdtwo}{\regmap''}{\opt}:\st, \lay \lcomp\rfs''}} \to
      {\conf{\frame{\cmdtwo}{\update{\regmap''}\vx{\sem\expr_{\regmap'', \lay}}}{\km[\syscall]}:\st, \lay \lcomp\rfs''}}.
    \end{equation*}
    By analyzing rule \ref{WL:Load}, we conclude that the mode flag
    $\opt$ must be $\km[\syscall]$.  Thus, we apply the IH and deduce
    the existence of a configuration
    \[
      {\conf{\frame{\csyscall \syscall {\vec \expr}\sep \cmd}{\regmap'}{\um}:\st', \lay \lcomp\rfs'''}},
    \]
    and a natural number $\nat'$ such that:
    \begin{equation*}
      \nstep {\nat'}  
      {\conf{\frame{\csyscall \syscall {\vec \expr}\sep \cmd}{\regmap'}{\um}:\st', \lay \lcomp\rfs'''}}
      {\conf{\frame{\cmemread \vx \expr \sep \cmdtwo}{\regmap''}{\km[\syscall]}:\st, \lay \lcomp\rfs''}},
    \end{equation*}
    and
    \[
      \nstep
      {\nat-\nat'-1}
      {\conf{\frame{\syss(\syscall) }{\update{\regmap_0} {\vec\vx}{\sem{\vec \expr}_{\regmap, \lay}}}{\um}, \lay \lcomp\rfs'''}}
      {\conf{\frame{\cmemread \vx \expr \sep \cmdtwo}{\regmap''}{\km[\syscall]}:\st'', \lay \lcomp\rfs''}},
    \]
    where $\km[\syscall](\st'')$ and  $\st''\cons\frame{\cmd}{\regmap'}{\um}\cons\st'=\st$.
    Finally, applying rule \ref{WL:Load}, we conclude:
    \begin{equation*}
      \step
      {\conf{\frame{\cmemread \vx \expr \sep \cmdtwo}{\regmap''}{\km[\syscall]}:\st, \lay \lcomp\rfs''}}
      {\conf{\frame{\cmdtwo}{\update{\regmap''}\vx{\sem\expr_{\regmap'', \lay}}}{\km[\syscall]}:\st, \lay \lcomp\rfs''}},
    \end{equation*}
    and
    \begin{equation*}
      \step
      {\conf{\frame{\cmemread \vx \expr \sep \cmdtwo}{\regmap''}{\km[\syscall]}:\st'', \lay \lcomp\rfs''}}
      {\conf{\frame{\cmdtwo}{\update{\regmap''}\vx{\sem\expr_{\regmap'', \lay}}}{\km[\syscall]}:\st'', \lay \lcomp\rfs''}},
    \end{equation*}
    which prove the claim.
    \proofcase{\ref{WL:SystemCall}} In this case, we can rewrite the main assumption as follows:  
    \begin{multline*}
      \nstep \nat \confone 
      {\conf{\frame{\csyscall \syscalltwo {\expr_1, \dots, \expr_k}\sep \cmd}{\regmap''}{\opt}:\st, \lay \lcomp\rfs'}} \to
      {\conf{\frame{\syss(\syscalltwo)}{\regmap}{\km[\syscalltwo]}: \frame{\cmd}{\regmap''}{\opt}:\st, \lay \lcomp\rfs''}}
    \end{multline*}
    We proceed by case analysis on $b$.    
    \begin{proofcases}
      \proofcase{$b = \um$}
      In this case, the first claim follows trivially. For the second one, we can verify that:
      \begin{equation*}
        \nstep 0
        {\conf{\frame{\syss(\syscalltwo)}{\regmap}{\km[\syscalltwo]}, \lay \lcomp\rfs''}}
        {\conf{\frame{\syss(\syscalltwo)}{\regmap}{\km[\syscalltwo]}, \lay \lcomp\rfs''}}.
      \end{equation*}

      Thus, the second claim holds for $\nat' = \nat$, $\syscalltwo=\syscall$, $\st' = \st$, and $\st'' = \nil$.
      \proofcase{$b = \km[\syscall]$} Here, we have a nested system call. We can apply the IH to the reduction:
      \begin{equation*}
        \nstep \nat \confone 
        {\conf{\frame{\csyscall \syscalltwo {\expr_1, \dots, \expr_k}\sep \cmd}{\regmap''}{\km[\syscall]}:\st, \lay \lcomp\rfs''}}.
      \end{equation*}
  The remaining steps follow similarly to the corresponding case in rule \ref{WL:Load}.
  \end{proofcases}
    \proofcase{\ref{WL:Pop}} In this case, we can rewrite the main assumption as follows:  
    \begin{equation*}
      \nstep \nat \confone 
      {\conf{\frame \cnil {\regmap''}{\opt}:\frame{\cmd}{\regmap}{\km[\syscall]}:\st, \lay \lcomp\rfs'}} \to\\
      {\conf{\frame{\cmd}{\update{\regmap'}{\ret}{\regmap''(\ret)}}{\km[\syscall]}:\st, \lay \lcomp\rfs'}}.
    \end{equation*}
    By analyzing rule \ref{WL:Pop}, we conclude that $\opt = \km[\syscall]$. This allows us to apply the induction hypothesis (IH). The remainder of the proof follows the same reasoning as in the case of rule \ref{WL:Load}.
  \end{proofcases}
\end{proof}

\begin{lemma}
  \label{lemma:ordinatytospec}
  For every system $\system = (\overline \rfs, \syss, \caps)$,
  store $\rfs \eqon{\Fn}\overline \rfs$,
  configuration $\conf{\frame{\cmd}{\regmap}{\km[\syscall]}, {\lay \lcomp \rfs}, \Ds, \Os}$,
  $\nat \in \Nat$,
  and every configuration $\conf{\st, {\lay \lcomp \rfs'}, \Ds, \Os}$
  such that
  \[
    \nesstep \nat {\conf{\frame{\cmd}{\regmap}{\km[\syscall]}, {\lay \lcomp \rfs}, \Ds, \Os}}
    {\conf{\st, {\lay \lcomp \rfs'}, \Ds, \Os}},
  \]
  there is a buffered memory $\bm \buf {(\lay \lcomp \rfs'')}$ and a sequence of observations $\Os$ such that 
  \[
    \nsstep \nat {\sframe {\frame{\cmd}{\regmap}{\km[\syscall]}} {\bm \nil {\lay \lcomp \rfs}} \bot} 
    {\sframe \st {\bm \buf {(\lay \lcomp \rfs'')}} \bot} {\dstep^\nat} \Os
  \]
  and $\overline {\bm\buf{(\lay \lcomp \rfs'')}} = {\lay \lcomp \rfs'}$.
\end{lemma}
\begin{proof}
  Notice that for \Cref{rem:stackinv1},
  the last transition cannot be derived using \ref{AL:Spec-Term}.
  By introspection of the rules, we also deduce that it cannot be derived using
  \ref{AL:Spec-Unsafe} or \ref{AL:Spec-Error}. With this additional observation, 
  we proceed by induction on $\nat$. In the inductive step, we assume that
  the last transition was not derived using any of the aforementioned rules.

  \begin{proofcases}
    \proofcase{0} Trivial.  \proofcase{$\nat+1$} In this case, the
    last rule used cannot be one of \ref{AL:Spec-Term},
    \ref{AL:Spec-Unsafe}, or \ref{AL:Spec-Error}.  Therefore, the
    $\nat$-th configuration cannot be a hybrid one.  Thus, we can
    restate the premise as follows:
    \begin{equation*}
      \nesstep \nat {\conf{\frame{\cmd}{\regmap}{\km[\syscall]}, {\lay \lcomp \rfs}, \Ds, \Os}}
      {\conf{\frame {\cmd'} {\regmap'} {\opt'}: \st', {\lay \lcomp \rfs''}, \Ds, \Os}}\ato\\
      {\conf{\st,{\lay \lcomp \rfs'}, \Ds, \Os}}.
    \end{equation*}
    We apply the IH to the first $\nat$ steps, showing that there
    exist $\buf$ and $\Os$ such that:
    \[
      \nsstep \nat {\sframe {\frame{\cmd}{\regmap}{\km[\syscall]}} {\lay \lcomp \rfs} \bot} 
      {\sframe {\frame {\cmd'} {\regmap'} {\km[\syscall]}:\st'} {\bm {\buf} {(\lay \lcomp \omega)}} \bot} {\dstep^\nat} \Os
    \]
    where
    $\overline {\bm\buf{(\lay \lcomp \omega)}} = {\lay \lcomp \rfs''}$
    and
    $\km[\syscall]({\frame {\cmd'} {\regmap'} {\km[\syscall]}:\st'})$.
    We need to show that
    \[
      \sstep {\sframe {\frame {\cmd'} {\regmap'} {\km[\syscall]}:\st'} {\bm {\buf} {(\lay \lcomp \omega)}} \bot}
      {\sframe{\st}{\bm{\buf'}{(\lay \lcomp \omega')}}{\bot}}{\dstep} \obs,
    \]
    and that
    $\overline {\bm{\buf'}{(\lay \lcomp \omega')}} = {\lay \lcomp
      \rfs'}$.  The proof proceeds by cases on the $\ato$ relation.
    Many cases are similar, so we present only the most significant
    ones.  Here, we also note that the last step cannot
    be derived using \ref{AL:Poison} or \ref{AL:Obs}, because this would contradict 
    \Cref{rem:stackinv1}, which shows that $\km(\frame {\cmd'} {\regmap'} {\opt'}: \st')$.
    In particular, this means that $\cmd'$ cannot be an attacker program.

    \begin{proofcases}
      \proofcase{\ref{AL:Fence}} In this case, the assumption rewrites as follows:
      \begin{equation*}
        \esstep  {\conf {\frame{\cfence\sep \cmdtwo}{\regmap}{\km[\syscall]}:\st',  {\lay \lcomp \rfs''}, \Ds, \Os}}{} \\
        {\conf {\frame{\cmdtwo}{
              \regmap
            }{\km[\syscall]}:\st',  {\lay \lcomp \rfs''}, \Ds, \Os}}.
      \end{equation*}
      The goal is to show that there exist an observation $\obs$, a buffer $\buf'$ and a store $\omega'$ such that:
      \begin{equation*}
        \sstep  {\sframe {\frame{\cfence\sep \cmdtwo}{\regmap}{\km[\syscall]}:\st'}{\bm{\buf}{(\lay \lcomp \omega)}}{\bot}}{} \dstep \obs \\
        {\sframe {\frame{\cmdtwo}{
              \regmap
            }{\km[\syscall]}:\st'}{ \bm{\buf'}{(\lay \lcomp \omega')}}{\bot}},
      \end{equation*}
      and
      $\overline {\bm{\buf'}{(\lay \lcomp \omega')}} = {\lay \lcomp
        \rfs''}$. Suitable buffers and stores for the target
      configuration are $\buf'=\nil$ and
      $\omega'=\overline {\bm{\buf} {(\lay \lcomp \omega)}}$, and the
      transition produces the observation $\onone$. Since
      ${\lay \lcomp \rfs''}=\overline {\bm{\nil}{\overline {\bm{\buf}
            {(\lay \lcomp \omega)}}}} $, which follows from the IH and
      the definition of $\overline \cdot$ on memories, the conclusion
      on stores holds.  Finally, we must prove that
      $\km[\syscall]({\frame{\cmdtwo}{\regmap}{\km[\syscall]}:\st'})$.
      This follows directly from the IH and the fact that $\cmdtwo$ is
      a sub-term of $\cfence\sep \cmdtwo$.

      \proofcase{\ref{AL:Load}} In this case, the assumption rewrites as follows:
      \begin{equation*}
        \esstep  {\conf {\frame{\cmemread \vx \expr\sep \cmdtwo}{\regmap}{\km[\syscall]}:\st',  {\lay \lcomp \rfs''}, \Ds, \Os}}{} \\
        {\conf {\frame{\cmdtwo}{
              \update\regmap\vx{{\lay \lcomp \rfs''}(\toAdd{\sem \expr_{\regmap, \lay}})}
            }{\km[\syscall]}:\st',  {\lay \lcomp \rfs''}, \Ds, \Os}}.
      \end{equation*}
      Our goal is to show that there exist an observation $\obs$ and a buffer $\buf'$ such that:
      \begin{multline*}
        \sstep  {\sframe {\frame{\cmemread \vx \expr\sep \cmdtwo}{\regmap}{\km[\syscall]}:\st'}{\bm{\buf}{(\lay \lcomp \omega)}}{\bot}}{} \dstep \obs \\
        {\sframe {\frame{\cmdtwo}{
              \update\regmap\vx{{\lay \lcomp \rfs''}(\toAdd{\sem \expr_{\regmap, \lay}})}
            }{\km[\syscall]}:\st'}{ \bm{\buf'}{(\lay \lcomp \omega')}}{\bot}},
      \end{multline*}
      where
      $\overline {\bm{\buf'}{(\lay \lcomp \omega)}} = {\lay \lcomp
        \rfs''}$. We choose $\buf'=\buf$ and $\omega=\omega'$ and the
      equality of the stores is a consequence of the IH. Because of
      our assumption on the target configuration, the only applicable rule
      is \ref{SI:Load-Step}, which produces the observation
      $\omem {\toAdd{\sem \expr_{\regmap, \lay}}}$.
      To prove the applicability of the speculative rule, we need to show that
      ${\lay \lcomp \rfs''}(\toAdd{\sem \expr_{\regmap, \lay}})=
      \bufread {\bm{\buf'}{(\lay \lcomp \omega)}} {\toAdd{\sem
          \expr_{\regmap, \lay}}} 0$, which follows from
      \Cref{rem:bufreadoverline}. Finally, we must observe that
      \[
        \km[\syscall]({\frame{\cmdtwo}{
            \update\regmap\vx{{\lay \lcomp \rfs''}(\toAdd{\sem \expr_{\regmap, \lay}})}
          }{\km[\syscall]}:\st'}),
      \]
      which is a direct consequence of the IH
      and of the fact that $\cmdtwo$ is a sub-term of $\cmemread \vx \expr\sep \cmdtwo$.

      \proofcase{\ref{AL:Store}} In this case, the assumption rewrites as follows:
      \begin{equation*}
        \esstep {\conf {\frame{\cmemass\expr\exprtwo\sep \cmdtwo}{\regmap}{\km[\syscall]}:\st',  {\lay \lcomp \rfs''}, \Ds, \Os}}{} \\
        {\conf {\frame{\cmdtwo}{\regmap}
            {\km[\syscall]}:\st',  \update{{\lay \lcomp \rfs''}} {\toAdd{\sem \expr_{\regmap, \lay}}} {\sem \exprtwo_{\regmap, \lay}}, \Ds, \Os}}.
      \end{equation*}
      The goal is to show that there exist an observation $\obs$ and a buffer $\buf'$ such that:
      \begin{equation*}
        \sstep  {\sframe {\frame{\cmemass \vx \expr\sep \cmdtwo}{\regmap}{\km[\syscall]}:\st'}{\bm{\buf}{(\lay \lcomp \omega)}}{\bot}}{} \dstep \obs \\
        {\sframe {\frame{\cmdtwo}{
              \regmap
            }{\km[\syscall]}:\st'}{ \bm{\buf'}{(\lay \lcomp \omega)}}{\bot}}
      \end{equation*}
      and we have $\overline {\bm{\buf'}{(\lay \lcomp \omega)}} = \update{{\lay \lcomp \rfs''}} {\toAdd{\sem \expr_{\regmap, \lay}}} {\sem \exprtwo_{\regmap, \lay}}$.
      By applying the rule \ref{SI:Store} to the speculative configuration, we observe that:
      \(
        \buf' = \bitem{\toAdd{\sem \expr_{\regmap, \lay}}} {\sem \exprtwo_{\regmap, \lay}} \cons \buf
      \).
      Therefore, the conclusion
      \[
        \overline {\bm{\bitem{\toAdd{\sem \expr_{\regmap, \lay}}} {\sem \exprtwo_{\regmap, \lay}} \cons \buf}{(\lay \lcomp \omega)}}= \update{{\lay \lcomp \rfs''}} {\toAdd{\sem \expr_{\regmap, \lay}}} {\sem \exprtwo_{\regmap, \lay}}
      \]
      is a direct consequence of the IH and the definition of $\overline \cdot$.
      
      \proofcase{\ref{AL:Call}} In this case, we observe that both
      this rule and \ref{SI:Call} share the same premises, meaning
      that \ref{SI:Call} can also be applied. By analyzing these
      rules, we deduce that if
      $\cmd' = \ccall \exprtwo {\expr_1, \dots, \expr_k} \sep
      \cmdtwo$, then the resulting target configurations are respectively:
      \[
        {\conf{\frame {{\lay \lcomp \rfs''}(\sem \exprtwo_{\regmap', \lay})} {\regmap_0'} {\km[\syscall]}: \frame {\cmdtwo} {\regmap'} {\km[\syscall]}: \st', {\lay \lcomp \rfs''}, \Ds, \Os}}
      \]
      and
      \[
        {\sframe{\frame {{\lay \lcomp \rfs''}(\sem \exprtwo_{\regmap', \lay})} {\regmap_0'} {\km[\syscall]}: \frame {\cmdtwo} {\regmap'} {\km[\syscall]}: \st'}{\bm\buf{(\lay \lcomp \omega)}}\bot},
      \]
      where $\regmap_0'=\update{\regmap_0}{\vx_1,\dots, \vx_k}{\sem {\expr_1}{\regmap', \lay}, \dots,  {\expr_k}{\regmap', \lay}}$.
      The conclusion regarding the buffered memory follows directly from the IH. Moreover, we can deduce that
      \[
        \km[\syscall](\frame {{\lay \lcomp \rfs''}(\sem \exprtwo_{\regmap', \lay})} {\regmap_0'} {\km[\syscall]}: \frame {\cmdtwo} {\regmap'} {\km[\syscall]}: \st')
      \]
      holds by the IH and the premises of the rule \ref{AL:Call}, which ensure that  $\sem \exprtwo_{\regmap', \lay}\in \underline \lay(\Fnk)$. This, in turn, implies the existence of some procedure $\fn \in \Fnk$ such that $\lay(\fn)=\sem \exprtwo_{\regmap', \lay}$. By the definition of $\cdot \lcomp \cdot$, we conclude that
      \[
        {\lay \lcomp \rfs''}(\sem \exprtwo_{\regmap', \lay}) =\omega(\fn) = \overline \rfs(\fn).
      \]
      Thus, it remains only to observe that $\km[\syscall](\overline \rfs(\fn))$ holds by the definition of the system, completing the proof.
    \end{proofcases}
  \end{proofcases}
\end{proof}

\begin{lemma}
  \label{lemma:thereisasyscall2}
  Consider a system $\system = (\rfs, \syss, \caps)$ and a configuration  
  \(
  \sframe {\frame \cmd \regmap \um} {{\psi}} {\bot}
  \)
  and $\psi =  \lay \lcomp {\rfs'}$
  with $\rfs' \eqon\Fn \rfs$. Let there be a speculative stack  
  \(
  \cfstack = \sframe {\frame \cmd {\regmap_1} {\km[\syscall]}:\st_1} {\psi_1} {\boolms}:\cfstack'
  \),
  a sequence of directives $\Ds$ that does not contain any $\dbt$ directives, a sequence of observations $\Os$, and a layout $\lay$ such that  
  \[
    \nsstep \nat {\sframe {\frame \cmd \regmap \um} {{\psi}} {\bot}} \cfstack \Ds \Os.
  \]
  Then, there exist a configuration  
  \(
  \sframe {\frame {\syss(\syscall)} {\regmap_2} {\km[\syscall]}:\st_2} {\psi_2} {\boolms'}
  \),
  a sequence of directives $\Ds'$, a sequence of observations $\Os'$, and a natural number $\nat' \leq \nat$ such that  
  \begin{equation*}
    \nsstep {\nat'} {\sframe {\frame {\syss(\syscall)} {\regmap_2} {\km[\syscall]}} {\psi_2} {\boolms'}}{}{\Ds'}{\Os'}\\   
    {\sframe {\frame \cmd {\regmap_1} {\km[\syscall]}:\overline \st} {\psi_1} {\boolms}:\overline \cfstack},
  \end{equation*}
  for some stacks $\overline \cfstack$ and $\overline \st$. Such that, in particular:  
  \begin{itemize}
  \item $\km[\syscall](\frame \cmd {\regmap_1} {\km[\syscall]}:\overline \st)$ holds.  
  \item There exists a stack $\st_2$ such that $\um(\st_2)$ and $\st_1 = \overline \st : \st_2$.  
  \end{itemize}

\end{lemma}

\begin{proof}
  By induction on $\nat$.
  \begin{proofcases}
    \proofcase{0} The claim holds by vacuity of the premise.
    \proofcase{$\nat+1$} The premise can be rewritten as follows:
    \[
      \nsstep \nat {\sframe {\frame \cmd \regmap \um} {{\psi}} {\bot}} T \Ds \Os \sto \dir  \obs \cfstack.
    \]
    We proceed by cases analysis on the rule that has been applied to show the last transition.
    Many of these cases are structurally similar, so we focus on the most representative and interesting ones.  
    Since $\Ds$ does not contain any $\dbt$ directives, we can assume without loss of generality that backtracking rules are not involved. 
    \begin{proofcases}
      \proofcase{\ref{SI:Op}}
        In this case we can assume that
        \[
          T = \sframe {\frame {\vx \ass \expr\sep\cmdtwo} {\regmap_3} {\opt}:\st_3} {{\psi_3}} {\boolms''}: T'.
        \]
        By introspection of the rule, we observe that the mode flag of
        the target configuration remains unchanged from the source
        configuration, implying that $\opt = \km[\syscall]$.  This
        allows us to apply the induction hypothesis (IH) to the first
        $\nat$ steps, yielding:
        \begin{equation*}
          \nsstep {\nat'} {\sframe {\frame {\syss(\syscall)} {\regmap_2} {\km[\syscall]}} {\psi_2} {\boolms'}}{}{\Ds'}{\Os'}\\
          {\sframe {\frame {\vx \ass \expr\sep\cmdtwo} {\regmap_3} {\opt}:\overline \st_3} {{\psi_3}} {\boolms''}: \overline T'},
        \end{equation*}
        where $\overline \st_3$ is a prefix of $\st_3$
        such that $\km[\syscall](\overline \st_3)$.
        Additionally, there exists a stack $G$ satisfying $\um(G)$ and $\st_3 = \overline \st_3 : G$.
        Then, we observe that:
        \begin{multline*}
          \sstep {\sframe {\frame {\vx \ass \expr\sep\cmdtwo} {\regmap_3} {\km[\syscall]}:\overline \st_3} {{\psi_3}} {\boolms''}: \overline T'}{}  \dstep\onone
          {\sframe {\frame {\cmdtwo} {\update{\regmap_3}\vx {\sem\exprtwo_{\regmap_3, \lay}}} {\km[\syscall]}:\overline \st_3} {{\psi_3}} {\boolms''}: \overline T'},
        \end{multline*}
        and similarly:
        \begin{multline*}
          \sstep {\sframe {\frame {\vx \ass \expr\sep\cmdtwo} {\regmap_3} {\km[\syscall]}:\st_3} {{\psi_3}} {\boolms''}: T'}{}  \dstep\onone
          {\sframe {\frame {\cmdtwo} {\update{\regmap_3}\vx {\sem\exprtwo_{\regmap_3, \lay}}} {\km[\syscall]}:\st_3} {{\psi_3}} {\boolms''}: T'},
        \end{multline*}
        These transitions establish respectively the existence of the
        reduction sequence starting from
        ${\sframe {\frame {\syss(\syscall)} {\regmap_2}
            {\km[\syscall]}} {\psi_2} {\boolms'}}$ and the structure
        of $\cfstack$.  We already established the properties of
        $\st_3, \overline \st_3$ and $G$ by the IH, so we conclude by
        observing that $\km[\syscall](\cmdtwo)$ follows from
        $\km[\syscall](\vx \ass \expr\sep\cmdtwo)$.

        \proofcase{\ref{SI:Load}}
        In this case we can assume that
        \[
          T = \sframe {\frame {\cmemread[\lbl] \vx \expr\sep\cmdtwo} {\regmap_3} {\opt}:\st_3} {{\psi_3}} {\boolms''}: T'.
        \]
        By introspection of the rule $\ref{SI:Load}$, we observe that
        the execution mode flag of the target configuration remains
        the same as that of the source configuration.  Thus, we
        conclude that $\opt = \km[\syscall]$. This allows us to apply
        the IH to the first $\nat$ steps,
        yielding:
        \begin{equation*}
          \nsstep {\nat'} {\sframe {\frame {\syss(\syscall)} {\regmap_2} {\km[\syscall]}} {\psi_2} {\boolms'}}{}{\Ds'}{\Os'}\\
          {\sframe {\frame {\cmemread[\lbl] \vx \expr\sep\cmdtwo} {\regmap_3} {\opt}:\overline \st_3} {{\psi_3}} {\boolms''}: \overline T'},
        \end{equation*}
        where $\overline \st_3$ is a prefix of $\st_3$ such that $\km[\syscall](\overline \st_3)$.  
        Furthermore, there exists a stack $G$ satisfying $\um(G)$ and $\st_3 = \overline \st_3 : G$.
        Next,  we define $\add$ as the value of $\toAdd{\sem \expr_{\regmap_3, \lay}}$ and  
        $(\val, \bool')$ as the pair returned by $\bufread{{\psi_3}} {\add} {i}$.
        We then observe the transitions:
        \begin{multline*}
          \sstep {\sframe {\frame {\cmemread[\lbl] \vx \expr\sep\cmdtwo} {\regmap_3} {\opt}:\st_3} {{\psi_3}} {\boolms''}: T'}{}  {\dload[\lbl]{i}}{\omem \add}\\
          \sframe {\frame {\cmdtwo} {\update{\regmap_3}\vx {\val}} {\opt}: \st_3} {{\psi_3}} {\boolms''\lor \bool'} :\\
          \sframe {\frame {\cmemread[\lbl] \vx \expr\sep\cmdtwo} {\regmap_3} {\opt}: \st_3} {{\psi_3}} {\boolms''}:  T'
        \end{multline*}
        and
        \begin{multline*}
          \sstep {\sframe {\frame {\cmemread[\lbl] \vx \expr\sep\cmdtwo} {\regmap_3} {\opt}:\overline \st_3} {{\psi_3}} {\boolms''}: \overline T'}{}  {\dload[\lbl] i}{\omem \add}\\
          \sframe {\frame {\cmdtwo} {\update{\regmap_3}\vx {\val}} {\opt}: \overline \st_3} {{\psi_3}} {\boolms''\lor \bool'} :\\
          \sframe {\frame {\cmemread[\lbl] \vx \expr\sep\cmdtwo} {\regmap_3} {\opt}: \overline \st_3} {{\psi_3}} {\boolms''}:  \overline T'.
        \end{multline*}
        As before, these transitions establish respectively the shape of the target stack $\cfstack$ and the existence of the reduction sequence starting from ${\sframe {\frame {\syss(\syscall)} {\regmap_2} {\km[\syscall]}} {\psi_2} {\boolms'}}$.
        We already established the properties of $\st_3, \overline \st_3$ and $G$ by the IH,
        so we conclude by observing that $\km[\syscall](\cmdtwo)$ follows from $\km[\syscall](\vx \ass \expr\sep\cmdtwo)$.

        \proofcase{\ref{SI:Call}} In this case, we assume that 
        \[
          T = \sframe {\frame {\ccall \expr {\exprtwo_1, \dots, \exprtwo_k}\sep\cmdtwo} {\regmap_3} {\opt}:\st_3} {{\psi_3}} {\boolms''}: T'.
        \]
        By analyzing the applied rule and noting that the execution
        flag of the target configuration is $\km[\syscall]$, we deduce
        that the same must hold for the source configuration. Thus, we
        conclude that $\opt = \km[\syscall]$, allowing us to apply the
        IH on the first $\nat$ steps, yielding:
        \begin{equation*}
          \nsstep {\nat'} {\sframe {\frame {\syss(\syscall)} {\regmap_2} {\km[\syscall]}} {\psi_2} {\boolms'}}{}{\Ds'}{\Os'}\\
          {\sframe {\frame {\ccall \expr {\exprtwo_1, \dots, \exprtwo_k}\sep\cmdtwo} {\regmap_3} {\opt}:\overline \st_3} {{\psi_3}} {\boolms''}: \overline T'},
        \end{equation*}
        where $\overline \st_3$ is a prefix of $\st_3$ such that
        $\km[\syscall](\overline \st_3)$ and there exists a stack $G$
        such that $\um(G)$ and $\st_3=\overline \st_3:G$.  Next, we
        define $\add$ as the value of
        $\toAdd{\sem \expr_{\regmap_3, \lay}}$.  From the premises of
        the rule \ref{SI:Call} and the definition of
        $\underline{\lay}$, we deduce that there exists some
        $\fn \in \Fn[\km]$ such that $\lay(\fn) = \add$.
        By \Cref{rem:simispreservedspec}, we infer that
        $\psi_3 = \bm{\buf'}{\lay\lcomp \rfs''}$ for some
        $\rfs'\eqon{\Fn}\rfs$.  From the definition of
        $\cdot \lcomp \cdot$ and the previous observations, we
        conclude that the executed procedure is precisely $\rfs(\fn)$.

        We now define $\overline{\regmap}$ as the register map
        obtained by evaluating the argument expressions in
        $\regmap_3$ and updating the argument registers of
        $\regmap_0$ accordingly.  By using the rule \ref{SI:Call}, we
        establish the existence of the following transitions:
        \begin{multline*}
          \sstep {\sframe {\frame {\ccall \expr {\exprtwo_1, \dots, \exprtwo_k}\sep\cmdtwo} {\regmap_3} {\km[\syscall]}: \st_3} {{\psi_3}} {\boolms''}:  T'}{}  {\dstep}{\ojump \add}\\
          {\sframe {\frame {\rfs(\fn)} {{\overline \regmap}} {\km[\syscall]}:\frame {\cmdtwo} {{\regmap_3}} {\km[\syscall]}: \st_3} {{\psi_3}} {\boolms''}:  T'},
        \end{multline*}
        and
        \begin{multline*}
          \sstep {\sframe {\frame {\ccall \expr {\exprtwo_1, \dots, \exprtwo_k}\sep\cmdtwo} {\regmap_3} {\km[\syscall]}:\overline \st_3} {{\psi_3}} {\boolms''}: \overline T'}{}  {\dstep}{\ojump \add}\\
          {\sframe {\frame {\rfs(\fn)} {{\overline \regmap}} {\km[\syscall]}:\frame {\cmdtwo} {{\regmap_3}} {\km[\syscall]}:\overline \st_3} {{\psi_3}} {\boolms''}: \overline T'}.
        \end{multline*}
        As in previous cases, these transitions establish both the
        structure of the target stack $\cfstack$ and the existence of
        the reduction sequence originating from
        ${\sframe {\frame {\syss(\syscall)} {\regmap_2}
            {\km[\syscall]}} {\psi_2} {\boolms'}}$.  We have
        already established the required properties of
        $\st_3, \overline{\st}''', G$ using the IH.  To lift these
        properties to satisfy the claim, we note that
        $\km[\syscall](\cmdtwo)$ follows from
        $\km[\syscall](\vx \ass \expr\sep\cmdtwo)$, and that
        $\km[\syscall](\rfs(\fn))$ holds by definition of $\rfs$,
        since $\fn \in \Fnk$.

        \proofcase{\ref{SI:System-Call}} In this case we can assume that
        \[
          T= \sframe {\frame {\csyscall \syscalltwo {\exprtwo_1, \dots, \exprtwo_k}\sep\cmdtwo} {\regmap_3} {\opt}:\st_3} {{\psi_3}} {\boolms''}: T'.
        \]
        We proceed by cases on $\opt = \um$:
        \begin{proofcases}
          \proofcase{$\opt = \um$}
          By introspection of the rule and the target configuration, we deduce that
          $\syscalltwo = \syscall$,  and that
          $\cfstack$ has the following shape:
          \[
            \sframe {\frame {\syss(\syscall)}{\regmap_0'}{\km[\syscall]}:\frame {\cmdtwo} {\regmap_3} {\um}:\st_3} {{\psi_3}} {\boolms''}: T',
          \]
          where $\regmap_0'$ is obtained by updating the argument
          registers of $\regmap_0$ with
          the evaluation of $\expr_1, \dots, \expr_k$. 
          To show the claim, it suffices to set $\nat'=0$, $\Ds'=\nil$, $\Os=\nil$,
          $\overline \st = \nil$, $\st_2=\st_1$.
          Additionally, we note that
          $\km[\syscall]({\frame {\syss(\syscall)}{\regmap_0'}{\km[\syscall]}})$
          holds for definition of $\syss$. Moreover $\um({\frame {\cmdtwo} {\regmap_3} {\opt}:\st_3})$ is a consequence of \Cref{rem:simispreservedspec}.
          \proofcase{$\opt \neq \um$} This case is analogous to the case of procedure calls. 
        \end{proofcases}
        \proofcase{\ref{SI:Pop}} In this case we assume:
        \[
          T = \sframe {\frame {\cnil} {\regmap_3} {\opt}:\st_3} {{\psi_3}} {\boolms''}: T'.
        \]
        By introspection of the rule, we deduce that
        \[
          \st_3= \frame \cmd {\update{\regmap_1}\ret\val} {\km[\syscall]}:\st_1
        \]
        for some $\val$.
        From this observation, and using \Cref{rem:simispreservedspec}, we deduce that $\opt = \km[\syscall]$. This allows us to apply the IH to the first $\nat$ steps, and we obtain:
        \begin{equation*}
          \nsstep {\nat'} {\sframe {\frame {\syss(\syscall)} {\regmap_2} {\km[\syscall]}} {\psi_2} {\boolms'}}{}{\Ds'}{\Os'}\\
          {\sframe {\frame {\cnil} {\regmap_3} {\opt}:\overline \st_3} {{\psi_3}} {\boolms''}: \overline T'},
        \end{equation*}
        where $\overline \st_3$ is a prefix of $\st_3$ such that $\km[\syscall](\overline \st_3)$ and there is a stack $G$ such that $\um(G)$ and $\st_3 = \overline \st_3:G$.
        This shows that, in particular,
        the topmost frame of $\overline \st_3$ must also be
        $\frame \cmd {\update{\regmap_1}\ret\val} {\km[\syscall]}$,
        so we have $\overline \st_3 = \frame \cmd {\update{\regmap_1}\ret\val} {\km[\syscall]}: \overline \st_1$ for some $\overline \st_1$.
        Thanks to this observation, and
        by introspection of the rule \ref{SI:Pop}, we observe that:
        \begin{equation*}
          \sstep {\sframe {\frame {\cnil} {\regmap_3} {\km[\syscall]}:\frame \cmd {\update{\regmap_1}\ret\val} {\km[\syscall]}:\st_1} {{\psi_3}} {\boolms''}: T'}{}  {\dstep}{\onone}\\
          {\sframe {\frame \cmd {\regmap_1} {\km[\syscall]}:\st_1} {{\psi_3}} {\boolms''}:  T'},
        \end{equation*}
        and also:
        \begin{equation*}
          \sstep {\sframe {\frame {\cnil} {\regmap_3} {\km[\syscall]}:\frame \cmd {\update{\regmap_1}\ret\val} {\km[\syscall]}: \overline \st_1} {{\psi_3}} {\boolms''}: \overline T'}{}  {\dstep}{\onone}\\
          {\sframe {\frame \cmd {\regmap_1} {\km[\syscall]}:\overline \st_1} {{\psi_3}} {\boolms''}:  \overline T'}.
        \end{equation*}
        Notice that we have $\km[\syscall](\overline \st_3)$ and $\overline \st_3 = \frame \cmd {\update{\regmap_1}\ret\val} {\km[\syscall]}: \overline \st_1$, so we deduce  $\km[\syscall](\overline \st_1)$. Also, observe $\st_1 = \overline \st_1 :G$. This concludes the proof.
      \end{proofcases}    
  \end{proofcases}
\end{proof}

\subsection{Proofs from \Cref{sec:sksenforcement} }
\label{sec:proofs3}

In this section, the language is extended with the instruction
$\cfence$ and with non-speculative call instructions. In order to make
our non-speculative semantics compatible with the extended language,
we enrich the non-speculative semantics of
\Cref{fig:stepexcerpt1,fig:stepexcerpt2} with the following rule for the
$\cfence$ instruction:
\[
  \Infer[WL][Fence]
  { \step
    {\ntc{\cfence\sep\cmd}{\regmap}{\opt}{\st}{\mem}}
    {\ntc{\cmd}{\regmap}{\opt}{\st}{\mem}}
  }
  {
  }
\]
and the following rules for safe calls:

\[
    \Infer[WL][SCall]{
    \step
    {\ntc{\cscall{\expr}{\vec \exprtwo}\sep\cmd}{\regmap}{\opt}{\st}{\mem}}
    {
      \ntc
      {\mem(\add)}
      {\regmap_0[\vec \vx \upd \sem{\vec \exprtwo}_{\regmap,\lay}]}
      {\opt}
      {\frame{\cmd}{\regmap}{\opt} : \st}
      {\mem}
    }
  }
  {\toAdd{\sem{\expr}_{\regmap, \lay}}=\add &
    \add \in \underline \lay(\Fn[\opt]) &
    \fbox{$\opt = \km[\syscall] \Rightarrow \add \in \underline \lay(\caps(\syscall))$}
  }
\]

\[
  \Infer[WL][SCall-Unsafe]{
    \step
    {\ntc{\cscall{\expr}{\vec \exprtwo}\sep\cmd}{\regmap}{\km[\syscall]}{\st}{\mem}}
    {\unsafe}
  }
  {
    \toAdd{\sem{\expr}_{\regmap, \lay}} = \add &
    \add \in \underline \lay(\Fn[\km]) &
    \fbox{$\add \not\in \underline \lay(\caps(\syscall))$}
  }
\]

\[
  \Infer[WL][SCall-Error]
  {\step
    {\ntc{\cscall \expr {\vec \exprtwo}\sep\cmd}{\regmap}{\opt}{\st}{\mem}}
    {\err}
  }
  {\toAdd{\sem\expr_{\regmap, \lay}} = \add &
    \add \notin \underline \lay(\Fn[\opt])
  }
\]

\noindent
The following remark is useful for the proof of \Cref{prop:mitigation}

\begin{remark}
  \label{rem:skstrans}
  If $\system$ is kernel safe, $\systrans$ is user-space semantics preserving, and
  $\systrans(\system) = \system'$, then $\system'$ is also kernel safe.
\end{remark}

\begin{proof}[Proof of \Cref{prop:mitigation}]
  We fix a system $\system = (\rfs, \syss, \caps) = \zeta(\system')$
  and a transformation $\systrans$ that preserves the system's
  semantics and enforces speculative safety. Claim (ii) follows
  directly from \Cref{lemma:ctxsemequiv}, so we focus on proving claim
  (i). The proof proceeds by contraposition.
  We assume there exists an unprivileged command $\speccmd \in \SpCmd$, an initial register map $\regmap$, a number of steps $\nat$, and a layout such that:
  \[
    \nesstep \nat {\conf{\frame \speccmd \regmap \um, \lay \lcomp \rfs, \nil, \nil}} \unsafe.
  \]
  First, we observe that $\nat \neq 0$. By inspecting the rules of the
  semantics, we note that the last rule applied must be one of the
  following: \ref{AL:Load-Unsafe}, \ref{AL:Store-Unsafe},
  \ref{AL:Call-Unsafe}, or \ref{AL:Spec-Unsafe}. We proceed by case
  analysis on these rules. In particular, the proofs for the first
  three rules are analogous, so we take the case of the rule
  \ref{AL:Load-Unsafe} as an example.
  \begin{proofcases}
    \proofcase{\ref{AL:Load-Unsafe}} By introspecting the rule,
    we deduce that there exists a configuration
    \[
      \conf{\frame {\cmemread \vx \expr} {\regmap'}{\km[\syscall]}:\st, \lay \lcomp\rfs', \Ds, \Os} 
    \]
    such that
    \begin{equation*}
      \nesstep \nat {\conf{\frame \speccmd \regmap \um, \lay \lcomp \rfs,\nil, \nil}}
      {\conf{\frame {\cmemread \vx \expr} {\regmap'}{\km[\syscall]}:\st, \lay \lcomp\rfs', \Ds, \Os}}\\ \ato \unsafe,
    \end{equation*}
    with $\rfs' \eqon{\Fn} \rfs$.  Using \Cref{rem:stackinv1}, we
    observe that $\st$ is the concatenation of a kernel mode stack
    $\st'$ and a user mode stack $\st''$. Therefore, we can apply
    \Cref{lemma:thereisasyscall1} and deduce that there exists a
    configuration
    \[
      {\conf{\frame {\syss(\syscall)} {\regmap_0[\vx_1, \dots, \vx_k\upd \val_1, \dots, \val_k]}{\km[\syscall]}, \lay \lcomp\rfs'', \Ds', \Os'}},
    \]
    with $\rfs'' \eqon{\Fn} \rfs$, a prefix $\st'$ of $\st$, and $\nat' \in \Nat$, such that:
    \begin{multline*}
      \lay \red {\conf{\frame {\syss(\syscall)} {\regmap_0[\vx_1, \dots, \vx_k\upd \val_1, \dots, \val_k]}{\km[\syscall]}, \lay \lcomp\rfs'', \Ds', \Os'}} \\ \ato^{\nat'}
      {\conf{\frame {\cmemread \vx \expr} {\regmap'}{\km[\syscall]}:\st', \lay \lcomp\rfs', \Ds, \Os}} \ato \unsafe.
    \end{multline*}
    With the application of \Cref{lemma:syscsym}, we deduce:
    \begin{equation*}
      \lay \red {\conf{\frame {\syss(\syscall)} {\regmap_0[\vx_1, \dots, \vx_k\upd \val_1, \dots, \val_k]}{\km[\syscall]}, \lay \lcomp\rfs''}} \\ \to^{\nat'}
      {\conf{\frame {\cmemread \vx \expr} {\regmap'}{\km[\syscall]}:\st', \lay \lcomp\rfs'}} \to \unsafe.
    \end{equation*}
    This shows that $\system$ is not kernel safe, which contradicts
    \Cref{rem:skstrans}, which states that $\system$ is safe.
    
    \proofcase{\ref{AL:Spec-Unsafe}} By inspecting the rule, we
    deduce that there is a hybrid configuration
    \[
      {\hconf\unsafe {\frame{\adversary}{\regmap}{\opt}\cons\st}\Ds\Os}
    \]
    such that 
    \[
      \esstep {\hconf\unsafe {\frame{\adversary}{\regmap}{\opt}\cons\st}\Ds\Os} \unsafe
    \]
    and
    \[
      \nesstep \nat  {\conf{\frame \speccmd \regmap \um, \lay \lcomp \rfs, \nil, \nil}} 
      {\hconf\unsafe {\frame{\adversary}{\regmap}{\opt}\cons\st}\Ds\Os}
    \]
    Using \Cref{rem:thereisaspecon}, we deduce that there exists a
    configuration
    \[
      {\sframe{\frame {\cmd} {\regmap'}{\um}}{\bm{\buf}{\lay \lcomp \rfs'}}{\bot}},
    \]
    with $\rfs' \eqon{\Fn} \rfs$, a sequence of directives $\Ds'$, a
    sequence of observations $\Os'$, and a natural number
    $\nat' \le \nat$, such that:
    \begin{equation*}
      \nsstep {\nat'}
      {\sframe{\frame {\cmd} {\regmap'}{\um}}{\bm{\buf}{\lay \lcomp \rfs'}}{\bot}}
      \unsafe {\Ds'}{\Os'}.
    \end{equation*}
    By \Cref{lemma:nobt}, we can assume without loss of generality
    that $\Ds'$ does not contain any $\dbt$ directive.  From
    \Cref{lemma:thereisasyscall2}, and by inspecting the semantics, we
    deduce that there are configurations
    \begin{equation*}
      {\sframe{\frame {\syss(\syscall)} {\regmap''}{\km[\syscall]}}{\bm{\buf'}{\lay \lcomp \rfs''}}{\bool}}
      \quad
      \text{and}
      \quad
      {\sframe{\frame {\cmdtwo} {\regmap'''}{{\km}_\syscall}:\st''}{\bm{\buf''}{\lay \lcomp \rfs'''}}{\boolms}},
    \end{equation*}
    a sequence of directives $\Ds''$, a sequence of observations
    $\Os''$ and a store $\rfs''$ such that:
    \begin{equation*}
      \nsstep {\nat''}
      {\sframe{\frame {\syss(\syscall)} {\regmap''}{\km[\syscall]}}{\bm{\buf'}{\lay \lcomp \rfs''}}{\bool}}
      {\sframe{\frame {\cmdtwo} {\regmap'''}{\km[\syscall]}:\st''}{\bm{\buf''}{\lay \lcomp \rfs'''}}{\boolms}}
      {\Ds''}{\Os''}\sto \dir \obs \unsafe .
    \end{equation*}
    We apply our assumption on $\zeta$, which gives:
    \begin{equation*}
      \nstep {*}
      {\sframe{\frame {\syss(\syscall)} {\regmap''}{\km[\syscall]}}{\overline{\bm{\buf'}{\lay \lcomp \rfs''}}}{\bool}}
      \unsafe.
    \end{equation*}
    This shows that $\system$ is unsafe, which contradicts
    \Cref{rem:skstrans}, thereby proving that the system is actually
    safe.
  \end{proofcases}
\end{proof}

\begin{proof}[Proof of \Cref{thm:mitigation}] The claim is a direct
consequence of \Cref{cor:fencesem2}, \Cref{lemma:fencesafeimpo} and
\Cref{prop:mitigation}.
\end{proof}

\begin{proof}[Proof of \Cref{thm:optmitigation}] The claim is a
direct consequence of \Cref{cor:fencesem2},
\Cref{lemma:optfencesafeimpo} and \Cref{prop:mitigation}.
\end{proof}

\begin{proof}[Proof of \Cref{thm:nospecmitigation}] The claim is a
direct consequence of \Cref{cor:fencesem2},
\Cref{lemma:nospecsafeimpo} and \Cref{prop:mitigation}.
\end{proof}


To aid in proving the upcoming results, we introduce the concept of
\emph{well-formedness} for the transformations $\fencetrans$,
$\optfencetrans$, and $\nospec$. These relations are designed to
express invariant properties that hold when the execution of a system
call begins and are preserved during kernel-space evaluation. The
corresponding predicates can be found in
\Cref{fig:transrel,fig:transrel2,fig:transrel3}.

For the well-formedness relation of $\optfencetrans$ in
\Cref{fig:transrel2}, we use the auxiliary predicate
$\Sigma(\cmd, m, e)$, which holds whenever the command has the form:
$\cmd = \optfencetrans[m][e](\cmd');
\optfencetrans[\top][\bot](\cmdtwo_1);\ldots;\optfencetrans[\top][\bot](\cmdtwo_h)$.
In addition, in the rule for configurations, we overload the operator
$\sep$ to represent the lexical concatenation of two programs $\cmd$
and $\cmdtwo$, with $\cnil$ serving as the null element for this
operation. Finally, note that since $k$ can be any natural number, we
also account for the case where $\cmdtwo_1, \ldots, \cmdtwo_k$ is an
empty sequence when $k = 0$.

\begin{figure*}[t]
  \flushleft
  \[
    \infer{\fwf {\rfs'}}{\forall \fn\in\Fnk. \exists \cmd. \rfs'(\fn)=\fencetrans(\cmd) & {\rfs} \eqon{\Fn} \rfs'}
    \quad
    \infer{\fwf{\frame \cmd \regmap {\km[\syscall]}:\st}}{\fwf{\st} & \exists \cmd'.\cmd=\fencetrans (\cmd')}
  \]
  
  \[
    \infer{\fwf\err}{}\quad
    \infer{\fwf\unsafe}{}
  \]
  
  \[
    \infer{\fwf{\sframe\st{\bm\buf{\lay\lcomp \rfs'}}\boolms}}{\fwf {\rfs'} & \fwf\st & \dom(\buf) \subseteq \underline \lay(\Ar)}
    \quad
    \infer{\fwf{\specconfone:\cfstack}}{\fwf \specconfone }\quad
    \infer{\fwf{\nil}}{}
  \]

  \caption{Well-formedness relation with respect to a system $\system=(\rfs, \syss, \caps)$ for the transformation $\fencetrans$.}
  \label{fig:transrel}
\end{figure*}

\begin{figure*}[t]
  \[
    \infer{\optfwf {\rfs'}}{\forall \fn\in\Fnk. \exists \cmd. \rfs'(\fn)=\optfencetrans[\top][\bot](\cmd) & {\rfs} \eqon{\Fn} \rfs'}
    \quad
    \infer{\optfwf{\sframe\nil{\bm\buf{\lay\lcomp \rfs'}}\boolms}}{\optfwf {\rfs'}  & \dom(\buf) \subseteq \underline \lay(\Ar)}
  \]
  
  \resizebox{\textwidth}{!}{$
    \infer{\optfwf{\sframe{\frame \cmd \regmap {\km[\syscall]}:\frame {\cmd_1} {\regmap_1} {\km[\syscall]}\ldots \frame {\cmd_h} {\regmap_h} {\km[\syscall]}}{\bm\buf{\lay\lcomp \rfs'}}\boolms}}
    {\optfwf {\rfs'} &
      \Sigma(\cmd, m, e), \land \Sigma(\cmd_1, \top, \bot)\land\ldots \land \Sigma(\cmd_k, \top, \bot)
      & \boolms \Rightarrow m
      & e \Rightarrow \buf = \nil
      & h, k \in \Nat
      & \dom(\buf) \subseteq \underline \lay(\Ar)}
  $}
  
  \[
    \infer{\optfwf{\specconfone:\cfstack}}{\optfwf \specconfone }\quad
    \infer{\optfwf{\nil}}{}
    \quad
    \infer{\optfwf\err}{}\quad
    \infer{\optfwf\unsafe}{}
  \]

    \caption{Well-formedness relation with respect to a system $\system=(\rfs, \syss, \caps)$ for the transformation $\optfencetrans$.}
  \label{fig:transrel2}
\end{figure*}

\begin{figure*}[t]
  \[
    \infer{\nsfwf {\rfs'}}{\forall \fn\in\Fnk. \exists \cmd. \rfs'(\fn)=\nospec(\cmd) & {\rfs} \eqon{\Fn} \rfs'}
    \quad
    \infer{\nsfwf{\frame \cmd \regmap {\km[\syscall]}:\st}}{\nsfwf{\st} & \exists \cmd'.\cmd=\nospec (\cmd')}
  \]

  \[
    \infer{\nsfwf\err}{}\quad
    \infer{\nsfwf\unsafe}{}
  \]
  
  \[
    \infer{\nsfwf{\sframe\st{\bm\nil{\lay\lcomp \rfs'}}\bot}}{\nsfwf {\rfs'} & \nsfwf\st}
    \quad
    \infer{\nsfwf{\specconfone:\cfstack}}{\nsfwf \specconfone }\quad
    \infer{\nsfwf{\nil}}{}
  \]
  \caption{Well-formedness relation with respect to a system $\system=(\rfs, \syss, \caps)$ for the system transformation $\nospec$.}
  \label{fig:transrel3}
\end{figure*}

\begin{lemma}
  \label{lemma:fencesafeimpo}
  The transformation $\fencetrans$  \emph{imposes speculative kernel safety}.
\end{lemma}

\begin{proof}
  We assume that there is a system $\system = (\rfs, \syss, \funs)$, a
  system call $\syscall$, a register map $\regmap$, a buffer $\buf$,
  an array store $\rfs' \eqon{\Fn} \rfs$, and a natural number $\nat$
  such that
  \[
    \nsstep[\system][\lay] \nat {\sframe{\frame {\syss(\syscall)}{\regmap} {\km[\syscall]}}  {\bm\buf{(\lay\lcomp\rfs')}}{\boolms}} \unsafe \Ds \Os
  \]
  for a sequence of directives $\Ds$ (that we assume to be free of
  $\dbt$ due to \Cref{lemma:nobt}) producing a sequence of
  observations $\Os$. By case analysis on $\nat$, we rule out the case
  where $\nat = 0$, as this would imply
  \[
    {\sframe{\frame {\syss(\syscall)} {\regmap} {\km[\syscall]}}{\bm\buf{(\lay\lcomp\rfs')}}{\boolms}} = \unsafe.
  \]
  Therefore, we assume that $\nat > 0$. By case analysis on the proof relation, we deduce that the rule used to show the last transition must be one of \ref{SI:Load-Unsafe}, \ref{SI:Store-Unsafe}, \ref{SI:Call-Unsafe}, or \ref{SI:Call-Step-Unsafe}.
  We will omit the case of rule \ref{SI:Store-Unsafe}, which is analogous to that of rule \ref{SI:Load-Unsafe}, treated below.
  \begin{proofcases}
    \proofcase{\ref{SI:Load-Unsafe}}
    In this case, our assumption rewrites as follows:
    \begin{multline*}
      \nsstep[\system][\lay] {\nat-1} {\sframe{\frame {\syss(\syscall)} {\regmap} {\km[\syscall]}}{\bm\buf{(\lay\lcomp\rfs')}}{\boolms}} {} {\Ds'} {\Os'} \\
      {\sframe{\frame {\cmemread \vx \expr\sep\cmdtwo} {\regmap'} {\km[\syscall]}:\st}{\bm{\buf'}{(\lay\lcomp\rfs'')}}{\boolms'}}:\cfstack  \sto {\dstep} {\omem\add} 
      \unsafe.
    \end{multline*}
    Observe that $\nat >1$; otherwise, if $\nat=1$, we would have
    ${\cmemread \vx \expr\sep\cmdtwo}= \syss(\syscall)$, but from
    $\system \in \im(\fencetrans)$, we would deduce that there exists
    a system $\system'=(\rfs'', \syss',\caps')$ such that
    $\fencetrans (\system') = \system$.  This would imply that
    $\syss(\syscall) = \fencetrans(\syss'(\syscall))$, so there is a
    command $\cmd \in \Cmd$ such that
    $\syss(\syscall) = \fencetrans(\cmd)$. However, by induction on
    the syntax of the command, we observe that this is not possible.
    We apply \Cref{lemma:fwf}, which leads us to deduce that there
    must be a stack of configurations
    \[
      \sframe{\frame{\cmd} {\regmap''} {\km[\syscall]}:\st''}{\bm{\buf''}{(\lay\lcomp\rfs''')}}{\boolms''}:\cfstack'
    \]
    such that:
    \begin{multline*}
      \nsstep[\system][\lay] {\nat-2} {\sframe{\frame {\syss(\syscall)} {\regmap} {\km[\syscall]}}{\bm\buf{(\lay\lcomp\rfs')}}{\boolms}} {}{\Ds''} {\Os''}\\
      {\sframe{\frame{\cmd} {\regmap''} {\km[\syscall]}:\st''}{\bm{\buf''}{(\lay\lcomp\rfs''')}}{\boolms''}:\cfstack'}  \sto \dir \obs\\
      {\sframe{\frame {\cmemread \vx \expr\sep\cmdtwo} {\regmap'} {\km[\syscall]}:\st}{\bm{\buf'}{(\lay\lcomp\rfs'')}}{\boolms'}}:\cfstack  \sto {\dstep} {\omem\add}.
      \unsafe.
    \end{multline*}
    In particular,
    $\fwf[\fencetrans]{\frame{\cmd} {\regmap''}
      {\km[\syscall]}:\st''}$ holds.  From this observation, we deduce
    that $\cmd = \fencetrans(\cmdtwo)$ for some $\cmdtwo \in \Cmd$,
    and that $\fwf[\fencetrans]{\st''}$ holds.  For these reasons, we
    can now proceed by cases on $\cmdtwo$ and the rule used to show
    the transition with the directive $\dir$ (knowing that
    $\dir \neq \dbt$ by assumption). We can rewrite the reduction
    above as follows:
    \begin{multline*}
      \nsstep[\system][\lay] {\nat-2} {\sframe{\frame {\syss(\syscall)} {\regmap} {\km[\syscall]}}{\bm\buf{(\lay\lcomp\rfs')}}{\boolms}} {}{\Ds''} {\Os''}\\
      {\sframe{\frame{\cfence\sep\cmemread \vx \expr\sep\cmdtwo} {\regmap''} {\km[\syscall]}:\st''}{\bm{\buf''}{(\lay\lcomp\rfs''')}}{\boolms''}:\cfstack}  \sto \dir \obs\\
      {\sframe{\frame {\cmemread \vx \expr\sep\cmdtwo} {\regmap'} {\km[\syscall]}:\st}{\overline{\bm{\buf''}{(\lay\lcomp\rfs''')}}}{\bot}}:\cfstack  \sto {\dstep} {\omem\add} 
      \unsafe.
    \end{multline*}
    In particular, none of the rules except for \ref{SI:Fence} can have been used. This also means that $\boolms' = \boolms'' = \bot$.
    From \Cref{lemma:nobtsteponly}, we deduce that there is $\nat'$ such that:
    \begin{multline*}
      \nsstep[\system][\lay] {\nat'} {\sframe{\frame {\syss(\syscall)} {\regmap} {\km[\syscall]}}{\bm\buf{(\lay\lcomp\rfs')}}{\boolms}} {}{\dstep^{\nat'}} {\Os'''}\\
      {\sframe{\frame {\cmemread \vx \expr\sep\cmdtwo} {\regmap'} {\km[\syscall]}:\st}{\overline{\bm{\buf''}{(\lay\lcomp\rfs''')}}}{\bot}}:\cfstack  \sto {\dstep} {\omem\add} 
      \unsafe.
    \end{multline*}
    Finally, we apply \Cref{lemma:stepsemsim} that shows:
    \begin{equation*}
      \nstep[\system][\lay] {\nat'} {\conf{\frame {\syss(\syscall)} {\regmap} {\km[\syscall]}, \overline{\bm\buf{(\lay\lcomp\rfs')}}}}{}\\
      {\conf{\frame {\cmemread \vx \expr\sep\cmdtwo}{\regmap'} {\km[\syscall]}:\st, \overline{\bm{\buf''}{(\lay\lcomp\rfs''')}}}}.
    \end{equation*}
    Finally, by assumption, we know that the rule \ref{SI:Load-Unsafe}
    has been applied to show the transition in the speculative
    semantics. By introspection of that rule, we conclude that its
    premises are also verified by the configuration
    \[
      {\conf{\frame {\cmemread \vx \expr\sep\cmdtwo}{\regmap'} {\km[\syscall]}:\st, \overline{\bm{\buf''}{(\lay\lcomp\rfs''')}}}}.
    \]
    This shows that \ref{WL:Load-Unsafe} applies, proving:
    \begin{equation*}
      \nstep[\system][\lay] {\nat'} 
      {\conf{\frame {\cmemread \vx \expr\sep\cmdtwo}{\regmap'} {\km[\syscall]}:\st, \overline{\bm{\buf''}{(\lay\lcomp\rfs''')}}}}
      \unsafe,
    \end{equation*}
    and, therefore, establishing the claim.

    \proofcase{\ref{SI:Call-Unsafe}}
    In this case, by following a similar approach to the previous case, we observe that
    \begin{multline*}
      \nsstep[\system][\lay] {\nat-1} {\sframe{\frame {\syss(\syscall)} {\regmap} {\km[\syscall]}}{\bm\buf{(\lay\lcomp\rfs')}}{\boolms}} {} {\Ds'} {\Os'} \\
      {\sframe{\frame {\ccall \expr {\expr_1, \ldots, \expr_n};\cmd'} {\regmap'} {\km[\syscall]}:\st}{\bm{\buf'}{(\lay\lcomp\rfs'')}}{\boolms'}}:\cfstack  \sto {\djump \add} {\ojump\add} 
      \unsafe.
    \end{multline*}
    As we did before, we observe that we must have $\nat > 1$. By
    applying \Cref{lemma:fwf}, we deduce that there must be a stack of
    configurations
    \[
      \sframe{\frame{\cmd} {\regmap''} {\km[\syscall]}:\st''}{\bm{\buf''}{(\lay\lcomp\rfs''')}}{\boolms''}:\cfstack'
    \]
    such that:
    \begin{multline*}
      \nsstep[\system][\lay] {\nat-2} {\sframe{\frame {\syss(\syscall)} {\regmap} {\km[\syscall]}}{\bm\buf{(\lay\lcomp\rfs')}}{\boolms}} {}{\Ds''} {\Os''}\\
      {\sframe{\frame{\cmd} {\regmap''} {\km[\syscall]}:\st''}{\bm{\buf''}{(\lay\lcomp\rfs''')}}{\boolms''}:\cfstack'}  \sto \dir \obs\\
      {\sframe{\frame {\ccall \expr {\expr_1, \ldots, \expr_n};\cmd'} {\regmap'} {\km[\syscall]}:\st}{\bm{\buf'}{(\lay\lcomp\rfs'')}}{\boolms'}}:\cfstack  \sto {\djump \add} {\ojump\add} 
      \unsafe,
    \end{multline*}
    In particular, we observe that
    $\fwf[\fencetrans]{\frame{\cmd} {\regmap''}
      {\km[\syscall]}:\st''}$ holds.  From this observation, we deduce
    that $\cmd = \fencetrans(\cmdtwo)$ for some $\cmdtwo \in \Cmd$,
    and that $\fwf[\fencetrans]{\st''}$ holds. However, this
    conclusion is absurd because, by examining the syntax of
    $\cmdtwo$, we can easily observe that none of the images of a
    command may reduce in one step to
    $\ccall \expr {\expr_1, \ldots, \expr_n}$.

    \proofcase{\ref{SI:Call-Step-Unsafe}} In this case, the command
    that caused the unsafe memory access can either be a
    $\ccall \cdot \cdot$ or a $\cscall \cdot \cdot$. By reasoning
    similar to the previous case, we can exclude the possibility of a
    $\ccall \cdot \cdot$, so we are left with:
    \begin{multline*}
      \nsstep[\system][\lay] {\nat-1} {\sframe{\frame {\syss(\syscall)} {\regmap} {\km[\syscall]}}{\bm\buf{(\lay\lcomp\rfs')}}{\boolms}} {} {\Ds'} {\Os'} \\
      {\sframe{\frame {\cscall \expr {\expr_1, \ldots, \expr_n};\cmd'} {\regmap'} {\km[\syscall]}:\st}{\bm{\buf'}{(\lay\lcomp\rfs'')}}{\boolms'}}:\cfstack  \sto {\dstep} {\ojump\add} 
      \unsafe.
    \end{multline*}
    Observe that $\nat > 1$; otherwise, if $\nat = 1$, we would have
    ${\cscall \expr {\expr_1, \ldots, \expr_n}} =
    \syss(\syscall)$. However, from $\system \in \im(\fencetrans)$, we
    deduce that there exists a system
    $\system' = (\rfs'', \syss', \caps')$ such that
    $\fencetrans (\system') = \system$. This would imply that
    $\syss(\syscall) = \fencetrans(\syss'(\syscall))$, and thus there
    is a command $\cmd \in \Cmd$ such that
    $\syss(\syscall) = \fencetrans(\cmd)$. But, by examining
    the syntax of the command, we observe that this is not
    possible. Therefore, we can apply \Cref{lemma:fwf}, which allows
    us to deduce that there is a stack of configurations
    \[
      \sframe{\frame{\cmd} {\regmap''} {\km[\syscall]}:\st''}{\bm{\buf''}{(\lay\lcomp\rfs''')}}{\boolms''}:\cfstack'
    \]
    such that:
    \begin{multline*}
      \nsstep[\system][\lay] {\nat-2} {\sframe{\frame {\syss(\syscall)} {\regmap} {\km[\syscall]}}{\bm\buf{(\lay\lcomp\rfs')}}{\boolms}} {}{\Ds''} {\Os''}\\
      {\sframe{\frame{\cmd} {\regmap''} {\km[\syscall]}:\st''}{\bm{\buf''}{(\lay\lcomp\rfs''')}}{\boolms''}:\cfstack'}  \sto \dir \obs\\
      {\sframe{\frame {\cscall \expr {\expr_1, \ldots, \expr_n};\cmd'} {\regmap'} {\km[\syscall]}:\st}{\bm{\buf'}{(\lay\lcomp\rfs'')}}{\boolms'}}:\cfstack  \sto {\dstep} {\ojump\add} 
      \unsafe.
    \end{multline*}
    Here, in particular,
    $\fwf[\fencetrans]{\frame{\cmd} {\regmap''}
      {\km[\syscall]}:\st''}$ holds.  From this observation, we deduce
    that $\cmd = \fencetrans(\cmdtwo)$ for some $\cmdtwo \in \Cmd$,
    and that $\fwf[\fencetrans]{\st''}$ holds. Reasoning similarly to
    what we did for the \ref{SI:Load-Unsafe} rule, we can rewrite the
    reduction above as follows:
    \begin{multline*}
      \nsstep[\system][\lay] {\nat-2} {\sframe{\frame {\syss(\syscall)} {\regmap} {\km[\syscall]}}{\bm\buf{(\lay\lcomp\rfs')}}{\boolms}} {}{\Ds''} {\Os''}\\
      {\sframe{\frame{\cfence\sep\cscall \expr {\expr_1, \ldots, \expr_n};\cmd'} {\regmap''} {\km[\syscall]}:\st''}{\bm{\buf''}{(\lay\lcomp\rfs''')}}{\boolms''}:\cfstack}  \sto \dir \obs\\
      {\sframe{\frame {\cscall \expr {\expr_1, \ldots, \expr_n};\cmd'} {\regmap'} {\km[\syscall]}:\st}{\overline{\bm{\buf''}{(\lay\lcomp\rfs''')}}}{\bot}}:\cfstack  \sto {\dstep} {\omem\add} 
      \unsafe.
    \end{multline*}
    In particular, none of the rules except for \ref{SI:Fence} can have been used.
    This also means that $\boolms'=\boolms''=\bot$.
    From \Cref{lemma:nobtsteponly}, we deduce that there is $\nat'$ such that:
    \begin{multline*}
      \nsstep[\system][\lay] {\nat'} {\sframe{\frame {\syss(\syscall)} {\regmap} {\km[\syscall]}}{\bm\buf{(\lay\lcomp\rfs')}}{\boolms}} {}{\dstep^{\nat'}} {\Os'''}\\
      {\sframe{\frame {\cscall \expr {\expr_1, \ldots, \expr_n};\cmd'} {\regmap'} {\km[\syscall]}:\st}{\overline{\bm{\buf''}{(\lay\lcomp\rfs''')}}}{\bot}}:\cfstack  \sto {\dstep} {\omem\add} 
      \unsafe.
    \end{multline*}
    We apply \Cref{lemma:stepsemsim}, which. shows:
    \begin{equation*}
      \nstep[\system][\lay] {\nat'} {\conf{\frame {\syss(\syscall)} {\regmap} {\km[\syscall]}, \overline{\bm\buf{(\lay\lcomp\rfs')}}}}{}\\
      {\conf{\frame {\cscall \expr {\expr_1, \ldots, \expr_n};\cmd'}{\regmap'} {\km[\syscall]}:\st, \overline{\bm{\buf''}{(\lay\lcomp\rfs''')}}}}.
    \end{equation*}
    By assumption, we know that the rule \ref{SI:Call-Step-Unsafe} has
    been applied to show the transition in the speculative
    semantics. Upon introspection of that rule, we conclude that its
    premises are also satisfied by the configuration
    \[
      {\conf{\frame {\cscall \expr {\expr_1, \ldots, \expr_n};\cmd'}{\regmap'} {\km[\syscall]}:\st, \overline{\bm{\buf''}{(\lay\lcomp\rfs''')}}}}.
    \]
    This shows that \ref{WL:SCall-Unsafe} applies to the transition
    \begin{equation*}
      \nstep[\system][\lay] {\nat'} 
      {\conf{\frame {\cscall \expr {\expr_1, \ldots, \expr_n};\cmd'}{\regmap'} {\km[\syscall]}:\st, \overline{\bm{\buf''}{(\lay\lcomp\rfs''')}}}}
      \unsafe,
    \end{equation*}
    and establishing the claim.
  \end{proofcases}
\end{proof}

\subsubsection{Technical Observations on the Speculative Semantics}

We begin by defining a predicate $\wf{\cdot}$ that captures the
well-formedness of a buffered memory with respect to a layout
$\lay$. For a system $\system = (\rfs, \syss, \caps)$ and a layout
$\lay$, this predicate is defined as follows:

\[
  \infer{\wf{\sconf{\err, \boolms}}}{}
  \quad\infer{\wf\unsafe}{}
  \quad\infer{\wf \nil}{}
  \quad\infer{\wf{\confone:\cfstack} }{\wf \confone &
    \wf \cfstack}
\]

\[
  \infer{\wf{\sframe\st {\bm\buf\mem} \boolms}}{\dom(\buf)\subseteq \lay(\Ar) & \mem = \lay \lcomp \rfs' & \rfs' \eqon{\Fn} \rfs &
  \st = \st_{\km}:\st_\um & \km(\st_{\km}) & \um(\st_\um)}
\]

\begin{remark}
  \label{rem:simispreservedspec}
  For every system $\system = (\rfs, \syss, \caps)$,  layout $\lay\in \Lay$, pair of configurations $\cfstack$ and $\cfstack'$, sequences of directives  $\Ds$ and of observations $\Os$, if $\wf\cfstack$ and $\nsstep*{\cfstack}{\cfstack'} \Ds \Os$, then we have $\wf{\cfstack'}$.
\end{remark}
\begin{proof}
  The proof goes by induction on the number of steps.
  \begin{proofcases}
    \proofcase{0} Trivial.
    \proofcase{$\nat + 1$}
    The claim is established through a case analysis of the rule used in the last transition. In most cases, the result follows directly from the IH. Therefore, we focus on the most relevant cases.
    \begin{proofcases}
      \proofcase{\ref{SI:Pop}}
        Since the target configuration contains a frame stack that is a suffix of the source configuration's frame stack, the conclusion directly follows from the IH.
      \proofcase{\ref{SI:Fence}} The configuration reached after $\nat$ steps can be expressed as follows:
    \begin{equation*}
      {\sframe{\frame {\cfence\sep\cmdtwo} \regmap \opt:\st}{\bm{\buf}{\lay\lcomp{\rfs'}}} \boolms\cons\cfstack''}.
    \end{equation*}
    By the IH, we know that this configuration satisfies $\wfp$. The target configuration obtained from this step has the following form:
    \begin{equation*}
      {\sframe{\frame {\cmdtwo} \regmap \opt:\st}{\bm \nil {\overline {\bm{\buf}{\lay\lcomp{\rfs'}}}}} \boolms\cons\cfstack''}.
    \end{equation*}
    The claim then follows from the IH and \Cref{rem:overlinewrtdom}.
    \proofcase{\ref{SI:Store}} We rewrite the configuration reached
      after $\nat$ steps as follows:
    \begin{equation*}
      {\sframe{\frame {\cmemass \expr \exprtwo\sep\cmdtwo} \regmap \opt:\st}{\bm{\buf}{\lay\lcomp{\rfs'}}} \boolms\cons\cfstack''}
    \end{equation*}
    By the IH, we know that this configuration satisfies the predicate $\wfp$. The target configuration after the transition takes the form:
    \begin{equation*}
      {\sframe{\frame {\cmdtwo} \regmap \opt:\st}{\bm{\bitem{\toAdd{\sem \expr_{\regmap, \lay}}}{\sem \exprtwo_{\regmap, \lay}}\cons \buf}{\lay\lcomp{\rfs'}}} \boolms\cons\cfstack''}.
    \end{equation*}
    From the premises of the rule, we have $\toAdd{\sem \expr_{\regmap, \lay}} \in \underline \lay(\Ar[\opt])$. The remaining conditions required to establish $\wfp$ follow directly from the IH.
    \proofcase{\ref{AL:Call}} The configuration reached after $\nat$ steps rewrites as follows:
    \begin{equation*}
      {\sframe{\frame {\ccall \expr{\expr_1,\dots,\expr_k}\sep\cmdtwo} \regmap \opt:\st}{\bm{\buf}{\lay\lcomp{\rfs'}}} \boolms\cons\cfstack''}.
    \end{equation*}
    According to the IH, this configuration satisfies $\wfp$. The
    target configuration after the call looks as follows:
    \begin{equation*}
      \sframe{\frame {\lay\lcomp{\rfs'}(\add)} {\regmap'} \opt:\frame {\cmdtwo} \regmap \opt:\st}{\bm{\buf}{\lay\lcomp{\rfs'}}} \boolms\cons\cfstack''
    \end{equation*}
    for some $\regmap'$ and with
    $\add = \toAdd{\sem \expr_{\regmap, \lay}}$. From the rule's
  premises, we know $\add \in \underline \lay(\Fn[\opt])$, indicating
  the existence of a function $\fn \in \Fn[\opt]$ such that
  $\lay(\fn) = \add$. Using the definition of $\cdot\lcomp \cdot$, we
  obtain $\lay\lcomp{\rfs'}(\add) = {\rfs'}(\fn)= \rfs(\fn)$, where
  the last equation is a consequence of the IH.  Next, we analyze the
  cases based on $\opt$.  If $\opt=\um$, the function body is
  unprivileged, which proves $\um(\lay\lcomp{\rfs'}(\add))$, as
  required.  If $\opt=\km$, we observe that
  $\lay\lcomp{\rfs'}(\add) \in \Cmd$, thus showing
  $\km[\syscall](\lay\lcomp{\rfs'}(\add))$, as needed.  The remaining
  conditions follow directly from the IH.

  \proofcase{\ref{AL:System-Call}} We begin by applying the IH, which establishes that the configuration reached after $\nat$ steps satisfies $\wfp$:
    \begin{equation*}
      \wf {\sframe{\frame {\csyscall \syscall{\expr_1,\dots,\expr_k}\sep \cmdtwo} \regmap {\opt}:\st}{\bm{\buf'}{\lay\lcomp{\rfs'}}} \boolms\cons\cfstack''}.
    \end{equation*}
    The proof proceeds by case analysis on $\opt$.
    \begin{proofcases}
      \proofcase{$\opt=\um$} From the IH we have:
      \[
        \um({\frame {\csyscall \syscall{\expr_1,\dots,\expr_k}\sep \cmdtwo} \regmap \um:\st}).
      \]
      The target configuration takes the form:
    \begin{equation*}
      {\sframe{\frame {\syss(\syscall)} {\regmap'} {\km[\syscall]}\cons\frame {\cmdtwo} \regmap {\um}\cons\st}{\bm{\buf'}{\lay\lcomp{\rfs'}}} \boolms\cons\cfstack''}.
    \end{equation*}
    for some $\regmap'$. To complete the proof, it suffices to observe that $\km(\frame {\syss(\syscall)} {\regmap'} {\km[\syscall]})$, which follows from the definition of $\syss$.
    \proofcase{$\opt=\km[\syscalltwo]$} By the IH, we know that $\frame {\csyscall \syscall{\expr_1,\dots,\expr_k}\sep \cmdtwo} \regmap {\km[\syscalltwo]} \cons \st = \frame {\syss(\syscall)} {\regmap'} {\km[\syscalltwo]}\cons\frame {\cmdtwo} \regmap {\km[\syscalltwo]}\cons \st_{\km}\cons\st_\um$ with and $\um(\st_\um)$ and $\km[\syscalltwo](\st_{\km})$. The target configuration is:
    \begin{equation*}
      {\sframe{\frame {\syss(\syscall)} {\regmap'} {\km[\syscalltwo]}\cons\frame {\cmdtwo} \regmap {\km[\syscalltwo]}\cons\st_{\km}\cons\st_\um}{\bm{\buf'}{\lay\lcomp{\rfs'}}} \boolms\cons\cfstack''}.
    \end{equation*}
   Observe that $\km[\syscalltwo](\frame {\syss(\syscall)} {\regmap'} {\km[\syscalltwo]}\cons\frame {\csyscall \syscall{\expr_1,\dots,\expr_k}\sep \cmdtwo} \regmap {\km[\syscalltwo]}\cons\st_{\km})$ holds, establishing the claim.
    \end{proofcases}
    \end{proofcases}
  \end{proofcases}
\end{proof}

\begin{remark}
  \label{rem:stackinv0}
  For every system $\system = (\rfs, \syss, \caps)$, every store $\rfs \eqon{\Fn} \rfs'$, program configuration $\confone = \conf{\st_{\km} : \st_\um, \lay \lcomp \rfs'}$ where $\km(\st_{\km})$ and $\um(\st_\um)$, and configuration $\conf{\st', \lay \lcomp \rfs''}$ that is reachable in $\nat$ steps from $\confone$, there exist a pair of stacks $\st_{\km}', \st_\um'$ such that $\st' = \st_{\km}' : \st_\um'$, and $\km(\st_{\km}')$ and $\um(\st_\um')$ hold.
\end{remark}

\begin{proof}
  We prove this by induction on the number of steps, $\nat$.
  \begin{proofcases}
    \proofcase{0} The base case is trivial.

    \proofcase{$\nat + 1$} The proof proceeds by cases on the rule applied to establish the last transition. Most rules are straightforward, so we focus on the more interesting cases:

    \begin{proofcases}
      \proofcase{\ref{WL:Pop}} The target configuration contains a stack that is a suffix of the stack in the source configuration. Hence, the conclusion follows trivially from the IH.

      \proofcase{\ref{WL:SystemCall}} In this case, we first apply the IH, which shows:
      \[
        \nstep \nat {\conf{\st_{\km} : \st_\um, \lay \lcomp \rfs}}{} \conf{\frame {\csyscall \syscall{\expr_1, \dots,\expr_k} \sep \cmdtwo} {\regmap'} \opt \cons \st', \lay \lcomp \rfs}
      \]
      and that the stack of frames $\frame {\csyscall \syscall{\expr_1, \dots,\expr_k}} {\regmap'} \opt \cons \st'$ is a concatenation of a pair of stacks $\st_{\km}', \st_\um'$ such that $\km(\st_{\km}')$ and $\um(\st_\um')$ hold. Now, consider the next transition:
      \[
        \step
        {\conf{\frame {\csyscall \syscall{\expr_1, \dots,\expr_k} \sep \cmdtwo} {\regmap'} \opt \cons \st_{\km}' \cons \st_\um', \lay \lcomp \rfs}}{}
        {\conf{\frame {\syss(\syscall)} {\regmap''} {\km[\syscalltwo]} : \frame {\cmdtwo} {\regmap'} \opt \cons \st_{\km}' \cons \st_\um', \lay \lcomp \rfs}}
      \]
      for an appropriate $\regmap''$. The claim requires verifying that the target configuration satisfies the desired properties. In particular, if $\opt = \km[\syscalltwo]$, then $\syscalltwo = \syscall$, and $\km(\frame {\syss(\syscall)} {\regmap''} {\km[\syscall]} \cons \frame {\cmdtwo} {\regmap'} \opt \cons \st_{\km}')$ holds trivially by our case analysis rule \ref{WL:SystemCall}. Otherwise, if $\opt = \um$, then $\st_{\km}' = \nil$ by the IH, and we simply need to observe that $\km(\frame {\syss(\syscall)} {\regmap''} {\km[\syscall]})$ holds.

      \proofcase{\ref{WL:Call}} This case is analogous to the previous one. Instead of setting the flag to $\km$ in the new frame, the rule copies it from the topmost frame of the source configuration, which does not invalidate the invariant. The target configuration looks as follows:
      \[
      \langle \frame {\lay \lcomp \rfs''(\add)} {\regmap_0[\vx_1 \upd \sem{\exprtwo_1}_{\regmap,\lay}, \dots, \vx_n \upd \sem{\exprtwo_n}_{\regmap,\lay}]} {\opt} :
      \frame {\cmd} {\regmap} {\opt} : \st, \lay \lcomp \rfs'' \rangle.
      \]
      From \Cref{rem:simispreservedspec}, we deduce that $\rfs'' \eqon{\Fn} \rfs$. Therefore, we have $\lay \lcomp \rfs''(\add)$, because from the premises of the rule, we know that $\add \in \underline \lay(\Fn[\opt])$. This implies that there exists a function $\fn \in \Fn[\opt]$ such that $\lay(\fn) = \add$. Using this observation and the definition of $\cdot \lcomp \cdot$, we deduce that $\lay \lcomp \rfs''(\add) = \rfs(\fn)$. Now, we proceed by cases on $\opt$. If $\opt = \um$, the body of the function is unprivileged, and this shows that $\um(\lay \lcomp \rfs''(\add))$ holds. If $\opt = \km$, we observe that the program is in $\Cmd$, which shows that $\km[\syscall](\lay \lcomp \rfs''(\add))$ holds, as required.
    \end{proofcases}
  \end{proofcases}
\end{proof}

\begin{lemma}
  \label{rem:hasu}
  Let $\mathsf{hasu}$ be defined on frame stacks as follows:
  \begin{align*}
    \mathsf{hasu}(\nil)&\defsym \bot & \mathsf{hasu}(f\cons\st)&\defsym \um(f) \lor \mathsf{hasu}(\st),
  \end{align*}
  and on configuration stacks as follows:
  \begin{align*}
    \mathsf{hasu}(\nil)&\defsym \top & \mathsf{hasu}(\conf{\st, \bmvar, \boolms} \cons \cfstack)&\defsym \um(f) \land \mathsf{hasu}(\cfstack).
  \end{align*}
  Let $\cfstack, \cfstack'$ be configuration stacks, $\dir$ a directive and $\obs$ an observation,
  if we have $\wf \cfstack, \mathsf{hasu}(\cfstack)$ and
  \[
    \sstep \cfstack {\cfstack'} \dir \obs, 
  \]
  then $\mathsf{hasu}(\cfstack')$ holds.
\end{lemma}
\begin{proof}
  The claim follows by case analysis on the transition rules. 
\end{proof}

\begin{remark}
  \label{rem:stackinv1}
  Let $\system = (\rfs, \syss, \caps)$ be a system. For every program configuration $\confone = \conf{\st_{\km} : \st_\um, \lay \lcomp \rfs', \Os, \Ds}$ where $\rfs \eqon{\Fn} \rfs'$, $\km(\st_{\km})$ and $\um(\st_\um)$, configuration $\conf{\st', \lay \lcomp {\rfs''}, \Os', \Ds'}$, and hybrid configuration $\hconf \cfstack {\st'}{\Os'}{\Ds'}$, that are reachable in $\nat$ steps from $\confone$, there exist a pair of stacks $\st_{\km}', \st_\um'$ such that $\st' = \st_{\km}' : \st_\um'$, and $\km(\st_{\km}')$, $\um(\st_\um')$ hold. Moreover, if the configuration is hybrid, we also have $\wf \cfstack$, $\um(\st')$, and all the frame stacks within $\cfstack$ contain at least a frame $f$ such that $\um(f)$ holds, i.e. $\mathsf{hasu}(\cfstack)$ (see \Cref{rem:hasu}).
\end{remark}
\begin{proof}
  The proof is by induction on the reduction step, and by case
  analysis on the rule used for the last transition step.  This proofs
  for the cases not dealing with hybrid configurations are analogous
  to the corresponding one \Cref{rem:stackinv0}.  Therefore, in the
  following, we only focus on the cases where one of the two configurations is hybrid.
  \begin{proofcases}
    \proofcase{\ref{AL:Spec-Init}} We rewrite the reduction as follows:
    \[
      \nesstep \nat \confone  \conf{\frame {\cspec \cmd;\cmdtwo} \regmap \opt\cons \st'', \lay \lcomp {\rfs''},\Ds', \Os'} \to \hconf{\conf{\frame \cmd \regmap \opt, \bm \nil {\lay \lcomp {\rfs''}, \bot}}} {\frame {\cmdtwo} \regmap \opt\cons \st''}{\Ds'} {\Os'} 
    \]
    From the IH, we deduce that
    $\um{\frame {\cspec \cmd;\cmdtwo} \regmap \opt\cons \st''}$. In
    particular, $\km(\frame {\cspec \cmd;\cmdtwo} \regmap \opt)$
    cannot hold as $\cspec \cmd \notin \Cmd$, proving all the
    additional claims we needed.
    \proofcase{\ref{AL:Spec-Dir}, \ref{AL:Spec-Step}, \ref{AL:Spec-Bt}} The claim is a direct consequence of the IH and \Cref{rem:hasu}.
    \proofcase{\ref{AL:Spec-Term}} The claim is a direct consequence of the IH. In particular, from the IH, we deduce that the source configuration of the last transition has the following form:
    \[
      \hconf{\conf{\frame \cnil \regmap \um, \bmvar, \bot}} {\st''}{\Ds'} {\Os'},
  \]
  where, in particular the mode-flag is $\um$. By combining this observation with the IH, we can deduce that the frame stack of the target configuration satisfies the predicate $\um(\cdot)$, establishing the claim.
  \end{proofcases}
\end{proof}

\begin{remark}
  \label{rem:speconum}
  For every system $\system = (\overline \rfs, \syss, \caps)$, number of steps $\nat$,
  configurations $\conf{\st_{\km}\cons\st_\um, {\lay \lcomp \rfs}, \Ds, \Os}$ and $\conf{\st, {\lay \lcomp \rfs'}, \Ds', \Os'}$
  such that $\rfs \eqon{\Fn}\overline \rfs$,
  if
  \[
    \nesstep \nat {\conf{\st_{\km}\cons\st_\um, {\lay \lcomp \rfs}, \Ds, \Os}}
    {\conf{\frame{\cspec \cmd;\adversary'}{\regmap'}{\opt}:\st, {\lay \lcomp \rfs'}, \Ds', \Os'}},
  \]
  then $\um(\frame{\cspec \cmd;\adversary'}{\regmap'}{\opt}:\st)$.
\end{remark}
\begin{proof}
  This result is a direct consequence of \Cref{rem:stackinv1}. According to the invariant property stated there, if the top frame on the stack of the target configuration is a kernel-mode stack, it implies that $\cspec \cmd;\adversary'\in\Cmd$, thus leading to a contradiction. Therefore, the claim holds.
\end{proof}

\begin{remark}
  \label{rem:noaddobs}
  Let  $\system=(\rfs, \syss, \caps)$ be a system. For every pair of speculative stacks $\cfstack, \cfstack'$, sequences of directives and observations $\Ds, \Os$, layout $\lay$, and number of steps $\nat$ such that:
  \[
    \nsstep \nat {\cfstack} {\cfstack'} \Ds {\Os},
  \]
   if $\add$ is an address such that $\add \notin \underline\lay(\Id)$, then $\Os$ does not contain the observations $\omem \add$ and $\ojump \add$.
 \end{remark}

 \begin{proof}
  We proceed by induction on $\nat$.
  \begin{proofcases}
    \proofcase{0} The base case is trivial.
    \proofcase{$\nat+1$} We apply IH and analyze the rule used in the last transition. For rules that do not produce the observations $\omem \add'$ and  $\ojump \add'$ for some $\add'$, the claim holds trivially.
    The rules that may produce such a transition are \ref{SI:Load-Step}, \ref{SI:Load}, \ref{SI:Load-Unsafe}, \ref{SI:Store}, \ref{SI:Store-Unsafe}, \ref{SI:Call}, and \ref{SI:Call-Unsafe}. Each of these rules has premises requiring that $\add' \in \underline\lay(\Ark)$, $\add' \in \underline\lay(\Fnk)$, $\add' \in \underline\lay(\Aru)$, or $\add' \in \underline\lay(\Fnu)$. 
    However, since it is given that $\add \notin \underline\lay(\Id)$, it follows that $\add$ does not belong to any of these sets. Consequently, we conclude that $\add' \neq \add$.
  \end{proofcases}
\end{proof}


\begin{remark}
  \label{rem:boolmsinc}
  Let $\sqsubseteq$ be the smallest partial order such that $\bot \sqsubseteq \top$. For every number of steps $\nat$, speculative configuration $\confone$, speculative stack $\cfstack$, and $\Ds$ without $\dbt$ directives, if there is a target configuration such that:
  \[
    \sstep[\system][\lay] {\sframe{\st}{\bmvar}{\boolms}:\cfstack} {{\sframe{\st'}{\bmvar}{\boolms'}}:\cfstack'}{\Ds} {\Os},
  \]
   then $\boolms \sqsubseteq \boolms'$.
\end{remark}
\begin{proof}
  The proof goes by induction on $\nat$. The base case is trivial, as $\sqsubseteq$ is reflexive. The inductive claim follows from the IH and by introspection of the rules.
\end{proof}

\begin{remark}
  \label{rem:thereisaspecon}
  Let $\system = (\rfs, \syss, \caps)$ be a system. Consider any configuration of the form  
  \[
    \conf{\frame \speccmd \regmap \um, \lay \lcomp \rfs, \nil, \nil}
  \]
  and any number of steps $\nat$, as well as sequences of directives and observations $\Ds$ and $\Os$. If  
  \[
    \nesstep \nat {\conf{\frame \speccmd \regmap \um, \lay \lcomp \rfs, \nil, \nil}} {\hconf\cfstack \st \Ds \Os},
  \]  
  then there exist $\nat' \le \nat$, a configuration  
  \[
    \sframe{\frame{\cmd'}{\regmap'}{\um}}{\lay\lcomp\rfs'}{\bot},
  \]  
  and sequences $\Ds'$ and $\Os'$ such that  
  \[
    \nsstep {\nat'}
    {\sframe{\frame{\cmd'}{\regmap'}{\um}}{\lay\lcomp\rfs'}{\bot}}
    \cfstack {\Ds'} {\Os'}.
  \]
\end{remark}

\begin{proof}
  The proof proceeds by induction on $\nat$.
  \begin{proofcases}
    \proofcase{0} This case follows from vacuity of the premise.
    \proofcase{$\nat+1$} The premise is:
      \[
        \nesstep {\nat} {\conf{\frame \speccmd \regmap \um, \lay \lcomp \rfs, \nil, \nil}} \conftwo \ato  {\hconf \cfstack \st \Ds \Os} 
      \]
      We proceed by cases on the rule that was applied to show the last transition, showing only some representative cases.
      \begin{proofcases}
        \proofcase{\ref{AL:Spec-Init}} We observe that
        \[
          \conftwo = \conf{\frame{\cspec {\cmd''}\sep \cmdtwo}{  \regmap''}\opt\cons \st, \lay\lcomp\rfs''}.
        \]
        The claim follows by introspection of the rule.  In
        particular, $\opt = \um$ follows from \Cref{rem:stackinv1}. It
        cannot be that $\opt=\km[\syscall]$ because, in that case, we
        the configuration could not contain the command $\cspec {\cmd'}\sep \cmdtwo$.
        \proofcase{\ref{AL:Spec-Step}} The IH proves the existence of the following reduction.
        \[
          \nsstep {\nat'}
          {\sframe{\frame{\cmd'}{\regmap'}{\um}}{\lay\lcomp\rfs'}{\bot}}
          \cfstack    {\Ds'} {\Os'}.
        \]
        from the premise of the rule, we conclude
        \[
          \sstep
          \cfstack {\cfstack'}    {\dstep} \obs,
        \]
        which establishes the claim.
      \end{proofcases}
  \end{proofcases}
\end{proof}

\begin{remark}
  \label{remark:nobttech}
  Let $\system$ be a system, $\lay$ a layout and $\confone$ a
  configuration. For every target stack
  $\cfstack$ such that $\cfstack = \specconftwo\cons\cfstack'$, for some non-empty $\cfstack'$,
  every sequences of directives $\Ds$ and
  observations $\Os$, number of steps $\nat$, if 
  \[
    \nsstep\nat  {\specconfone} {\cfstack}\Ds \Os,
  \]
  then there are $\nat'<\nat$, ${\Ds'}$, ${\Os'}$ such that: 
  \[
    \nsstep{\nat'} {\specconfone} {\cfstack'} {\Ds'} {\Os'}.
  \]
\end{remark}
\begin{proof}
  By induction on $\nat$.
  \begin{proofcases}
    \proofcase{$\nat=0$} The claim we need to establish in this derivation is:
    for every $\cfstack = \specconftwo\cons\cfstack'$ for some non-empty $\cfstack'$,
    sequences of directives $\Ds$ and
    observations $\Os$, if 
    \[
      \nsstep0  {\specconfone} {\cfstack}\Ds \Os,
    \]
    there are $\nat'<0$, ${\Ds'}$, ${\Os'}$ such that: 
    \[
      \nsstep{\nat'} {\specconfone} {\cfstack'} {\Ds'} {\Os'}.
    \]
    The conclusion follows by vacuity of the premises, in particular,
    it cannot be that
    $\specconfone \cons \nil = \cfstack = \specconftwo\cons\cfstack'$ with
    $\cfstack'$ non-empty.
    
    \proofcase{$\nat+1$} Here, the claim is the following: for every
    $\cfstack$ such that $\cfstack = \specconftwo\cons\cfstack'$ for
    some non-empty $\cfstack'$, sequences of directives $\dir:\Ds$ and
    observations $\obs:\Os$, if
    \[
      \nsstep\nat  {\specconfone} {\cfstack''}\Ds \Os\sto \dir \obs {\cfstack},
    \]
    there are $\nat'<\nat+1$, ${\Ds'}$, ${\Os'}$ such that: 
    \[
      \nsstep{\nat'} {\specconfone} {\cfstack'} {\Ds'} {\Os'}.
    \]
    We go by cases on the directive used for the last step:
    \begin{proofcases}
      \proofcase{$\dstep$} Observe that $\cfstack'' = \specconftwo'\cons\cfstack'$ for some $\specconftwo'$, so the claim is a consequence of the IH.
      
      \proofcase{$\dload i, \dbranch b, \djump \add$} The witness we need to introduce is the $\nat$ step transition 
      \[
        \nsstep\nat {\specconfone} {\cfstack''}\Ds \Os.
      \]

      \proofcase{$\dbt$} We go by cases on the rule that has been applied. It must be one of \ref{SI:Backtrack-Top}, \ref{SI:Backtrack-Bot}.
      \begin{proofcases}
        \proofcase{\ref{SI:Backtrack-Top}} In this case, we go by cases on the $\nat+1$-th target stack. If it is empty or it has one element, the claim holds by vacuity of the premise. If it has more than one element, the claim follows by two applications of the IH.
        
        \proofcase{\ref{SI:Backtrack-Bot}} The claim holds by vacuity of the premise: the $\nat+1$-th target stack has just one element.
      \end{proofcases}
    \end{proofcases}
  \end{proofcases}
\end{proof}

\begin{remark}
  \label{lemma:misspecunset}
  For every $\nat$, initial configuration $\specconfone=\sframe{\frame \cmd \regmap \um} {\mem} \bot$, if:
  \[
    \nsstep \nat  {\specconfone} {\specconftwo} {\Ds} {\Os},
  \]
  and $\specconftwo$ has a mis-speculation flag ($\specconftwo \neq \unsafe$) then the mis-speculation flag of $\specconftwo$ must be $\bot$.
\end{remark}
\begin{proof}
  The proof proceeds by induction. The base case is trivial. In the inductive step, we perform a case analysis on the directive used for the last transition:
  \begin{proofcases}
    \proofcase{$\dstep$} By inspecting this fragment of the semantics, we deduce that the $\nat$-th and the $\nat+1$-th target configurations have the same number of configurations. Therefore, the $\nat$-th target stack must contain exactly one configuration. As a result, we can apply the IH  and deduce that the mis-speculation flag of the $\nat$-th target configuration is $\bot$. Furthermore, we observe that for all the rules producing a transition with the $\dstep$ directive, the flag of the source configuration and the target configuration remains the same. Hence, we conclude the proof for this case.
    
    \proofcase{$\dload i, \dbranch b, \djump \add$} For all the rules that match these directives, the claim holds due to the vacuity of the premise: the $\nat+1$-th target configuration stack has a height greater than $1$. The only exception is \ref{SI:Load-Err}. In this case, the proof follows the same reasoning as the case where the rule uses the $\dstep$ directive, which has already been considered in the previous step.  
    \proofcase{$\dbt$} We go by case analysis on the applied rule, which must be one of \ref{SI:Backtrack-Top} or \ref{SI:Backtrack-Bot}.
    \begin{proofcases}
          \proofcase{\ref{SI:Backtrack-Bot}} In this case, the claim follows directly from the definition of the rule.
          \proofcase{\ref{SI:Backtrack-Top}} Let $\cfstack_\nat$ denote the $\nat$-th target configuration stack. Note that its height can neither be $0$ nor $1$. In the first case, there would be no subsequent transition, while in the second case, applying the IH would imply that the mis-speculation flag is $\bot$, contradicting the assumption about the applied rule.
          If the height of $\cfstack_\nat$ is greater than $2$, then the height of the $\nat+1$-th target stack is greater than $1$, so the claim holds due to the vacuity of the premise. 
          Finally, if the height is exactly $2$, then $\cfstack_\nat$ must have the shape $\specconftwo_\nat \cons \specconftwo_{\nat}'$. By \Cref{remark:nobttech}, we know that there is a sequence of transitions from $\specconfone$ to $\specconftwo_\nat'$ of length $\nat' \le \nat$. By applying the IH to this sequence, we can show that $\specconftwo_\nat'$ has the mis-speculation flag unset. We conclude by noting that the $\nat+1$-th target configuration stack is precisely $\specconftwo_\nat'$.  
    \end{proofcases}
  \end{proofcases}
\end{proof}

\begin{lemma}
  \label{lemma:nobtsteponly}
  Let $\nat$ be a natural number, $\confone$ a configuration, $\cfstack$ a stack, and $\Ds$, $\Os$ sequences of directives and observations, respectively. Suppose that:
  \[
    \nsstep{\nat} {\confone} {\conftwo : \cfstack} {\Ds} {\Os}.
  \]
  If the mis-speculation flag of $\conftwo$ is $\bot$, then there exist a sequence of observations $\overline \Os$ and a natural number $\nat' \le \nat$ such that:
  \[
    \nsstep{\nat'} {\confone} {\conftwo} {\dstep^{\nat'}} {\overline \Os}.
  \]
\end{lemma}

\begin{proof}
  The proof proceeds by induction on $\nat$. The base case is trivial. 
  For the inductive case, we consider the directive used in the $\nat+1$-th transition.
  The premise tells us that:
  \[
    \nsstep{\nat} {\confone} {\conftwo: \cfstack} {\Ds} {\Os} \sto \dir \obs {\conftwo': \cfstack'}  
  \]
  and we know that the mis-speculation flag of $\conftwo'$ is $\bot$. 
  Our goal is to show that there exist $\overline \Os'$ and $\nat'' \le \nat+1$ such that:
  \[
    \nsstep{\nat''} {\conftwo} {\conftwo'} {{\dstep^{\nat''}}} {{\overline \Os}}
  \]

  \begin{proofcases}
    \proofcase{$\dstep$} 
    Since the directive is $\dstep$, we observe that the mis-speculation flag of $\conftwo$ must be $\bot$, because no rule associated with $\dstep$ changes this flag. 
    Therefore, we can directly apply the IH, stating that:
    \[
      \nsstep{\nat'} {\confone} {\conftwo} {{\dstep^{\nat'}}} {{\overline \Os}}
    \]

    Next, we examine the rules matching the transition:
    \[
      \sstep {\conftwo: \cfstack} {\conftwo': \cfstack'} \dir \obs
    \]
    By analyzing all the applicable rules, we note that their premises
    depend solely on $\conftwo$ and not on $\cfstack$. Consequently,
    if any of these rules applies to the transition above, it must
    also apply to the configuration $\conftwo$. A case analysis of
    these rules reveals that the resulting configuration is precisely
    $\conftwo''$.
    \proofcase{$\dload i$} The only two rules matching
    this directive, given our premises, are \ref{SI:Load-Err} and
    \ref{SI:Load}. We only take in exam the latter rule, as the former
    one can be reduced using the same rule with the $\dstep$
    directive.
    If the rule \ref{SI:Load} is applied, we can express the transition  
    \[
      \sstep {\conftwo: \cfstack}  {\conftwo': \cfstack'}  \dir \obs
    \]  
    as follows:
    \begin{multline*}
      \sstep {\sframe{\frame{\cmemread[\lbl] \vx  \expr\sep\cmd}{\regmap}{\opt}\cons \st} {\bm\buf\mem}{\boolms} \cons \cfstack} {}{{{\dload[\lbl] i}}} {\omem \add} \\
      \sframe{\frame{\cmd}{\update \regmap x \val}{\opt}\cons\st}{\bm\buf\mem}{\boolms\lor\bool'} \cons
      \sframe{\frame{\cmemread[\lbl] \vx  \expr\sep\cmd}{\regmap}{\opt}\cons\st}{\bm\buf\mem}{\boolms} \cons \cfstack
    \end{multline*}
    From the premises of the rule, we deduce that the following conditions hold:
    \begin{varitemize}
    \item $\toAdd{\sem\expr_{\regmap, \lay}} = \add$
    \item $\bufread {\bm\buf\mem} \add i = \val, \bool'$
    \item $\add \in \underline \lay(\Ar[\opt])$
    \item If $\opt = \km[\syscall]$, then $\add \in \underline \lay(\caps(\syscall))$
    \end{varitemize}
    In particular, from the main premise of the claim, we know that $\bool'$ cannot be $\top$. Hence, it follows that $\bufread {\bm\buf\mem} \add i = \val, \bot$. 
    Our goal is to show that the following transition holds:
    \begin{equation*}
      \sstep {\sframe{\frame{\cmemread[\lbl] \vx  \expr\sep\cmd}{\regmap}{\opt}\cons \st} {\bm\buf\mem}{\boolms}} {}{{{\dstep}}} {\omem \add}
      \sframe{\frame{\cmd}{\update \regmap x \val}{\opt}\cons\st}{\bm\buf\mem}{\boolms\lor\bool'}
    \end{equation*}
    To establish this, we note that if we can show $\bufread {\bm\buf\mem} \add 0 = \val, \bot$, the premises of the rule \ref{SI:Load-Step} are satisfied, establishing the transition the proof. The conclusion $\bufread {\bm\buf\mem} \add 0 = \bufread {\bm\buf\mem} \add i=\val, \bot$  follows from \Cref{remark:onbuflookup}.
    \proofcase{$\djump \add$} The only rules matching
    this directive, given our premises, are \ref{SI:Call-Err} and
    \ref{SI:Call}. We only take in exam the latter rule, as proof for the former
    one is analogous to the latter.
    If the rule \ref{SI:Call} is applied, we can express the transition  
    \[
      \sstep {\conftwo: \cfstack}  {\conftwo': \cfstack'}  \dir \obs
    \]  
    as follows:
    \begin{multline*}
      \lay \red \sframe{\frame{\ccall \expr {\vec\exprtwo}\sep\cmd}{\regmap}{\opt}\cons\st}{\bmem}{\boolms}\cons\cfstack
      \sto{\djump \add}{\ojump \add} \\\sframe{\frame{\mem(\add)} {\regmap_0'}{\opt}\cons\frame{\cmd}{\regmap}{\opt} \cons \st}{\bmem}{\boolms \lor (\add \neq \add')}\cons \sframe{\frame{\ccall \expr {\vec\exprtwo}\sep\cmd}{\regmap}{\opt}\cons\st}{\bmem}{\boolms}\cons \cfstack
    \end{multline*}
    From the premises of the rule, we deduce that the following properties hold:
    \begin{varitemize}
    \item $\toAdd{\sem\expr_{\regmap, \lay}} = \add$
    \item $\add \in \underline \lay(\Ar[\opt])$
    \item If $\opt = \km[\syscall]$, then $\add \in \underline \lay(\caps(\syscall))$
    \end{varitemize}
    In particular, from the main premise of the claim, we know that $\add= \add'$. Our goal is to show that the following transition holds:
    \begin{equation*}
      \lay \red \sframe{\frame{\ccall \expr {\vec\exprtwo}\sep\cmd}{\regmap}{\opt}\cons\st}{\bmem}{\boolms}
      \sto{\djump \add}{\ojump \add} \\\sframe{\frame{\mem(\add)} {\regmap_0'}{\opt}\cons\frame{\cmd}{\regmap}{\opt} \cons \st}{\bmem}{\boolms \lor (\add \neq \add')}
    \end{equation*}
    To establish this, we simply observe that premises of the rule \ref{SI:Call-Step} are satisfied, establishing the transition the proof.
    \proofcase{$\dbranch \bool$} Analogous to the previous case.
    \proofcase{$\dbt$}
    If the applied rule is \ref{SI:Backtrack-Bot}, the claim directly follows from the IH. 
    Otherwise, the applied rule is \ref{SI:Backtrack-Top}. We go by case analysis on the height of the configuration stack $\conftwo: \cfstack$.
    If the height is $0$, we reach a terminal configuration, which is a contradiction.
    If the height is $1$, then the topmost configuration must have its mis-speculation flag unset due to \Cref{lemma:misspecunset}. This contradicts the premise of the applied rule, which requires the flag to be set to $\top$. 
    If the height is greater than $1$, we can express the configuration stack as $\conftwo: \cfstack = \conftwo \cons \conftwo' \cons \cfstack'$. By introspection of the rule and the main assumption of this claim, we deduce that the mis-speculation flag of $\conftwo'$ is $\bot$. 
    Applying \Cref{remark:nobttech} to the $\nat$-step reduction from $\confone$ to $\conftwo \cons \conftwo' \cons \cfstack'$, we establish the existence of a $\nat' \le \nat$ reduction from $\confone$ to $\conftwo' \cons \cfstack'$. 
    By the IH, there exists a reduction of length $\nat'' \le \nat'$ from $\confone$ to $\conftwo'$, which uses only the $\dstep$ directive. This completes the proof.
  \end{proofcases}
\end{proof}

\begin{lemma}
  \label{lemma:nobt}
  For every system $\system = (\rfs, \syss, \caps)$,
  speculative configuration $\specconfone = \sframe {\frame \cmd \regmap \opt} {\mem} \bot$
  non-empty speculative stack $\cfstack$,
  sequence of directives $\Ds$ and
  sequence of observations $\Os$,
  if:
  \[
    \nsstep \nat \specconfone {\cfstack} \Ds \Os,
  \]
  then there are $\nat'\le\nat$, $\Ds'$, $\Os'$
  such that 
  \[
    \nsstep {\nat'} \specconfone {\cfstack} {\Ds'} {\Os'}.
  \]
  and $\Ds'$ does not contain any $\dbt$ directive.
\end{lemma}
\begin{proof}
  By induction on $\nat$.
  \begin{proofcases}
    \proofcase{0} This case is trivial.
    \proofcase{$\nat+1$} The premise states that
    \[
      \nsstep \nat \specconfone {\cfstack'} \Ds \Os \sto\dir \obs \cfstack.
    \]
    By the induction hypothesis (IH), we know that there exist $\nat' \le \nat$,
    a sequence of directives $\Ds'$ (without any $\dbt$ directives), and a sequence of observations $\Os'$ such that:
    \[
      \nsstep {\nat'} \specconfone {\cfstack'} {\Ds'} {\Os'}.
    \]
    We need to show that if
    \[
      \sstep {\cfstack'} {\cfstack} \dir \obs,
      \tag{\dag}
    \]
    then there exists $\nat'' \le \nat+1$, a sequence of directives $\Ds''$, and a sequence of observations $\Os''$ such that:
    \[
      \nsstep {\nat''} \specconfone {\cfstack} {\Ds''} {\Os''}.
    \]
    Additionally, it must hold that $\Ds''$ does not contain any $\dbt$ directives. We proceed by analyzing the directive $\dir$ used in (\dag).
    \begin{proofcases}
      \proofcase{$\dir \neq \dbt$} In this case, the claim directly follows from the IH.
      \proofcase{$\dir = \dbt$} The rule that has been applied can either be \ref{SI:Backtrack-Top} or \ref{SI:Backtrack-Bot}.
      \begin{proofcases}
        \proofcase{\ref{SI:Backtrack-Top}} In this case, by introspection of the rule we conclude that $\cfstack' = \specconftwo : \cfstack''$, that the mis-speculation flag of $\specconftwo$ is set to $\top$. By \Cref{lemma:misspecunset}, we know that $\cfstack'' \neq \nil$ (otherwise, the configuration $\conftwo$ could not have its mis-speculation flag set to $\top$). 
        This observation allows us to apply \Cref{remark:nobttech} and establish that there is a sequence of transitions from $\cfstack$ to $\cfstack''$ with length $\nat' < \nat$. The conclusion then follows by applying the IH to this intermediate result. This completes the sub-case.
        
        \proofcase{\ref{SI:Backtrack-Bot}} In this case, $\cfstack$ contains only a single configuration. We use \Cref{lemma:nobtsteponly} to demonstrate that there is a sequence of transitions from $\specconfone$ to $\cfstack$, consisting solely of the $\dstep$ directive, which establishes our claim.
      \end{proofcases}
    \end{proofcases}
  \end{proofcases}
\end{proof}

\begin{lemma}
  \label{lemma:syscsym}
  For every system $\system (\rfs, \syss, \caps)$ and configurations 
  \[
    \conf {\frame {\syss(\syscall)} \regmap {\km[\syscall]}, \lay \lcomp {\rfs'}, \Os, \Ds}
  \]
  with $\rfs' \eqon{\Fn} \rfs$, and 
  \[
    \conf {\frame {\cmd} {\regmap'} {\km[\syscall]} : \st, \lay \lcomp {\rfs''}, \Os', \Ds'},
  \]
  if:
  \[
    \nesstep n {\conf {\frame {\syss(\syscall)} \regmap {\km[\syscall]}, \lay \lcomp {\rfs'}, \Os, \Ds}}
    {\conf {\frame {\cmd} {\regmap'} {\km[\syscall]} : \st, \lay \lcomp {\rfs''}, \Os', \Ds'}}
  \]
  then:
  \[
    \nesstep n {\conf {\frame {\syss(\syscall)} \regmap {\km[\syscall]}, \lay \lcomp {\rfs'}}}
    {\conf {\frame {\cmd} {\regmap'} {\km[\syscall]} : \st, \lay \lcomp {\rfs''}}}.
  \]
\end{lemma}

\begin{proof}
  The proof proceeds by induction on $\nat$. The base case is straightforward. For the inductive step, we perform a case analysis on the last applied rule. The result follows by examining the premises of the rules and applying the IH.
  It is important to note that, due to \Cref{rem:stackinv1}, the following rules cannot be applied: 
  \ref{AL:Poison}, \ref{AL:Obs}, \ref{AL:Obs-End}, \ref{AL:Spec-Init}, \ref{AL:Spec-Dir}, \ref{AL:Spec-Step}, \ref{AL:Spec-Bt}, and \ref{AL:Spec-Term}.
\end{proof}

The following lemma establishes that the output of a \emph{system-call semantics 
preserving} transformation remains semantically equivalent to its input.

\begin{lemma}
  \label{lemma:ctxsemequiv}
  Let $\system = ({\rfs}_1, \syss_1, \caps_1)$ be a system, and let  
  $\systrans(\system) = ({\rfs}_2, \syss_2, \caps_2)$.  
  If $\systrans$ is \emph{system-call semantics preserving}, then the systems  
  $\system$ and $\systrans(\system)$ are \emph{semantically equivalent}.
\end{lemma}

\begin{proof}
  To establish the main claim, we first prove the following auxiliary claim.
  Let
  \[
    \confone_1 = \conf{\frame \cmd \regmap \um, \lay \lcomp {\rfs}_1}
    \quad \text{and} \quad
    \confone_2 = \conf{\frame \cmd \regmap \um, \lay \lcomp {\rfs}_2}.
  \]
  For every $\nat \in \mathbb{N}$ and every pair of systems 
  $\system_1$ and $\system_2$ whose system calls have the same semantics 
  in the sense of \Cref{def:sempres}, we have:
  \[
    \nstep[\system_1] \nat {\confone_1} {\conf{\st, \lay \lcomp {\rfs}_1'}} \land \um(\st_1)
    \Rightarrow
    \nstep[\system_2] * {\confone_2} {\conf{\st, \lay \lcomp {\rfs}_2'}} 
    \land {\rfs}_1' \eqon{\Idu \cup \Ark} {\rfs}_2'.
  \]
  We prove this by induction on $\nat$:
  \begin{proofcases}
    \proofcase{0} This case is trivial.
    \proofcase{$\nat+1$} In this case, we assume:
    \[
      \nstep[\system_1] \nat {\confone_1} {\conf{\st, \lay \lcomp {\rfs}_1'}} \to {\conf{\st', \lay \lcomp {\rfs}_1''}} \land \um(\st').
    \]
    We apply the IH to the $\nat$-step transition, obtaining:
    \[
      \nstep[\system_2] * {\confone_2} {\conf{\st, \lay \lcomp {\rfs}_2'}} \land {\rfs}_1' \eqon{\Idu\cup\Ark} {\rfs}_2'.
    \]
    Our goal is to establish:
    \[
      \nstep[\system_2] * {\conf{\st, \lay \lcomp {\rfs}_2'}} {\conf{\st', \lay \lcomp {\rfs}_2''}} \land {\rfs}_1'' \eqon{\Idu\cup\Ark} {\rfs}_2''.
    \]
    The proof proceeds by case analysis on the rule used in:
    \[
      \lay \red[\system_1] {\conf{\st, \lay \lcomp {\rfs}_1'}} \to {\conf{\st', \lay \lcomp {\rfs}_1''}}.
      \tag{\dag}
    \]
    In most cases, the identity of the frame stack and the equivalence
    ${\rfs}_1' \eqon{\Idu\cup\Ark} {\rfs}_2'$ suffice to show that the
    same rule applies to the corresponding transition in $\system_2$.
    To illustrate this, we consider \ref{WL:Load} as an
    example. Another notable case is \ref{WL:Pop}, which we address
    afterward.
    \begin{proofcases}
      \proofcase{\ref{WL:Load}} By analyzing the rule, we observe that
      $\um(\st)$ holds. In particular: the execution-mode flag of the
      topmost frame remains unchanged in the transition.  The
      expression within the load instruction does not contain
      kernel-space identifiers; otherwise, it would break the
      invariant of \Cref{rem:stackinv0}.
      From these observations and the structure of \ref{WL:Load}, we
      deduce that: ${\rfs}_1''={\rfs}_1'$,
      $\st = \frame {\cmemread \vx \expr\sep \cmdtwo}{\regmap'}{\um}
      \cons \overline \st$, and that
      $\st' = \frame {\cmdtwo}{\update{\regmap'}{\vx}{{\rfs}_1'(\ar,
          i)} }{\um} \cons \overline \st$, for some
      $(\ar, i) \in \Aru\times \Nat$ such that
      $\lay(\ar) + i = \sem\expr_{\lay, \regmap'}$.  Moreover, by
      introspection of rule \ref{WL:Load}, we confirm that it also
      applies to establish:
      \[
        \nstep[\system_2] {} {\conf{\st, \lay \lcomp {\rfs}_2'}}
        {\conf{\st', \lay \lcomp {\rfs}_2'}}
      \]
      In particular, the stack frame is $\st'$ because
      ${\rfs}_1'(\ar, i) = {\rfs}_2'(\ar, i)$.
      \proofcase{\ref{WL:Pop}}
      We can refine the assumption (\dag) as follows:
      \[
        \lay \red[\system_1] {\conf{\frame \cnil
            {\regmap'}{\opt}:\overline \st, \lay \lcomp {\rfs}_1'}}
        \to {\conf{\overline \st', \lay \lcomp {\rfs}_1'}}
      \]
      where, in particular, $\overline \st'$ is obtained by updating
      the register map of $\overline \st$ according to the
      rule. Additionally, we have $\um(\overline \st')$ which also
      implies $\um(\st)$. We proceed by cases on $\opt$.
      
      \begin{proofcases}
        \proofcase{$\opt=\um$} In this case, we can apply the IH, and
        the remaining part of the proof follows analogously to the
        case of rule \ref{WL:Load}.

        \proofcase{$\opt=\km[\syscall]$}
        Here, we cannot apply the IH, but we can use
        \Cref{lemma:thereisasyscall0} to prove the existence of the
        reductions:
        \begin{multline*}
          \nstep[\system_1] {\nat'} {\confone_1}  {\conf{\frame{\csyscall \syscall {\expr_1, \ldots, \expr_k};\cmdtwo}\regmap \um\cons \st'', \lay \lcomp {\rfs}_1''}} \to\\
          {\conf{\frame{\syss_1(\syscall)}{\regmap_0'} {\km[\syscall]}\cons \frame{\cmdtwo}\regmap \um\cons \st'', \lay \lcomp {\rfs}_1''}} \to^{\nat-\nat'-1}
          {\conf{\frame \cnil
              {\regmap'}{\km[\syscall]}:\overline \st, \lay \lcomp {\rfs}_1'}}
          \to {\conf{\overline \st', \lay \lcomp {\rfs}_1'}}
          \tag{$*$}
        \end{multline*}
        and
        \begin{equation*}
          \lay\red[\system_1]   
          {\conf{\frame{\syss_1(\syscall)}{\regmap_0'} {\km[\syscall]}, \lay \lcomp {\rfs}_1''}} \to^{\nat-\nat'-1}
          {\conf{\frame \cnil
              {\regmap'}{\km[\syscall]}:\overline \st'', \lay \lcomp {\rfs}_1'}},
          \tag{\ddag}
        \end{equation*}
        with
        $\overline \st''\cons \frame{\cmdtwo}\regmap \um\cons \st'' =
        \overline \st$ and $\km(\overline \st'')$.  In ($*$), we
        tacitly unrolled the first reduction step of the system call
        invocation.  From $\km(\overline \st'')$ and
        $\um(\overline \st)$, we deduce that $\overline \st''$ must be
        empty. 
        
        Next, we apply the IH to the first $\nat'$ steps of the
        reduction in ($*$), obtaining:
        \begin{multline*}
          \nstep[\system_2] {*} {\confone_2}  {\conf{\frame{\csyscall \syscall {\expr_1, \ldots, \expr_k};\cmdtwo}\regmap \um\cons \st'', \lay \lcomp {\rfs}_2''}}\to\\
          {\conf{\frame{\syss_2(\syscall)}{\regmap_0'} {\km[\syscall]}\cons\frame{\cmdtwo}\regmap \um\cons \st'', \lay \lcomp {\rfs}_2''}},
        \end{multline*}
        where ${\rfs}_2''\eqon{\Idu\cup \Ark}{\rfs}_2''$.  Here, we
        have once again expanded the first step of the system call
        using rule \ref{WL:SystemCall}.  By the system call semantics
        equivalence property, we deduce:
        \[
          \Eval[\system_1][\lay]{\syss_1(\syscall), \regmap_0', \km[\syscall], {\rfs}_1''}\simeq
          \Eval[\system_2][\lay]{\syss_2(\syscall), \regmap_0', \km[\syscall], {\rfs}_2''}.
        \]
        Using this equivalence, we conclude that: 
        \begin{equation*}
          \lay\red[\system_2]  
          {\conf{\frame{\syss_2(\syscall)}{\regmap_0'} {\km[\syscall]}, \lay \lcomp {\rfs}_2''}} \to^{\nat-\nat'-1}
          {\conf{\frame \cnil
              {\regmap''}{\km[\syscall]}, \lay \lcomp {\rfs}_2'}}
        \end{equation*}
        By combining this result and (\ddag), we obtain
        ${\rfs}_1'\eqon{\Idu\cup \Ark}{\rfs}_2'$ and
        $\regmap''(\ret)= \regmap'(\ret)$.  We conclude by applying \Cref{lemma:contextpluggingtech}.
      \end{proofcases}
    \end{proofcases}
  \end{proofcases}
  Now that the auxiliary claim has been established, we specialize it
  to $\system_1 = \system$ and $\system_2 = \systrans(\system)$.  The
  proof proceeds by case analysis on
  $\Eval[\system_1][\lay]{\cmd, \regmap, \um, {\rfs}_1}$:
  \begin{proofcases}
    \proofcase{$(\val, \rfs)$} The conclusion is a direct consequence
    of the auxiliary claim.

    \proofcase{$\err$} In this case we have:
    \[
      \nstep[\system_1][\lay] * {\confone_1} {\confone_1'} \to \err,
    \]
    notice that $\confone_1'$ is not a terminal configuration.  Now we
    go by case analysis on whether ${\confone_1'}$ is in user-mode or
    not (i.e. if its frame stack is in user-mode).
    \begin{proofcases}
      \proofcase{$\um({\confone_1'})$} In this case, we apply the
      auxiliary claim, and we deduce:
      \[
        \nstep[\system_2][\lay] * {\confone_2} {\confone_2'}.
      \]
      In particular, the frame stacks of $\confone_1'$ and
      $\confone_2'$ are identical, and their stores are identical on
      user-space identifiers. By case analysis on the rule that was
      used to prove:
      \[
        \step[\system_1][\lay] {\confone_1'} \err,
      \]
      we deduce that the same rule can also be used to prove:
      \[
        \step[\system_2][\lay] {\confone_2'} \err.
      \]
      \proofcase{$\km(\confone_1')$} We apply \Cref{lemma:thereisasyscall0}, which proves that:
      \begin{multline*}
        \nstep[\system_1][\lay] * {\confone_1} {\conf {\frame {\csyscall \syscall{\expr_1, \ldots, \expr_k}\sep\cmdtwo} {\regmap'}\um\cons \st, \lay \lcomp {\rfs}_1'}}\to\\ {\conf {\frame {\syss_1(\syscall)} {\regmap_0'} {\km[\syscall]}\cons \frame {\cmdtwo} {\regmap'}\um\cons \st, \lay \lcomp {\rfs}_1'}}\to^* {\conf {\st'\cons\frame {\cmdtwo} {\regmap'}\um\cons \st, \lay \lcomp {\rfs}_1''}} \to \err,
      \end{multline*}
      and that:
      \begin{equation*}
        \nstep[\system_1][\lay] *  {\conf {\frame {\syss_1(\syscall)} {\regmap_0'} {\km[\syscall]}, \lay \lcomp {\rfs}_1'}} {\conf {\st', \lay \lcomp {\rfs}_1''}}
      \end{equation*}
      In the first conclusion, we tacitly replaced $\confone_1'$ with
      ${\conf {\st'\cons\frame {\cmdtwo} {\regmap'}\um\cons \st, \lay
          \lcomp {\rfs}_1''}}$ and we made the first reduction step
      after the system call invocation explicit.  By case analysis on
      the rules that has been used to prove
      $\step[\system_1][\lay] {\confone_1'} \err$, we observe that the
      same rule can also be used to prove
      $\step[\system_1][\lay]{\conf {\st', \lay \lcomp
          {\rfs}_1''}}\err$. This proves:
      \begin{equation*}
        \nstep[\system_1][\lay] *  {\conf {\frame {\syss_1(\syscall)} {\regmap_0'} {\km[\syscall]}, \lay \lcomp {\rfs}_1'}}  \err,
      \end{equation*}
      and by applying the assumption on the preservation of system call semantics, we deduce:
      \begin{equation*}
        \nstep[\system_2][\lay] *  {\conf {\frame {\syss_2(\syscall)} {\regmap_0'} {\km[\syscall]}, \lay \lcomp {\rfs}_2'}}  \err.
      \end{equation*}
      By applying the auxiliary claim, we prove:
      \begin{multline*}
        \nstep[\system_2][\lay] * {\confone_2} {\conf {\frame {\csyscall \syscall{\expr_1, \ldots, \expr_k}\sep\cmdtwo} {\regmap'}\um\cons \st, \lay \lcomp {\rfs}_2'}}\to\\ {\conf {\frame {\syss_2(\syscall)} {\regmap_0'} {\km[\syscall]}\cons \frame {\cmdtwo} {\regmap'}\um\cons \st, \lay \lcomp {\rfs}_2'}},
      \end{multline*}
      Here, we unrolled the first step of the execution of the system
      call using the rule \ref{WL:SystemCall}. We conclude with an
      application of \Cref{lemma:contextpluggingtech}, which proves
      $\nstep[\system_2][\lay] * {\confone_2} \err$.
    \end{proofcases}
    \proofcase{$\unsafe$} This case is analogous to the previous one.
    \proofcase{$\Omega$} In this case we must show:
    \[
      \forall \nat. \exists \conftwo_\nat. \nstep[\system_2][\lay] \nat {\confone_2} {\conftwo_\nat}.
    \]
    The proof proceeds by contraposition. Assume that
    $\exists \nat.\lay \red[\system_2] \confone_2 \not\to^\nat$. Let
    $\overline{\nat}$ be the least $\nat$ such that
    $\nstep[\system_2][\lay] {\overline{\nat}} {\confone_2} \conftwo$
    for some $\conftwo$. Observe that $\conftwo$ must now be a
    terminal configuration. By reasoning identically to the previous
    cases, but swapping the roles of $\system_1$ and $\system_2$, we
    deduce that
    $\nstep[\system_1][\lay] {\overline{\nat}} {\confone_1}
    {\conftwo'}$ for some terminal configuration $\conftwo'$, and this
    contradicts our assumption.
  \end{proofcases}
\end{proof}

\HERE

\subsubsection{Technical Observations on the $\fencetrans$ Transformation}
\label{sec:appfencetrans}

We say that two programs $\cmd,\cmd'$ are in relation $\cmd \fencerel\cmd'$ if and only if
$\cmd'$ can be obtained by substituting instructions $\ccall \expr {\exprtwo_1, \ldots, \exprtwo_k}$ with $\cscall \expr {\exprtwo_1, \ldots, \exprtwo_k}$ and by placing $\cfence$ instruction inside $\cmd$.

We extend this relation to configurations by stipulating

\[
  \infer{\nil \fencerel \nil}{} \quad
  \infer{\rfs \fencerel \rfs'}{\rfs \eqon{\Idu\cup \Ark} \rfs' & \forall \fn \in \Fnk. \rfs'(\fn) \fencerel (\rfs(\fn))}
  \quad
  \infer{\syss \fencerel \syss'}{\forall \syscall \in \Sys. \syss(\syscall) \fencerel \syss(\syscall)'}
  \quad
    \infer{(\rfs, \syss, \caps) \fencerel (\rfs', \syss', \caps)}{
      \rfs \fencerel \rfs' &
      \syss \fencerel \syss'
    }
\]
\[    
  \infer{\frame \cmd \regmap {\km[\syscall]}:\st \fencerel \frame {\cmd'} {\regmap} {\km[\syscall]}:\st'}{\st\fencerel\st' & \cmd \fencerel \cmd'}\quad
  \infer{\frame \cmd \regmap {\um}:\st \fencerel \frame {\cmd} {\regmap} {\um}:\st'}{\st\fencerel\st'}
\]
\[
  \infer{\err \fencerel \err}{}\quad
  \infer{\unsafe \fencerel \unsafe}{} \quad
  \infer{\conf{\st, \lay \lcomp \rfs} \fencerel
    \conf{\st', \lay \lcomp \rfs}}{\rfs\fencerel \rfs' & \st \fencerel\st'}
\]

\begin{lemma}
  \label{lemma:fencestep}
  For every pair of systems $\system \fencerel \system'$, and for any configurations 
   $\confone = \conf{\frame \cmd \regmap \opt : \st, \lay \lcomp \rfs}$
   and 
   $\conftwo$ such that $\confone \fencerel \conftwo$, if
  \[
    \step[\system][\lay]  {\confone}
    {\confone'},
  \]
  then there exists a configuration $\conftwo'$ such that
  \[
    \nstep[\system'][\lay] * {\conftwo}
    {\conftwo'},
  \]
  and $\confone' \fencerel \conftwo'$.
\end{lemma}

\begin{proof}
  We rewrite the configurations $\confone, \conftwo$ as follows:
  \[
    \confone = \conf{\frame \cmd \regmap \opt : \st, \lay \lcomp \rfs}
    \tag{\dag}
  \]
  and
  \[
    \conftwo = \conf{\frame {\cmd'} \regmap \opt : \st', \lay \lcomp \rfs'}.
    \tag{\ddag}
  \]
  with $\st \fencerel \st'$ and $\rfs \fencerel \rfs'$. We proceed by cases on $\opt$.

  \begin{proofcases}
    \proofcase{$\um$}
    The proof goes by cases on the rule that has been applied to $\cmd$. We show some of the most relevant cases.
    In all of the cases within this sub-derivation, we have $\cmd=\cmd'$ by the assumption $\confone \fencerel \conftwo$. 
    \begin{proofcases}
      \proofcase{\ref{WL:Op}}
      The assumption is
      \begin{equation*}
        \step[\system][\lay]  {\conf{\frame {\vx \ass \expr\sep \cmdtwo} \regmap \um : \st,  \lay \lcomp \rfs}}{}\\
        {\conf{\frame {\cmdtwo} {\update\regmap\vx{\sem \expr_{\regmap, \lay}}} \um : \st,  \lay \lcomp \rfs}}.
      \end{equation*}
      By applying the same rule to $\conftwo$, we obtain:
      \begin{equation*}
        \step[\system'][\lay]  {\conf{\frame {\vx \ass \expr\sep \cmdtwo} \regmap \um : \st',  \lay \lcomp \rfs'}}{}\\
        {\conf{
            \frame {\cmdtwo} {\update\regmap\vx{\sem \expr_{\regmap, \lay}}} \um : \st',
            \lay \lcomp \rfs'}}
      \end{equation*}
      which establishes the claim.

      \proofcase{\ref{WL:Store}} 
      The assumption is
      \begin{equation*}
        \step[\system][\lay]  {\conf{\frame {\cmemass \expr\exprtwo\sep \cmdtwo} \regmap \um : \st,  \lay \lcomp \rfs}}{}\\
        {\conf{\frame {\cmdtwo} {\regmap} \um : \st,  \update{\lay \lcomp {\rfs}}\add\val}},
      \end{equation*}
      where $\val = \sem \exprtwo_{\regmap, \lay}$ and $\add = \toAdd{\sem \expr_{\regmap, \lay}}$.
      Observe that $\add \in \underline \lay(\Aru)$, meaning that there exists a pair $(\ar, i)$ such that $\lay(\ar)+i=\add$ and $\ar \in \Aru$.
      Thus, we have that 
      \(
      \update{\lay \lcomp {\rfs}}\add\val={\lay \lcomp {\update{{\rfs}}{(\ar, i)} \val}}
      \)
      due to \Cref{rem:memupdtostupd}.
      Applying the same rule to $\conftwo$, we obtain:
      \begin{equation*}
        \step[\system'][\lay]  {\conf{\frame {\cmemass \expr\exprtwo\sep \cmdtwo} \regmap \um : \st',  \lay \lcomp \rfs'}}{}\\
        {\conf{\frame {\cmdtwo} {\regmap} \um : \st', \lay \lcomp { \update{{\rfs'}}{(\ar, i)}\val}}},
      \end{equation*}
      and to complete the proof, we observe that $\update{{\rfs}}{(\ar, i)}\val \fencerel \update{{\rfs'}}{(\ar, i)}\val$, which follows from the assumption $\ar \in \Aru$.
      \proofcase{\ref{WL:Call}} 
      The assumption is
      \begin{equation*}
        \step[\system][\lay]  {\conf{\frame {\ccall \expr{\vec\exprtwo}\sep \cmdtwo} \regmap \um : \st,  \lay \lcomp \rfs}}{}\\
        \conf{\frame {\lay \lcomp \rfs(\toAdd{\sem \expr_{\regmap, \lay}}) } {\regmap_0'} \um : \conf{\frame {\cmdtwo} {\regmap} \um : \st,  \lay \lcomp \rfs}},
      \end{equation*}
      where $\regmap_0'$ is a shorthand for $\regmap_0[\vec\vx \upd \sem{\vec \exprtwo}_{\regmap, \lay}]$. From the premise of the rule, we infer that there exists $\fn \in \Fnu$ such that
      \(
      \toAdd {\sem \expr_{\regmap, \lay}} = \lay(\fn)
      \)
      Thus, from the definition of $\lay \lcomp \rfs$ and the definition of the $\fencerel$ relation, we deduce that
      \(
      \lay \lcomp \rfs = \rfs(\fn) = \rfs'(\fn)
      \),
      where $\rfs'$ is the store of (\ddag). Consequently, these observations allow us to demonstrate that applying the same rule to $\conftwo$ results in
      \begin{equation*}
        \step[\system'][\lay]  {\conf{\frame {\ccall \expr{\vec \exprtwo}\sep \cmdtwo} \regmap \um : \st',  \lay \lcomp \rfs'}}{}\\
        {\conf{\frame {\lay \lcomp \rfs(\toAdd{\sem \expr_{\regmap, \lay}}) } {\regmap_0'} \um : \conf{\frame {\cmdtwo} {\regmap} \um : \st',  \lay \lcomp \rfs'}}},
      \end{equation*}
      thereby proving the claim.

      \proofcase{\ref{WL:SystemCall}} 
      The assumption is
      \begin{equation*}
        \step[\system][\lay]  {\conf{\frame {\csyscall \syscall{\vec\exprtwo}\sep \cmdtwo} \regmap \um : \st,  \lay \lcomp \rfs}}{}\\
        \conf{\frame {\syss(\syscall) } {\regmap_0'} {\km[\syscall]} : \conf{\frame {\cmdtwo} {\regmap} \um : \st,  \lay \lcomp \rfs}},
      \end{equation*}
      where $\regmap_0'$ is a shorthand for
      \(
        \regmap_0[\vec \vx \upd \sem{\vec \exprtwo}]
      \).
      By applying the same rule to the configuration $\conftwo$ of (\ddag), we obtain
      \begin{equation*}
        \step[\system'][\lay]  {\conf{\frame {\csyscall \syscall{\vec \exprtwo}\sep \cmdtwo} \regmap \um : \st',  \lay \lcomp \rfs'}}{}\\
        {\conf{\frame {\syss'(\syscall)} {\regmap_0'} {\km[\syscall]} : \conf{\frame {\cmdtwo} {\regmap} \um : \st',  \lay \lcomp \rfs'}}},
      \end{equation*}
      To conclude, we simply observe that
      \(
        \frame {\syss(\syscall)} {\regmap_0'} {\km[\syscall]} \fencerel
        \frame {\syss'(\syscall)} {\regmap_0'} {\km[\syscall]}
      \),
      which follows by the definition of the relation $\fencerel$.

      \proofcase{\ref{WL:Store-Error}} 
      The assumption is
      \begin{equation*}
        \step[\system][\lay]  {\conf{\frame {\cmemass \expr\exprtwo\sep \cmdtwo} \regmap \um : \st,  \lay \lcomp \rfs}}\err,
      \end{equation*}
      Let $\add = \toAdd{\sem \expr_{\regmap, \lay}}$. From the premises of the rule, we deduce that $\add \notin \underline \lay(\Aru)$.
      This suffices to establish:
      \begin{equation*}
        \step[\system][\lay]  {\conf{\frame {\cmemass \expr\exprtwo\sep \cmdtwo} \regmap \um : \st',  \lay \lcomp \rfs'}}{\err},
      \end{equation*}
      thereby proving the claim.
    \end{proofcases}
    \proofcase{$\km[\syscall]$} 
    Under this assumption, most cases are analogous to the corresponding ones for user-mode execution. Similarly, we proceed by case analysis on the rule that establishes the transition, focusing on the most significant cases.

    \begin{proofcases}
      \proofcase{\ref{WL:Op}} 
      The assumption is
      \begin{equation*}
        \step[\system][\lay]  {\conf{\frame {\vx \ass \expr\sep \cmdtwo} \regmap {\km[\syscall]} : \st,  \lay \lcomp \rfs}}{}\\
        {\conf{\frame {\cmdtwo} {\update\regmap\vx{\sem \expr_{\regmap, \lay}}} {\km[\syscall]} : \st,  \lay \lcomp \rfs}},
      \end{equation*}
      From the assumption on the initial configurations, the configuration $\conftwo$ in (\dag) has the form ${\conf{\frame {\cmd'} \regmap {\km[\syscall]} : \st',  \lay \lcomp \rfs'}}$, where $\cmd'$ is one of the following commands:
      \begin{align*}
        &\vx \ass \expr\sep \cmdtwo'
        &&\vx \ass \expr\sep \cfence\sep \cmdtwo'
        &&\cfence \sep\vx \ass \expr\sep \cmdtwo'
      \end{align*}
      with $\cmdtwo \fencerel \cmdtwo'$. In the first two cases, the same rule can be applied to $\conftwo$, yielding:
      \begin{equation*}
        \step[\system'][\lay]  \conftwo
        {\conf{
            \frame {\cmdtwo''} {\update\regmap\vx{\sem \expr_{\regmap, \lay}}} {\km[\syscall]} : \st',
            \lay \lcomp \rfs'}}
      \end{equation*}
      for $\cmdtwo'' \in \{ \cmdtwo', \cfence;\cmdtwo' \}$.
      The claim follows from the observation that $\cmdtwo \fencerel \cmdtwo''$.
      In the third case, we observe that:
      \begin{equation*}
        \nstep[\system'] 2 {\conf{\frame {\cfence \sep\vx \ass \expr\sep \cmdtwo'} \regmap {\km[\syscall]} : \st',  \lay \lcomp \rfs'}}{}\\
        {\conf{
            \frame {\cmdtwo'} {\update\regmap\vx{\sem \expr_{\regmap, \lay}}} {\km[\syscall]} : \st',
            \lay \lcomp \rfs'}}
      \end{equation*}
      The conclusion follows directly since $\cmdtwo \fencerel \cmdtwo'$ holds by assumption.
      \proofcase{\ref{WL:Call}} 
      The assumption is
      \begin{equation*}
        \step[\system][\lay]  {\conf{\frame {\ccall \expr{\vec \exprtwo}\sep \cmdtwo} \regmap {\km[\syscall]} : \st,  \lay \lcomp \rfs}}{}\\
        \conf{\frame {\lay \lcomp \rfs(\sem \expr_{\regmap, \lay}) } {\regmap_0'} {\km[\syscall]} : \conf{\frame {\cmdtwo} {\regmap} {\km[\syscall]} : \st,  \lay \lcomp \rfs}},
      \end{equation*}
      where $\regmap_0'$ is a shorthand for
      \(
        \regmap_0[\vec\vx \upd\sem{\vec \exprtwo}_{\regmap, \lay}]
      \).
      From the premise of the rule, we observe that there exists $\fn \in \Fnk$ such that
      $\toAdd{\sem \expr_{\regmap, \lay}} = \lay(\fn)$. Consequently, from the definition of $\lay \lcomp \rfs$,
      we deduce that $\lay \lcomp \rfs = \rfs(\fn)$, and we conclude $\rfs(\fn) \fencerel \rfs'(\fn)$ by assumption.
      Finally, we observe that $\conftwo$ in (\ddag) is:
      \(
        {\conf{\frame {\cmd' \sep \cmdtwo'} \regmap {\km[\syscall]} : \st',  \lay \lcomp \rfs'}},
      \)
      where $\cmd'$ is one of the following commands:
      \begin{align*}
        &\ccall \expr{\vec \exprtwo} 
        &&\cscall \expr{\vec \exprtwo}
        &&\cfence \sep \ccall \expr{\vec \exprtwo} 
        &&\cfence \sep \cscall \expr{\vec \exprtwo}
      \end{align*}
      In any of these cases, using the rules \ref{WL:Fence}, \ref{WL:Call}, and \ref{WL:SCall}, we can demonstrate the following reduction:
      \begin{equation*}
        \nstep[\system'][\lay] * {\conf{\frame {\ccall \expr{\vec\exprtwo}\sep \cmdtwo'} \regmap {\km[\syscall]} : \st',  \lay \lcomp \rfs'}}{}\\
        {\conf{\frame {\lay \lcomp \rfs'(\toAdd{\sem \expr_{\regmap, \lay}}) } {\regmap_0'} {\km[\syscall]} : \conf{\frame {\cmdtwo'} {\regmap} {\km[\syscall]} : \st',  \lay \lcomp \rfs'}}}.
      \end{equation*}
      To conclude, we must observe that 
      \(
        \lay \lcomp \rfs(\sem \expr_{\regmap, \lay}) \fencerel \lay \lcomp \rfs'(\sem \expr_{\regmap, \lay})
      \),
      which follows from the assumptions $\toAdd{\sem \expr_{\regmap, \lay}} = \lay(\fn)$, $\fn \in \Fnk$, and $\rfs \fencerel \rfs'$.
    \end{proofcases}
  \end{proofcases}
\end{proof}

\begin{corollary}
  \label{cor:fencesem1}
  For every two systems
  $\system = ({\rfs}_1, \syss_1, \caps)  \fencerel ({\rfs}_2, \syss_2, \caps) =  \system'$,
  every layout $\lay$, unprivileged command $\cmd$, and registers $\regmap$, we have
  \[
    \Eval[\system][\lay]{\cmd, \regmap, \um, {\rfs}_1}
    \simeq
    \Eval[\system'][\lay]{\cmd, \regmap, \um, {\rfs}_2},
  \]
  where the equivalence is that of \Cref{def:sempres}, i.e.,
  $(v, \rfs[1]) \simeq (v, \rfs[2])$ if
  $\rfs[1] \eqon{\Idu} \rfs[2]$, and it coincides with equality otherwise.
\end{corollary}

\begin{proof}
  Observe that the initial configurations are in $\fencerel$ relation, therefore the claim
  is a trivial consequence of \Cref{lemma:fencestep}.
\end{proof}

\begin{corollary}
  \label{cor:fencesem2}
  Transformations $\fencetrans, \optfencetrans, \nospec$
  preserve the semantics of the systems.
\end{corollary}

\begin{proof}
  Observe that for any $\systrans \in \{\fencetrans, \optfencetrans, \nospec\}$,
  we always have $\system \fencerel \systrans(\system)$. Hence, the conclusion
  follows as a trivial consequence of \Cref{cor:fencesem1}.
\end{proof}

\begin{lemma}
  \label{lemma:fwf}
  Fix a configuration stack $\cfstack$ such that $\fwf[\fencetrans] \cfstack$ and a system $\system \in \im(\fencetrans)$.  For every layout $\lay$,  $\dbt$-free sequence of directives $\dir:\Ds$, sequence of observations $\obs:\Os$, and speculative configuration stack $\cfstack'$ such that $\lnot (\fwf[\fencetrans]{\cfstack'})$, there exists a stack $\cfstack''$ satisfying the following properties:
\[
  (\nsstep \nat \cfstack {\cfstack'} {\Ds:\dir} {\Os:\obs}) \Rightarrow (\nsstep {\nat-1} \cfstack {\cfstack''}{\Ds} {\Os})\quad \text{and} \quad \fwf[\fencetrans]{\cfstack''}
\]
\end{lemma}
\begin{proof}
  We proceed by induction on $\nat$.
  \begin{proofcases}
    \proofcase{0} This case is absurd.
    
    \proofcase{$\nat+1$} Assume that
    \[
      (\nsstep {\nat+1} \cfstack {\cfstack'}{\Ds:\dir} {\Os:\obs}) 
    \]
    and that $\lnot (\fwf[\fencetrans]{\cfstack'})$. We need to prove that
    \[
      (\nsstep {\nat} \cfstack {\cfstack''}{\Ds} {\Os})
    \]
    for some $\cfstack''$ such that $\fwf[\fencetrans]{\cfstack''}$.
    The proof goes by contraposition: assume that the claim does not hold, i.e., $\lnot \fwf[\fencetrans]{\cfstack''}$.
    If $\nat = 0$, then $\cfstack''$ would be exactly $\cfstack$, contradicting our assumption. Hence, we can assume that $\nat > 0$.
    In this case, we reach a contradiction with \Cref{rem:fwftech}, as we have two consecutive reduction steps with intermediate configuration stacks not satisfying $\fwf[\fencetrans] \cdot$.
  \end{proofcases}
\end{proof}

\begin{remark}
  \label{rem:fwftech}
  Fix a system $\system \in \im(\fencetrans)$.
  For every layout $\lay$ and stack of speculative configurations $\cfstack$ such that $\fwf[\fencetrans] \cfstack$, one of the following three cases holds:
  \begin{varitemize}
  \item $\nf \cfstack{}$.
  \item For every directive $\dir \neq \dbt$ and observation $\obs$, if
    \(
      \sstep \cfstack {\cfstack'} \dir \obs,
    \)
    then $\fwf[\fencetrans]{\cfstack'}$.
  \item For every pair of directives $\dir_1, \dir_2 \neq \dbt$ and every pair of observations $\obs_1, \obs_2$, if
    \(
      \nsstep 2 \cfstack {\cfstack'} {\dir_1:\dir_2} {\obs_1:\obs_2},
    \)
    then $\fwf[\fencetrans]{\cfstack'}$.
  \end{varitemize}
\end{remark}
\begin{proof}
If the initial configuration stack is terminal, the conclusion is trivial. Therefore, we assume that the stack has the form
\(
\sframe{\frame {\fencetrans(\cmd)} \regmap {\km[\syscall]} : \st}{\bm\buf{\lay \lcomp \rfs'}} \boolms : \cfstack
\).
The proof proceeds by cases on $\cmd$.

\begin{proofcases}
  \proofcase{$\cnil$} In this case, we have the following transition:
  \begin{multline*}
    \sstep{\sframe{\frame {\fencetrans(\cnil)} \regmap {\km[\syscall]}\cons{\frame \cmdtwo {\regmap'} {\km[\syscall]}}\cons\st'}
      {\bm\buf{\lay \lcomp \rfs'}}\boolms\cons\cfstack}{}{\dstep} {\onone} \\
    {\sframe{{\frame \cmdtwo {\update{\regmap'}\ret{\regmap(\ret)}} {\km[\syscall]}}\cons\st'}
      {\bm\buf{\lay \lcomp \rfs'}}\boolms\cons\cfstack}.
  \end{multline*}
  The claim follows directly from the assumption on the initial configuration stack, yielding:
  \[
    \fwf[\fencetrans]{{\frame \cmdtwo {\regmap'} {\km[\syscall]}} : \st'}.
  \]
  \proofcase{$\cmemread \vx \exprtwo \sep \cmdtwo$}
Observe that the initial configuration stack has the following form:
\[
  \sframe{\frame {\cfence \sep \cmemread \vx \exprtwo \sep \fencetrans(\cmdtwo)} \regmap {\km[\syscall]} : \st'}
  {\bm\buf{\lay \lcomp \rfs'}} \boolms : \cfstack.
\]
If $\boolms = \top$, the first claim holds trivially because, without backtracking, this configuration cannot reduce further. Therefore, we assume that $\boolms = \bot$. 
After the first reduction step, which applies rule \ref{SI:Fence}, one of the following rules must apply: \ref{SI:Load-Step}, \ref{SI:Load}, \ref{SI:Load-Unsafe}, or \ref{SI:Load-Err}. If one of the last two rules applies, the claim is trivial. Since the behavior of rules \ref{SI:Load} and \ref{SI:Load-Step} is analogous, we focus on the former. In this case, we have:
\begin{multline*}
  \nsstep 2 {\sframe{\frame {\cfence \sep \cmemread \vx \exprtwo \sep \fencetrans(\cmdtwo)} \regmap {\km[\syscall]} : \st'}
  {\bm\buf{\lay \lcomp \rfs'}} \boolms : \cfstack}{} {\dstep : \dload i} {\onone : \omem \add} \\
  \sframe{\frame {\fencetrans(\cmdtwo)} {\update\regmap\vx \val} {\km[\syscall]} : \st'} 
  {\bm\nil {\overline{\bm\buf{\lay\lcomp \rfs'}}}}{\boolms \lor \bool'} \\
  \sframe{\frame {\cmemread \vx \exprtwo \sep \fencetrans (\cmdtwo)} \regmap {\km[\syscall]} : \st'}
  {\bm\nil {\overline{\bm\buf{\lay \lcomp \rfs'}}}}{\boolms} : \cfstack
\end{multline*}
where $\val, b' = \bufread {\bm\nil {\overline{\bm\buf{\lay \lcomp \rfs'}}}} \add i$ with $\add = \toAdd{\sem \exprtwo_{\regmap,\lay}}$. 
Notice that the $\cfence$ instruction preceding the load instruction has flushed the memory. To show the claim, we need to verify that:
\[
  \fwf[\fencetrans]{\sframe{\frame {\fencetrans(\cmdtwo)} \regmap {\km[\syscall]} : \st'} 
  {\bm\nil {\overline{\bm\buf{\lay \lcomp \rfs'}}}}{\boolms \lor \bool'}}.
  \tag{$*$}
\]
In particular, we need to establish $\fwf[\fencetrans]{\bm\nil {\overline{\bm\buf{\lay \lcomp \rfs'}}}}$. By the assumption on the initial configuration stack, we know that the domain of $\buf$ does not contain any function address. Therefore, we can apply \Cref{rem:overlinewrtdom}, which ensures that the resulting memory is:
\(
  \overline{\bm\buf{\lay \lcomp \rfs'}} = \lay \lcomp \rfs'' 
\)
for some $\rfs''$ such that $\rfs'' \eqon{\Fn} \rfs' \eqon{\Fn} \rfs$. The other properties needed to establish ($*$) follow directly from the assumption on the initial configuration stack.

  \proofcase{$\cmemass \expr \exprtwo \sep \cmdtwo$}
We can rewrite the initial configuration stack as follows:
\[
  \sframe{\frame {\cfence \sep \cmemass \expr \exprtwo \sep \fencetrans(\cmdtwo)} \regmap {\km[\syscall]} : \st'}
  {\bm\buf{\lay \lcomp \rfs'}} \boolms : \cfstack.
\]
If $\boolms = \top$, the first claim holds trivially. Otherwise, the first reduction step must use rule \ref{SI:Fence}. For the next step, one of the rules \ref{SI:Store}, \ref{SI:Store-Unsafe}, or \ref{SI:Store-Err} must apply. If one of the last two rules applies, the conclusion is immediate. Therefore, we focus on the case where rule \ref{SI:Store} applies. In this case, we have:
\begin{multline*}
  \nsstep 2 {\sframe{\frame {\cfence \sep \cmemass \expr \exprtwo \sep \fencetrans(\cmdtwo)} \regmap {\km[\syscall]} : \st'}
  {\bm\buf{\lay \lcomp \rfs'}} \boolms : \cfstack}{} {\dstep : \dstep} {\onone : \omem \add} \\
  \sframe{\frame {\fencetrans(\cmdtwo)} {\regmap} {\km[\syscall]} : \st'} 
  {\bm{\bitem \add \val \cons \nil}{\overline{\bm\buf{\lay \lcomp \rfs'}}}}{\boolms} : \cfstack.
\end{multline*}
where $\add = \toAdd{\sem \expr_{\regmap,\lay}}$, and $\val = \sem \exprtwo_{\regmap,\lay}$.
Observe that to establish the claim, we need to verify that the following property holds:
\[
  \fwf[\fencetrans]{\sframe{\frame {\fencetrans(\cmdtwo)} {\regmap} {\km[\syscall]} : \st'} 
  {\bm{\bitem \add \val \cons \nil}{\overline{\bm\buf{\lay \lcomp \rfs'}}}}{\boolms}}.
\]
We conclude
$\fwf[\fencetrans]{\bm\nil {\overline{\bm\buf{\lay \lcomp \rfs'}}}}$ as in the previous case, so we only need to establish that $\add \in \Ark$. This follows directly from the premises of rule \ref{SI:Store}. The other properties needed to establish the claim follow from the assumption on the initial configuration stack.

\proofcase{$\ccall \expr {\vec \exprtwo} \sep \cmdtwo$} In this case, we can rewrite the initial configuration stack as follows:
\[
  \sframe{\frame {\cfence \sep \cscall \expr {\vec \exprtwo} \sep \fencetrans(\cmdtwo)} \regmap {\km[\syscall]} : \st'}
  {\bm\buf{\lay \lcomp \rfs'}} \boolms : \cfstack.
\]
Observe that the first reduction step (if any) must apply rule \ref{SI:Fence}. After that step, one of the following rules must apply: \ref{SI:Call-Step}, \ref{SI:Call-Step-Unsafe}, or \ref{SI:Call-Step-Err}. If one of the last two rules applies, the claim follows trivially. Otherwise, we observe the following reduction:
\begin{multline*}
  \nsstep 2 {\sframe{\frame {\fencetrans(\cmemass \expr \exprtwo \sep \cmdtwo)} \regmap {\km[\syscall]} : \st'}
  {\bm\buf{\lay \lcomp \rfs'}} \boolms : \cfstack}{} {\dstep : \dstep} {\onone : \ojump \add} \\
  \sframe{\frame {\overline{\bm\buf{\lay \lcomp \rfs'}}(\add)} {\regmap} {\km[\syscall]} \cons \frame {\fencetrans(\cmdtwo)} {\regmap} {\km[\syscall]} \cons \st'}
  {\bm\nil{\overline{\bm\buf{\lay \lcomp \rfs'}}}}{\boolms \lor \bool'} : \cfstack,
\end{multline*}
with, $\add = \toAdd{\sem \expr_{\regmap,\lay}}$ and $\val = \sem \exprtwo_{\regmap,\lay}$. 
We now need to verify:
\[
  \fwf[\fencetrans]{\sframe{\frame {\fencetrans(\cmdtwo)} {\regmap} {\km[\syscall]} : \st'} 
  {\bm{\bitem \add \val \cons \nil}{\overline{\bm\buf{\lay \lcomp \rfs'}}}}{\boolms}}.
\]
We can establish
$\fwf[\fencetrans]{\bm\nil {\overline{\bm\buf{\lay \lcomp \rfs'}}}}$ by reasoning similarly to the previous cases. Therefore,  we focus on establishing that the loaded program
$\overline{\bm\buf{\lay \lcomp \rfs'}}(\add)$ is equal to $\fencetrans(\cmd')$ for some command $\cmd'$. 
By inspecting rule \ref{SI:Call-Step}, we deduce that $\add \in \underline \lay(\Fn[\km])$. This means that there exists a function in $\Fn[\km]$ such that $\lay(\fn) = \add$. Applying \Cref{rem:overlinewrtdom}, we deduce that
\(
  \overline{\bm\buf{\lay \lcomp \rfs'}} = \lay \lcomp \rfs'' 
\)
for some $\rfs'' \eqon{\Fn} \rfs' \eqon{\Fn} \rfs$. By the definition of $\cdot \lcomp \cdot$, this implies that
\(
  \overline{\lay \lcomp \rfs''}(\add) = \rfs(\add),
\)
which is equal to $\fencetrans(\cmd')$ for some command $\cmd'$ by the hypothesis on $\system$.
  \end{proofcases}
\end{proof}

\newcommand{\coolms}{c_{\mathit{ms}}}

\begin{lemma}
  \label{lemma:stepsemsim}
  For every stack $\cfstack$ and configurations
  \(
    \specconfone = \sframe{\frame {\cmd_1} {\regmap_1} {\km[\syscall]}:\st_1} {\bm{\buf_1}{\mem_1}}{\boolms}
  \)
  and 
  \(
    \specconftwo = \sframe{\frame {\cmd_2} {\regmap_2} {\km[\syscall]}:\st_2} {\bm{\buf_2}{\mem_2}}{\coolms},
  \)
  if 
  \begin{equation*}
    \sstep {\sframe{\frame {\cmd_1} {\regmap_1} {\km[\syscall]}:\st_1} {\bm{\buf_1}{\mem_1}}{\boolms}\cons \cfstack} {}\dstep \obs \\
    {\sframe{\frame {\cmd_2} {\regmap_2} {\km[\syscall]}:\st_2} {\bm{\buf_2}{\mem_2}}{\coolms}\cons \cfstack} 
  \end{equation*}
  for some observation $\obs$, then
  \[
    \step {\ntc{\cmd_1} {\regmap_1} {\km[\syscall]} {\st_1} {\overline {\bm{\buf_1}{\mem_1}}}}
    {\ntc {\cmd_2} {\regmap_2} {\km[\syscall]}{\st_2} {\overline {\bm{\buf_2}{\mem_2}}}}.
  \]
\end{lemma}
\begin{proof}
    The proof proceeds by cases on the transition relation. Most cases are straightforward, therefore, we focus on those involving memory interactions, which require more care.
  \begin{proofcases}
    \proofcase{\ref{SI:Fence}} We rewrite the assumption as follows:
    \begin{equation*}
      \sstep {\sframe{\frame {\cfence\sep\cmd} {\regmap_1} {\km[\syscall]}:\st_1} {\bm{\buf_1}{\mem_1}}{\boolms}\cons \cfstack}{} \dstep \obs \\
      {\sframe{\frame {\cmd} {\regmap_1} {\km[\syscall]}:\st_1} {\overline{\bm{\buf_1}{\mem_1}}}{\boolms}\cons \cfstack} 
    \end{equation*}
    The claim is 
    \begin{equation*}
      \step {\conf{\frame {\cfence\sep\cmd} {\regmap_1} {\km[\syscall]}:\st_1, \overline {\bm{\buf_1}{\mem_1}}}}
      {\conf{\frame {\cmd_1} {\regmap_1} {\km[\syscall]}:\st_1, \overline {\bm{\buf_1}{\mem_1}}}},
    \end{equation*}
    which is trivially true.
    \proofcase{\ref{SI:Store}} We rewrite the assumption as:
    \begin{equation*}
      \sstep {\sframe{\frame {\cmemass \expr \exprtwo\sep\cmd} {\regmap_1} {\km[\syscall]}:\st_1} {\bm{\buf_1}{\mem_1}}{\boolms}\cons \cfstack} {} \dstep \obs\\
      {\sframe{\frame {\cmd} {\regmap_1} {\km[\syscall]}:\st_1} {{\bm{\bitem \add \val:\buf_1}{\mem_1}}}{\boolms}\cons \cfstack},
    \end{equation*}
    where $\add = \toAdd{\sem \expr_{\regmap_1,\lay}}$ and $\val = {\sem \exprtwo_{\regmap_1,\lay}}$.
    To show the claim, we observe that 
    \begin{equation*}
      \step {\conf{\frame {\cmemass \expr \exprtwo\sep\cmd} {\regmap_1} {\km[\syscall]}:\st_1}, \overline {\bm{\buf_1}{\mem_1}}}{}\\
      {\conf{\frame {\cmd_1} {\regmap_1} {\km[\syscall]}:\st_1, \update{\overline {\bm{\buf_1}{\mem_1}}}\add\val}},
    \end{equation*}
    and we know that $\overline {\bm{\bitem \add \val:\buf_1}{\mem_1}} = \update{\overline {\bm{\buf_1}{\mem_1}}}\add\val$ by definition of $\overline \cdot$.
    \proofcase{\ref{SI:Load}} We rewrite the assumption as follows:
    \begin{equation*}
      \sstep {\sframe{\frame {\cmemread \vx \expr\sep\cmd} {\regmap_1} {\km[\syscall]}:\st_1} {\bm{\buf_1}{\mem_1}}{\boolms}\cons \cfstack} {} \dstep \obs\\
      {\sframe{\frame {\cmd} {\update{\regmap_1}\vx {\val}} {\km[\syscall]}:\st_1} {{\bm{\buf_1}{\mem_1}}}{\boolms}\cons \cfstack},
    \end{equation*}
    where $\val = \bufread {\bm{\buf_1}{\mem_1}} {\toAdd{\sem \expr_{\regmap_1,\lay}}} 0$.
    To show the claim, it suffices to apply \Cref{rem:bufreadoverline}.
  \end{proofcases}
\end{proof}

\subsubsection{Technical Observations on the $\optfencetrans$ Transformation}
\label{sec:appoptfencetrans}

The next result is crucial to show that $\optfencetrans$ imposes \emph{speculative kernel safety}.
It relies on a predicate $\optfencepred$ that is defined as follows:

\[
  \infer{\optfencepred(\sframe{\frame{\cmd\sep \cmdtwo}{\regmap}{\opt}\cons\st}{\bmem}{\boolms}\cons\cfstack)}
  {
    \cmd\neq \ccall \expr{\vec \exprtwo}
    &
    \cmd\in \{\cmemread \vx \expr, \cmemass \expr \exprtwo, \cscall \expr {\vec \exprtwo}\} \Rightarrow \boolms = \bot}
\]





\begin{lemma}
  \label{lemma:optfenceonestep}
  For every system $\system = (\rfs'', \syss, \caps)\in \im(\optfencetrans)$,
  directive  $\dir \neq \dbt$, and
  speculative stack of configurations $\sframe{\frame{\prog}{\regmap}{\km[\syscall]}\cons\st}{\bm {\buf} {\lay \lcomp \rfs}}{\boolms}\cons\cfstack$ such that:
  \begin{itemize}
    \item[(H1)]\label{hypo:opth1} $\optfencepred(\sframe{\frame{\prog}{\regmap}{\km[\syscall]}\cons\st}{\bm {\buf} {\lay \lcomp \rfs}}{\boolms}\cons\cfstack)$,
    \item[(H2)]\label{hypo:opth2} $\optfwf {\sframe{\frame{\prog}{\regmap}{\km[\syscall]}\cons\st}{\bm {\buf} {\lay \lcomp \rfs}}{\boolms}\cons\cfstack}$,
    \item[(H3)]\label{hypo:opth3} $\lnot (\nf {{\sframe{\prog}{\regmap}{\km[\syscall]}\cons\st}{\bm {\buf} {\lay \lcomp \rfs}}{\boolms}\cons\cfstack} \dir)$,
  \end{itemize}
  there exists a set of configurations $Z$ containing:
  \begin{itemize}
  \item a speculative stack of configurations $\sframe{\frame{\prog'}{\regmap'}{\km[\syscall]}\cons\st'}{\bm {\buf'} {\lay \lcomp {\rfs'}}}{\boolms'}\cons\cfstack'$  such that:
    \begin{itemize}
    \item[(C1)] $\sframe{\frame{\prog'}{\regmap'}{\km[\syscall]}\cons\st'}{\bm {\buf'} {\lay \lcomp {\rfs'}}}{\boolms'}\cons\cfstack' \in \optfencepred$,
    \item[(C2)] $\optfwf {\sframe{\frame{\prog'}{\regmap'}{\km[\syscall]}\cons\st'}{\bm {\buf'} {\lay \lcomp {\rfs'}}}{\boolms'}\cons\cfstack'}$;
    \end{itemize}
  \item a speculative stack of configurations $(\err, \bot)\cons\cfstack'$;
  \item the configuration $\unsafe$,
  \end{itemize}
  and either:
  \[
    \sstep {\sframe{\frame{\prog}{\regmap}{\km[\syscall]}\cons\st}{\bm {\buf} {\lay \lcomp \rfs}}{\boolms}\cons\cfstack} {} \dir \obs z \text{ with }z \in Z,
    \tag{C3A}
  \]
  or
  \[
    \sstep {\sframe{\frame{\prog}{\regmap}{\km[\syscall]}\cons\st}{\bm {\buf} {\lay \lcomp \rfs}}{\boolms}\cons\cfstack} {} \dir \obs
    \cfstack'' \sto \dstep \onone
    z\text{ with }z \in Z\text{ and }\cfstack'' \in \optfencepred.
    \tag{C3B}
  \]
  for some $\obs \in \Obs$
\end{lemma}
\begin{proof}
  By (H2), we deduce that: 
  \[
    \prog = \optfencetrans[m][e] (\cmd); \vec \progtwo,
  \]
  where $\vec\progtwo$ is a (possibly empty) sequence of commands $\progtwo_1, \ldots, \progtwo_h$ such that for every $1 \le i \le h$, it holds that 
  \[
    \progtwo_i = \optfencetrans[\top][\bot] (\cmd_i).
  \]
  We proceed by cases on $\vec \progtwo$:
  \begin{proofcases}
    \proofcase{$\vec \progtwo=\cnil$} The proof proceeds by cases on $\cmd$
    \begin{proofcases}
      \proofcase{$\cmd = \cnil$}
      In this case, from (H3) and by analyzing the applicable rules, we deduce that the only directive compatible with this command is $\dstep$, and that the transition rule applied must be \ref{SI:Pop}. Furthermore, by inspecting this rule, we conclude that $\st$ is not empty. Consequently, from (H2), we derive that  
      \(
      \st = {\frame{\optfencetrans[\top][\bot] (\cmd_1)}{\regmap_1}{\km[\syscall]}\cons \ldots \cons \frame{\optfencetrans[\top][\bot] (\cmd_k)}{\regmap_k}{\km[\syscall]}}
      \).
      Thus, the only reachable stack of configurations is 
      \(
      \sframe{\st}{\bm {\buf} {\lay \lcomp \rfs}}{\boolms}\cons\cfstack.
      \)
      We include this configuration in $Z$, while the other elements of $Z$ are not relevant for this part of the proof and can be instantiated with any suitable values.
      This stack satisfies conditions (C1), (C2), and (C3A).
      Claim (C1) holds because, by analyzing the definition of $\optfencetrans[\top][\bot](\cdot)$, we deduce that $\optfencetrans[\top][\bot] (\cmd_1)$ cannot be a member of the set $\{\cmemread \vx \expr, \cmemass \expr \exprtwo, \cscall \expr {\vec \exprtwo}\}$. Claim (C2) follows from (H2).

      \proofcase{$\cmd = \stat; \cmdtwo$} The proof proceeds by case analysis on the instruction $\stat$.
      \begin{proofcases}
        \proofcase{$\stat = \cskip$} We begin by noting that 
        \(
          \optfencetrans[m][e](\cskip;\cmdtwo) = \cskip; \optfencetrans[m][e] (\cmdtwo)
        \).
        Similar to the case where $\cmd = \cnil$, from (H3) and by examining the applicable rules, we conclude that the only directive compatible with this command is $\dstep$, and that the transition rule applied must be \ref{SI:Skip}. In addition, the resulting configuration stack is 
        \[
          \sframe{\frame{\optfencetrans[m][e] (\cmdtwo)}{\regmap}{\opt}\cons\st}{\bm {\buf} {\lay \lcomp \rfs}}{\boolms}\cons\cfstack,
        \]
        which we include in $Z$, along with any other suitable configurations that are not relevant to this part of the proof. This stack satisfies conditions (C1), (C2), and (C3A).
        To establish (C1), we proceed by case analysis on the value of $m$.
        \begin{proofcases}
          \proofcase{$m = \top$} By definition of $\optfencetrans[\top][e](\cdot)$, we infer that $\optfencetrans[m][e] (\cmdtwo)$ cannot belong to the set $\{\cmemread \vx \expr, \cmemass \expr \exprtwo, \cscall \expr {\vec \exprtwo}\}$ and it cannot be a call instruction.
          
          \proofcase{$m = \bot$} In this case, (H2) implies that $\boolms = \bot$, hence (C1) holds trivially.
        \end{proofcases}

        Finally, claim (C2) follows directly from (H2), since this transition does not alter the state.
        \proofcase{$\stat = \vx\ass \expr$} This case is analogous to the previous one, with the only difference that the transition causes a modification to the register map, which does not influence the proof. 
        \proofcase{$\stat = \cmemass \expr \exprtwo$} the proof proceeds by case analysis on $m$.
        \begin{proofcases}
          \proofcase{$m = \top$} In this case, by examining the definition of $\optfencetrans[\top][e](\cdot)$, we deduce that $\optfencetrans[\top][e](\cmemass \expr \exprtwo;\cmdtwo) = \cfence; \cmemass \expr \exprtwo \sep \optfencetrans[\bot][\bot] (\cmdtwo)$. From (H3), and by introspection of the rules, we infer that $\dir =\dstep$ and that the rule for showing the first transition is \ref{SI:Fence}. Following this transition, another transition occurs, again with the directive $\dstep$. As a result, the entire reduction has the following structure:
          \[
            \sstep {\sframe{\frame{\cfence; \cmemass \expr \exprtwo \sep \optfencetrans[\bot][\bot] (\cmdtwo)}{\regmap}{\km[\syscall]}\cons\st}{\bm {\buf} {\lay \lcomp \rfs}}{\boolms}\cons\cfstack} {} \dir \obs
            \cfstack'' \sto \dstep \onone
            z
          \]
          Depending on the rule applied in the final transition (specifically, \ref{SI:Store}, \ref{SI:Store-Err}, or \ref{SI:Store-Unsafe}), the final configuration $z$ can take one of the following forms:
          \begin{itemize}
          \item     $\sframe{\frame{\optfencetrans[\bot][\bot] (\cmdtwo)}{\regmap}{\km[\syscall]}\cons\st} {\bm {\bitem {\sem \expr_{\regmap, \lay} }{\sem \exprtwo_{\regmap, \lay} }} {\overline {\bm {\buf} {\lay \lcomp \rfs}}} }{\bot}\cons\cfstack$,
          \item     $(\err, {\bot})\cons\cfstack$, or
          \item     $\unsafe$.
          \end{itemize}
          These three possible outcomes are the elements of the set $Z$. It is important to note that the existence of the initial reduction step---a direct consequence of (H3)---implies that $\boolms = \bot$. Consequently, $\cfstack''$ satisfies the $\optfencepred$ condition.
          Moreover, the final configuration can be any element of $Z$, depending on the value of $\toAdd{\sem \expr_{\regmap, \lay}}$.
          Finally, we must show that $\sframe{\frame{\optfencetrans[\bot][\bot] (\cmdtwo)}{\regmap}{\km[\syscall]}\cons\st} {\bm {\bitem {\sem \expr_{\regmap, \lay} }{\sem \exprtwo_{\regmap, \lay} }} {\overline {\bm {\buf} {\lay \lcomp \rfs}}} }{\bot}\cons\cfstack$ satisfies (C1) and (C2).
          Claim (C1) holds trivially, because the mis-speculation flag of the target configuration is $\bot$, and the image of $\optfencetrans$ cannot be a call instruction.
          For (C2),  we are required to establish:
          \begin{enumerate}
          \item ${\overline {\bm {\buf} {\lay \lcomp \rfs}}} = \lay\lcomp \rfs'$ for some $\rfs'$ such that $\optfwf {\rfs'}$.
        \item$\st = \frame {\cmd_0} {\regmap_0} {\opt_0}:\ldots:\frame {\cmd_k} {\regmap_k} {\opt_k} $ for $\cmd_0, \ldots, \cmd_k$ such that $\Sigma(\cmd_0, \top, \bot), \ldots,\Sigma(\cmd_k, \top, \bot)$.
          \item $\cmd = \optfencetrans[m][e](\cmd') \land \boolms \Rightarrow m \land e \Rightarrow \buf = \nil$
          \item $\dom({\bitem {\toAdd{\sem \expr_{\regmap, \lay}} }{\sem \exprtwo_{\regmap, \lay} }}) \subseteq \underline \lay(\Ar)$
          \end{enumerate}
          From (H2), we conclude that the domain of $\buf$ is a subset of $\Ar$. Consequently, we can apply \Cref{rem:overlinewrtdom} to establish that $\rfs' \eqon{\Fn} \rfs \eqon{\Fn} \rfs''$. This follows directly from (H2), and hence $\optfwf{\rfs'}$, given that $\system \in \im(\optfencetrans)$.
          Point (2) is a direct consequence of (H2), as the frame stack beneath the topmost configuration remains unchanged. Point (3) can be easily verified through a direct examination of the target configuration. Lastly, point (4) follows by introspection of the rule \ref{SI:Store}, which has been applied to demonstrate the transition.
 
          \proofcase{$m = \bot$} In this case, by analyzing the definition of $\optfencetrans[\bot][e](\cdot)$, we conclude that $\optfencetrans[\top][e](\cmemass \expr \exprtwo;\cmdtwo) = \cmemass \expr \exprtwo \sep \optfencetrans[\bot][\bot] (\cmdtwo)$. From (H3), and by examining the rules, we determine that $\dir = \dstep$, and the applicable transition rule must be one of \ref{SI:Store}, \ref{SI:Store-Err}, or \ref{SI:Store-Unsafe}. Consequently, the target configuration can be one of the following:
          \begin{itemize}
          \item $\sframe{\frame{\optfencetrans[\bot][\bot] (\cmdtwo)}{\regmap}{\km[\syscall]}\cons\st} {\bm {\bitem {\toAdd{\sem \expr_{\regmap, \lay}} }{\sem \exprtwo_{\regmap, \lay} }} {\overline {\bm {\buf} {\lay \lcomp \rfs}}} }{\bot}\cons\cfstack$,
          \item $(\err, {\bot})\cons\cfstack$, or
          \item $\unsafe$.
          \end{itemize}
          These three possible outcomes form the set $Z$. Note that from (H2), it follows that $\boolms = \bot$. From this point onward, the proof proceeds using the same strategy as in the previous case.

       \end{proofcases}
        \proofcase{$\stat = \cmemread \vx \expr$} The proof proceeds by case analysis on $m$.
        \begin{proofcases}
          \proofcase{$m=\top$} In this case, by analyzing the definition of $\optfencetrans[\top][e](\cdot)$, we deduce that $\optfencetrans[m][e](\cmemread \vx \expr;\cmdtwo) = \cfence; \cmemread \vx \expr \sep \optfencetrans[\bot][\top] (\cmdtwo)$. From (H3), and by analyzing the semantics, we conclude that $\dir = \dstep$, and the rule for the first transition must be \ref{SI:Fence}. Subsequently, another transition follows, using the directive $\dir' \in \{\dstep\} \cup \{\dload i \mid i \in \Nat\}$. This implies that the entire reduction has the following form:
          \[
            \sstep {\sframe{\frame{\cfence; \cmemread \vx \expr \sep \optfencetrans[\bot][\top] (\cmdtwo)}{\regmap}{\km[\syscall]}\cons\st}{\bm {\buf} {\lay \lcomp \rfs}}{\boolms}\cons\cfstack} {} \dstep \onone
            \cfstack'' \sto {\dir'} \obs
            z\text{ for some  }\obs, z.          
          \]
          Depending on the rule that is applied in the last transition, the final configuration can be one of the following:
          \begin{itemize}
          \item     $\sframe{\frame{\optfencetrans[\bot][\top] (\cmdtwo)}{\update \regmap \vx  {\overline {\bm {\buf} {\lay \lcomp \rfs}}(\toAdd{\sem \expr_{\regmap, \lay}})}}{\km[\syscall]}\cons\st} {\bm {\nil} {\overline {\bm {\buf} {\lay \lcomp \rfs}}} }{\bot}\cons\cfstack$, if the rule is \ref{SI:Load} or \ref{SI:Load-Step},
          \item     $(\err, {\bot})\cons\cfstack$, if the rule is \ref{SI:Load-Err} or
          \item     $\unsafe$, if the rule was \ref{SI:Load-Unsafe}.
          \end{itemize}
          We take these three possible outcomes as the elements of the set $Z$.
          Notice that the existence of the first reduction step---which, in turn, follows from (H3)---implies that $\boolms = \bot$. Therefore, $\cfstack''$ satisfies $\optfencepred$. 

          Finally, we must establish that the configuration $\sframe{\frame{\optfencetrans[\bot][\top] (\cmdtwo)}{\update \regmap \vx  {\overline {\bm {\buf} {\lay \lcomp \rfs}}(\sem \expr_{\regmap, \lay})}}{\km[\syscall]}\cons\st} {\bm {\nil} {\overline {\bm {\buf} {\lay \lcomp \rfs}}} }{\bot}\cons\cfstack$ satisfies both (C1) and (C2). Claim (C1) holds trivially because the mis-speculation flag of the target configuration is $\bot$. Additionally, (C2) can be established using the same strategy as in the previous case.

          \proofcase{$m=\bot$} In this case, by analyzing the definition of $\optfencetrans[\bot][e](\cdot)$, we deduce that  $\optfencetrans[\bot][e](\cmemread \vx \expr;\cmdtwo) = \cmemread \vx \expr \sep \optfencetrans[\lnot e][e] (\cmdtwo)$.
          From (H3) and by analyzing the semantics, we conclude that the applicable transition rule must be one of \ref{SI:Load}, \ref{SI:Load-Step}, \ref{SI:Load-Err}, or \ref{SI:Load-Unsafe}, and the directive is $\dir \in \{\dstep\} \cup \{\dload i \mid i \in \Nat\}$.
          Depending on the applied rule, the target configuration can be one of the following:
          \begin{itemize}
          \item     $\sframe{\frame{\optfencetrans[\lnot e][e] (\cmdtwo)}{\update \regmap \vx  \val}{\km[\syscall]}\cons\st} { { {\bm {\buf} {\lay \lcomp \rfs}}} }{\boolms \lor f}\cons\cfstack$, where ${\bufread {\bm {\buf} {\lay \lcomp \rfs}} {\sem \expr_{\regmap, \lay}} i} = \val , f$ for some $i \in \Nat$, if the rule is \ref{SI:Load} or \ref{SI:Load-Step},
          \item     $(\err, {\bot})\cons\cfstack$, if the rule is \ref{SI:Load-Err} or
          \item     $\unsafe$, if the rule was \ref{SI:Load-Unsafe}.
          \end{itemize}
          We take these three possible outcomes as the elements of the set $Z$.  Finally, we are required to show that $\sframe{\frame{\optfencetrans[\lnot e][e] (\cmdtwo)}{\update \regmap \vx  \val}{\km[\syscall]}\cons\st} { { {\bm {\buf} {\lay \lcomp \rfs}}} }{\boolms \lor f}\cons\cfstack$ satisfies (C1) and (C2). The proof proceeds by case analysis on $e$:
          \begin{proofcases}
            \proofcase{$e= \bot$} In this case, by analyzing the definition of $\optfencetrans[\top][\bot](\cdot)$, we deduce that $\optfencetrans[\top][\bot] (\cmdtwo)$ cannot be a member of the set $\{\cmemread \vx \expr, \cmemass \expr \exprtwo, \cscall \expr {\vec \exprtwo}\}$. This establishes (C1). Moreover, (C2) follows directly from (H2).
            \proofcase{$e= \top$} Condition (C1) follows from the fact that $\boolms = \bot$ and the definition of $\optfencetrans$. From (H2), we deduce that $\buf = \nil$. Then, by analyzing the definition of                ${\bufread {\bm {\buf} {\lay \lcomp \rfs}} {\sem \expr_{\regmap, \lay}} i}$,
            we conclude that $f = \bot$. Consequently, (C2) follows from (H2) and the fact that $f \lor \boolms = \bot \land \buf = \nil$.
          \end{proofcases}
        \end{proofcases}
        \proofcase{$\stat = \ccall  \expr {\vec \exprtwo}\sep \cmdtwo$} By analyzing the definition of $\optfencetrans[m][e](\cdot)$, we deduce that $\optfencetrans[m][e]( \cscall \expr {\vec \exprtwo}) = \cfence; \cscall  \expr {\vec \exprtwo} \sep \optfencetrans[\top][\bot] (\cmdtwo)$.
        From (H3) and by inspecting the rules, we conclude that the first transition is governed by \ref{SI:Fence}, followed by another transition with the directive $\dstep$. The complete reduction takes the form:
        \[
          \sstep {\sframe{\frame{\cfence; \cscall  \expr {\vec \exprtwo} \sep \optfencetrans[\top][\bot] (\cmdtwo)}{\regmap}{\km[\syscall]}\cons\st}{\bm {\buf} {\lay \lcomp \rfs}}{\boolms}\cons\cfstack} {} \dstep \onone
          \cfstack'' \sto {\dstep} \obs
          z\text{ for some  }\obs, z.          
        \]
        Depending on the rule applied in the last transition, the final configuration can be:
        \begin{itemize}
        \item     $\sframe{\frame{\optfencetrans[\top][\bot] (\cmdtwo')}{\update {\regmap_0} {\vec \vx}  {\sem {\vec \exprtwo}_{\regmap, \lay}}} {\km[\syscall]}\cons
            \frame{\optfencetrans[\top][\bot] (\cmdtwo)}{\regmap}{\km[\syscall]}
            \cons\st} {\bm {\nil} {\overline {\bm {\buf} {\lay \lcomp \rfs}}} }{\bot}\cons\cfstack$, if the rule is \ref{SI:Call} or \ref{SI:Call-Step}, \DD{Notice that here we are approximating: we can compile functions starting with $e = \top, m=\bot$. In practice we do not do it because it is easier to implement.}
        \item     $(\err, {\bot})\cons\cfstack$, if the rule is \ref{SI:Call-Err} or
        \item     $\unsafe$, if the rule was \ref{SI:Call-Step-Unsafe}.
        \end{itemize}
        These three possible outcomes form the set $Z$. Given the existence of the first reduction step---ensured by (H3)---it follows that $\boolms=\bot$, implying that $\cfstack''$ satisfies $\optfencepred$.
        Next, we must verify that $\sframe{\frame{\optfencetrans[\top][\bot] (\cmdtwo')}{\update {\regmap_0} {\vec \vx}  {\sem {\vec \exprtwo}_{\regmap, \lay}}} {\km[\syscall]}\cons
          \frame{\optfencetrans[\top][\bot] (\cmdtwo)}{\regmap}{\km[\syscall]}
          \cons\st} {\bm {\nil} {\overline {\bm {\buf} {\lay \lcomp \rfs}}} }{\bot}\cons\cfstack$ satisfies (C1) and (C2). Claim (C1) holds trivially, because the mis-speculation flag of the target configuration is $\bot$. Claim (C2) requires proving the following:
        \begin{enumerate}
        \item ${\overline {\bm {\buf} {\lay \lcomp \rfs}}} = \lay\lcomp \rfs'$ for some $\rfs'$ such that $\optfwf {\rfs'}$.
        \item All  frames in the stack of the final configuration carry commands  $\cmd_0, \ldots, \cmd_k$ such that $\Sigma(\cmd_0, \top, \bot), \ldots,\Sigma(\cmd_k, \top, \bot)$.
        \item $ \bot \Rightarrow \top  \land \bot \Rightarrow \nil = \nil$
        \item $\dom(\nil) \subseteq \underline \lay(\Ar)$
        \end{enumerate}
        The proof of point (1) is analogous to the one given for stores. Point (2) is a consequence of (H2), and follows by introspection of the target configuration. Points (3) and (4) (where we already replaced $m, e, \boolms$ and $\mu$ with their actual value for this sub-derivation) are trivial.
        
        \proofcase{$\stat =\cif \expr {\cmd_1} {\cmd_2} ; \cmdtwo$}
        By analyzing the definition of $\optfencetrans[m][e](\cdot)$, we deduce that:
        \[
          \optfencetrans[m][e](\cif \expr {\cmd_1} {\cmd_2} ; \cmdtwo) = \cif \expr {\optfencetrans[\top][\bot](\cmd_1)}{\optfencetrans[\top][\bot](\cmd_2)} \sep \optfencetrans[\top][\bot] (\cmdtwo).
        \]
        From (H3) and by inspecting the rules, we conclude that the rule for the first transition must be either \ref{SI:If} or \ref{SI:If-Branch}. The reduction therefore takes the following form:
        \begin{multline*}
          {\sframe{\frame{\cif \expr {\optfencetrans[\top][\bot](\cmd_1)}{\optfencetrans[\top][\bot](\cmd_2)} \sep \optfencetrans[\top][\bot] (\cmdtwo)}{\regmap}{\km[\syscall]}\cons\st}{\bm {\buf} {\lay \lcomp \rfs}}{\boolms}\cons\cfstack} {}\sto \dir {\obranch b}\\
          {\sframe{\frame{\optfencetrans[\top][\bot](\cmd_i) \sep \optfencetrans[\top][\bot] (\cmdtwo)}{\regmap}{\km[\syscall]}\cons\st}{\bm {\buf} {\lay \lcomp \rfs}}{\boolms}\cons\cfstack'}
          \text{ for some  }\cfstack', b, \dir, i.
        \end{multline*}
        No $\err$ or $\unsafe$ state can be reached in this case, so we define a set $Z$ containing the above target configuration and some arbitrary configurations to comply with the requirements over $Z$. 
        We now need to show that
        \(
          {\sframe{\frame{\optfencetrans[\top][\bot](\cmd_i) \sep \optfencetrans[\top][\bot] (\cmdtwo)}{\regmap}{\km[\syscall]}\cons\st}{\bm {\buf} {\lay \lcomp \rfs}}{\boolms}\cons\cfstack}
        \)
        satisfies (C1) and (C2).
        The validity of (C1) follows directly from the definition of the function $\optfencetrans[\top][\bot](\cdot)$.
        To verify (C2), we must prove the following:
        \begin{enumerate}
        \item ${\lay \lcomp \rfs} \eqon{\Fn} \lay\lcomp \rfs''$, which directly follows from (H2).
        \item The stack $\st$ is of the form $\frame {\cmd_0} {\regmap_0} {\opt_0}:\ldots:\frame {\cmd_k} {\regmap_k} {\opt_k}$, where $\cmd_0, \ldots, \cmd_k$ satisfies $\Sigma(\cmd_0, \top, \bot),\\ \ldots,\Sigma(\cmd_k, \top, \bot)$.  This follows trivially from (H2), as the frame stack remains unchanged.
        \item $\Sigma(\optfencetrans[\top][\bot](\cmd_i) \sep \optfencetrans[\top][\bot] (\cmdtwo))$, that is trivial.
        \item $\dom(\buf) \subseteq \underline \lay(\Ar)$, which is a consequence of (H2).
        \end{enumerate}
      \end{proofcases}
    \end{proofcases}
    \proofcase{$\vec \progtwo \neq \cnil$} This remaining case is analogous to the case where $\vec \progtwo=\cnil$ and $\cmd\neq \cnil$, as the premises (H1), (H2), (H3) also hold for the first element of $\vec \progtwo$.  
  \end{proofcases}
\end{proof}

\begin{lemma}
  \label{lemma:optfencensteps}
  Let $\system = (\rfs, \syss, \caps)$ be a system. For every configuration stacks $\conf{\conf{\prog, \regmap, \km[\syscall]}, \bm\buf{(\lay \lcomp \rfs')}, \boolms}\cons \cfstack$ and $\cfstack'$, directives $\Ds$, observations $\Os$, and number of steps $\nat$, if the following conditions hold:
  \begin{itemize}
  \item   \(
    \conf{\conf{\prog, \regmap, \km[\syscall]}, \bm\buf{(\lay \lcomp \rfs')}, \boolms}\cons \cfstack \sto{\Ds}{\Os}^n \cfstack',
    \)
  \item $\cfstack' \notin \{\unsafe, (\err, \boolms'):\cfstack''\}$,
  \item $\dbt\notin \Ds$,
  \item $\optfencepred(\conf{\conf{\prog, \regmap, \km[\syscall]}, \bm\buf{(\lay \lcomp \rfs')}, \boolms}\cons \cfstack)$,
  \item $\optfwf {\conf{\conf{\prog, \regmap, \km[\syscall]}, \bm\buf{(\lay \lcomp \rfs')}, \boolms}\cons \cfstack}$,
  \end{itemize}
  then we have  $\optfencepred(\cfstack')$.
\end{lemma}
\begin{proof}
  Direct consequence of \Cref{lemma:optfenceonestep}.
\end{proof}

\begin{lemma}
  \label{lemma:optfencesyscall}
  Let $\optfencetrans(\system) = (\rfs, \syss, \caps)$ be a system. For every configuration stacks $\conf{\conf{\syss(\syscall), \regmap, \km[\syscall]}, \bm\buf{(\lay \lcomp \rfs')}, \boolms}$ and $\cfstack$, directives $\Ds$, observations $\Os$, and number of steps $\nat$, if the following conditions hold:
  \begin{itemize}
  \item \(
      \lay \red[\optfencetrans(\system)] \conf{\conf{\syss(\syscall), \regmap, \km[\syscall]}, \bm\buf{(\lay \lcomp \rfs')}, \boolms} \sto{\Ds}{\Os}^n \cfstack,
    \)
    \item $\cfstack \notin \{\unsafe, (\err, \boolms'):\cfstack''\}$,
    \item $\rfs'\eqon{\Fn}\rfs$,  $\dom(\buf)\subseteq\underline \lay(\Ar)$, and
    \item $\dbt\notin \Ds$,
    \end{itemize}
    then we have $\optfencepred(\cfstack)$.
\end{lemma}
\begin{proof}
  By the definition of $\optfencetrans(\system)$, the body of a syscall is translated with initial flags $m = \top$ and $e = \bot$. This implies that the initial configuration has the form:
  \[
    \conf{\conf{\optfencetrans[\top][\bot](\cmd), \regmap, \km[\syscall]}, \bm\buf{(\lay \lcomp \rfs')}, \boolms}.
  \]
  By analyzing the definition of $\optfencetrans[\top][\bot](\cdot)$, we observe that $\optfencetrans[\top][\bot](\cmd)$ cannot be a potentially unsafe command. Therefore
  \(
    \optfencepred(\conf{\conf{\optfencetrans[\top][\bot](\cmd), \regmap, \km[\syscall]}, \bm\buf{(\lay \lcomp \rfs')}, \boolms})
  \)
  holds.
  Additionally, we have:
  \[
    \optfwf {\conf{\conf{\optfencetrans[\top][\bot](\cmd), \regmap, \km[\syscall]}, \bm\buf{(\lay \lcomp \rfs')}, \boolms}}
  \]
  due to the assumptions on $\buf$ and $\rfs'$. The claim then follows by applying \Cref{lemma:optfencensteps}.
\end{proof}
\begin{lemma}
  \label{lemma:optfencesafeimpo}
  The transformation $\optfencetrans$ imposes \emph{speculative kernel safety}.
\end{lemma}
\begin{proof}
  More precisely, we need to show that: given $\system = (\rfs, \syss, \caps)$, 
  for every buffer $\buf$ with $\dom(\buf) \subseteq \underline \lay(\Ar)$
  and store $\rfs' \eqon{\Fun} \optfencetrans(\rfs)$, if
  \[
    \lay \red[\systrans(\system)] \conf{\conf{\syss(\syscall), \regmap, \km[\syscall]}, \bm\buf{(\lay \lcomp \rfs')}, \boolms} \sto{\Ds}{\Os}^* \unsafe,
  \]
  then
  \[
    \lay \red[\systrans(\system)] \conf{\conf{\syss(\syscall), \regmap, \km[\syscall]}, \overline{\bm\buf{(\lay \lcomp \rfs')}}} \to^* \unsafe.
  \]

  We start by observing that, due to \Cref{lemma:nobt}, we can assume without loss of generality that $\dbt \notin \Ds$. By analyzing the semantics, the final configuration preceding the $\unsafe$ state must have carried one of the following commands: 
  $\cmemass \expr \exprtwo,  \cmemread \vx \expr,  \ccall \expr{\vec \exprtwo}$, or $\cscall \expr{\vec \exprtwo}$.
  Moreover, the transition to the $\unsafe$ state can be performed using any of the rules \ref{SI:Load-Unsafe}, \ref{SI:Store-Unsafe}, \ref{SI:Call-Step-Unsafe}, or \ref{SI:Call-Unsafe}.
  By applying \Cref{lemma:optfencesyscall}, we can exclude the last case (\ref{SI:Call-Unsafe}). Furthermore, we know that if the command is any of the other three, the mis-speculation flag of the final configuration before $\unsafe$ is $\bot$. 
  Therefore, the conclusion of this proof follows similarly to the argument in \Cref{lemma:fencesafeimpo}. By analyzing the rule used in the last transition and applying a combination of \Cref{lemma:nobtsteponly} and \Cref{lemma:stepsemsim}, we demonstrate that under the speculative semantics, the configuration:
  \[
    \conf{\conf{\syss(\syscall), \regmap, \km[\syscall]}, \overline{\bm\buf{(\lay \lcomp \rfs')}}}
  \]
  can reach a configuration that satisfies the premises of the corresponding speculative unsafe rules \ref{WL:Load-Unsafe}, \ref{WL:Store-Unsafe}, or \ref{WL:Call-Unsafe}.
  Thus, we conclude that it also reaches the $\unsafe$ state.
\end{proof}
\subsubsection{Technical Observations on the $\nospec$ Transformation}
\label{sec:apponospectrans}

\begin{lemma}
  \label{lemma:nospeconestep}
  For every system $\system = (\rfs'', \syss, \caps)\in \im(\nospec)$,
  directive  $\dir \neq \dbt$, and
  speculative stack of configurations $\sframe{\frame{\cmd}{\regmap}{\km[\syscall]}\cons\st}{\bm {\buf} {\lay \lcomp \rfs}}{\boolms}\cons\cfstack$ such that:
  \begin{itemize}
    \item[(H1)]\label{hypo:opth2} $\nsfwf {\sframe{\frame{\cmd}{\regmap}{\km[\syscall]}\cons\st}{\bm {\buf} {\lay \lcomp \rfs}}{\boolms}\cons\cfstack}$,
    \item[(H2)]\label{hypo:opth3} $\lnot (\nf {{\sframe{\cmd}{\regmap}{\km[\syscall]}\cons\st}{\bm {\buf} {\lay \lcomp \rfs}}{\boolms}\cons\cfstack} \dir)$,
  \end{itemize}
  there is a set $Z$ containing:
  \begin{itemize}
  \item a speculative stack of configurations $\sframe{\frame{\cmd'}{\regmap'}{\km[\syscall]}\cons\st'}{\bm {\buf'} {\lay \lcomp {\rfs'}}}{\boolms'}\cons\cfstack'$ and an observation $\obs$ such that
    \begin{itemize}
    \item[(C1)] $\nsfwf {\sframe{\frame{\cmd'}{\regmap'}{\km[\syscall]}\cons\st'}{\bm {\buf'} {\lay \lcomp {\rfs'}}}{\boolms'}\cons\cfstack'}$,
    \end{itemize}
  \item a speculative stack $(\err, \bot)\cons\cfstack'$,
  \item the configuration $\unsafe$,
  \end{itemize}
  and either:
  \[
    \sstep {\sframe{\frame{\cmd}{\regmap}{\km[\syscall]}\cons\st}{\bm {\buf} {\lay \lcomp \rfs}}{\boolms}\cons\cfstack} {} \dir \obs z \text{ with }z \in Z,
    \tag{C2A}
  \]
  \[
    \sstep {\sframe{\frame{\cmd}{\regmap}{\km[\syscall]}\cons\st}{\bm {\buf} {\lay \lcomp \rfs}}{\boolms}\cons\cfstack} {} \dir \obs
    \cfstack'' \sto \dstep \onone
    z\text{ with }z \in Z\text{ and }\cfstack'' \in \optfencepred,
    \tag{C2B}
  \]
  or
  \[
     {\sframe{\frame{\cmd}{\regmap}{\km[\syscall]}\cons\st}{\bm {\buf} {\lay \lcomp \rfs}}{\boolms}\cons\cfstack} {}  \sto \dir \obs \cfstack'', \cfstack'' \in \optfencepred \text{ and } \nf {\cfstack''}{}
    \tag{C2C}
  \]
\end{lemma}
\begin{proof}
  The proof goes by cases on $\cmd$:
  \begin{proofcases}
    \proofcase{$\cmd = \cnil$} In this case, from (H3) and by introspection of the rules, we deduce that the only directive compatible with this command is $\dstep$, and that the transition rule that is applied must be \ref{SI:Pop}. By analyzing this rule, we also deduce that $\st$ is not empty. Therefore, from (H1), we conclude that  $\st = {\frame{\nospec (\cmd_1)}{\regmap_1}{\km[\syscall]}\cons \ldots \cons \frame{\nospec (\cmd_k)}{\regmap_k}{\km[\syscall]}}$and hence the source configuration can only reach the target stack configuration $\sframe{\st}{\bm {\buf} {\lay \lcomp \rfs}}{\bot}\cons\cfstack$ which we put in $Z$, along with any two other suitable elements, which are not relevant for this part of the proof. The target configuration stack satisfies both (C1) and (C2A). Claim (C1) follows directly from (H1). Claim  (C2A) can be verified through introspection of the rule \ref{SI:Pop}. 
    \proofcase{$\cmd = \stat; \cmdtwo$} The proof proceeds by induction on the instruction $\stat$.
    \begin{proofcases}
      \proofcase{$\stat = \cskip$} We start by observing that $\nospec(\cskip;\cmdtwo) = \cskip; \nospec (\cmdtwo)$. From (H3) and by introspection of the semantics, we deduce that the only directive compatible with this command is $\dstep$, and that the transition rule that must be applied is \ref{SI:Skip}. Applying this rule, we obtain the target stack configuration $\sframe{\frame{\nospec (\cmdtwo)}{\regmap}{\opt}\cons\st}{\bm {\nil} {\lay \lcomp \rfs}}{\bot}\cons\cfstack$, which we include in the set $Z$ along with other two suitable configurations that are not relevant for this part of the proof. It is trivial to see that the target stack satisfies (C1), and (C2A).
      \proofcase{$\stat = \vx\ass \expr$} This case is analogous to the previous one, with the difference that the register file is modified, without influencing the proof. 
      \proofcase{$\stat = \cmemass \expr \exprtwo$}  By introspection of the definition of $\nospec$, we deduce that  $\nospec(\cmemass \expr \exprtwo;\cmdtwo) = \cmemass \expr \exprtwo \sep \cfence\sep \nospec (\cmdtwo)$. From (H2) and by introspection of the semantics, we identify that the first transition can be shown with either \ref{SI:Store}, \ref{SI:Store-Err}, or \ref{SI:Store-Unsafe}. The target configuration can take one of the following forms
      \begin{itemize}
      \item     $\sframe{\frame{\cfence \sep \nospec (\cmdtwo)}{\regmap}{\km[\syscall]}\cons\st} {{ {\bm {\bitem {\sem \expr_{\regmap, \lay} }{\sem \exprtwo_{\regmap, \lay} }} {\lay \lcomp \rfs}}} }{\bot}\cons\cfstack$. Notice that, by applying the \ref{SI:Fence} rule, this configuration transitions to
        \(
          \sframe{\frame{\nospec (\cmdtwo)}{\regmap}{\km[\syscall]}\cons\st} {\bm \nil {\overline {\bm {\bitem {\sem \expr_{\regmap, \lay} }{\sem \exprtwo_{\regmap, \lay} }} {\lay \lcomp \rfs}}} }{\bot}\cons\cfstack.
        \)
      \item     $(\err, {\bot})\cons\cfstack$, or
      \item     $\unsafe$.
      \end{itemize}
      We define the set $Z$ as follows:
      \[
        Z=\{\unsafe, (\err, {\bot})\cons\cfstack, \sframe{\frame{\nospec (\cmdtwo)}{\regmap}{\km[\syscall]}\cons\st} {\bm \nil {\overline {\bm {\bitem {\sem \expr_{\regmap, \lay} }{\sem \exprtwo_{\regmap, \lay} }} {\lay \lcomp \rfs}}} }{\bot}\cons\cfstack\}.
      \]
      If the target configuration is $(\err, {\bot}) \cons \cfstack$ or $\unsafe$, claim (C2A) holds trivially. Otherwise, we need to establish (C1) and (C2B).
      For (C2B), we have the transition:
      \begin{multline*}
        \sframe{\frame{\cfence \sep \nospec (\cmdtwo)}{\regmap}{\km[\syscall]}\cons\st} {{{\bm {\bitem {\sem \expr_{\regmap, \lay} }{\sem \exprtwo_{\regmap, \lay} }} {\lay \lcomp \rfs}}} }{\bot}\cons\cfstack \sto \dstep \onone\\
        \sframe{\frame{\nospec (\cmdtwo)}{\regmap}{\km[\syscall]}\cons\st} {\bm \nil {\overline {\bm {\bitem {\sem \expr_{\regmap, \lay} }{\sem \exprtwo_{\regmap, \lay} }} {\lay \lcomp \rfs}}} }{\bot}\cons\cfstack,
      \end{multline*}
      and we are required to show that $\optfencepred(\sframe{\frame{\cfence \sep \nospec (\cmdtwo)}{\regmap}{\km[\syscall]}\cons\st} {{{\bm {\bitem {\sem \expr_{\regmap, \lay} }{\sem \exprtwo_{\regmap, \lay} }} {\lay \lcomp \rfs}}} }{\bot}\cons\cfstack)$ holds, which is trivial. 
      For (C1), we need to establish:
      \[
        \nsfwf{          \sframe{\frame{\nospec (\cmdtwo)}{\regmap}{\km[\syscall]}\cons\st} {\bm \nil {\overline {\bm {\bitem {\sem \expr_{\regmap, \lay} }{\sem \exprtwo_{\regmap, \lay} }} {\lay \lcomp \rfs}}} }{\bot}\cons\cfstack
        }.
      \]
      To this aim, it is crucial to observe that $\nsfwf{\bm \nil {\overline {\bm {\bitem {\sem \expr_{\regmap, \lay} }{\sem \exprtwo_{\regmap, \lay} }} {\lay \lcomp \rfs}}} }$ holds.
      This follows from the fact that the premises of \ref{SI:Store} require $\sem \expr_{\regmap, \lay} \in \underline \lay(\Ark)$, making the claim a direct consequence of \Cref{rem:overlinewrtdom}.
      The other conditions required for (C1) follow directly by introspection of the configuration and from (H1).
      
      \proofcase{$\stat = \cmemread \vx \expr$} By introspection of the definition of $\nospec$, we deduce that $\nospec(\cmemread \vx \expr;\cmdtwo) = \cmemread \vx \expr \sep \nospec (\cmdtwo)$. From (H2), and by introspection of the semantics, we conclude that the first transition uses one of the following rules: \ref{SI:Load}, \ref{SI:Load-Step}, \ref{SI:Load-Err}, or \ref{SI:Load-Unsafe}. The directive employed of the transition is $\dir \in \{\dstep\}\cup\{\dload i\mid i \in \Nat\}$.
      Depending on the rule that is applied, the target configuration can take one of the following forms:
        \begin{itemize}
        \item     $\sframe{\frame{\nospec (\cmdtwo)}{\update \regmap \vx  \val}{\km[\syscall]}\cons\st} { { {\bm {\nil} {\lay \lcomp \rfs}}} }{f}\cons\cfstack$, where ${\bufread {\bm {\nil} {\lay \lcomp \rfs}} {\sem \expr_{\regmap, \lay}} i} = \val , f$ for some $i \in \Nat$, if the rule is \ref{SI:Load} or \ref{SI:Load-Step},
        \item     $(\err, {\bot})\cons\cfstack$, if the rule is \ref{SI:Load-Err} or
        \item     $\unsafe$, if the rule was \ref{SI:Load-Unsafe}.
        \end{itemize}
        We define the set $Z$ as follows:
        \[
          Z = \{\unsafe, (\err, {\bot}) \cons \cfstack, \sframe{\frame{\nospec (\cmdtwo)}{\update \regmap \vx \val}{\km[\syscall]}\cons\st} { {\bm {\nil} {\lay \lcomp \rfs}}} {f}\cons\cfstack\},
        \]
        and we observe that we are in the case where (C2A) holds. Finally, we must show that the configuration  $\sframe{\frame{\nospec (\cmdtwo)}{\update \regmap \vx  \val}{\km[\syscall]}\cons\st} { { {\bm {\nil} {\lay \lcomp \rfs}}} }{f}\cons\cfstack$ satisfies (C1). To establish (C1), we analyze definition of 
          ${\bufread {\bm {\buf} {\lay \lcomp \rfs}} {\sem \expr_{\regmap, \lay}} i}$,  deducing that that $f=\bot$ because the buffer is empty, so (C1) is satisfied.
          
      \proofcase{$\stat = \ccall  \expr {\vec \exprtwo}\sep \cmdtwo$} By introspection of the definition of $\nospec$, we deduce that $\nospec( \ccall \expr {\vec \exprtwo}) = \cscall  \expr {\vec \exprtwo} \sep \nospec (\cmdtwo)$. From (H2), and by analyzing the semantics, we observe that the configuration reduces using the directive $\dstep$. Therefore, the reduction takes the following form:
      \[
        \sstep {\sframe{\frame{\cscall  \expr {\vec \exprtwo} \sep \nospec (\cmdtwo)}{\regmap}{\km[\syscall]}\cons\st}{\bm {\nil} {\lay \lcomp \rfs}}{\bot}\cons\cfstack} {}  {\dstep} \obs
        z\text{ for some  }\obs, z.          
      \]
      Depending on the applied rule, the target configuration can be one of the following:
      \begin{itemize}
      \item     $\sframe{\frame{\lay \lcomp \rfs(\toAdd{\sem\expr_{\regmap, \lay}})}{\update {\regmap_0} {\vec \vx}  {\sem {\vec \exprtwo}_{\regmap, \lay}}} {\km[\syscall]}\cons
          \frame{\nospec (\cmdtwo)}{\regmap}{\km[\syscall]}
          \cons\st} {\bm {\nil} {\lay \lcomp \rfs} }{\bot}\cons\cfstack$, if the applied rule is \ref{SI:Call} or \ref{SI:Load-Step},
      \item     $(\err, {\bot})\cons\cfstack$, if the applied rule is \ref{SI:Call-Err} or
      \item     $\unsafe$, if the applied rule is \ref{SI:Call-Step-Unsafe}.
      \end{itemize}
      We collect these three possible configurations as the elements of the set $Z$, and we establish (C2A). If the target configuration is either the second or third one, the conclusion follows trivially. In the first case, we have $\toAdd{\sem\expr_{\regmap, \lay}} \in \underline \lay(\Addk)$. Therefore $\lay \lcomp \rfs(\toAdd{\sem\expr_{\regmap, \lay}})= \rfs(\fn)$ for some $\fn \in \Fnk$. Using (H1), we deduce that the function body must be of the form  $\nospec (\cmdtwo')$ for some $\cmdtwo$.
      Finally, we are required to show that $\sframe{\frame{\nospec (\cmdtwo')}{\update {\regmap_0} {\vec \vx}  {\sem {\vec \exprtwo}_{\regmap, \lay}}} {\km[\syscall]}\cons \frame{\nospec (\cmdtwo)}{\regmap}{\km[\syscall]} \cons\st} {\bm {\nil} {\lay \lcomp \rfs} }{\bot}\cons\cfstack$ satisfies (C1). To this aim, we need to prove:
      \[
        \nsfwf{\sframe{\frame{\nospec (\cmdtwo')}{\update {\regmap_0} {\vec \vx}  {\sem {\vec \exprtwo}_{\regmap, \lay}}} {\km[\syscall]}\cons \frame{\nospec (\cmdtwo)}{\regmap}{\km[\syscall]} \cons\st} {\bm {\nil} {\lay \lcomp \rfs} }{\bot}\cons\cfstack},
      \]
      which is a direct consequence of (H1).

      \proofcase{$\stat =\cif \expr {\cmd_1} {\cmd_2} ; \cmdtwo$} By introspection of the definition of $\nospec$, we deduce that:
      \[
        \nospec(\cif \expr {\cmd_1} {\cmd_2} ; \cmdtwo) = \cif \expr {\cfence; \nospec(\cmd_1)}{\cfence; \nospec(\cmd_2)} \sep \nospec (\cmdtwo).
      \]
      From (H2), and by introspection of the semantics, we observe that the rule governing the first transition is either \ref{SI:If} or \ref{SI:If-Branch}. Thus, the first reduction step is as follows:
      \begin{multline*}
        {\sframe{\frame{\cif \expr {\cfence; \nospec(\cmd_1)}{\cfence; \nospec(\cmd_2)} \sep \nospec (\cmdtwo)}{\regmap}{\km[\syscall]}\cons\st}{\bm {\nil} {\lay \lcomp \rfs}}{\bot}\cons\cfstack} {}\sto \dir {\obranch b}\\
        {\sframe{\frame{\cfence;\nospec(\cmd_i) \sep \nospec (\cmdtwo)}{\regmap}{\km[\syscall]}\cons\st}{\bm {\nil} {\lay \lcomp \rfs}}{\boolms}\cons\cfstack'}
        \text{ for some  }\cfstack', b, \boolms, i, \dir
      \end{multline*}
      Depending on the directive $\dir$ two cases arise:
      \begin{proofcases}
        \proofcase{$\dir \in \{\dstep, \dbranch {\sem \expr_{\regmap, \lay}}\}$} In this case, we have $\boolms = \bot$. Therefore, the \ref{SI:Fence} rule can be applied, leading to the following configuration:
        \[
          {\sframe{\frame{\nospec(\cmd_i \sep \cmdtwo)}{\regmap}{\km[\syscall]}\cons\st}{\bm {\nil} {\lay \lcomp \rfs}}{\boolms}\cons\cfstack}.
        \]
        This configuration satisfies (C2A) for a suitable set $Z$, whose precise definition is omitted. By introspection of this configuration and using (H1), we conclude that it also satisfies (C1).
        \proofcase{$\dir = \dbranch {\lnot  \sem \expr_{\regmap, \lay}}$} By analyzing the semantics, we observe that $\boolms = \top$, and therefore, the configuration cannot reduce further. In this case, (C2C) holds for some suitable set $Z$, whose definition is omitted for brevity.
      \end{proofcases}
    \end{proofcases}
  \end{proofcases}
\end{proof}

\begin{lemma}
  \label{lemma:nospecnsteps}
  Let $\nospec(\system) = (\rfs, \syss, \caps)$ be a system. For every configuration stacks $\conf{\conf{\prog, \regmap, \km[\syscall]}, \bm\buf{(\lay \lcomp \rfs')}, \boolms}\cons \cfstack$ and $\cfstack'$, directives $\Ds$,  observations $\Os$, and number of steps $\nat$, if the following conditions hold:
  \begin{itemize}
  \item
    \(
    \lay \red \conf{\conf{\prog, \regmap, \km[\syscall]}, \bm\buf{(\lay \lcomp \rfs')}, \boolms}\cons \cfstack \sto{\Ds}{\Os}^n \cfstack',
    \)
  \item $\cfstack' \notin \{\unsafe, (\err, \boolms'):\cfstack''\}$,
  \item $\dbt\notin \Ds$,
  \item $\nsfwf {\conf{\conf{\prog, \regmap, \km[\syscall]}, \bm\buf{(\lay \lcomp \rfs')}, \boolms}\cons \cfstack}$,
  \end{itemize}
  then we have  $\optfencepred(\cfstack')$.
\end{lemma}
\begin{proof}
  Notice that, from the premises of this lemma, assumptions (H1) and
  (H2) of \Cref{lemma:optfenceonestep} are satisfied.  Moreover, we
  observe that:
  \[
    \nsfwf {\conf{\conf{\prog, \regmap, \km[\syscall]}, \bm\buf{(\lay
          \lcomp \rfs')}, \boolms}\cons \cfstack} \implies
    \optfencepred(\cfstack), \tag{$*$}
  \]
  Therefore, the conclusion is a direct consequence of
  \Cref{lemma:optfenceonestep}, which ensures that either
  $\optfencepred(\cfstack')$ or $\nsfwf{\cfstack'}$ holds at every
  step.
  
\end{proof}

\begin{lemma}
  \label{lemma:nospecsyscall}
  Let $\nospec(\system) = (\rfs, \syss, \caps)$ be a system. For every configuration stacks $\conf{\conf{\syss(\syscall), \regmap, \km[\syscall]}, \bm\buf{(\lay \lcomp \rfs')}, \boolms}$ and $\cfstack$, directives $\Ds$,  observations $\Os$, and number of steps $\nat$, if the following conditions hold:
  \begin{itemize}
  \item \(
    \lay \red \conf{\conf{\syss(\syscall), \regmap, \km[\syscall]}, \bm\buf{(\lay \lcomp \rfs')}, \boolms} \sto{\Ds}{\Os}^n \cfstack
    \),
  \item$\cfstack \notin \{\unsafe, (\err, \boolms'):\cfstack''\}$,
  \item 
    $\rfs'\eqon{\Fn}\rfs$,
  \item   $\dom(\buf)\subseteq\underline \lay(\Ar)$, and 
  \item $\dbt\notin \Ds$,
  \end{itemize}
  then we have $\optfencepred(\cfstack)$.
\end{lemma}
\begin{proof}
  Notice that, by definition of $\nospec$, the body of system calls are translated with an initial $\cfence$ instruction. This implies that the initial configuration has the following form:
  \[
    \conf{\conf{\cfence; \nospec(\cmd), \regmap, \km[\syscall]}, \bm\buf{(\lay \lcomp \rfs')}, \boolms}.
  \]
  The proof proceeds by case analysis on $\nat$.
  \begin{proofcases}
    \proofcase{$\nat=0$} Trivial.
    \proofcase{$\nat>0$} By examining the semantics, we observe that $\boolms = \bot$, and that the first rule that is applied is \ref{SI:Fence}. Therefore the first transition is:
    \[
        \lay \red  \conf{\conf{\cfence; \nospec(\cmd), \regmap, \km[\syscall]}, \bm\buf{(\lay \lcomp \rfs')}, \bot} \to \conf{\conf{\nospec(\cmd), \regmap, \km[\syscall]}, \bm \nil {\overline {\bm\buf{(\lay \lcomp \rfs')}}}, \bot}.
    \]
    Consequently, we conclude with an application of \Cref{lemma:nospecnsteps}. In particular, the premise
    \[
      \nsfwf{\conf{\conf{\nospec(\cmd), \regmap, \km[\syscall]}, \bm \nil {\overline {\bm\buf{(\lay \lcomp \rfs')}}}, \bot}}
    \]
    is established as follows:
    \begin{itemize}
    \item By examining of the configuration, we note that the mis-speculation flag is $\bot$, and that both the write buffer and the frame stack are empty.
    \item By \Cref{rem:overlinewrtdom}, we conclude that ${\overline {\bm\buf{(\lay \lcomp \rfs')}}} \eqon{\Fn} {\rfs'}$.
    \end{itemize}
  \end{proofcases}
\end{proof}

\begin{lemma}
  \label{lemma:nospecsafeimpo}
  The transformation $\nospec$ imposes \emph{specuative kernel safety}.
\end{lemma}
\begin{proof}
  The proof is analogous to the one of \Cref{lemma:optfencesafeimpo}.
\end{proof}

    


\section{Discussion on SESES Performance Evaluation}
\label{app:seses}
In \Cref{sec:discussion}, we observe that SESES alone is insufficient to prevent the speculative attacks we consider there. However, we also note that combining SESES with the \texttt{lvi-cfi} pass and a BTB mitigation technique, such as retpoline or eIBRS, could prevent those attacks. Since SESES is not compatible with retpoline, we attempted to evaluate a configuration with SESES and eIBRS.

We made considerable efforts to measure the overhead of SESES under this configuration. Unfortunately, we encountered previously undocumented link-time bugs in \texttt{lvi-cfi} that prevented successful compilation of the Linux kernel. 

To estimate the potential overhead of SESES and \texttt{lvi-cfi}, we analyzed the output of the aborted compilation. Specifically, we counted the number of \texttt{lfence} instructions in the \texttt{vmlinux.o} file and compared it to the counts resulting from our own transformations. The results are shown in the following table:

\begin{table}[h]
  \centering
  \begin{tabular}[t]{cccc}
    \toprule
    \makecell{Simple Fencing\\ Transformation} & \makecell{Optimized fencing\\Transformation} & \makecell{Non-speculating\\  Transformation}& SESES + \texttt{lvi-cfi}\\
    \midrule
    $1083\cdot 10^3$ & $530\cdot 10^3$ & $1045\cdot 10^3$ & $1123 \cdot 10^3$\\
    \bottomrule
  \end{tabular}
    \caption{Count of \texttt{lfence} instructions for different transformations.}
\end{table}

While this analysis does not provide runtime overhead measurements, it
supports the hypothesis that the performance cost of SESES combined
with \texttt{lvi-cfi} would be similar to that of our Simple Fencing
Transformation. 

\end{document}